\documentclass[onecolumn,superscriptaddress,longbibliography,aps]{revtex4-2}

\usepackage{tikz}
\usepackage[margin=1in]{geometry}
\usepackage[utf8]{inputenc}
\usepackage[colorlinks=true, linkcolor=blue, citecolor=blue, urlcolor=blue]{hyperref}
\usepackage{amsmath,amsthm,amsfonts}
\usepackage{bbold}
\usepackage[normalem]{ulem}
\usepackage{hyperref}
\usepackage[ruled]{algorithm2e}
\usepackage{qcircuit}
\usepackage{thm-restate}
\usepackage{enumitem}
\usepackage{bm}

\usepackage{mathtools}
\DeclarePairedDelimiter\ceil{\lceil}{\rceil}
\DeclarePairedDelimiter\floor{\lfloor}{\rfloor}

\newtheorem{thm}{Theorem}
\newtheorem{lem}[thm]{Lemma}
\newtheorem{defn}[thm]{Definition}
\newtheorem{cor}[thm]{Corollary}

\newcommand{\norm}[1]{\left\lVert#1\right\rVert}
\newcommand{\lb}{\left(}
\newcommand{\rb}{\right)}
\newcommand{\ls}{\left[}
\newcommand{\rs}{\right]}
\newcommand{\lc}{\left\{}
\newcommand{\rc}{\right\}}
\newcommand{\bep}{\bm{\varepsilon_x}}

\newcommand{\br}{\textbf{r}}
\newcommand{\bp}{\textbf{p}}

\newcommand{\bx}{\textbf{x}}
\newcommand{\xt}{\widetilde{x}}
\newcommand{\Xt}{\widetilde{X}}
\newcommand{\bxt}{\widetilde{\bx}}
\newcommand{\bpt}{\widetilde{\bp}}

\newcommand{\Ht}{\widetilde{H}}
\newcommand{\Lt}{\widetilde{L}}
\newcommand{\rt}{\widetilde{r}}

\newcommand{\ket}[1]{|#1\rangle}
\newcommand{\bra}[1]{\langle#1|}
\newcommand{\braket}[2]{\langle #1 | #2 \rangle}
\newcommand{\ketbra}[2]{|#1\rangle\!\langle#2|}
\newcommand{\ev}[1]{\left\langle #1 \right\rangle}
\newcommand{\hamt}{\mathtt{HAM} \text{-} \mathtt{T}}
\newcommand{\qft}{\mathtt{QFT}}
\newcommand{\one}{\mathbb{1}}

\newcommand{\UofTP}{\affiliation{
Department of Physics, University of Toronto, Canada}}

\newcommand{\UofT}{\affiliation{
Department of Computer Science, University of Toronto, Canada}}

\newcommand{\PNNL}{\affiliation{Pacific Northwest National Laboratory, Richland WA, USA}}

\newcommand{\CIFAR}{\affiliation{Canadian Institute for Advanced Research, Toronto, Canada}}

\begin{document}

\title{Efficient Quantum Optimization via Dynamical Simulation}

\author{Ahmet Burak Catli}
\email{catli@cs.toronto.edu}
\UofTP

\author{Sophia Simon}
\email{sophia.simon@mail.utoronto.ca}
\UofTP

\author{Nathan Wiebe}
\email{nawiebe@cs.toronto.edu}
\UofT
\PNNL
\CIFAR

\begin{abstract}
    We provide several quantum algorithms for continuous optimization that do not require gradient estimation. Instead, we encode the optimization problem into the dynamics of a physical system and coherently simulate the time evolution. We focus on the setting where the objective function can \emph{only} be accessed via a phase oracle. 
    Our first two algorithms can find local optima of a differentiable function $f: \mathbb{R}^N \rightarrow \mathbb{R}$ by simulating either classical or quantum dynamics with friction via a time-dependent Hamiltonian. We show that for the benchmark problem of optimizing a locally quadratic objective function, these methods require a total of $O(N^2\kappa^2/h_x^2\epsilon)$ queries to a phase oracle to find an $\epsilon$-approximate local optimum, where $\kappa$ is the condition number of the Hessian matrix and $h_x$ is the discretization spacing. In contrast, we show that methods based on gradient descent require $O(N^{3/2}(1/\epsilon)^{\kappa \log(3)/4})$ queries. This corresponds to an exponential separation between the query upper bounds for the benchmark problem.
    Our third algorithm can find the global optimum of $f$ by preparing a classical low-temperature thermal state via simulation of the classical Liouvillian operator associated with the Nosé Hamiltonian. We use results from the quantum thermodynamics literature to bound the thermalization time for the discrete system. Additionally, we analyze barren plateau effects that commonly plague quantum optimization algorithms and observe that our approach is vastly less sensitive to this problem than standard gradient-based optimization.
    Our results suggests that these dynamical optimization approaches may be far more scalable for future quantum machine learning, optimization and variational experiments than was widely believed.
\end{abstract}

\maketitle

\newpage
\tableofcontents

\section{Introduction}

Optimization tasks have long been a target for quantum computers starting from early proposals such as the D\"urr-H\o yer optimization algorithm~\cite{durr1996quantum}, the quantum approximate optimization algorithm (QAOA) \cite{farhi2014quantum}, least squares fitting~\cite{wiebe2012quantum} and quantum algorithms for semi-definite programming~\cite{brandao2017quantum,van2017quantum}.  The most commonly used approach to solving continuous optimization problems on quantum computers involves the use of gradient descent to find a local optimum for a given objective function $f$. This approach has been widely used in QAOA as well as quantum machine learning and variational eigensolver results~\cite{farhi2014quantum,schuld2020circuit,peruzzo2014variational}. In these settings, the objective function is encoded in the amplitudes of a quantum state, meaning that obtaining an accurate bit string representation of even a single value of the objective function can be expensive. The standard approach then involves devising a quantum algorithm that estimates the gradient by sampling and then uses a classical computer to update the parameters of the function~\cite{peruzzo2014variational,farhi2014quantum,kieferova2017tomography,schuld2019evaluating}. Similar ideas are also explored in the context of simulating molecular dynamics on quantum computers, wherein the forces are computed under the Born-Oppenheimer approximation on a quantum computer and a classical computer is used to update the nuclear positions under Newton's equations of motion~\cite{o2022efficient}.

In cases where the objective function can only be accessed via the amplitudes or phases of a quantum state, such as in the examples given above, approaches based on gradient descent all possess a single bottleneck: the evaluation of the gradients on a quantum computer.  In particular, if we consider an optimization problem with $N$ parameters, the work of~\cite{gilyen2019optimizing,van2020convex} shows that the best performance that can be attainable for computing a gradient vector within error $\epsilon$ in the Euclidean norm, when having only amplitude or phase oracle access to the objective function, requires $O(\sqrt{N}/\epsilon)$ queries. This is highly problematic as it implies that even with the best possible gradient estimation procedure, achieving $9$ digits of accuracy will likely require billions of gate operations for even the simplest of optimization problems.  Further, for optimization problems with vanishing gradients~\cite{mcclean2018barren}, the cost of evaluating small gradients can be truly catastrophic as it can lead to exponential costs in navigating the optimization landscape. 

Alternatives to gradient based approaches have been considered in the past.  One popular approach is the Quantum Hamiltonian Descent (QHD) approach of~\cite{leng2023qhd, leng2023nonconvex}, which examines simulating dissipative dynamics to find a global optimum. Ref.~\cite{leng2023nonconvex} in particular provides evidence of quantum advantage for non-convex optimization problems.
Other approaches such as~\cite{chen2024langevin} use quantum Langevin dynamics for optimization.  The common thread between these approaches is that they eschew direct gradient evaluation and instead push the burden to simulating quantum dynamics. More specifically, references~\cite{leng2023qhd, leng2023nonconvex} introduce a time-dependent Hamiltonian to simulate damping and adiabatic evolution in order to prepare the global minimum of a given function while~\cite{chen2024langevin} simulates non-unitary dynamics directly.

Our work follows in a similar spirit to the above gradient-free strategies but focuses on different aspects. 
First, we mainly consider the setting where the objective function can \emph{only} be accessed via a phase oracle, while prior work focuses on the case where one is given bit oracle access to the objective function. For QHD~\cite{leng2023qhd, leng2023nonconvex}, the bit oracle is ultimately converted to a phase oracle but the key point to notice here is that allowing bit oracle access to the objective function makes strategies based on standard gradient descent much more viable. Indeed, our algorithms do not yield any asymptotic advantage over gradient descent for our local optimization benchmark problem if one has bit oracle access to the objective function.
Another difference compared to Refs.~\cite{leng2023qhd, leng2023nonconvex} is that we utilize a time-dependent Hamiltonian to find local rather than global optima which requires a different analysis. In particular, the spectral gap of the time-dependent Hamiltonian is not directly relevant in our setting.
Furthermore, in contrast to prior work, we tackle the problem of finding the global optimum of an arbitrary twice differentiable function by performing unitary quantum dynamics in a higher dimensional Hilbert space that will solve the optimization problem in the reduced space and is inspired in part by previous work that utilized this approach to address issues faced in gradient evaluation for simulating chemical dynamics~\cite{Simon2024Liouvillian}.

For both local and global optimization, we show that our algorithms have provable convergence guarantees under appropriate assumptions on the optimization landscape and further offer substantial computational advantages relative to existing gradient-based methods in scenarios where the objective function can only be access via the amplitudes of a quantum state. In particular, our time-dependent Hamiltonian approach allows us to prove an upper bound on the phase oracle query complexity for the benchmark problem of finding the optimum of a convex quadratic function that has exponentially better scaling w.r.t.~the condition number of the Hessian matrix than the corresponding upper bounds for standard gradient descent methods. We also provide evidence that this exponential separation between the upper bounds persists for more general smooth and strongly convex functions.

It should also be considered that specific hardware platforms might have different costs for different tasks and a method that does not require explicit estimates of gradients might provide advantages not captured by the asymptotic scalings. For example, for a neutral atom quantum computer a measurement might take three orders of magnitude longer than a gate operation~\cite{neutralatoms}. In that case, a fully coherent quantum algorithm may outperform a quantum algorithm which requires repeated measurements to estimate gradients.

The remainder of this paper is laid out as follows.  In Section~\ref{sec:cohOpt} we discuss the general setting of our optimization problems and the oracles that we assume for computing the objective function for both the setting where bit and phase oracles are considered.  
Our main results are summarized in Section~\ref{sec:main} and the assumptions that are needed in order to verify that the preconditions are met. Next, we discuss the local optimization algorithms which are based on simulating either quantum or classical dynamics in the presence of friction in Section~\ref{sec:local}. Then, in Section~\ref{sec:convex_quadratic}, we provide upper bounds on the number of queries needed by the local quantum optimization algorithms to find the optimal value of a convex quadratic function within error $\epsilon$. We also discuss extensions to generic smooth and strongly convex functions. Note here that while our analysis focuses on the case where the function is strongly convex, any sufficiently smooth optimization problem can be closely approximated by a strongly convex function for initializations that are sufficiently close to a local optimum. Thus, we choose to think of these results as pertaining to local optimization.
Section~\ref{sec:global} contains the discussion of our global approach, which uses ideas from quantum thermodynamics to assess the complexity of preparing the global optimum of an objective function by preparing an approximation to a near-zero temperature thermal distribution over the parameters of the model and shows efficiency of the method under specific assumptions made about the gap of the Liouvillian. Next, we compare our coherent quantum algorithms to gradient-based methods in Section~\ref{sec:barren}. In particular, we compare the query complexity of our local quantum optimization algorithm to the query complexity of standard gradient descent methods for ill-conditioned optimization problems. Further, we discuss our global optimization algorithm in the context of barren plateaus and vanishing gradients in variational models before concluding in Section~\ref{sec:conclusion}.

\section{Coherent Optimization}
\label{sec:cohOpt}

The central problem that we address here is that of optimizing an objective function on a quantum computer without the need to compute gradients.  We specifically investigate two strategies for achieving this optimization.  The first approach is based on simulations of dissipative dynamics in a coherent setting via time-dependent Hamiltonian simulation.  This approach drives the system into a local optimum. The second approach is a global optimization strategy wherein the simulated dynamics drives the system into a classical low-temperature thermal state.  This second approach will, upon success, find a state that is close to the global optimum rather than a local optimum but requires qualitatively different assumptions in order to reach such a state.  For this reason, we consider both strategies.

There are of course several approaches to quantum optimization that could be considered and have been considered in the past. Here we aim to perform a strong form of optimization wherein we perform a mapping that will, with high probability, transform a set of initial quantum states into a new set of quantum states that are in the space of optimal solutions up to some discretization. As we require the algorithm to output the answers as a bit string, the measurement process for this is shockingly simple: we simply measure our input register in the computational basis. This process cannot, however, be unitary for arbitrary input states because any such optimization process must, in cases with a single global optimum, map multiple input parameters to the same point. This prevents the function from maintaining inner products between inputs and thus the overall process cannot be unitary in general. We address this by restricting the set of input states in the local optimization approach and by performing a subsystem trace in the global optimization approach. 
Below we formalize the problem of finding a local or global optimum of a given objective function using as few function evaluations as possible.

\begin{defn}[Continuous Quantum Optimization Problem]
    Let $\bx \in \ls -x_{\max}, x_{\max} \rs^N \subseteq \mathbb{R}^N$ be a real valued vector and let ${\bf z} \in \mathbb{Z}_{2^n}^N$ be a corresponding $n$-bit encoding of these values such that $\bx={\bf c} + {\bf z} h_x$ for some constant vector ${\bf c}$ and grid spacing $h_x > 0$. Further, let $f: \ls -x_{\max}, x_{\max} \rs^N \rightarrow [0,f_{\max})$ be a twice differentiable objective function and let $S_{\mathrm{opt}}$ be a fixed set of optimal points of $f$ such that for any $\bx_S \in S_{\mathrm{opt}}$ it holds that $\nabla f (\bx_S) = 0$.
    The problem then is to find a $\mathbf{z}^*$ such that $|f({\bf c}+{\bf z}^* h_x) - f(\bx_S)| \leq \epsilon$ with probability at least $1-\delta$ for some $\bx_S \in S_{\mathrm{opt}}$ using a minimum number of queries to a set of oracles that compute $f$.
\end{defn}

In the following, for notational simplicity, we will often refer to quantum states using real valued arguments rather than the discrete values used in their encoding.  Specifically, for $\bx \in \mathbb{R}^N$ such that $\bx = {\bf c}+ [z_0 h_x,z_1 h_x,\ldots,z_{N-1} h_x]$, we define
\begin{equation}
    \ket{\bx} := \ket{[x_0 -c_0]/h_x}\cdots \ket{[x_{N-1}-c_{N-1}]/h_x}.
\end{equation}
Note that in this discrete setting the set of optimal points does not precisely correspond to the set of vectors such that $\nabla f(\bx) = 0$; however, for twice differentiable functions  it can easily be seen that the optima defined above will coincide with such points in the limit as $h_x \rightarrow 0$.

The aim of the Continuous Quantum Optimization Problem is to minimize the number of queries made to oracles that compute the value of $f$ but the definition of the oracles used in the problem is left purposely vague in the above definition.  This is because there are a host of different oracle settings that could be considered that substantially change the query complexity.  We consider three families of settings here: a bit oracle and two types of phase oracles.  We define these oracles as follows:

\begin{defn}[Bit oracle for the objective function]
    We say that $O_f^{(b)}$ is a bit oracle for a function $f: \ls -x_{\max}, x_{\max} \rs^N \rightarrow~[-f_{\max}, f_{\max}]$ if for any computational basis state $\ket{\bx} \in \mathbb{Z}_{2^n}^N$,
    \begin{equation}
        O_f^{(b)} \ket{\bx}\ket{y}_b = \ket{\bx}\ket{y \oplus \widetilde{f}(\bx)}_b,
    \end{equation}
    where $\widetilde{f}$ is a $b$-bit approximation of $f$ such that $\max_{\bx}|f(\bx) - \tilde{f}(\bx) | \le \max_{\bx}  |f(\bx)| /2^{b-1}$.
\label{def:bit_oracle}
\end{defn}

\begin{defn}[Phase oracle for the objective function]
    We say that $O_f^{(p)}$ is a phase oracle for a function $f: \ls -x_{\max}, x_{\max} \rs^N \rightarrow [-f_{\max}, f_{\max}]$ if for any computational basis state $\ket{\bx} \in \mathbb{Z}_{2^n}^N$,
    \begin{equation}
         O_f^{(p)} \ket{\bx} = e^{i \frac{f(\bx)}{2f_{\max}} }\ket{\bx}.
    \end{equation}
\label{def:phase_oracle}
\end{defn}

\begin{defn}[Phase oracles for the partial derivatives of the objective function]
    Let $f: \ls -x_{\max}, x_{\max} \rs^N \rightarrow \mathbb{R}$ and let it be promised that $\max_\bx \left| \frac{\partial f}{\partial x_k} \right| \leq f'_{\max}$ for all $k \in [N]$.
    We say that $O_{f',k}^{(p)}$ is a phase oracle for the $k$-th partial derivative of $f$ if for any computational basis state $\ket{\bx} \in \mathbb{Z}_{2^n}^N$,
    \begin{equation}
         O_{f',k}^{(p)} \ket{\bx} = e^{\frac{i}{2f'_{\max}} \frac{\partial f(\bx)}{\partial x_k}} \ket{\bx}.
    \end{equation}
\label{def:phase_oracle_derivatives}
\end{defn}

In principle, the phase oracles $O_{f',k}^{(p)}$ for the partial derivatives of $f$ can be constructed approximately via a finite difference scheme by using the phase oracle $O_f^{(p)}$ in Definition~\ref{def:phase_oracle}, see e.g.~\cite{Simon2024Liouvillian} for more details. For simplicity, however, we will use the $O_{f',k}^{(p)}$ oracles in cases where we require access to the partial derivatives of $f$.

Note that a single query to a bit oracle can implement a query to a phase oracle, but the converse is not true unless further assumptions are made about the range of $f(\bx)$~\cite{gilyen2019optimizing}. This means that a bit oracle is more powerful than a phase oracle as seen in the difference in query complexity discussed in~\cite{gilyen2019optimizing}. 
For applications in quantum machine learning, a third type of oracle known as a probability oracle is needed.  This oracle acts as follows:
\begin{equation}
    O_f \ket{\bx}\ket{0} = \ket{\bx} \left(\sqrt{\frac{f(\bx)}{f_{\max}}}\ket{0}\ket{\chi(\bx)} + \sqrt{1-\frac{f(\bx)}{f_{\max}}}\ket{1}\ket{\phi(\bx)}  \right),
\end{equation}
for arbitrary states $\ket{\chi(\bx)}$ and $\ket{\phi(\bx)}$.
This form of an oracle is used in quantum rejection sampling and in the Harrow Hassidim and Lloyd algorithm~\cite{Harrow2009linearsystems}.  Further, in cases where an expectation value of a quantum state is desired that is parameterized by $\bx$, the Hadamard-test circuit can be used to estimate the expectation value of a unitary $V$ against a particular state via the LCU Lemma~\cite{childs2012hamiltonian}:

\begin{align}
    &\ket{\bx}(H \otimes I)(\ketbra{0}{0}\otimes I + \ketbra{1}{1}\otimes V )(H \otimes I)\ket{0}\ket{\psi(\bx)} = \ket{\bx}(\ket{0}\frac{I+V}{2}\ket{\psi(\bx)} + \ket{1}\frac{I-V}{2}\ket{\psi(\bx)})\nonumber\\
    &\qquad\qquad\qquad:=\ket{\bx} \left(\sqrt{\frac{1+ {\rm Re}(\bra{\psi(\bx)} V \ket{\psi(\bx)})}{2}}\ket{0}\ket{\chi(\bx)} + \sqrt{\frac{1- {\rm Re}(\bra{\psi(\bx)} V \ket{\psi(\bx)})}{2}}\ket{1}\ket{\phi(\bx)}  \right).
\end{align}
This circuit can clearly be viewed as a probability oracle such that $f(x)$ encodes the real part of the expectation value of a unitary.  This is a fundamental task in quantum machine learning and variational algorithms as the training objective functions or estimates of the groundstate energy can be written as linear combinations of expectation values of unitaries.

We do not explicitly consider probability oracles in the following despite their obvious relevance within the field.  This is because an amplitude amplification unitary can be used to convert such an oracle to a phase oracle using poly-logarithmic overhead~\cite{gilyen2019optimizing} and our underlying algorithms directly use phase oracles.

Note that when comparing different algorithms, we assume that we only have access to a single type of oracle. In particular, in the phase oracle setting, the assumption is that the objective function can \emph{only} be accessed via a phase oracle.

\subsection{Example of Bit and Phase Oracles}
As an example of a phase oracle, we can consider the case in quantum machine learning of nearest neighbor classification~\cite{wiebe2015quantum}.  In this setting, we have a series of training vectors $\ket{\psi_j}$ and we wish to compute the inner products to find, for fixed $k$, ${\rm argmin}_j(|\braket{\psi_j}{\psi_k}|)$.  Using the swap test, a probability oracle $U_P$ can be constructed such that
\begin{equation}
    U_P\ket{j}\ket{k}\ket{0} = \sqrt{\frac{1}{2} + \frac{|\braket{\psi_j}{\psi_k}|}{2} } \ket{j}\ket{k}\ket{\phi_{\rm sym}}\ket{0} + \sqrt{\frac{1}{2} - \frac{|\braket{\psi_j}{\psi_k}|}{2} } \ket{j}\ket{k}\ket{\phi_{\rm asym}}\ket{1}
\end{equation}
where $\ket{\phi_{\rm sym}}$ is a state in the symmetric subspace of $\ket{\psi_k}$ and $\ket{\psi_j}$ and $\ket{\phi_{\rm asym}}$ is a vector in the anti-symmetric space.  Constructing the amplitude amplification unitary that marks $\ket{0}$ in the above state provides a unitary that has eigenvalues $\exp \lb \pm i \arcsin \lb\sqrt{\frac{1}{2}+\frac{|\braket{\psi_j}{\psi_k}|^2}{2}} \rb \rb$.  We can then use an external control on an ancillary qubit, along with qubitization~\cite{gilyen2019optimizing}, to generate an $\epsilon$-approximation to a rotation of the form (for ancillary qubit in computational basis state $\ket{c}$)
\begin{equation}
    \ket{c}\ket{j}\ket{k}\ket{0} \mapsto e^{ic\left( \frac{1}{2} + \frac{|\braket{\psi_j}{\psi_k}|^2}{2}\right)}\ket{c}\ket{j}\ket{k}\ket{0}
\end{equation}
which is clearly a phase oracle.  This approach allows a probability oracle to be converted into a phase oracle at low cost without requiring amplitude estimation and illustrates the core of the proposal in~\cite{gilyen2019optimizing}.

On the other hand, if we assume that there exists $q\in \mathbb{Z}_{2^m}$ such that $ (1+|\braket{\psi_j}{\psi_k}|^2) = q\pi/2^m$ for positive integer $m$, then we can learn $q$, which is a bit encoding of the phase, using quantum phase estimation which requires $2^m$ repetitions of the circuit.  This means that we can construct a bit oracle from the phase oracle at cost exponential in the number of bits.  This cost is further optimal because any improvements to the conversion would violate quantum lower bounds for phase estimation.

If we have a bit oracle representation of the probability, then we would have an oracle $U_b$ such that
\begin{equation}
    U_b \ket{j}\ket{k}\ket{0} \mapsto \ket{j}\ket{k}\Biggr|{\left[\frac{1}{2}  + \frac{|\braket{\psi_j}{\psi_k}|^2}{2}\right]}\Biggr\rangle,
\end{equation}
where $[\cdot]$ denotes an $m$-bit binary encoding.  Under these circumstances, note that if we apply an adder then
\begin{equation}
    {\rm ADD} \left(\Biggr|{\left[\frac{1}{2}  + \frac{|\braket{\psi_j}{\psi_k}|^2}{2}\right]}\Biggr\rangle {\rm QFT} \ket{1}\right) = e^{i2\pi \left[\frac{1}{2}  + \frac{|\braket{\psi_j}{\psi_k}|^2}{2}\right]/2^m } \left(\Biggr|{\left[\frac{1}{2}  + \frac{|\braket{\psi_j}{\psi_k}|^2}{2}\right]}\Biggr\rangle {\rm QFT} \ket{1}\right),
\end{equation}
which converts the result into a phase oracle at cost $O(m)$~\cite{gidney2018halving}.  This shows that bit oracles, if available, are a powerful resource but can be expensive to construct from a native phase or probability oracle.
For these reasons, we discuss both types of oracles here but focus our discussion on phase oracles owing to their ubiquity in algorithms like nearest neighbor classification.

\section{Main Results}
\label{sec:main}

In this section, we give a brief summary of the main results of our paper.  Specifically, we provide upper bounds on the query complexity for the following two tasks: (a) finding the optimum of a convex quadratic function via a local approach based on time-dependent Hamiltonian simulation of damping and (b) finding the global optimum of a general differentiable function via a global approach based on preparing a low-temperature Gibbs distribution over the optimization parameters. Unless stated otherwise, we use $\norm{\cdot}$ to refer to the 2-norm of a vector or the induced 2-norm (spectral norm) of a matrix, depending on the context.

The following informal theorem provides a bound on the complexity of the local approach for the task of optimizing a convex quadratic function. A more precise statement and a tighter bound are given in Theorem~\ref{thm:main}. We also provide a generalization that yields analogous results under the weaker promise of strong convexity only. Note that we choose convex quadratic optimization as a benchmark problem because it allows us to prove relatively tight convergence bounds. However, as discussed in Section~\ref{sec:local}, our local optimization algorithms can also be used to find local optima of non-convex functions.

\begin{thm}[Coherent convex quadratic optimization; informal version of Theorem~\ref{thm:main}]
    Let $\epsilon > 0$ be an error tolerance, let $A \in \mathbb{R}^{N \times N}$ be positive with eigenvalues $0 < \lambda_0 =: \lambda_{\min} \leq \lambda_1 \leq \dots \leq \lambda_{N-2} \leq \lambda_{N-1} =: \lambda_{\max}$ and let $f(\bx) = \lb \bx- \bx^* \rb^\top A \lb \bx- \bx^* \rb$.
    Assume having access to a quantum state $\ket{\widetilde{\psi}_0 \lb \bx \rb} \in \mathbb{C}^{2^{Nn}}$ with sufficiently smooth amplitudes and let the corresponding vector of position expectation values be given by $\bra{\widetilde{\psi}_0} \hat{\bx} \ket{\widetilde{\psi}_0} =: \bx_0 \in \mathbb{R}^N$.
    Further, assume that $f$ can only be accessed via a phase oracle $O_f^{(p)}$ as given in Definition~\ref{def:phase_oracle}.
    Then there exists a quantum algorithm that can solve the Continuous Quantum Optimization Problem by finding an $\bx'$ such that $\left| f\lb \bx' \rb - f\lb \bx^* \rb \right| \leq \epsilon$ with probability at least $1 - \delta$ using 
    \begin{equation}
        \widetilde{O} \lb \frac{N^2}{h_x^2 \epsilon} \frac{\lambda_{\max}^2}{\lambda_{\min}} \norm{\bx_0 - \bx^*}^2 x_{\max}^2  \log \lb 1/\delta \rb \rb
    \end{equation}
    queries to controlled-$O_f^{(p)}$ and its inverse.
    If instead $f$ can be accessed via a bit oracle $O_f^{(b)}$ as given in Definition~\ref{def:bit_oracle}, then the same problem can be solved using
    \begin{equation}
        \widetilde{O} \lb \frac{N}{h_x^2 \sqrt{\lambda_{\min}}} \log^2 \lb\frac{\lambda_{\max} \norm{\bx_0 - \bx^*} x_{\max}}{\epsilon} \rb \log \lb 1/\delta \rb \rb
    \end{equation}
    queries to $O_f^{(b)}$.
\label{thm:main_informal}
\end{thm}

In the case of bit oracle access to the objective function, our query upper bound in the above theorem seems to depend only logarithmically on $\lambda_{\max}$. However, in order to ensure that there exists a grid point sufficiently close to the true minimizer such that the error in the function value is at most $\epsilon$, we require the grid spacing $h_x$ to be at most $\sqrt{\epsilon/\lambda_{\max}}$. 
In fact, in Section~\ref{sec:convex_quadratic} we provide evidence that $h_x \sim \sqrt{\lambda_{\min} \epsilon}/\lambda_{\max}$.
This means that we generally do not expect any asymptotic advantage of our dynamical optimization algorithms over gradient descent based methods in this setting since, for example, conjugate gradient descent requires only $O \lb \sqrt{\kappa} \log (1/\epsilon) \rb$ queries to a bit oracle to find the optimum of a convex quadratic function~\cite{NesterovLectures}.
However, one might be able to utilize interpolation schemes such that in certain cases a larger grid spacing with milder dependence on $\epsilon$ and $\lambda_{\max}$ could be used. We leave a detailed analysis of this idea for future work.

The assumption that the information is passed to the system through a bit oracle is of course a strong one. Amplitude oracles, and in turn phase oracles, are much more prevalent in the case of QML or VQE applications.  In the event that 
the objective function can only be accessed via a phase oracle, we find that quantum algorithms that simulate frictional dynamics can lead to exponentially better scaling with the condition number than vanilla gradient descent which we show might scale exponentially with the condition number in the worst case. This makes the upper bounds yielded by our method for local optimization to be, to our knowledge, the best asymptotic scaling available for quantum algorithms that involve variational optimization such as VQE, QML and QAOA. 

Our global approach uses results from the classical open systems literature to simulate thermalization of a system of classical particles. Each optimization parameter gets mapped to a single classical particle. The objective function in this case acts as a potential for the classical particles.  If a low-temperature thermal state of the classical system can be prepared, then we can use this approach to find a global rather than a local optimum for our optimization problem. We consider the following classical Nosé Hamiltonian for an extended system consisting of the original system and a heat bath~\cite{Nose1984,Nose1984partition_function,Simon2024Liouvillian}:
\begin{equation}
    H_{N} = \frac{\mathbf{p}^2}{2m s^2} + V(\mathbf{x}) + \frac{p_s^2}{2Q} + g\beta^{-1} \ln(s),
\end{equation}
where $\mathbf{p}$ is the vector of momentum variables of the classical particles, $\mathbf{x}$ is the vector of position variables of the classical particles, $s\in (0,\infty]$ is a bath degree of freedom, $p_s$ is its canonical momentum, $\beta$ is the inverse temperature of the final distribution and $g$ and $Q$ are free parameters that can be chosen to impact the rate with which the reduced phase space distribution over $(\bp,\bx)$ approaches a thermal state with temperature $1/\beta$.  This approach can be thought of as an open systems analogue of~\cite{leng2023qhd} with the modification that the dynamics are classical rather than quantum. The dynamics of the phase space probability density are governed by a classical Liouvillian operator, $L_N$, which is essentially the analogue of a quantum Hamiltonian for classical dynamics. Using the Koopman-von Neumann formalism, the evolution under the Liouvillian can be written as unitary dynamics acting on a quantum state, i.e.~it can be transformed into a Schrödinger equation.
Once recast as a Schr\"{o}dinger equation, we use Hamiltonian simulation to simulate the dynamics. The value of the objective function is returned from a quantum phase oracle, which renders the resulting algorithm not classically simulatable, and in turn allows a weaker oracle to be used for optimization than previously considered.  

A challenge of this approach, akin to~\cite{leng2023qhd}, is that we need to evolve long enough to ensure that the distribution has equilibrated.  This time scales inversely with the spectral gap of the Liouvillian and can cause the algorithm to require exponential run time.  However, as efficient optimization would result in ${\rm NP}\subseteq {\rm BQP}$, it is highly unlikely that efficient algorithms for global optimization are possible without utilizing structure of the particular problem instance. We provide estimates of the equilibration time needed for the reduced system to approach a thermal distribution using ideas from quantum thermodynamics~\cite{Linden2009thermal}, suggesting that new algorithmic applications may be gleaned from foundational work in thermodynamics.

Below we provide an informal statement of our theorem for the number of oracle calls and Toffoli gates needed to sample from a low-temperature thermal state using a discretization of the Liouvillian for the Nos\'e process.

\begin{thm}[Global Optimization Theorem; informal version of Theorem~\ref{thm:mainGlobal}]
\label{thm:mainGlobalInformal}
    Let $\delta>0$ be an error tolerance and consider the problem of drawing a sample from the marginal distribution over $\bx$ from a distribution that is within total variational distance $\delta$ from the dilation of the thermal state $e^{-\beta(\bp^2/2m +V(\bx))}/\mathcal{Z}$. 
    There exists a quantum algorithm for solving this problem that uses a number of oracle queries and a number of Toffoli gates that scale as
    \begin{equation}
        N_{\rm query,MC}\in \tilde{O}\left( \frac{\alpha'' N}{\gamma \delta^2}\right),\qquad N_{\rm Toff, MC} \in \tilde{O}\left(\frac{
        \alpha''\left((N+d)\log(g) \right)}{\gamma \delta^2} \right),
    \end{equation}
    respectively, where $\gamma$ is the ($\beta$-dependent) spectral gap of the discretized Liouvillian $L_N$ and $\alpha''$ is a block encoding constant for the Liouvillian that depends on the discretization scale  as well as the volume of support of the microcanonical phase space volume which is assumed to be compact.
\end{thm}

We see from these results that we can find the global optimum as well as local optima by changing the dynamics to allow us to explore multiple local optima.  Specifically, under the assumption that the discretization error caused by choosing a discretization scale of $h_{\min}$ is sufficient and the spectral gap of the Liouvillian (analogous to a spectral gap of a Markov chain) is sufficiently large then we can in fact find a close approximation to a globally optimal point by leveraging the fact that the continuum Nos\'e Liouvillian equilibrates to the global optimum.  This result is further particularly interesting because it relies on a number of results from quantum thermodynamics which are needed to estimate the timescale required for the evolution to reach microcanonical equilibrium.

In Table~\ref{tab:compare} we consider the setting of having only phase oracle access to the objective function and compare the query complexities of our local and global optimization algorithms to hybrid quantum-classical approaches which only use the quantum computer to compute gradients.
Importantly, we compare general purpose quantum optimization algorithms with each other, meaning that none of these algorithms rely on any particular structure of the objective function. For the comparison of the local approaches, we choose convex quadratic optimization as a benchmark problem since it allows us to provide concrete convergence and query bounds. 
In particular, the query upper bounds that we are able to prove for this benchmark problem are exponentially better in the condition number, $\lambda_{\max}/\lambda_{\min}$, for our fully quantum approach than the hybrid approach.

For global optimization problems, the query upper bound for our fully quantum approach has exponentially better scaling with the inverse spectral gap $\gamma^{-1}$ and the inverse error tolerance $\delta^{-1}$ than the hybrid approach. Of course, global optimization is a hard problem in general, so any algorithm is expected to require a number of queries that scales super-polynomially with $N$, the number of optimization parameters, in the worst case. In particular, $\gamma$ can be exponentially small in $N$ in which case the query upper bound for our fully quantum algorithm scales exponentially in $N$. However, the query upper bound for the hybrid approach in that case is doubly exponential in $N$, so the upper bound for the fully quantum algorithm is still exponentially better than the upper bound for the hybrid approach.

As these are all upper bounds, we cannot say definitively that the fully quantum approaches are always asymptotically superior. Lower bounds would be needed to prove an asymptotic advantage of fully quantum approaches over hybrid approaches.  We further show that the quadratic assumption in the local case can be relaxed to allow strongly convex functions; however, the bounds are often worse and in a sufficiently small neighborhood about the optimal solution, for any strongly convex function, the convergence guarantees for the quadratic case will still apply.  For this reason, we focus on the quadratic case while emphasizing that further work is needed to provide tight bounds for more general functions.

\begin{table}[t]
    \renewcommand{\arraystretch}{2}
    \centering
    \begin{tabular}{|c|c|c|}
    \hline
         &Hybrid quantum classical with phase oracle & Fully quantum with phase oracle \\
         \hline
         Local optimization & $\widetilde{O} \lb N^{3/2}\left(\frac{2\lambda_{\max}\|\bx_0 - \bx^*\|^2}{\epsilon}\right)^{(\lambda_{\max}/\lambda_{\min}) \log(3)/4} \rb$ (Thm.~\ref{thm:quantumGradDesAlg}) & {$\widetilde{O}\left(\frac{N^2 \lambda_{\max}^2}{\epsilon \lambda_{\min}^2h_{\min}^2}\right)$} (Thm.~\ref{thm:main}) \\
         Global optimization & $\widetilde{O}\left(\frac{N^{3/2}e^{2\sqrt{N}/\gamma\delta^2}}{\gamma^3 \delta^6}\right)$ (Thm.~\ref{thm:hybridGlobal})& $\widetilde{O}\left( \frac{N^2}{\gamma \delta^2 h_{\min}^2}\right)$ (Thm.~\ref{thm:mainGlobal})\\
         \hline
    \end{tabular}
    \caption{Comparison of upper bounds on the number of oracle calls for local and global optimization problems using either hybrid quantum-classical optimization algorithms or our fully quantum optimization algorithms. 
    We assume that the objective function can only be accessed via a phase oracle (see Defs.~\ref{def:phase_oracle} and \ref{def:phase_oracle_derivatives}). For local optimization, we consider the benchmark problem of finding the optimum of a quadratic convex function within error $\epsilon$ since this setting allows us to prove concrete convergence and query bounds. For global optimization, we consider the problem of drawing a sample from a distribution within total variational distance~$\delta$ from a discrete approximation to the global optimum of a general differentiable function. In both cases, we consider $N$ optimization parameters and take the remaining parameters to be constant for simplicity except $\lambda_{\max},\lambda_{\min}$ (the largest/smallest eigenvalues of the Hessian matrix), the minimum grid spacing $h_{\min}$ (which is expected to scale at most polynomially with $\lambda_{\max}^{-1},\lambda_{\min}$) and $\gamma$ (the spectral gap of the discrete Nos\'e Liouvillian at near-zero temperature). We only include factors of the initial distance from the optimum, $\|\bx_0 - \bx^*\|$, if the associated exponent is not clearly constant.}
\label{tab:compare}
\end{table}

\section{Quantum Algorithms for Coherent Local Optimization}
\label{sec:local}

Consider a twice continuously differentiable nonnegative function $f: \mathbb{R}^N \rightarrow \mathbb{R}_{\geq 0}$ which satisfies
\begin{equation}
    f(\bx) \rightarrow + \infty \quad \text{as} \quad \norm{\bx} \rightarrow +\infty \quad \text{for all} \quad \bx \in \mathbb{R}^N.
\label{local_opt_assumption}
\end{equation}
The goal now is to find a local minimum of $f$. That is to say, we aim to find an $\bx^* \in \mathbb{R}^N$ that satisfies $(\nabla f)(\bx^*) = 0$ and $H_f (\bx^*) \succ 0$ where the latter condition states that the Hessian of $f$ evaluated at $\bx^*$ is positive definite.

A standard tool for attacking such a minimization problem is gradient descent which iteratively updates the proposed solution by evaluating the gradient of $f$ at the current point.
In this paper, we present alternative approaches inspired by classical dynamics with friction. While our main algorithm deals with the simulation of quantum-mechanical systems with friction, it will be helpful to consider the classical analogue first in order to develop some intuition for dynamics with friction. The basic idea is to treat $\bx$ as the position vector of a single classical particle moving in $N$ dimensions in a potential given by $f$. Crucially, the particle is subject to a velocity-dependent friction force, $\mathbf{F}_{\rm fric} := - \beta \dot{\bx}$, where $\beta \geq 0$ determines the strength of the applied friction and $\dot{\bx} := \frac{d \bx}{dt}$.
The equations of motion for the classical particle can then be written in the following vector form:
\begin{equation}
    m \ddot{\bx} = -\nabla f(\bx) - \beta \dot{\bx}. 
\label{eom_friction_general}
\end{equation}
Intuitively, due to the friction, the particle will slow down over time, eventually settling into a local minimum of $f$. More precisely, assume that $f$ is analytic and let $S_{\mathrm{stationary}} := \{\bx \in \mathbb{R}^N | \nabla f(\bx) = 0 \}$ denote the set of stationary points of $f$. Then it can be shown that the solution to Eq.~\eqref{eom_friction_general} converges to a stationary point of $f(\bx)$ in the limit as $t \rightarrow \infty$, see~\cite{Haraux1998diffeq_convergence} for a proof.
Thus, if we were able to efficiently simulate the dynamics associated with Eq.~\eqref{eom_friction_general} on a quantum computer, we could efficiently find a local minimum/stationary point of $f$ simply by measuring the position vector of the particle after sufficiently long time. The question is just how to map such a classical dissipative system onto a quantum computer. Based on the Liouvillian formalism of classical mechanics, Ref.~\cite{Simon2024Liouvillian} shows how to efficiently simulate the dynamics of a classical system assuming that the system can be described by a classical Hamiltonian. As it turns out, there does exists a time-dependent Hamiltonian whose associated equations of motion are given by Eq.~\eqref{eom_friction_general}. 

\begin{defn}[Hamiltonian with friction]
    Let $\beta, m, t > 0$ and let $f: \mathbb{R}^N \rightarrow \mathbb{R}$ be differentiable. Then we define the following time-dependent friction Hamiltonian:
    \begin{equation}
        H_{\mathrm{fric}}(t) := e^{- \beta t/m} \sum_{j=1}^N \frac{p_j^2}{2m} + e^{\beta t/m} f(\bx),
    \end{equation}
    where $p_j$ can be interpreted either as the classical momentum conjugate to $x_j$ or as a quantum operator satisfying the canonical commutation relations such that $[x_j, p_k]  = i \delta_{jk}$. 
\label{def:friction_ham}
\end{defn}

For our purposes, we will often set $m=1$ for simplicity, especially in query complexity bounds.

\begin{lem}[Equations of motion for the friction Hamiltonian]
    The equations of motion associated with the classical friction Hamiltonian $H_{\mathrm{fric}}(t)$ are given by
    \begin{equation}
        m \ddot{\bx} = -\nabla f(\bx) - \beta \dot{\bx}.
    \end{equation}
\end{lem}

\begin{proof}
    Hamilton's equations of motion are as follows:
    \begin{equation}
    \begin{split}
        \dot{x}_j &= \frac{\partial H}{ \partial p_j} = e^{-\beta t/m}\frac{p_j}{m} \\
        \dot{p}_j &= -\frac{\partial H}{ \partial x_j} = - e^{\beta t/m} \frac{\partial f}{\partial x_j}.
    \end{split}
    \label{eom2}
    \end{equation}
    This implies that
    \begin{equation}
        m\ddot{x}_j = e^{-\beta t/m} \dot{p}_j - \frac{\beta}{m} e^{-\beta t/m} p_j = -\frac{\partial f}{\partial x_j} - \beta \dot{x}_j,
    \end{equation}
    which is the coordinate version of Eq.~\eqref{eom_friction_general}.
\end{proof}

Next, let us briefly explain the Liouvillian formulation of classical mechanics which allows us to easily simulate the dynamics of classical systems on a quantum computer. The Liouvillian formalism is centered around the phase space probability density $\rho (\bx, \bp, t)$ of the classical system which obeys the following equation of motion:
\begin{equation}
    \frac{\partial \rho}{\partial t} = -iL\rho,
\label{liouville_eq}
\end{equation}
where $L$ is the Liouvillian operator given by
\begin{equation}
    L := -i \sum_{j=1}^N \lb \frac{\partial H}{\partial p_j} \frac{\partial}{\partial x_j} - \frac{\partial H}{\partial x_j} \frac{\partial}{\partial p_j} \rb,
\label{liouvillian}
\end{equation}
with $H$ being the system's Hamiltonian. The key point is that $L$ is Hermitian, meaning that the time evolution operator generated by $L$ is unitary. This allows us to use existing Hamiltonian simulation methods to efficiently simulate the dynamics of $\rho$ on a quantum computer by effectively treating $L$ as just another type of Hamiltonian. For technical reasons, we actually end up simulating the dynamics of the Koopman-von Neumann wave function $\psi_{KvN} (\bx, \bp)$ associated with $\rho$ rather than the dynamics of $\rho$ itself. The Koopman-von Neumann wave function can be thought of as the square root of the probability density $\rho$ and, for the purpose of simulating dynamics, it can be treated like any other quantum wave function. Importantly, $\psi_{KvN}$ obeys the same equation of motion as $\rho$ shown in Eq.~\eqref{liouville_eq}, see Ref.~\cite{Simon2024Liouvillian} for more details.

Note that the classical Liouvillian operator associated with the friction Hamiltonian $H_{\mathrm{fric}}(t)$ is defined as follows:
\begin{defn}[Liouvillian with friction]
    Let $\beta, m, t > 0$, let $f: \mathbb{R}^N \rightarrow \mathbb{R}$ be differentiable and let $H_{\mathrm{fric}}(t)$ be the classical Hamiltonian defined in Definition~\ref{def:friction_ham}. Then we define the following time-dependent friction Liouvillian:
    \begin{equation}
        L_{\mathrm{fric}}(t) := -i \sum_{j=1}^N \lb \frac{\partial H_{\mathrm{fric}}(t)}{\partial p_j} \frac{\partial}{\partial x_j} - \frac{\partial H_{\mathrm{fric}}(t)}{\partial x_j} \frac{\partial}{\partial p_j} \rb
        = -i \sum_{j=1}^N \lb e^{-\beta t/m} \frac{p_j}{m} \frac{\partial}{\partial x_j} - e^{\beta t/m} \frac{\partial f}{\partial x_j} \frac{\partial}{\partial p_j} \rb.
    \end{equation}
\label{def:friction_liouvillian}
\end{defn}
The corresponding unitary time evolution operator for the phase space density is given by 
\begin{equation}
    U_L(t) := \mathcal{T} \ls e^{-i\int_0^t L_{\mathrm{fric}}(s) \mathrm{d}s} \rs,
\label{friction_liouvillian_ev}
\end{equation}
where $\mathcal{T}\ls \cdot \rs$ denotes the time-ordering operator.

So far, we have assumed that the friction Hamiltonian $H_{\mathrm{fric}}$ in Definition~\ref{def:friction_ham} is a classical Hamiltonian. However, in principle, $H_{\mathrm{fric}}$ can also be viewed as a quantum Hamiltonian.
In the quantum case, it can be shown that an initial wave function, when evolved under the unitary time evolution operator 
\begin{equation}
    U_H(t) := \mathcal{T} e^{-i\int_0^t H_{\mathrm{fric}}(s) \mathrm{d}s},
\label{friction_ham_ev}
\end{equation}
converges to a superposition of delta peaks centered at the stationary points of $f$, see Ref.~\cite{Smedt1986CK_asymptotics} for a proof in the 1-dimensional case. Thus, simply implementing $U_H(t)$ also allows us to find a local minimum/stationary point of the objective function $f$.
In the following, we will discuss how to implement discrete versions of both $U_L(t)$ and $U_H(t)$ on a quantum computer and how to use them to find a local optimum of $f$.
 
Our first algorithm simulates classical dynamics with friction whereas our second algorithm simulates quantum dynamics with friction. Note that our main result utilizes only the second algorithm for simulating quantum dynamics with friction as it leads to better bounds. Nonetheless, we include a discussion of the algorithm for simulating classical dynamics with friction for completeness since the analysis of our main algorithm leverages results from the analysis of the first algorithm.

Both algorithms utilize the interaction picture to reduce the impact of the exponential factor $e^{\beta t/m}$, appearing in the definition of $H_{\mathrm{fric}}$, on the running time.
The interaction picture is effectively just a unitary transformation on the Hilbert space under consideration. In the following, we will derive the interaction picture for a Hamiltonian $H(t)$ which is time-dependent in the Schrödinger picture. Specifically, we assume that $H(t)$ is of the form $H(t) = A(t) + B(t)$ where $A(t)$ and $B(t)$ are both Hermitian for all $t$.
Let $\ket{\psi(t)}_S$ be a quantum state in the Schrödinger picture and define 
\begin{equation}
    U_B(t) := \mathcal{T} \ls e^{- i \int_0^t B(s) \mathrm{d}s} \rs,
\end{equation}
where
\begin{equation}
    \mathcal{T} \ls e^{- i \int_0^t B(s) \mathrm{d}s} \rs := \lim_{r \rightarrow \infty} \prod_{j=1}^r e^{-iB\lb  \frac{j-1}{r} t \rb \frac{t}{r}} 
\end{equation}
is a time-ordered operator exponential.
Then the quantum state in the interaction picture is given by
\begin{equation}
    \ket{\psi(t)}_I := U_B^\dagger(t) \ket{\psi(t)}_S.
\end{equation}
Note that
\begin{equation}
    U_B^\dagger(t) = \lim_{r \rightarrow \infty} \prod_{j=r}^1 e^{iB\lb  \frac{j-1}{r} t \rb \frac{t}{r}} = \lim_{r \rightarrow \infty} \prod_{j=r}^1 e^{-iB\lb  \frac{j-1}{r} t \rb \lb - \frac{t}{r} \rb}  = \mathcal{T} \ls e^{- i \int_t^0 B(s) \mathrm{d}s} \rs,
\end{equation}
meaning that generically, $ U_B^\dagger(t) \neq \mathcal{T} e^{i \int_0^t B(s) \mathrm{d}s}$.

An operator $M(t)$ in the Schrödinger picture takes the following form in the interaction picture:
\begin{equation}
    M_{I}(t) = U_B^\dagger(t) M(t) U_B(t). 
\end{equation}
The differential equation governing the time evolution of $\ket{\psi(t)}_I$ is given by
\begin{equation}
\begin{split}
    i \partial_t \ket{\psi(t)}_I &= i \partial_t \lb U_B^\dagger(t) \ket{\psi(t)}_S \rb = i \lb \frac{\partial U_B^\dagger(t)}{\partial t} \ket{\psi(t)}_S + U_B^\dagger(t) \frac{\partial \ket{\psi(t)}_S}{\partial t} \rb \\
    &= - U_B^\dagger(t) B(t) \ket{\psi(t)}_S + U_B^\dagger(t) \lb A(t) + B(t) \rb \ket{\psi(t)}_S \\
    &=  U_B^\dagger(t) A(t) U_B(t) U_B^\dagger(t) \ket{\psi(t)}_S \\
    &= A_I(t) \ket{\psi(t)}_I,
\end{split}
\end{equation}
where we defined
\begin{equation}
    A_I(t) := U_B^\dagger(t) A(t) U_B(t),
\label{intpic_op}
\end{equation}
and used the fact that
\begin{equation}
    \partial_t \mathcal{T} \ls e^{-i \int_t^0 B(s) \mathrm{d}s} \rs = i \mathcal{T} \ls e^{-i \int_t^0 B(s) \mathrm{d}s} \rs B(t),
\end{equation}
which follows from the second fundamental theorem of calculus. 
Specifically,
\begin{equation}
    \partial_t \int_t^0 B(s) \mathrm{d}s = -B(t).
\end{equation}
This implies that
\begin{align}
    \ket{\psi(t)}_I &= \mathcal{T} \ls e^{- i \int_0^t A_I(s) \mathrm{d}s} \rs \ket{\psi(0)}_I \\
    \ket{\psi(t)}_S &= \mathcal{T} \left[ e^{- i \int_0^t B(s) \mathrm{d}s} \right] \mathcal{T} \left[ e^{- i \int_0^t A_I(s) \mathrm{d}s} \right] \ket{\psi(0)}_S.
    \label{intpic_to_schrodinger}
\end{align}
Note that in the special case where $\left[ B(t), B(t') \right] = 0$ for all $t,t' \geq 0$ we have that
\begin{align}
    U_B(t) &= \mathcal{T} \left[ e^{- i \int_0^t B(s) \mathrm{d}s} \right] = e^{- i \int_0^t B(s) \mathrm{d}s} \\
    U_B^\dagger(t) &= \mathcal{T} \left[ e^{- i \int_t^0 B(s) \mathrm{d}s} \right] = e^{- i \int_t^0 B(s) \mathrm{d}s} = e^{i \int_0^t B(s) \mathrm{d}s}.
\end{align}

In order to simulate the time evolution under a time-dependent Hamiltonian, we assume that we have access to a time-dependent matrix encoding of the Hamiltonian as defined below.
\begin{defn}[Time-dependent matrix encoding~\cite{Low2019InteractionPic}]
    Given a matrix $H(s) : [0,t] \rightarrow \mathbb{C}^{2^{n_s} \times 2^{n_s}}$, integer $M >0$, and a promise $\norm{H} \leq \alpha$, assume there exists a unitary oracle $\hamt \in \mathbb{C}^{M 2^{n_a + n_s} \times M 2^{n_a + n_s}}$ such that
    \begin{align}
        &\hamt = 
        \begin{pmatrix}
            H/\alpha & \cdot \\
            \cdot & \cdot 
        \end{pmatrix}, 
        \quad H = \mathrm{Diagonal} [H(0), H(t/M), \dots, H((M-1)t/M] \\
        & \implies \lb \bra{0}_a \otimes \mathbb{1} \rb \hamt \lb \ket{0}_a \otimes \mathbb{1} \rb = \sum_{m=0}^{M-1} \ketbra{m}{m} \otimes \frac{H(mt/M)}{\alpha}.
    \end{align}
\label{def:ham-t}
\end{defn}

We simulate the time-ordered operator exponential by implementing a truncated Dyson series as shown in~\cite{Low2019InteractionPic,kieferova2019simulating}. This approach has the best known scaling for simulating generic time-dependent Hamiltonians; however, the constant factors involved can make implementing it a challenge.  The central idea behind the approach is simple: we take the Dyson series for the time-evolution operator and discretize the integrals to reduce the result to a discrete sum of unitaries that can be implemented using LCU methods~\cite{childs2012hamiltonian}.
 For convenience, let us restate the main result of these works here.

\begin{lem}[Multi-segment Hamiltonian simulation by a truncated Dyson series~\cite{Low2019InteractionPic}]
    Let $H(s): [0,t] \rightarrow \mathbb{C}^{2^{n_s} \times 2^{n_s}}$, and let it be promised that $\max_s \norm{H(s)} \leq \alpha$, and let $\langle \| \dot{H} \| \rangle := \frac{1}{t} \int_0^t \norm{\frac{\mathrm{d} H(s)}{\mathrm{d}s}} \mathrm{d} s$. Further, let $\tau := t/\lceil 2 \alpha t \rceil$ and assume $H_j(s) := H((j-1) \tau + s) : s \in [0, \tau]$ is accessed by an oracle $\hamt_j$ of the form specified in Definition~\ref{def:ham-t} with
    \begin{equation}
        M \in O \lb  \frac{t}{\alpha \epsilon} \lb \langle \| \dot{H} \| \rangle + \max_s \norm{H(s)}^2 \rb \rb.
    \end{equation}
    For all $|t| \geq 0$ and $\epsilon > 0$, an operation $W$ can be implemented with failure probability at most $O (\epsilon)$ such that
    \begin{equation}
        \norm{W - \mathcal{T} \left[ e^{-i \int_0^t H(s) \mathrm{d}s} \right]} \leq \epsilon
    \end{equation}
    with the following cost:
    \begin{enumerate}
        \item Queries to all $\hamt_j$: $O \lb \alpha t \frac{\log \lb \alpha t/\epsilon \rb}{\log \log \lb \alpha t /\epsilon \rb} \rb$,
        \item Qubits: $n_s + O \lb n_a + \log \lb \frac{t}{\alpha \epsilon} \lb \langle \| \dot{H} \| \rangle + \max_s \norm{H(s)}^2 \rb \rb \rb$,
        \item Primitive gates: $O \lb \alpha t \lb n_a + \log \lb \frac{t}{\alpha \epsilon} \lb \langle \| \dot{H} \| \rangle + \max_s \norm{H(s)}^2 \rb \rb \rb \frac{\log \lb \alpha t/\epsilon \rb}{\log \log \lb \alpha t /\epsilon \rb} \rb$.
    \end{enumerate}
\label{lem:time-dep-sim}
\end{lem}
As the number of query operations made to the $\hamt_j$ oracles is independent of the derivative of the Hamiltonian, the work of~\cite{Low2019InteractionPic} recognizes that transforming into an interaction frame can provide a substantial advantage for simulating the dynamics by taking the term with the largest coefficient in the LCU expansion and transforming it into a rapidly varying time-dependent Hamiltonian.  In our work, however, we need to perform the interaction picture transformation for a time-dependent Hamiltonian owing to the native time-dependence of the dissipative Hamiltonian that we simulate.
The following lemma shows how to utilize the interaction picture for simulating time-dependent Hamiltonians.

\begin{restatable}[Interaction picture simulations for time-dependent Hamiltonians]{lem}{TdepIntpic}
    Let $A(s): [0,t] \mapsto \mathbb{C}^{2^{n_s} \times 2^{n_s}}$, $B(s): [0,t] \mapsto \mathbb{C}^{2^{n_s} \times 2^{n_s}}$ be Hermitian and let it be promised that $\norm{A(s)} \leq \alpha_A$, $\norm{\dot{A}(s)} \leq \alpha_A'$ and $\norm{B(s)} \leq \alpha_B$ for all $s \in [0,t]$. Assume access to the following unitary oracles $\hamt_j$:
    \begin{equation}
    \begin{split}
         &\lb \bra{0}_a \otimes \mathbb{1}_n \rb \hamt_j \lb \ket{0}_a \otimes \mathbb{1}_n \rb = \sum_{m=0}^{M-1} \ketbra{m}{m} \otimes \frac{A_I(t_j + \tau m/M)}{\alpha_A} \\
         &= \sum_{m=0}^{M-1} \ketbra{m}{m} \otimes \frac{U_B((j-1)\tau + \tau m/M)^\dagger A((j-1)\tau + \tau m/M) U_B((j-1)\tau + \tau m/M)}{\alpha_A},
    \end{split}
    \end{equation}
    where $U_B(s) := \mathcal{T} \ls e^{- i \int_0^s B(s') \mathrm{d}s'} \rs$, $A_I(s) := U_B^\dagger(s) A(s) U_B(s)$,
    $\tau \in O \lb \alpha_A^{-1} \rb$ and $M \in O \lb \frac{t}{\epsilon} \lb \alpha_A + \alpha_B + \frac{\alpha_A'}{\alpha_A} \rb \rb$.
    For all $|t| \geq 0$, the time-evolution operator $\mathcal{T} \ls e^{-i \int_0^t \lb A(s) + B(s) \rb \mathrm{d}s} \rs$ may be approximated to error $\epsilon$ with the following cost:
    \begin{enumerate}
        \item Simulations of $U_B(t)$: $1$,
        \item Queries to all $\hamt_j$: $O \lb \alpha_A t \frac{\log \lb \alpha_A t/\epsilon \rb}{\log \log \lb \alpha_A t/\epsilon \rb} \rb$,
        \item Qubits: $n_s + O \lb n_a + \log \lb \frac{t}{\epsilon} \lb \alpha_A + \alpha_B + \frac{\alpha_A'}{\alpha_A} \rb \rb \rb$,
        \item Primitive gates: $O \lb \alpha_A t \lb n_a + \log \lb \frac{t}{\epsilon} \lb \alpha_A + \alpha_B + \frac{\alpha_A'}{\alpha_A} \rb \rb \rb \frac{\log \lb \alpha_A t/\epsilon \rb}{\log \log \lb \alpha_A t/\epsilon \rb} \rb$.
    \end{enumerate}
\label{lem:intpic-t-dep-sim}
\end{restatable}
The proof of the above lemma is similar to the proof of Lemma~\ref{lem:time-dep-sim} in~\cite{Low2019InteractionPic}. The details can be found in Appendix~\ref{app:intpic-t-dep}.

In general, we might not be able to implement the $\hamt_j$ oracles exactly. 
The following lemma provides a bound on the allowable error in the $\hamt_j$ oracles.

\begin{restatable}[Robust interaction picture simulations]{lem}{Robustintpic}
    Let $A(s): [0,t] \mapsto \mathbb{C}^{2^{n_s} \times 2^{n_s}}$, $B(s): [0,t] \mapsto \mathbb{C}^{2^{n_s} \times 2^{n_s}}$ be differentiable Hermitian operator valued functions and let it be promised that $\norm{A(s)} \leq \alpha_A$, $\norm{\dot{A}(s)} \leq \alpha_A'$ and $\norm{B(s)} \leq \alpha_B$ for all $s \in [0,t]$. Let $\epsilon \in (0,1)$ be an error tolerance and let $S \in \Theta \lb \alpha_A t \rb$ be an integer. Assume having access to the following unitary oracles $\widetilde{\hamt_j}$:
    \begin{equation}
         \lb \bra{0}_a \otimes \mathbb{1}_n \rb \widetilde{\hamt}_j \lb \ket{0}_a \otimes \mathbb{1}_n \rb = \sum_{m=0}^{M-1} \ketbra{m}{m} \otimes \frac{\widetilde{A}_I(t_j + \tau m/M)}{\widetilde{\alpha}_A} 
    \end{equation}
    where $\norm{\widetilde{A}_I(s)} \leq \widetilde{\alpha}_A$ for all $s \in [0,t]$ and
    \begin{equation}
        \norm{\widetilde{A}_I(s) - A_I(s)} \leq \frac{\epsilon}{4 S}
    \end{equation}
    with $A_I \lb s \rb$ as given in Eq.~\eqref{intpic_op}.
    Further, let $\tau := \frac{t}{\lceil 2 e \alpha  t\rceil}$ where $\alpha := \max \left\{\alpha_A, \widetilde{\alpha}_A \right\} \leq \alpha_A + \epsilon$. Additionally, $t_j := (j-1) \tau$ and $M \in O \lb \frac{t}{\epsilon} \lb \alpha_A + \alpha_B + \frac{\alpha_A'}{\alpha_A} \rb \rb$.
    Then the time-evolution operator $\mathcal{T} \ls e^{-i \int_0^t \lb A(s) + B(s) \rb \mathrm{d}s} \rs$ can be approximated to error $\epsilon$ with the following cost:
    \begin{enumerate}
        \item Simulations of $U_B(t)$: $1$,
        \item Queries to all $\widetilde{\hamt}_j$: $O \lb \alpha_A t \frac{\log \lb \alpha_A t/\epsilon \rb}{\log \log \lb \alpha_A t/\epsilon \rb} \rb$,
        \item Qubits: $n_s + O \lb n_a + \log \lb \frac{t}{\epsilon} \lb \alpha_A + \alpha_B + \frac{\alpha_A'}{\alpha_A} \rb \rb \rb$,
        \item Primitive gates: $O \lb \alpha_A t \lb n_a + \log \lb \frac{t}{\epsilon} \lb \alpha_A + \alpha_B + \frac{\alpha_A'}{\alpha_A} \rb \rb \rb \frac{\log \lb \alpha_A t/\epsilon \rb}{\log \log \lb \alpha_A t/\epsilon \rb} \rb$.
    \end{enumerate}
\label{lem:robust_intpic_sim}
\end{restatable}

The proof of the above lemma is quite similar to the proof of Lemma~\ref{lem:intpic-t-dep-sim}. It can be found in Appendix~\ref{app:robust_intpic}.

\subsection{Liouvillian Quantum Optimization Algorithm}

The first method that we will consider is a local optimization algorithm involving classical dynamics with friction. Specifically, we encode the parameters of the optimization problem in the position variables of a classical system whose dynamics are governed by the friction Liouvillian given in Definition~\ref{def:friction_liouvillian}. As argued earlier, after sufficiently long time evolution, the position variables will approach a local optimum/stationary point of the objective function. Thus, the goal is to implement the unitary time evolution operator $U_L(t)$ given in Eq.~\eqref{friction_liouvillian_ev} as efficiently as possible. The central advantage of this method is that it allows us to leverage our classical insights about the dynamics of the system to guide the optimization problem.  The drawback is that, relative to purely quantum dynamics, the need for qubits is greater because both position and momentum need to be separately encoded in this approach. Additionally, we require access to the derivatives of the objective function which is not the case for the quantum dynamics approach discussed in the next section.

In order to simulate time evolution $U_L(t)$ generated by the time-dependent friction Liouvillian $L_{\mathrm{fric}}$ on a quantum computer we need to discretize it.

\begin{defn}[Discretized Friction Liouvillian]
    Let $\beta, m > 0$, let $f: \mathbb{R}^N \rightarrow \mathbb{R}$ be differentiable, let $D_{x,j}$ denote the finite difference approximation of some fixed order to $\partial_{x,j}$ on $2^n$ grid points where $j \in \lc 1,2, \dots, N \rc$ and similarly, let $D_{p,j}$ denote the finite difference approximation of some fixed order to $\partial_{p,j}$ on $2^n$ grid points.
    Then we define the following time-dependent friction Liouvillian on $2 N n$ qubits:
    \begin{equation}
    \begin{split}
        \widetilde{L}(t) &:= - \frac{i}{m} e^{-\beta t/m} \sum_{j=1}^N \Big( \mathbb{1}_2 \otimes \cdots \otimes \mathbb{1}_{j-1} \otimes D_{x,j} \otimes \mathbb{1}_{j+1} \otimes \cdots \otimes \mathbb{1}_N \\
        &\quad \underbrace{\qquad \qquad \qquad \quad  \otimes \mathbb{1}_1 \otimes \mathbb{1}_2 \otimes \cdots \otimes \mathbb{1}_{j-1} \otimes \sum_{p_j} p_j \ketbra{p_j}{p_j} \otimes \mathbb{1}_{j+1} \otimes \cdots \otimes \mathbb{1}_N \Big)}_{=: A_L(t)} \\
        &\qquad + \underbrace{i e^{\beta t/m} \sum_{j=1}^N \lb \sum_{\bx} \frac{\partial f}{\partial x_j} \ketbra{\bx}{\bx} \otimes \mathbb{1}_1 \otimes \mathbb{1}_2 \otimes \cdots \otimes \mathbb{1}_{j-1} \otimes D_{p,j} \otimes \mathbb{1}_{j+1} \otimes \cdots \otimes \mathbb{1}_N \rb}_{=: B_L(t)}.
    \end{split}
    \end{equation}
    The first $Nn$ qubits encode the $N$ discretized position variables each of which can take on the following $2^n$ values: $x_{\max} - h_x, x_{\max} - 2h_x \dots, 0 , -h_x, \dots, - x_{\max}$ with $h_x$ denoting the position grid spacing. Similarly, the remaining $Nn$ qubits encode the $N$ discretized momentum variables each of which can take on the following $2^n$ values: $p_{\max} - h_p, p_{\max} - 2h_p \dots, 0 , -h_p, \dots, - p_{\max}$ with $h_p$ denoting the momentum grid spacing.
    
    Further, we define
    \begin{align}
        \alpha_{A_L} &:= \max_{s \in [0,t]} \norm{A_L(s)}, \\
        \alpha_{B_L} &:= \max_{s \in [0,t]} \norm{B_L(s)}.
    \end{align}
\label{def:discretized_liouvillian}
\end{defn}

To simplify notation, we will often write $D_{x,j}$ to mean $\mathbb{1}_1 \otimes \mathbb{1}_2 \otimes \cdots \otimes \mathbb{1}_{j-1} \otimes D_{x,j} \otimes \mathbb{1}_{j+1} \otimes \cdots \otimes \mathbb{1}_{N}$ and similarly for $D_{p,j}$.

Since the discretized Liouvillian $\widetilde{L}$ is still Hermitian, we can treat it as just another time-dependent Hamiltonian and apply the results from the previous section for simulating time-dependent Hamiltonians. The following lemma provides some useful bounds on various quantities associated with $\Lt(t)$.

\begin{lem}[Properties of $\Lt(t)$]
    Consider the Liouvillian $\Lt(t)$ from Definition~\ref{def:discretized_liouvillian}. Let $f'_{\max} := \max_{\bx,j} |\frac{\partial f}{\partial x_j}|$ and let $p_{\max}$ be an upper bound on any of the discrete momentum variables.
    Then for all $s \in [0,t]$ the following statements are true:
    \begin{enumerate}
        \item $\alpha_{A_L} \leq N \norm{D_x} \frac{p_{\max}}{m}$, where $\norm{D_x} := \max_j \norm{D_{x,j}}$.
        \item $\alpha_{B_L} \leq e^{\beta t/m} f'_{\max} \norm{D_p}$, where $\norm{D_p} := \max_j \norm{D_{p,j}}$.
        \item $\langle \| \dot{\Lt} \| \rangle = \frac{1}{t} \int_0^t \norm{\frac{\mathrm{d} \Ht(s)}{\mathrm{d}s}} \mathrm{d} s \leq \frac{\beta}{m} \lb \alpha_{A_L} + \alpha_{B_L} \rb$.
    \end{enumerate}
\label{lem:props_disc_L}
\end{lem}

\begin{proof}
    The first inequality follows from the fact that $e^{-\beta t/m}$ is at most $1$ such that each of the $N$ summands of $A_L$ is upper bounded by $\norm{D_x} \frac{p_{\max}}{m}$. The second inequality follows from similar considerations. The third inequality follows from the fact that
    \begin{equation}
        \frac{1}{t} \int_0^t \norm{\frac{\mathrm{d} \Lt(s)}{\mathrm{d}s}} \mathrm{d} s \leq \max_s \norm{\frac{\mathrm{d} \Lt(s)}{\mathrm{d}s}} \leq \max_s \lc \frac{\beta}{m} \norm{A_L(s)} + \frac{\beta}{m} \norm{B_L(s)} \rc
        \leq \frac{\beta}{m} \lb \alpha_{A_L} + \alpha_{B_L} \rb.
    \end{equation}
\end{proof}

Note that $A_L(t)$ can be decomposed into a linear combination of unitaries, $A_L(t) = \sum_k \alpha_k(t) U_k$ with $U_k$ unitary, such that $\max_t \sum_k |\alpha_k(t)| \in O \lb \alpha_{A_L} \rb$. This can be accomplished by decomposing $D_{x,j}$ into a linear combination of unitary adders and implementing $\sum_{p_j} p_j \ketbra{p_j}{p_j}$ via the alternating sign trick, see Refs.~\cite{Berry2014alter_sign_trick, Simon2024Liouvillian} for more details.

Next, let us show how to utilize interaction picture simulations to efficiently implement the time evolution under $\Lt$. To do so, we need to discuss how to implement the corresponding $\hamt_j$ oracles as described in Definition~\ref{def:ham-t}.

\begin{lem}[Approximate implementation of $\hamt_j$ for the Liouvillian approach]
    Let $\beta, m > 0$ and for all $s \in [0,t]$ let $A(s) \equiv A_L(s)$ and $B(s) \equiv B_L(s)$ according to Definition~\ref{def:discretized_liouvillian}. Further, define
    \begin{align}
        U_{B}(s) &:= \mathcal{T} e^{- i \int_0^s B(s') \mathrm{d}s'} \\
        A_I(s) &:= U_B^\dagger(s) A(s) U_B(s)
    \end{align}
    and consider the same setting as in Lemma~\ref{lem:robust_intpic_sim}.
    We can implement each unitary $\widetilde{\hamt}_j$ using 
    \begin{equation}
        \widetilde{O} \lb N \frac{m}{\beta} e^{\beta t/m} f'_{\max} \norm{D_p} \log^2 \lb \frac{\alpha_{A_L} S}{\epsilon} \rb \rb
    \end{equation}
    queries to all phase oracles $O_{f',k}^{(p)}$ and controlled-$O_{f',k}^{(p)}$ as given in Definition~\ref{def:phase_oracle_derivatives}.
\label{lem:hamt_liouvillian}
\end{lem}

\begin{proof}
    First, note that $B(s)$ commutes with $B(s')$ for all $s, s' \in [0,t]$. Further, for all $k \in [N]$ the discrete derivative operators obey
    \begin{equation}
        D_{p,k} := -i \qft \sum_{q_k} g_k(q_k)  \ketbra{q_k}{q_k} \qft^{-1},
    \end{equation}
    where $g_k: \mathbb{R} \rightarrow \mathbb{R}$ is specified by the finite difference scheme of $D_{p,k}$ used for approximating $\partial_{p,k}$ and $\qft$ denotes the quantum Fourier transform.
    This means that
    \begin{equation}
    \begin{split}
        U_B(t) &= \mathcal{T} e^{-i \int_0^t B(s) \mathrm{d}s} = e^{\frac{m}{\beta} \lb e^{\beta t/m} - 1 \rb \sum_k \sum_{\bx} \frac{\partial f}{\partial x_k} \ketbra{\bx}{\bx} \otimes D_{p,k}} \\
        &= \prod_{k=1}^N \lb 
        \one \otimes \qft_k \cdot \lb  \sum_{\bx} \sum_{q_k} e^{-i \frac{m}{\beta} \lb e^{\beta t/m} - 1 \rb \frac{\partial f}{\partial x_k} g_k(q)} \ketbra{\bx}{\bx} \otimes \ketbra{q_k}{q_k} \rb \cdot \one \otimes \qft_k^{-1} \rb.
    \end{split}
    \end{equation}
    Similarly,
    \begin{equation}
        U_B^\dagger(t) = \prod_{k=1}^N \lb 
        \one \otimes \qft_k \cdot \lb  \sum_{\bx} \sum_{q_k} e^{i \frac{m}{\beta} \lb e^{\beta t/m} - 1 \rb \frac{\partial f}{\partial x_k} g_k(q)} \ketbra{\bx}{\bx} \otimes \ketbra{q_k}{q_k} \rb \cdot \one \otimes \qft_k^{-1} \rb.
    \end{equation}
    Then we can decompose the error-free $\hamt_j$ oracles as follows:
    \begin{equation}
    \begin{split}
        \hamt_j &= \underbrace{\lb \sum_{m'=0}^{M-1} \ketbra{m'}{m'} \otimes \mathbb{1}_a \otimes U_B^\dagger(t_j + \tau m'/M) \rb}_{=: V_B^\dagger(j)} \times \underbrace{\lb \sum_{m'=0}^{M-1} \ketbra{m'}{m'} \otimes O_A(t_j + \tau m'/M) \rb}_{=: V_A(j)} \\
        &\qquad \times \underbrace{\lb \sum_{m'=0}^{M-1} \ketbra{m'}{m'} \otimes \mathbb{1}_a \otimes U_B(t_j + \tau m'/M) \rb}_{=: V_B(j)},
    \end{split}
    \end{equation}
    where $t_{j,m'} := (j-1) \tau + m' \tau/M$ and $\lb \bra{0}_a \otimes \mathbb{1}_n \rb O_A(t_j + \tau m'/M) \lb \ket{0}_a \otimes \mathbb{1}_n \rb = A(t_j + \tau m'/M)/\alpha_{A_L}$.
    
    Let us now show how to implement a unitary approximation $\widetilde{V}_B(j)$ to $V_B(j)$.
    For simplicity, we ignore the ancilla register labeled $a$ used in the implementation of $V_A(j)$ as the $V_B(j)$'s act trivially on that register. Let
     \begin{equation}
        F'_k := \sum_\bx \frac{1}{2 f'_{\max}} \frac{\partial f}{\partial x_k} \ketbra{\bx}{\bx},
    \end{equation}
    such that $O_{f',k}^{(p)} = e^{iF'_k}$. Note that $\norm{F'_k} \leq \frac{1}{2}$. This allows us to apply Corollary 71 of Ref.~\cite{gilyen2019quantum} which states that we can implement a $\lb \frac{2}{\pi}, 2, \widetilde{\epsilon} \rb$-block-encoding $U_{F'_k}$ of $F'_k$ using $O \lb \log \lb 1/\widetilde{\epsilon} \rb \rb$ queries to controlled-$O_{f',k}^{(p)}$ and its inverse where we demand that
    \begin{equation}
        \widetilde{\epsilon} \in \Theta \lb \frac{\epsilon}{\alpha_{A_L} S N \frac{m}{\beta} e^{\beta t/m} f'_{\max} \norm{D_{p}}} \rb.
    \end{equation}
    Once we have $U_{F'_k}$ we can use Corollary 62 of Ref.~\cite{gilyen2019quantum} to construct a $\lb 1, 4, \frac{\epsilon}{16 N \alpha_{A_L}  S} \rb$-block-encoding of 
    \begin{equation}
        e^{-i \frac{m}{\beta} \lb e^{\beta t_{j,m'}/m} - 1 \rb 2  f'_{\max} F'_k g_k(q_k)} = e^{-i \frac{m}{\beta} \lb e^{\beta t_{j,m'}/m} - 1 \rb g_k(q_k) \sum_\bx \frac{\partial f}{\partial x_k} \ketbra{\bx}{\bx}} =: U_{F'_k, q_k}
    \end{equation}
    with $k$ and $q_k$ being fixed using $O \lb \frac{m}{\beta} e^{\beta t/m} f'_{\max} \norm{D_p} + \log \lb N \alpha_{A_L} S/\epsilon \rb \rb$ queries to $U_{F'_k}$ and controlled-$U_{F'_k}$ where we used the fact that $|g_k(q_k)| \leq \norm{D_p}$. 
    
    Note though that we require controlled applications of $U_{F'_k, q_k}$ for the implementation of $V_B(j)$ since $t_{j,m'}$ is controlled by the $\ket{m'}$ register and additionally, $g_k(q_k)$ is controlled by the $\ket{q_k}$ register which is the Fourier transformed register associated with the momentum variable $p_k$.
    The idea for implementing $U_{F'_k, q_k}$ in a controlled fashion is as follows:
    \begin{enumerate}
        \item For all $k \in [N]$ do:
        \begin{enumerate}
            \item Compute an $\widetilde{\widetilde{\epsilon}}$-precise approximation $\widetilde{h}(j,m',q_k)$ of $h(j, m',q_k) := \frac{m}{\beta} \lb e^{\beta t_{j,m'}/m} - 1 \rb f'_{\max} g_k(q_k)$ into an ancilla register of size $R = \left\lceil \log \lb \frac{m}{\beta} e^{\beta t/m}  f'_{\max} \norm{D_p}/\widetilde{\widetilde{\epsilon}} \rb \right\rceil$, i.e.
        \begin{equation}
            \ket{m'}\ket{\bx}\ket{q_k}\ket{0} \rightarrow \ket{m'}\ket{\bx}\ket{q_k}\ket{\widetilde{h}(j,m',q_k)},
        \end{equation}
        where $\widetilde{\widetilde{\epsilon}} \leq \frac{\epsilon}{16 N \alpha_{A_L} S}$.
    
        \item For all $r \in [R]$ do:
        
        Controlled by the $r$-th ancilla qubit implement
        \begin{equation}
            \sum_{\bx }e^{-i\widetilde{\widetilde{\epsilon}} \, 2^r \frac{\partial f / \partial x_k}{f'_{\max}}} \ketbra{\bx}{\bx}
        \end{equation}
        within error $\frac{\epsilon}{16 N \alpha_{A_L} S}$ using 
        \begin{equation}
        \begin{split}
             &O \lb \lb  \widetilde{\widetilde{\epsilon}} \, 2^r + \log \lb \alpha_{A_L} S N/\epsilon \rb \rb \log \lb \frac{\alpha_{A_L} S N \frac{m}{\beta} e^{\beta t/m} f'_{\max} \norm{D_p}}{\epsilon} \rb \rb \\
             &\subseteq \widetilde{O} \lb \frac{m}{\beta} e^{\beta t/m} f'_{\max} \norm{D_p} \log^2 \lb \frac{\alpha_{A_L} S N}{\epsilon} \rb \rb
        \end{split}
        \end{equation}
        queries to $O_{f',k}^{(p)}$ and controlled-$O_{f',k}^{(p)}$.
        \end{enumerate}
    \end{enumerate}
    By the triangle inequality, we then have that
    \begin{equation}
        \norm{\widetilde{V}_B(j) - V_B(j)} \leq N \lb \frac{\epsilon}{16 \alpha_{A_L} N S} + \frac{\epsilon}{16 \alpha_{A_L} N S} \rb = \frac{\epsilon/4S}{2 \alpha_{A_L}}.
    \end{equation}
    The same analysis holds for $\widetilde{V}^\dagger_B(j)$ meaning
    \begin{equation}
        \norm{\widetilde{V}^\dagger_B(j) - V^\dagger_B(j)} \leq N \lb \frac{\epsilon}{16 \alpha_{A_L} N S} + \frac{\epsilon}{16 \alpha_{A_L} N S} \rb = \frac{\epsilon/4S}{2 \alpha_{A_L}}.
    \end{equation}

    So far, we have shown how to implement $V_B(j)$ within error $\epsilon/\lb 8 \alpha_{A_L} S \rb$ using the phase oracles $O_{f',k}^{(p)}$ and controlled-$O_{f',k}^{(p)}$. Let us now bound the overall error associated with implementing $\hamt_j$.
    By the triangle inequality we have that
    \begin{equation}
    \begin{split}
        \norm{\widetilde{\hamt}_j - \hamt_j} &= \norm{\widetilde{V}_B^\dagger(j) V_A(j)\widetilde{V}_B(j) - V_B^\dagger(j) V_A(j) V_B(j)} \\
        &\leq \norm{\widetilde{V}_B^\dagger(j) V_A(j)\widetilde{V}_B(j) - \widetilde{V}_B^\dagger(j) V_A(j)V_B(j)} + \norm{\widetilde{V}_B^\dagger(j) V_A(j)V_B(j) - V_B^\dagger(j) V_A(j) V_B(j)} \\
        &\leq \norm{\widetilde{V}_B(j) - V_B(j)} + \norm{\widetilde{V}^\dagger_B(j) - V^\dagger_B(j)} \\
        &\leq \frac{\epsilon/4S}{\alpha_{A_L}}.
    \end{split}
    \end{equation}
    This implies that
    \begin{equation}
        \norm{\widetilde{A}_I(s) - A_I(s)} \leq \alpha_{A_L} \norm{ \lb \bra{0}_a \otimes \mathbb{1}_n \rb \widetilde{\hamt}_j \lb \ket{0}_a \otimes \mathbb{1}_n \rb -  \lb \bra{0}_a \otimes \mathbb{1}_n \rb \hamt_j \lb \ket{0}_a \otimes \mathbb{1}_n \rb}
        \leq \frac{\epsilon}{4S},
    \end{equation}
    as desired.
\end{proof}

\begin{lem}[Liouvillian simulation of friction with phase oracle access]
    Let $\epsilon \in (0,1)$ be an error tolerance and let $t > 0$. Further, let $f: \mathbb{R}^N \rightarrow \mathbb{R}$ be differentiable and let $\Lt(t)$ be the corresponding discretized Liouvillian as given in Def.~\ref{def:discretized_liouvillian}. Assume having access to the phase oracles $\lc O_{f', k}^{(p)} \rc_{k=1}^N$ as described in Definition~\ref{def:phase_oracle_derivatives} as well as their controlled versions.
    Then an operation $W$ can be implemented with failure probability at most $O(\epsilon)$ such that
    \begin{equation}
        \norm{W - \mathcal{T} \left[ e^{-i \int_0^t \Lt(s) \mathrm{d}s} \right]} \leq \epsilon
    \end{equation}
    using a total number of queries to all $O_{f', k}^{(p)}$ and controlled-$O_{f', k}^{(p)}$ that scales as
    \begin{equation}
        \widetilde{O} \lb \alpha_{A_L} N\frac{m}{\beta} e^{\beta t/m}  f'_{\max} \norm{D_p} t \log^3 \lb \frac{1}{\epsilon} \rb \rb \subseteq \widetilde{O} \lb N^2 \norm{D_x} p_{\max} e^{\beta t/m}  \frac{f'_{\max}}{\beta} \norm{D_p} t \log^3 \lb \frac{1}{\epsilon} \rb \rb.
    \end{equation}
\label{lem:liouvillian_friction}
\end{lem}

\begin{proof}
    Follows directly from Lemmas~\ref{lem:robust_intpic_sim}, \ref{lem:props_disc_L} and \ref{lem:hamt_liouvillian} and the fact that we can simulate $\mathcal{T} e^{- i \int_0^t B_L(s) \mathrm{d}s}$ within error $O \lb \epsilon \rb$ using
     \begin{equation}
         \widetilde{O} \lb N \frac{m}{\beta} e^{\beta t/m} f'_{\max} \norm{D_p}  \log^2 \lb \frac{\alpha_{A_L} S}{\epsilon} \rb \rb \subseteq \widetilde{O} \lb \alpha_{A_L} N\frac{m}{\beta} e^{\beta t/m}  f'_{\max} \norm{D_p} t \log^3 \lb \frac{1}{\epsilon} \rb \rb
     \end{equation}
    queries to all $O_{f', k}^{(p)}$ and controlled-$O_{f', k}^{(p)}$ by following the same strategy as in the proof of Lemma~\ref{lem:hamt_liouvillian}.
\end{proof}

\subsection{Hamiltonian Quantum Optimization Algorithm}
\label{sec:quantum_ham_approach}

In this section, we show how to efficiently implement the time evolution under the following discretized quantum friction Hamiltonian:

\begin{defn}[Discretized Quantum Friction Hamiltonian]
    Let $\beta, m, t > 0$, let $f: \mathbb{R}^N \rightarrow [-f_{\max}, f_{\max}]$ and let $D_{x,j}^2$ denote the finite difference approximation of some fixed order to $\partial_{x,j}^2$ on $2^n$ grid points where $j \in \lc 1,2, \dots, N \rc$. Then we define the following time-dependent Hamiltonian on $Nn$ qubits:
    \begin{equation}
        \Ht(t) := \underbrace{e^{-\beta t/m} \frac{1}{2m} \sum_{j=1}^N \mathbb{1}_1 \otimes \mathbb{1}_2 \otimes \cdots \otimes \mathbb{1}_{j-1} \otimes D_{x,j}^2 \otimes \mathbb{1}_{j+1} \otimes \cdots \otimes \mathbb{1}_N}_{=: A_H(t)} + \underbrace{e^{\beta t/m} \sum_{\bx} f(\bx) \ketbra{\bx}{\bx}}_{=: B_H(t)}.
    \end{equation}
    Each of the $N$ discretized position variables is encoded in $n$ qubits and each can take on the following $2^n$ values: $x_{\max} - h_x, x_{\max} - 2h_x \dots, 0 , -h_x, \dots, - x_{\max}$, where $h_x$ denotes the grid spacing. 
    
    Further, we define
    \begin{align}
        \alpha_{A_H} &:= \max_{s \in [0,t]} \norm{A_H(s)}, \\
        \alpha_{B_H} &:= \max_{s \in [0,t]} \norm{B_H(s)}.
    \end{align}
\label{def:discretized_ham}
\end{defn}
To simplify notation, we will often write $D_{x,j}^2$ to mean $\mathbb{1}_1 \otimes \mathbb{1}_2 \otimes \cdots \otimes \mathbb{1}_{j-1} \otimes D_{x,j}^2 \otimes \mathbb{1}_{j+1} \otimes \cdots \otimes \mathbb{1}_N$.

Now assume that $f$ is twice continuously differentiable, nonnegative and obeys
\begin{equation}
    f(\bx) \rightarrow + \infty \quad \text{as} \quad \norm{\bx} \rightarrow \infty \quad \text{for all} \quad \bx \in \mathbb{R}^N.
\end{equation}
Then, as discussed previously, we expect that sufficiently long time evolution under $\Ht(t)$ will lead to a concentration of the probability distribution associated with a given initial quantum state around stationary points of $f$ as long as the spatial discretization errors are negligible. This allows us to find a stationary point of $f$ with high probability simply by measuring the time-evolved quantum state in the computational basis.
A proof of convergence in continuous space with $N=1$ can be found in~\cite{Smedt1986CK_asymptotics}. While it should be rather straightforward to extend the proof to the case $N>1$, we leave that task for future work.

The following lemma provides some useful bounds regarding various quantities associated with $\Ht(t)$.
\begin{lem}[Properties of $\Ht(t)$]
    Consider the Hamiltonian $\Ht(t)$ from Definition~\ref{def:discretized_ham}.
    Then for all $s \in [0,t]$ the following statements are true:
    \begin{enumerate}
        \item $\alpha_{A_H} \leq N \frac{\norm{D_x^2}}{2m}$, where $\norm{D_x^2} := \max_j \norm{D_{x,j}^2}$.
        \item $\alpha_{B_H} \leq e^{\beta t/m} f_{\max}$.
        \item $\langle \| \dot{\Ht} \| \rangle = \frac{1}{t} \int_0^t \norm{\frac{\mathrm{d} \Ht(s)}{\mathrm{d}s}} \mathrm{d} s \leq \frac{\beta}{m} \lb \alpha_{A_H} + \alpha_{B_H} \rb$.
    \end{enumerate}
\label{lem:props_disc_H}
\end{lem}

\begin{proof}
    The first inequality follows from the fact that $e^{-\beta t/m}$ is at most $1$ such that each of the $N$ summands of $A_H$ is upper bounded by $\frac{\norm{D_x^2}}{2m}$. The second inequality follows directly from the fact that $B_H$ is diagonal in the computational basis. The third inequality follows from the fact that
    \begin{equation}
        \frac{1}{t} \int_0^t \norm{\frac{\mathrm{d} \Ht(s)}{\mathrm{d}s}} \mathrm{d} s \leq \max_s \norm{\frac{\mathrm{d} \Ht(s)}{\mathrm{d}s}} 
        \leq \max_s \lc \frac{\beta}{m} \norm{A_H(s)} + \frac{\beta}{m} \norm{B_H(s)} \rc 
        \leq \frac{\beta}{m} \lb \alpha_{A_H} + \alpha_{B_H} \rb.
    \end{equation}
\end{proof}

Note that $A_H(t)$ can be decomposed into a linear combination of unitaries, $A_H(t) = \sum_k \alpha_k(t) U_k$ with $U_k$ unitary, such that $\max_t \sum_k |\alpha_k(t)| \in O \lb \alpha_{A_H} \rb$. This can be accomplished by decomposing $D_{x,j}^2$ into a linear combination of unitary adders, see Ref.~\cite{Simon2024Liouvillian} for more details. Further, note that
\begin{equation}
    \norm{D_x^2} \in O \lb \frac{1}{h_x^2} \rb,
\end{equation}
since $D_{x,j}^2$ is a finite-difference approximation to the second-order derivative $\frac{\partial^2}{\partial_{x,j}^2}$. This implies that
\begin{equation}
    \alpha_{A_H} \in O \lb \frac{N}{m h_x^2}\rb,
\label{bound_A_H}
\end{equation}
which we used to obtain Theorem~\ref{thm:main_informal} from Theorem~\ref{thm:main}.

According to Lemma~\ref{lem:time-dep-sim}, we can thus implement an operation $W$ with failure probability at most $O \lb \epsilon \rb$ such that
\begin{equation}
    \norm{W - \mathcal{T} \left[ e^{-i \int_0^t \Ht(s) \mathrm{d}s} \right]} \leq \epsilon
\end{equation}
with the following cost:
\begin{enumerate}
    \item Queries to all $\hamt_j$: $O \lb \lb \alpha_{A_H} + \alpha_{B_H} \rb t \frac{\log \lb \lb \alpha_{A_H} + \alpha_{B_H} \rb t/\epsilon \rb}{\log \log \lb \lb \alpha_{A_H} + \alpha_{B_H} \rb t /\epsilon \rb} \rb \subseteq \widetilde{O} \lb  \lb N \frac{\norm{D_x^2}}{m} t + e^{\beta t/m} f_{\max} \rb \log \lb 1/\epsilon \rb \rb$,
    \item Qubits: $Nn + O \lb n_a + \log \lb \frac{t}{\epsilon} \lb \frac{\beta}{m} + \alpha_{A_H} + \alpha_{B_H}  \rb \rb \rb$,
    \item Primitive gates: $O \lb \lb \alpha_{A_H} + \alpha_{B_H} \rb t \lb n_a + \log \lb \frac{t}{\epsilon} \lb \frac{\beta}{m}  + \alpha_{A_H} + \alpha_{B_H} \rb \rb \rb \frac{\log \lb \lb \alpha_{A_H} + \alpha_{B_H} \rb t/\epsilon \rb}{\log \log \lb \lb \alpha_{A_H} + \alpha_{B_H} \rb t /\epsilon \rb} \rb$.
\end{enumerate}

Note that the above upper bound on the number of queries to all $\hamt_j$ oracles scales exponentially with the evolution time $t$. In the following, we discuss how to utilize the interaction picture to improve the query complexity of the simulation. In particular, we show how to exponentially improve the dependence on $t$.
As we ultimately wish to provide an upper bound on the number of queries to the objective function $f$, let us first discuss how to implement the $\hamt_j$ oracles.

\begin{restatable}[Approximate implementation of $\mathtt{HAM} \text{-} \mathtt{T}_j$]{lem}{HamT}
    Let $\beta, m > 0$, let $f: \mathbb{R}^N \rightarrow [-f_{\max}, f_{\max}]$ and for all $s \in [0,t]$ let
    \begin{align}
        A_H(s) &:= e^{- \beta s/m} \sum_{j=1}^N \frac{D_{x,j}^2}{2m} \\
        B_H(s) &:= e^{\beta s/m}\sum_{\bx} f(\bx) \ketbra{\bx}{\bx} \\
        U_B(s) &:= \mathcal{T} e^{- i \int_0^s B_H(s') \mathrm{d}s'} \\
        A_I(s) &:= U_B^\dagger(s) A_H(s) U_B(s)
    \end{align}
    and consider the same setting as in Lemma~\ref{lem:robust_intpic_sim}. 
    We can implement each unitary $\widetilde{\hamt}_j$ using either 4 queries to an $\epsilon'$-precise bit oracle $O_f^{(b)}$ of the objective function $f$ with
    \begin{equation}
        \epsilon' \leq \frac{\epsilon/4S}{4\alpha_{A_H} \frac{m}{\beta} \lb e^{\beta t/m} - 1  \rb},
    \end{equation}
    or 
    \begin{equation}
        \widetilde{O} \lb \frac{m}{\beta} e^{\beta t/m} f_{\max}  \log^2 \lb \frac{\alpha_{A_H} S}{\epsilon} \rb \rb
    \end{equation}
    queries to a controlled phase oracle $O_f^{(p)}$ of the objective function $f$. 
\label{lem:hamt}
\end{restatable}

The proof of the above lemma is quite similar to the proof of Lemma~\ref{lem:hamt_liouvillian}. It can be found in Appendix~\ref{app:hamt}.

\begin{lem}[Quantum simulation of friction with bit oracle access]
    Let $\epsilon \in (0,1)$ be an error tolerance and let $t > 0$. Further, let $f: \mathbb{R}^N \rightarrow \mathbb{R}$ be some function and let $\Ht(t)$ be the corresponding discretized Hamiltonian as given in Def.~\ref{def:discretized_ham}. Assume having access to an $\epsilon'$-precise bit oracle $O_f^{(b)}$ of $f$ as described in Definition~\ref{def:bit_oracle} with $1/\epsilon' \in O \lb \frac{\alpha_{A_H}^2 t \frac{m}{\beta} e^{\beta t/m}}{\epsilon} \rb$.
    Then an operation $W$ can be implemented with failure probability at most $O (\epsilon)$ such that
    \begin{equation}
        \norm{W - \mathcal{T} \left[ e^{-i \int_0^t \Ht(s) \mathrm{d}s} \right]} \leq \epsilon
    \end{equation}
    using a number of queries to $O_f^{(b)}$ that scales as
    \begin{equation}
         O \lb \alpha_{A_H} t \frac{\log \lb \alpha_{A_H} t/\epsilon \rb}{\log \log \lb \alpha_{A_H} t/\epsilon \rb} \rb 
         \subseteq \widetilde{O} \lb \frac{N \norm{D_x^2} t}{m} \log \lb 1/\epsilon \rb \rb.
    \end{equation}
\label{lem:friction}
\end{lem}

\begin{proof}
    Follows directly from Lemmas~\ref{lem:robust_intpic_sim}, \ref{lem:props_disc_H} and \ref{lem:hamt} and the fact that we can simulate $U_B(t) = \mathcal{T} e^{- i \int_0^t B(s) \mathrm{d}s}$ within error $O \lb \epsilon \rb$ using $O(1)$ queries to $O_f^{(b)}$ by following the same strategy as in the proof of Lemma~\ref{lem:hamt}.
\end{proof}

\begin{lem}[Quantum simulation of friction with phase oracle access]
    Let $\epsilon \in (0,1)$ be an error tolerance and let $t > 0$. Further, let $f: \mathbb{R}^N \rightarrow \mathbb{R}$ be some function and let $\Ht(t)$ be the corresponding discretized Hamiltonian as given in Def.~\ref{def:discretized_ham}. Assume having access to a phase oracle $O_f^{(p)}$ of $f$ as described in Definition~\ref{def:phase_oracle}.
    Then an operation $W$ can be implemented with failure probability at most $O(\epsilon)$ such that
    \begin{equation}
        \norm{W - \mathcal{T} \left[ e^{-i \int_0^t \Ht(s) \mathrm{d}s} \right]} \leq \epsilon
    \end{equation}
    using a number of queries to controlled-$O_f^{(p)}$ and its inverse that scales as
    \begin{equation}
        \widetilde{O} \lb \frac{m}{\beta} e^{\beta t/m} f_{\max} \alpha_{A_H} t \log^3 \lb \frac{1}{\epsilon} \rb \rb \subseteq \widetilde{O} \lb \frac{f_{\max}}{\beta} e^{\beta t/m} N \norm{D_x^2} t \log^3 \lb \frac{1}{\epsilon} \rb \rb
    \end{equation}
\label{lem:friction_phase}
\end{lem}

\begin{proof}
     Follows directly from Lemmas~\ref{lem:robust_intpic_sim}, \ref{lem:props_disc_H} and \ref{lem:hamt} and the fact that we can simulate $\mathcal{T} e^{- i \int_0^t B_H(s) \mathrm{d}s}$ within error $O \lb \epsilon \rb$ using 
     \begin{equation}
         \widetilde{O} \lb \frac{m}{\beta} e^{\beta t/m} f_{\max}  \log^2 \lb \frac{\alpha_{A_H} S}{\epsilon} \rb \rb \subseteq \widetilde{O} \lb \frac{f_{\max}}{\beta} e^{\beta t/m} N \norm{D_x^2} t \log^3 \lb \frac{1}{\epsilon} \rb \rb
     \end{equation}
    queries to controlled-$O_f^{(p)}$ and its inverse by following the same strategy as in the proof of Lemma~\ref{lem:hamt}.
\end{proof}

The above lemma effectively provides a query upper bound for finding a local optimum of any twice continuously differentiable nonnegative function $f$ that satisfies the assumption given in Eq.~\eqref{local_opt_assumption}. Note though that the query complexity depends on the evolution time $t$ which is left as a user-specified input parameter. Providing a tight bound on the required $t$ for a given objective function is difficult unless we make further assumptions on the objective function. In the next section, we consider the benchmark problem of optimizing a convex quadratic function for which we can prove relatively tight convergence and query bounds.

\section{Convex Quadratic Optimization}
\label{sec:convex_quadratic}

The previous section provided estimates of the complexity of simulating a dynamical system for a fixed evolution time that will, if $t$ is chosen appropriately, yield a local optimum. In this section, we prove upper bounds on $t$ for the task of finding the optimum of a convex quadratic function $f: \mathbb{R}^N \rightarrow \mathbb{R}$ of the following form:
\begin{equation}
    f(\bx) = \frac{1}{2}\lb \bx - \bx^* \rb^\top A \lb \bx - \bx^* \rb + c,
\end{equation}
where $A \in \mathbb{R}^{N \times N}$ is positive definite, $\bx^*$ is the vector corresponding to the minimum of $f$ and $c \in \mathbb{R}$ is a constant.
The motivation behind considering such a convex quadratic function stems from the observation that any twice continuously differentiable function $g: \mathbb{R}^N \rightarrow \mathbb{R}$ which has at least one local minimum can be approximated by a convex quadratic function around the local minimum as long as its second derivatives at the local minimum are positive. To be more specific, assume that $\bx^*$ is a local minimum of $g$ meaning $\lb \nabla g \rb \Big|_{\bx = \bx^*} = 0$ and $H_g (\bx^*) \succ 0$ where the latter condition states that the Hessian of $g$ evaluated at $\bx^*$ is positive definite.
Taylor expanding $g$ around $\bx^*$ up to second order yields
\begin{equation}
\begin{split}
       g(\bx) &= g(\bx^*) + \lb \bx - \bx^* \rb \cdot \lb \nabla g\rb \Big|_{(\bx = \bx^*)} + \frac{1}{2} \lb \bx - \bx^* \rb^\top H_g(\bx^*) \lb \bx - \bx^* \rb + O \lb \norm{\bx - \bx^*}^3 \rb \\
       &= g(\bx^*) + \frac{1}{2} \lb \bx - \bx^* \rb^\top H_g(\bx^*) \lb \bx - \bx^* \rb + O \lb \norm{\bx - \bx^*}^3 \rb.
\end{split}
\end{equation}
This shows that locally, close to $\bx^*$, $g(\bx)$ is well approximated by a convex quadratic function. Bounding the running time of our quantum algorithms for convex quadratic functions thus allows us to understand the convergence behavior of our algorithms close to the local optima of a generic differentiable non-convex function.

In the following, we will first discuss the main aspects of the classical Liouvillian approach before focusing on the quantum Hamiltonian approach.
A central concept in the proofs is the $\epsilon$-equilibration time of the underlying dynamical system as defined below. 

\begin{defn}[Equilibration time]
    Let $f: \mathbb{R}^N \rightarrow \mathbb{R}$ be differentiable and let $S_{\mathrm{stationary}} := \{\bx \in \mathbb{R}^N | \nabla f(\bx) = 0 \}$ denote the set of stationary points of $f$. For a probability density function over $\bx$ that evolves under a given time-evolution operator, denoted $\rho(\bx,t)$, and for $\rho(\bx,0)=\rho_{\mathrm{init}}$, we call $t^*$ the $\epsilon$-equilibration time of $f$ w.r.t.~$\rho_{\mathrm{init}}$ if $t^*$ is the smallest time such that for all $t \geq t^*$ a sample position vector $\bx'(t) \sim \rho(t)$ from the time evolved probability distribution satisfies the following condition with probability at least $2/3$:
    \begin{equation}
        \left| f \lb \bx'(t) \rb - f(\bx_S) \right| \leq \epsilon,
    \end{equation}
    for some $\bx_S \in S$.
\label{def:eq_times}
\end{defn}
We will consider either the dynamics generated by the classical friction Liouvillian given in Definition~\ref{def:friction_liouvillian} or the quantum friction Hamiltonian given in Definition~\ref{def:friction_ham}. Note that the above definition of equilibration time is not restricted to convex quadratic functions. It can be applied to virtually any function. However, generically, it will be difficult to provide tight bounds.

\subsection{Classical Liouvillian Approach}
\label{sec:liouvillian_approach}

We first analyze the classical approach to give a basis of comparison for the quantum approach which we will discuss afterwards. Let us begin by discussing the Liouvillian framework in a bit more detail. Consider a classical system with initial positions $\bx_0 \in \mathbb{R}^N$ and initial momenta $\bp_0 \in \mathbb{R}^N$. Let $\mathbf{X}_{\bx_0, \bp_0}(t)$ and $\mathbf{P}_{\bx_0, \bp_0}(t)$ denote the solutions to Hamilton's equations of motion such that $\mathbf{X}_{\bx_0, \bp_0}(0) = \bx_0$, $\mathbf{P}_{\bx_0, \bp_0}(0) = \bp_0$ and
\begin{align}
    \frac{dX_j}{dt} = \frac{\partial H(\bx, \bp)}{\partial p_j}\Big|_{\bx = \mathbf{X}(t), \bp = \mathbf{P}(t)} \\
    \frac{dP_j}{dt} = -\frac{\partial H (\bx, \bp)}{\partial x_j}\Big|_{\bx = \mathbf{X}(t), \bp = \mathbf{P}(t)}
\end{align}
where $H$ is the classical Hamiltonian of the system and we dropped the subscripts for notational simplicity. It is well known that the following probability distribution is a solution to the Liouville equation shown in Eq.~\eqref{liouville_eq}:
\begin{equation}
    \widetilde{\rho}(\bx,\bp,t) := \prod_{j=1}^N \lb \delta \lb x_j - X_{j}(t) \rb \delta \lb p_j - P_{j}(t) \rb \rb,
\end{equation}
where $\delta \lb \cdot \rb$ denotes the Dirac delta distribution. This fact can be verified explicitly by plugging the proposed solution back into Eq.~\eqref{liouville_eq}. In the following lemma, we show that having a probability distribution over initial positions and momenta also results in a valid solution to Liouville's equation. 

\begin{lem}[Solutions to Liouville's equation]
    Let $\mathbf{X}_{\bx_0, \bp_0}(s): [0,t] \mapsto \mathbb{R}^N$ and $\mathbf{P}_{\bx_0, \bp_0}(s): [0,t] \mapsto \mathbb{R}^N$ be solutions to Hamilton's equations of motion with Hamiltonian $H$ and initial conditions $\mathbf{X}(0) = \bx_0$ and $\mathbf{P}(0) = \bp_0$. Further, let $\rho(\bx_0,\bp_0)$ be any probability distribution over initial positions and momenta, meaning $\rho(\bx_0,\bp_0) \geq 0$ for all $\bx_0 \in \mathbb{R}^N$ and $\bp_0 \in \mathbb{R}^N$ and $\int_{\mathbb{R}^{2N}} \rho(\bx_0,\bp_0) d\bx_0 d\bp_0 = 1$. Then
    \begin{equation}
        R(\bx,\bp,t) := \int_{\mathbb{R}^{2N}} \rho(\bx_0,\bp_0) \prod_{j=1}^N \lb \delta \lb x_j - X_{\bx_0, \bp_0, j}(t) \rb \delta \lb p_j - P_{\bx_0, \bp_0, j}(t) \rb \rb d\bx_0 d\bp_0,
    \end{equation}
    is the unique probability distribution associated with the initial distribution $\rho(\bx_0,\bp_0)$ that satisfies Liouville's equation after time $t$.
\label{lem:solutions}
\end{lem}

\begin{proof}
    To verify that $R(\bx,\bp,t)$ is indeed a valid probability distribution for all $t$, note that $R(\bx,\bp,t) \geq 0$ for all $\bx$, $\bp$ and $t$ because $\rho(\bx_0,\bp_0)$ is non-negative. Further,
    \begin{equation}
    \begin{split}
        \int_{\mathbb{R}^{2N}} R(\bx,\bp,t) d\bx d\bp &= \int_{\mathbb{R}^{2N}} \lb  \int_{\mathbb{R}^{2N}} \rho(\bx_0,\bp_0) \prod_{j=1}^N \lb \delta \lb x_j - X_{\bx_0, \bp_0, j}(t) \rb \delta \lb p_j - P_{\bx_0, \bp_0, j}(t) \rb \rb d\bx_0 d\bp_0 \rb d\bx d\bp \\
        &= \int_{\mathbb{R}^{2N}} \rho(\bx_0,\bp_0)  \lb \int_{\mathbb{R}^{2N}} \prod_{j=1}^N \lb \delta \lb x_j - X_{\bx_0, \bp_0, j}(t) \rb \delta \lb p_j - P_{\bx_0, \bp_0, j}(t) \rb \rb d\bx d\bp \rb d\bx_0 d\bp_0 \\
        &= \int_{\mathbb{R}^{2N}} \rho(\bx_0,\bp_0) d\bx_0 d\bp_0 = 1.
    \end{split}
    \end{equation}
    The following calculation shows that $R(\bx,\bp,t)$ satisfies Liouville's equation:
    \begin{equation}
    \begin{split}
        \frac{\partial R}{\partial t} &= \int_{\mathbb{R}^{2N}} \rho(\bx_0,\bp_0) \frac{\partial}{\partial t} \prod_{j=1}^N \lb \delta \lb x_j - X_{\bx_0, \bp_0, j}(t) \rb \delta \lb p_j - P_{\bx_0, \bp_0, j}(t) \rb \rb d\bx_0 d\bp_0 \\
        &= \int_{\mathbb{R}^{2N}} \rho(\bx_0,\bp_0) \lb -i L  \rb \prod_{j=1}^N \lb \delta \lb x_j - X_{\bx_0, \bp_0, j}(t) \rb \delta \lb p_j - P_{\bx_0, \bp_0, j}(t) \rb \rb d\bx_0 d\bp_0 \\
        &= -iL \int_{\mathbb{R}^{2N}} \rho(\bx_0,\bp_0) \prod_{j=1}^N \lb \delta \lb x_j - X_{\bx_0, \bp_0, j}(t) \rb \delta \lb p_j - P_{\bx_0, \bp_0, j}(t) \rb \rb d\bx_0 d\bp_0 \\
        &= -iL R.
    \end{split}
    \end{equation}
    Uniqueness follows from the uniqueness assumptions of Hamilton's equations of motion.
\end{proof}

The next lemma shows how to compute expectation values as a function of time w.r.t.~the time-evolved phase space density assuming that we know the solutions to Hamilton's equations of motion.

\begin{lem}[Expectation values from particle trajectories]
    Let $A(\bx,\bp): \mathbb{R}^{2N} \rightarrow \mathbb{R}$ be a function of position and momentum and let $\mathbf{X}_{\bx_0, \bp_0}(t)$ and $\mathbf{P}_{\bx_0, \bp_0}(t)$ be solutions to Hamilton's equations of motion with Hamiltonian $H$ and initial conditions $\mathbf{X}(0) = \bx_0$ and $\mathbf{P}(0) = \bp_0$. Furthermore, let $\rho(\bx_0,\bp_0)$ be a probability distribution over initial positions and momenta. Then the expectation value of $A$ w.r.t.~the time-evolved probability density at time $t$ is given by
    \begin{equation}
        \ev{A}_t = \int_{\mathbb{R}^{2N}} A \lb \mathbf{X}_{\bx_0, \bp_0}(t), \mathbf{P}_{\bx_0, \bp_0}(t) \rb \rho(\bx_0,\bp_0) d\bx_0 d\bp_0. 
    \end{equation}
\label{lem:expectation}
\end{lem}

\begin{proof}
    Lemma~\ref{lem:solutions} shows that the solution to Liouville's equation can be written as follows:
    \begin{equation}
        R(\bx,\bp,t) := \int_{\mathbb{R}^{2N}} \rho(\bx_0,\bp_0) \prod_{j=1}^N \lb \delta \lb x_j - X_{\bx_0, \bp_0, j}(t) \rb \delta \lb p_j - P_{\bx_0, \bp_0, j}(t) \rb \rb d\bx_0 d\bp_0.
    \end{equation}
    Thus,
    \begin{equation}
    \begin{split}
        \ev{A}_t &= \int_{\mathbb{R}^{2N}} A(\bx,\bp) R(\bx,\bp,t) d\bx d\bp \\
        &= \int_{\mathbb{R}^{2N}} A(\bx,\bp) \lb \int_{\mathbb{R}^{2N}} \rho(\bx_0,\bp_0) \prod_{j=1}^N \lb \delta \lb x_j - X_{\bx_0, \bp_0, j}(t) \rb \delta \lb p_j - P_{\bx_0, \bp_0, j}(t) \rb \rb d\bx_0 d\bp_0 \rb d\bx d\bp \\
        &= \int_{\mathbb{R}^{2N}} \rho(\bx_0,\bp_0) \lb  \int_{\mathbb{R}^{2N}} A(\bx,\bp) \prod_{j=1}^N \lb \delta \lb x_j - X_{\bx_0, \bp_0, j}(t) \rb \delta \lb p_j - P_{\bx_0, \bp_0, j}(t) \rb \rb d\bx d\bp \rb d\bx_0 d\bp_0 \\
        &= \int_{\mathbb{R}^{2N}} A \lb \mathbf{X}_{\bx_0, \bp_0}(t), \mathbf{P}_{\bx_0, \bp_0}(t) \rb \rho(\bx_0,\bp_0) d\bx_0 d\bp_0.
    \end{split}
    \end{equation}
\end{proof}
In the following, we will often talk about the expectation value of an observable $A$ w.r.t.~some initial probability distribution $\rho_0(\bx, \bp)$. We denote that expectation value as follows:
\begin{equation}
    \ev{A}_0 := \int_{\mathbb{R}^{2N}} A(\bx,\bp) \rho_0(\bx, \bp) d\bx d\bp.
\end{equation}

Before providing bounds on the $\epsilon$-equilibration time of a general multivariate convex quadratic function, let us consider the problem of optimizing a convex quadratic function of a single variable. Specifically, we consider
\begin{equation}
    f(x) = \frac{a}{2} (x-x^*)^2 + c \quad \text{ with } a>0 \text{ and } c \in \mathbb{R}.
\end{equation}
This can be viewed as the potential function of a harmonic oscillator in one dimension with spring constant $a$ and equilibrium position at $x=x^*$. The constant $c$ translates to an energy offset which does not affect the dynamics in any way. 
For ease of notation, define $\xt := x - x^*$.
Then the equation of motion of a damped harmonic oscillator in terms of $\xt$ reads
\begin{equation}
    m \ddot{\xt} = - a \xt - \beta \dot{\xt},
\label{damped_SHO}
\end{equation}
where $m > 0$ is a tunable mass parameter, $\dot{\xt} := \frac{d\xt}{dt}$, $\ddot{\xt} := \frac{d^2\xt}{dt^2}$ and $\beta > 0$ is the friction coefficient.
There are three different types of solutions to Eq.~\eqref{damped_SHO}.
\begin{enumerate}
    \item Critically damped: $\beta^2 - 4ma = 0$. The solution decays exponentially fast to the equilibrium without any oscillations. Fastest approach to equilibrium.
    \item Underdamped: $\beta^2 - 4ma < 0$. The solution is oscillatory but is enveloped by an exponentially decaying function.
    \item Overdamped: $\beta^2 - 4ma > 0$. The solution decays exponentially fast to the equilibrium without any oscillations but slower than in the critically damped regime.
\end{enumerate}
In the following, we will focus on the critically damped and underdamped harmonic oscillators as they provide the best convergence behavior. 

The definition below allows us to express our results more compactly.
\begin{defn}[Modified Lambert $W$ function]
    Let $x < 0$ and let $W_{-1}$ denote the $-1$ branch of the Lambert $W$ function. Then we define the modified Lambert $W$ function as follows:
    \begin{equation}
        \widetilde{W}(x):=
        \begin{cases}
            W_{-1}(x) & \mathrm{if} \quad -\frac{1}{e} \leq x < 0 \\
            0 & \mathrm{if} \quad x < -\frac{1}{e}.
        \end{cases}
    \end{equation}
\end{defn}

Let us first discuss the equilibration time of the critically damped harmonic oscillator.

\begin{lem}[Equilibration time of the critically damped harmonic oscillator]
    Let $f(x) = \frac{a}{2} (x-x^*)^2 + c$ with $a~>~0$ and $c \in \mathbb{R}$ constant and let $\beta^2 = 4ma$. Given an initial phase space density $\rho_0(x, p)$, the $\epsilon$-equilibration time $t^*$ can be upper bounded as follows:
    \begin{equation}
        t^* \leq \max \left\{ t^{\ev{x}}_1, t^{\ev{x}}_2,  t^{\sigma}_1, t^{\sigma}_2, t^{\sigma}_3 \right\},
    \end{equation}
    where 
    \begin{align}
        t^{\ev{x}}_1 &:= \frac{1}{\gamma} \log \lb \sqrt{\frac{8a}{\epsilon}} |\ev{x}_0 - x^*| \rb \\
        t^{\ev{x}}_2 &:= -\frac{1}{\gamma} \widetilde{W} \lb - \sqrt{\frac{\epsilon}{8a}} \frac{\gamma}{|\ev{r}_0|} \rb \\
        t^{\sigma}_1 &:= \frac{1}{2 \gamma} \log \lb \frac{18 a \sigma^2_x}{\epsilon} \rb \\
        t^{\sigma}_2 &:= - \frac{1}{2 \gamma} \widetilde{W} \lb - \frac{\gamma \epsilon}{18 a | \mathrm{cov}_0 \lb x, r \rb|} \rb \\
        t^{\sigma}_3 &:= - \frac{1}{\gamma} \widetilde{W} \lb - \gamma \sqrt{\frac{\epsilon}{18 a \sigma^2_{r}}} \rb,
    \end{align}
    with $\gamma := \frac{\beta}{2m}$, $r := \frac{p}{m} + \gamma (x - x^*)$, $\sigma_x^2 := \ev{x^2}_0 - \ev{x}_0^2$, $\mathrm{cov}_0 \lb x, r \rb := \ev{x  r}_0 - \ev{x}_0\ev{r}_0$ and $\sigma^2_{r} := \ev{r^2}_0 - \ev{r}_0^2$.
\label{lem:critical_time_bound}
\end{lem}

\begin{proof}
    Let us consider the coordinate system where all positions are shifted by $x^*$ such that the minimum is at $\xt_{\min} = 0$. Then the general solution for a critically damped harmonic oscillator with initial position $\xt_0 := x_0 - x^*$ and initial momentum $p_0$ is of the form
    \begin{equation}
    \begin{split}
        \Xt_{\xt_0,p_0}(t) &= e^{-\frac{\beta t}{2m}} \lb \xt_0 + \lb \frac{p_0}{m} + \frac{\beta}{2m} \xt_0 \rb t \rb \\
        &= e^{-\gamma t} \lb \xt_0 + \lb \frac{p_0}{m} + \gamma \xt_0 \rb t \rb \\
        &= e^{-\gamma t} \lb \xt_0 + r_0 t \rb \\
        &= X_{x_0, p_0}(t) - x^*.
    \end{split}
    \end{equation}
    According to Lemma~\ref{lem:expectation} we thus have that
    \begin{equation}
    \begin{split}
        \ev{\xt}_t &= \int_{\mathbb{R}^{2}} e^{-\gamma t} \lb \xt_0 + r_0 t \rb \rho_0(\xt_0, p_0) d\xt_0 dp_0 \\
        &= e^{-\gamma t} \lb \ev{\xt}_0 + \ev{r}_0 t \rb.
    \end{split}
    \end{equation}
    Therefore,
    \begin{equation}
        |\langle \xt \rangle_t - \xt_{\min}| = |\ev{\xt}_t| \leq e^{-\gamma t} |\ev{\xt}_0| + t e^{-\gamma t} |\ev{r}_0|,
    \label{critical_upper_bound}
    \end{equation}
    where we used the fact that $t \geq 0$.
    To get $\epsilon'$-close to the minimum at $\xt_{\min} = 0$ such that
    \begin{equation}
        |\langle \xt \rangle_t - \xt_{\min}| = |\langle x \rangle_t - x^*| \leq \epsilon' \quad \forall t \geq t^*,
    \end{equation}
    it then suffices to ensure that
    \begin{align}
        e^{-\gamma t} |\ev{\xt}_0| &\leq \frac{\epsilon'}{2} \; \forall t \geq t^* \implies t^* \geq \frac{1}{\gamma} \log \lb \frac{2 |\ev{\xt}_0 |}{\epsilon'} \rb = \frac{1}{\gamma} \log \lb \frac{2 |\ev{x}_0 - x^*|}{\epsilon'} \rb \label{critical_t1_x}\\
        e^{-\gamma t} |\ev{r}_0| t &\leq \frac{\epsilon'}{2}  \; \forall t \geq t^* \implies t^* \geq - \frac{1}{\gamma} \widetilde{W} \lb -\frac{\gamma \epsilon'}{2 |\ev{r}_0|} \rb. \label{critical_t2_x}
    \end{align}
   
    Showing that the average position is close to the minimum is not sufficient for our purposes as we might have to obtain a lot of samples to compute the average position. Ideally, we only need a small number of high-quality samples. This can be guaranteed by picking a large enough $t$ such that the position variance at time $t$, $\sigma^2(t) := \ev{x^2}_t - \ev{x}_t^2$, is sufficiently small. Specifically, Chebyshev's inequality states that a position sample $x'_t \sim R(x, p, t)$ from the time evolved phase space density satisfies
    \begin{equation}
         P\lb |x'_t - \ev{x}_t| \geq \epsilon' \rb \leq \frac{\sigma^2(t)}{\epsilon'^2}.
    \end{equation}
    We want this failure probability to be at most $1/3$ for all $t \geq t^*$ such that we can boost the success probability close to at least $1 - \delta$ using only $\log \lb 1/\delta \rb$ samples and taking the sample that yields the smallest value of $f$. Therefore, it suffices to ensure that
    \begin{equation}
        \sigma^2(t) = \ev{x^2}_t - \ev{x}_t^2 \leq \frac{\epsilon'^2}{3}. 
    \label{var_inequality}
    \end{equation}
    Note that $\ev{\xt^2}_t - \ev{\xt}_t^2 = \ev{(x - b)^2}_t - \ev{x - b}_t^2 = \ev{x^2}_t - \ev{x}_t^2 = \sigma^2(t)$.
    Direct computation reveals that 
    \begin{equation}
    \begin{split}
        \ev{\xt^2}_t - \ev{\xt}_t^2 &= \int_{\mathbb{R}^{2}} e^{- 2 \gamma t} \lb \xt_0 + r_0 t \rb^2 \rho(\xt_0, p_0) d\xt_0 dp_0 \\
        &\quad - \lb \int_{\mathbb{R}^{2}} e^{-\gamma t} \lb \xt_0 + r_0 t \rb \rho(\xt_0, p_0) d\xt_0 dp_0 \rb^2 \\
        &= e^{-2 \gamma t} \lb \ev{\xt^2}_0 + 2t \ev{\xt r}_0 + t^2 \ev{r^2}_0 \rb \\
        &\quad - e^{-2 \gamma t} \lb \ev{\xt}_0^2 + 2t \ev{\xt}_0\ev{r}_0 + t^2 \ev{r}_0^2 \rb \\
        &= e^{-2\gamma t} \lb \sigma^2_x + 2 t \, \mathrm{cov}_0 \lb \xt, r \rb + t^2 \sigma^2_r \rb,
    \end{split}
    \end{equation}
    where $\sigma_x^2 = \sigma^2(0)$, $\mathrm{cov}_0 \lb \xt, r \rb = \ev{\xt r}_0 - \ev{\xt}_0\ev{r}_0$ and $\sigma^2_r = \ev{r^2}_0 - \ev{r}_0^2$. Note that $\mathrm{cov}_0 \lb \xt, r \rb = \mathrm{cov}_0 \lb x, r \rb$.
    Thus,
    \begin{equation}
         \sigma^2(t) \leq e^{-2 \gamma t} \sigma^2_x + 2 t e^{-2 \gamma t}  | \mathrm{cov}_0 \lb x, r \rb| + t^2 e^{-2 \gamma t} \sigma^2_r.
    \label{var_bound}
    \end{equation}
    The inequality in Eq.~\eqref{var_inequality} is satisfied for all $t \geq t^*$ if
    \begin{align}
        e^{-2 \gamma t} \sigma^2_x &\leq \frac{\epsilon'^2}{9} \implies t^* \geq \frac{1}{2 \gamma} \log \lb \frac{9 \sigma^2_x}{\epsilon'^2} \rb \label{critical_t1} \\
        2 t e^{-2 \gamma t}  | \mathrm{cov}_0 \lb x, r \rb| &\leq \frac{\epsilon'^2}{9} \implies t^* \geq - \frac{1}{2\gamma} \widetilde{W} \lb - \frac{\gamma \epsilon'^2}{9 | \mathrm{cov}_0 \lb x, r \rb|} \rb \label{critical_t2}\\
        t^2 e^{-2 \gamma t} \sigma^2_r &\leq \frac{\epsilon'^2}{9} \implies t^* \geq - \frac{1}{\gamma} \widetilde{W} \lb - \gamma \sqrt{\frac{\epsilon'^2}{9 \sigma^2_r}} \rb. \label{critical_t3}
    \end{align}
    
    Lastly, let us bound $\epsilon'$ in terms of $\epsilon$. Note that with probability at least $2/3$
    \begin{equation}
        |x'_t - x^*| \leq |x'_t - \ev{x}_t| + |\ev{x}_t - x^*| \leq 2 \epsilon'.
    \end{equation}
    In case of success we have that
    \begin{equation}
    \begin{split}
        |f(x_t') - f\lb x^* \rb| &= \left| \frac{a}{2} \lb x'_t - x^* \rb^2 + c - c \right| \\
        &= \frac{a}{2} \left| \lb x_t' - x^* \rb^2  \right| \\
        &= \frac{a}{2} \left|x'_t - x^* \right|^2 \\
        &\leq 2 a \epsilon'^2.
    \end{split}
    \end{equation}
    To ensure that $|f(x'_t) - f\lb x^* \rb| \leq \epsilon$ with probability at least $2/3$, it then suffices to choose $\epsilon' = \sqrt{\frac{\epsilon}{2 a}}$. Plugging this back into the expressions in Eqs.~\eqref{critical_t1_x}, \eqref{critical_t2_x}, \eqref{critical_t1}, \eqref{critical_t2} and \eqref{critical_t3} and taking the maximum yields the final bound on the $\epsilon$-equilibration time $t^*$.
\end{proof}

A similar bound can be obtained for the $\epsilon$-equilibration time of an underdamped harmonic oscillator. The details can be found in Appendix~\ref{app:underdamped}.

Let us now move to the multivariate setting.
The following lemma provides a bound on the $\epsilon$-equilibration time of a set of damped coupled harmonic oscillators all of which are critically damped or underdamped.

\begin{restatable}[Equilibration time of damped coupled harmonic oscillators]{lem}{Multivariate}
    Let $f(\bx) = \frac{1}{2}\lb \bx - \bx^* \rb^\top A \lb \bx - \bx^* \rb + c$ with $c \in \mathbb{R}$ constant and $A \in \mathbb{R}^{N \times N}$ positive definite with eigenvalues $0 < \lambda_{\min} :=  \lambda_0 \leq \lambda_1 \leq \cdots \leq \lambda_{N-1} =: \lambda_{\max}$. Further, let $\beta, m > 0$ be constants such that $\beta^2 - 4m\lambda_j \leq 0$ for all eigenvalues $\lambda_j$ of $A$. Given an initial phase space density $\rho_0(\bx, \bp)$, the $\epsilon$-equilibration time $t^*$ can be upper bounded as follows:
    \begin{equation}
        t^* \leq \left\{ t^{\ev{x}}_{1}, t^{\ev{x}}_{2},  t^{\sigma}_{1}, t^{\sigma}_{2}, t^{\sigma}_{3} \right\},
    \end{equation}
    where
    \begin{align}
        t^{\ev{x}}_{1} &:= \frac{1}{\gamma} \log \lb \sqrt{\frac{8\lambda_{\max}}{\epsilon}} \norm{\ev{\bx}_0 - \bx^*} \rb \\
        t^{\ev{x}}_{2} &:= - \frac{1}{\gamma} \widetilde{W} \lb - \sqrt{\frac{\epsilon}{8\lambda_{\max}}} \frac{\gamma}{ \norm{\ev{\br}_0}} \rb \\
        t^{\sigma}_{1} &:= \frac{1}{2 \gamma} \log \lb \frac{18 \lambda_{\max} \sum_{j=1}^N \sigma^2_{x,j}}{\epsilon} \rb \\
        t^{\sigma}_{2} &:= - \frac{1}{2 \gamma} \widetilde{W} \lb - \frac{\gamma \epsilon}{18 \lambda_{\max}  \left| \sum_{j=1}^N \mathrm{cov}\lb x_j, r_j \rb_0 \right|} \rb \\
        t^{\sigma}_{3} &:= - \frac{1}{\gamma} \widetilde{W} \lb - \gamma \sqrt{\frac{\epsilon}{18 \lambda_{\max} \sum_{j=1}^N \sigma^2_{r,j}}} \rb,
    \end{align}
    with $\gamma := \frac{\beta}{2m}$, $\br := \frac{\bp}{m} + \gamma \lb \bx - \bx^* \rb$, $\sigma^2_{x,j} :=  \ev{x_j^2}_0 - \ev{x_j}_0^2$, $\mathrm{cov} \lb x_j, r_j \rb_0 := \ev{x_j r_j}_0 - \ev{x_j}_0\ev{r_j}_0$ and $\sigma^2_{r,j} := \ev{r_j^2}_0 - \ev{r_j}_0^2$.
\label{lem:multivariate_time}
\end{restatable}
We prove this lemma by reducing the $N$-variable optimization problem to $N$ independent single-variable optimization problems. The details can be found in Appendix~\ref{app:multivariate_time}.

While the above bound on the $\epsilon$-equilibration time is quite tight, it is somewhat difficult to parse. The following definition will be useful for stating looser but easier to understand asymptotic bounds on the equilibration time.

\begin{defn}[Upper bound on initial parameters in the Liouvillian setting]
    Let $\gamma = \beta/2m$, let $\bx^* \in \mathbb{R}^N$ be the minimum of $f$ as in Lemma~\ref{lem:multivariate_time} and let $\rho_0(\bx, \bp)$ be an initial phase space probability density. Further, for all $j \in [N]$, let $\sigma^2_{x,j} = \ev{x_j^2}_0 - \ev{x_j}^2_0$, $\sigma^2_{p,j} = \ev{p_j^2}_0 - \ev{p_j}^2_0$ and $\mathrm{cov}_0\lb x_j, p_j \rb = \ev{x_j p_j}_0 - \ev{x_j}_0 \ev{p_j}_0$. Then we define
    \begin{equation}
        \widetilde{\chi}_0 := \max \lc \gamma^2 \norm{\ev{\bx}_0 - \bx^*}^2, \frac{\norm{\ev{\bp}_0}^2}{m^2}, \gamma^2 \sum_{j=1}^N \sigma^2_{x,j},  \frac{\gamma}{m} \left| \sum_{j=1}^N \mathrm{cov}\lb x_j, p_j \rb_0 \right|, \frac{1}{m^2}\sum_{j=1}^N \sigma^2_{p,j} \rc.
    \end{equation}
\label{def:initial_conditions_liouvillian}
\end{defn}

\begin{restatable}[Asymptotic bound on the $\epsilon$-equilibration time in the Liouvillian setting]{cor}{EquilibrationTime}
    Consider the same setting as in Lemma~\ref{lem:multivariate_time} and assume that $\gamma \in \Theta \lb \sqrt{\lambda_{\min}} \rb$.
    Then the $\epsilon$-equilibration time $t^*$ in Lemma~\ref{lem:multivariate_time} satisfies the following asymptotic bound:
    \begin{equation}
        t^* \in  O \lb \frac{1}{\sqrt{\lambda_{\min}}} \log \lb \frac{\lambda_{\max}}{\lambda_{\min}} \frac{ \widetilde{\chi}_0}{\epsilon} \rb \rb,
    \end{equation}
    where $ \widetilde{\chi}_0$ is as in Definition~\ref{def:initial_conditions_liouvillian}.
\label{cor:equilibration_time_liouvillian}
\end{restatable}
The proof of the above corollary can be found in Appendix~\ref{app:equilibration_time}.

Before stating the main theorem of this section, let us briefly fix some new notation and also recall some old notation.

\begin{defn}[Expectation values in the continuum and discrete setting; Liouvillian approach]
    Let $U_L(t) = \mathcal{T} e^{-i\int_0^t L_{\mathrm{fric}}(s) \mathrm{d}s}$ be the time evolution operator as defined in Eq.~\eqref{friction_liouvillian_ev}. Further, let $\Lt(t)$ be the discretized friction Hamiltonian associated with $L_{\mathrm{fric}}(t)$ according to Definition~\ref{def:discretized_liouvillian}.
    Then we define the following time evolution operator $\widetilde{U}_L(s): [0,t] \rightarrow \mathbb{C}^{2^{2Nn} \times 2^{2Nn}}$:
    \begin{equation}
        \widetilde{U}_L(s) := \mathcal{T} e^{-i \int_0^s \Lt(s') \mathrm{d}s'}.
    \end{equation}
    Next, let $\hat{x}_j$ denote the continuum position operator acting on the $j$-th position variable and let $\hat{\xt}_j$ denote the corresponding discretized and truncated position operator which is diagonal in the computational basis with diagonal values $x_{\max} - h_x, x_{\max} - 2h_x \dots, 0 , -h_x, \dots, - x_{\max}$. Each position (and momentum) variable is discretized over $2^n$ grid points implying that $\hat{\xt}_j \in \mathbb{C}^{2^n \times 2^n}$ for all $j \in [N]$. When clear from context, we simply write $\hat{\xt}_j$ to mean $\mathbb{1}_1 \otimes \mathbb{1}_2 \otimes \cdots \otimes \mathbb{1}_{j-1} \otimes \hat{\xt}_j \otimes \mathbb{1}_{j+1} \otimes \cdots \otimes \mathbb{1}_{N} \bigotimes_{k=1}^N \one_k$. 
    Now, let $\ket{\psi_0 \lb \bx,\bp \rb}$ be an initial Koopman-von Neumann wave function in the continuum and let $\ket{\widetilde{\psi}_0 \lb \bx,\bp \rb}$ denote the corresponding discretized KvN wave function obtained by evaluating $\psi_0 \lb \bx,\bp \rb$ on the $2^{2Nn}$ grid points associated with $\Lt(t)$ and normalizing the resulting state.
    Then we define the following quantities:
    \begin{equation}
        \ev{x_j}_t := \bra{\psi_0} U_L^\dagger(t) \hat{x}_j U_L(t)\ket{\psi_0}
    \end{equation}
    denotes the time evolved expectation value of the continuum position operator w.r.t.~the exact evolution in the continuum, 
    \begin{equation}
        \ev{\xt_j}_t := \bra{\widetilde{\psi}_0} \widetilde{U}_L^\dagger(t) \hat{\xt}_j \widetilde{U}_L(t)\ket{\widetilde{\psi}_0}
    \end{equation}
    denotes the time evolved expectation value of the discretized position operator w.r.t.~the exact discretized evolution operator,
    \begin{equation}
        \sigma_{j}^2(t) := \ev{\hat{x}_j^2}_t - \ev{\hat{x}_j}^2_t
    \end{equation}
    denotes the time evolved variance of the $j$-th continuum position variable and
    \begin{equation}
        \widetilde{\sigma}_{j}^2 (t) := \ev{\hat{\xt}_j^2}_t - \ev{\hat{\xt}_j}^2_t
    \end{equation}
    denotes the time evolved variance of the $j$-th position variable w.r.t.~the exact discretized evolution operator $\widetilde{U}_L(t)$.
\label{def:expectation_values_liouvillian}
\end{defn}

Now we are ready to state the main theorem of this section which provides an upper bound on the running time of the dissipative Liouvillian simulation algorithm for the task of finding the optimum of a convex quadratic function.

\begin{thm}[Coherent convex quadratic optimization via classical dynamics]
    Let $\epsilon > 0$ be an error tolerance. Let $A \in \mathbb{R}^{N \times N}$ be positive with eigenvalues $0 < \lambda_0 =: \lambda_{\min} \leq \lambda_1 \leq \dots \leq \lambda_{N-2} \leq \lambda_{N-1} =: \lambda_{\max}$ and assume that we know $\lambda_{\min}$ and $\lambda_{\max}$ each within some constant factor.
    Further, let $f(\bx) = \lb \bx- \bx^* \rb^\top A \lb \bx- \bx^* \rb + c$ for some constant $c \in \mathbb{R}$ and consider the corresponding discrete friction Liouvillian $\widetilde{L}$ as given in Definition~\ref{def:discretized_liouvillian}.
    Let $\ket{\psi_0 (\bx, \bp)}$ be an initial Koopman-von Neumann wave function and assume having access to a quantum state $\ket{\widetilde{\psi}_0 \lb \bx, \bp \rb} \in \mathbb{C}^{2^{2Nn}}$ encoding the discretized initial KvN wave function such that for all $0 \leq t \leq t^*$ with $t^*$ being the equilibration time from Lemma~\ref{lem:multivariate_time} it holds that
    \begin{align}
        \norm{\ev{\bxt}_t - \ev{\bx}_t} &\leq \sqrt{\frac{\epsilon}{18\lambda_{\max}}} \\
        \left| \sum_{j=1}^N \widetilde{\sigma}_{j}^2 (t) - \sum_{j=1}^N \sigma_{j}^2(t) \right| &\leq \frac{\epsilon}{24\lambda_{\max}}.
    \end{align}
    Further, assume having access to the partial derivatives of $f$ via the phase oracles $O_{f',k}^{(p)}$ given in Definition~\ref{def:phase_oracle_derivatives}. 
    Then there exists a quantum algorithm that can solve the Continuous Quantum Optimization Problem by finding an $\bx'$ such that $\left| f\lb \bx' \rb - f\lb \bx^* \rb \right| \leq \epsilon$ with probability at least $1 - \delta$ using 
    \begin{equation}
        \widetilde{O} \lb \frac{N^{3/2}}{\epsilon} \frac{\lambda_{\max}^2}{\lambda_{\min}^2} \alpha_{A_L} x_{\max} \norm{D_p}  \widetilde{\chi}_0 \log \lb 1/\delta \rb \rb.
    \end{equation}
    queries to $O_{f',k}^{(p)}$ and controlled-$O_{f',k}^{(p)}$ where $\alpha_{A_L}$ is given in Definition~\ref{def:discretized_liouvillian} and $\widetilde{\chi}_0$ is given in Definition~\ref{def:initial_conditions_liouvillian}.
\label{thm:optimization_liouvillian}
\end{thm}

\begin{proof}
    Lemma~\ref{lem:multivariate_time} provides an upper bound on the equilibration time $t^*$ of damped coupled classical harmonic oscillators in continuous space whose potential is given by $f(\bx)$ and whose equilibrium configuration corresponds to the minimum $\bx^*$ of $f(\bx)$. The main idea is to simulate the dynamics of such damped coupled harmonic oscillators, which are governed by $U_L(t) = \mathcal{T} e^{-i \int_0^t L_{\mathrm{fric}}(s) \mathrm{d}s}$, on a quantum computer and then use the fact that the probability distribution after time $t^*$ is strongly localized around the minimum at $\bx^*$ such that we only need to draw a small number of samples in order to determine $f(\bx^*)$ within error $\epsilon$.

    In order to simulate the dynamics of such a continuous system on a quantum computer, we need to discretize it. Let us consider the discrete Liouvillian $\widetilde{L}(t)$ as given in Definition~\ref{def:discretized_liouvillian} and let $\widetilde{U}_L(t) = \mathcal{T} e^{-i \int_0^t \Lt(s) \mathrm{d}s}$ denote the corresponding time evolution operator. Lemma~\ref{lem:liouvillian_friction} shows how to implement an $\widetilde{\epsilon}$-precise approximation $W(t)$ to $\widetilde{U}_L(t)$ using a total number of queries to all $O_{f', k}^{(p)}$ and controlled-$O_{f',k}^{(p)}$ that scales as
    \begin{equation}
        \widetilde{O} \lb \alpha_{A_L} N \frac{m}{\beta} e^{\beta t/m} f'_{\max} \norm{D_p} t \log^3 \lb \frac{1}{\widetilde{\epsilon}} \rb \rb.
    \end{equation}
    
    Now, let us consider the expectation value of the position operator in the continuum and in the discrete setting according to Definition~\ref{def:expectation_values_liouvillian} and let
    \begin{equation}
        \ev{\widetilde{\xt}_j}_t := \bra{\widetilde{\psi}_0} W^\dagger(t) \hat{\xt}_j W(t)\ket{\widetilde{\psi}_0}
    \end{equation}
    denote the time evolved expectation value of the discretized position operator for the $j$-th variable w.r.t.~the approximate discretized evolution operator $W(t)$.
    By the triangle inequality, we then have that 
    \begin{equation}
        \norm{\ev{\widetilde{\bxt}}_{t^*} - \bx^*} \leq \underbrace{\norm{ \ev{\widetilde{\bxt}}_{t^*} - \ev{\bxt}_{t^*}}}_{=: \epsilon_{\mathrm{sim}}} + \underbrace{\norm{\ev{\bxt}_{t^*} - \ev{\bx}_{t^*}}}_{=: \epsilon_{\mathrm{dis}}} + \underbrace{\norm{\ev{\bx}_{t^*} - \bx^*}}_{=: \epsilon_{\mathrm{eq}}}.
    \end{equation}
    Let $\epsilon_x$ be an error tolerance to be bounded later.
    In order for the above error to be at most $\epsilon_x$ such that $\norm{\ev{\widetilde{\bxt}}_{t^*} - \bx^*} \leq \epsilon_x$, it suffices to ensure that
    \begin{align}
        \epsilon_{\mathrm{sim}} &\leq \epsilon_x/3 \\
        \epsilon_{\mathrm{dis}} &\leq \epsilon_x/3 \\
        \epsilon_{\mathrm{eq}} &\leq \epsilon_x/3.
    \end{align}
    We also need to ensure that the concentration bounds used in Lemma~\ref{lem:multivariate_time} are still valid in the discrete setting since otherwise we might have to draw a lot of samples to obtain a good estimate of $\bx^*$.
    Let $\widetilde{\bxt}'_t \sim |\widetilde{\widetilde{\psi}}_t(\bx)|^2$ be a sample position vector from the time-evolved discrete phase space density associated with $\ket{\widetilde{\widetilde{\psi}}_t} := W(t) \ket{\widetilde{\psi}_0}$. According to the multivariate Chebyshev inequality the probability that $\widetilde{\bxt}'_t$ is far from the mean vector $\ev{\widetilde{\bxt}}_{t}$ is upper bounded as follows:
    \begin{equation}
         P\lb \norm{\widetilde{\bxt}'_t - \ev{\widetilde{\bxt}}_{t}} \geq \epsilon_x \rb \leq \frac{\sum_j \widetilde{\widetilde{\sigma}}_{j}^2 (t)}{\epsilon_x^2},
    \end{equation}
    where
    \begin{equation}
        \widetilde{\widetilde{\sigma}}_{j}^2 (t) := \ev{\widetilde{\widetilde{x}}_j^2}_t - \ev{\widetilde{\widetilde{x}}_j}^2_t
    \end{equation}
    denotes the time-evolved variance of the $j$-th position variable w.r.t.~the approximate discretized evolution operator $W(t)$.
    By the triangle inequality we then have that
    \begin{equation}
        \left| \sum_j \widetilde{\widetilde{\sigma}}_{j}^2 (t) - \sum_j \sigma_{j}^2(t) \right| \leq \underbrace{\left| \sum_j \widetilde{\widetilde{\sigma}}_{j}^2 (t) - \sum_j \widetilde{\sigma}_{j}^2(t) \right|}_{=: \epsilon'_{\mathrm{sim}}} + \underbrace{\left| \sum_j \widetilde{\sigma}_{j}^2 (t) - \sum_j \sigma_{j}^2(t) \right|}_{=: \epsilon'_{\mathrm{dis}}}.
    \end{equation}
    Now, as long as
    \begin{align}
        \sum_j \sigma_j^2(t) &\leq \frac{\epsilon_x^2}{6}, \\
        \epsilon'_{\mathrm{sim}} &\leq \frac{\epsilon_x^2}{12}, \\
        \epsilon'_{\mathrm{dis}} &\leq \frac{\epsilon_x^2}{12},
    \end{align}
    it holds that
    \begin{equation}
    \begin{split}
        P\lb \norm{\widetilde{\bxt}'_t - \ev{\widetilde{\bxt}}_{t}} \geq \epsilon_x \rb &\leq \frac{\sum_j \widetilde{\widetilde{\sigma}}_{j}^2 (t)}{\epsilon_x^2} \\
        &\leq \frac{1}{\epsilon_x^2} \lb \sum_j \sigma_j^2(t) + \epsilon'_{\mathrm{sim}} + \epsilon'_{\mathrm{dis}} \rb \\
        &\leq \frac{1}{3}.
    \end{split}
    \end{equation}
    This means that with probability at least $2/3$ we have that
    \begin{equation}
        \norm{\widetilde{\bxt}'_t - \bx^*} \leq \norm{\widetilde{\bxt}'_t - \ev{\widetilde{\bxt}}_{t}} + \norm{\ev{\widetilde{\bxt}}_{t} - \bx^*} \leq 2\epsilon_x.
    \end{equation}
    Drawing $O \lb \log \lb 1/\delta \rb \rb$ samples and taking the sample that leads to the smallest value of $f$ allows us to boost the success probability to at least $1 - \delta$. Specifically, the probability that not a single sample out of $v$ many samples is $\epsilon_x$-close to $\ev{\widetilde{\bxt}}_{t}$ is at most $\lb 1/3 \rb^v$. Thus, $1/{3^v} \leq \delta$ if $v \geq \log \lb 1/\delta \rb/\log \lb 3\rb$.

    By the Cauchy-Schwarz inequality we then have that
    \begin{equation}
    \begin{split}
        \left| f \lb \widetilde{\bxt}'_{t} \rb - f(\bx^*) \right| &= \left| \frac{1}{2}\lb \widetilde{\bxt}'_{t} - \bx^* \rb^\top A \lb \widetilde{\bxt}'_{t} - \bx^* \rb \right| \\
        &\leq \frac{1}{2} \norm{\widetilde{\bxt}'_{t} - \bx^*} \norm{A \lb \widetilde{\bxt}'_{t} - \bx^* \rb} \\
        &\leq \frac{1}{2} \norm{A} \norm{\widetilde{\bxt}'_{t} - \bx^*}^2 \\ 
        &\leq 2 \norm{A} {\epsilon_x}^2 = 2 \lambda_{\max} {\epsilon_x}^2.
    \end{split}
    \label{function_error_bound_liouville}
    \end{equation}
    Hence, in order to ensure that $\left| f \lb \widetilde{\bxt}'_{t} \rb - f(\bx^*) \right| \leq \epsilon$, it suffices to choose $\epsilon_x = \sqrt{\frac{\epsilon}{2\lambda_{\max}}}$.

    Let us now discuss how to achieve the various error bounds.
    First, according to Corollary~\ref{cor:equilibration_time_liouvillian}, we require
    \begin{equation}
        t^* \in  O \lb \frac{1}{\sqrt{\lambda_{\min}}} \log \lb \frac{\lambda_{\max}}{\lambda_{\min}} \frac{ \widetilde{\chi}_0}{\epsilon_x} \rb \rb \subseteq O \lb \frac{1}{\sqrt{\lambda_{\min}}} \log \lb \frac{\lambda_{\max}}{\lambda_{\min}} \frac{ \widetilde{\chi}_0}{\epsilon} \rb \rb,
    \end{equation}
    in order for $\epsilon_{\mathrm{eq}} \leq \epsilon_x/3$ and $\sum_j \sigma_j^2(t) \leq \epsilon_x^2/6$. Note that in the above bound we picked the friction coefficient $\beta$ such that $\beta \in \Theta \lb \sqrt{\lambda_{\min}} \rb$. 

    The conditions on the discretization errors, $\epsilon_{\mathrm{dis}} \leq \epsilon_x/3 \leq \sqrt{\frac{\epsilon}{18\lambda_{\max}}}$ and $\epsilon_{\mathrm{dis}}' \leq \epsilon_x^2/12 \leq \frac{\epsilon}{24\lambda_{\max}}$ are true by assumption. Implicitly, this requires us to choose sufficiently small grid spacings $h_x$ and $h_p$ for the finite difference approximations of the discretized partial derivatives $\partial_{x,j}$ and $\partial_{p,j}$.
   
    Next, let us discuss how to ensure that $\epsilon_{\mathrm{sim}} \leq \epsilon_x/3$ and $\epsilon'_{\mathrm{sim}} \leq \epsilon_x^2/12$. Let $\ket{\widetilde{\psi}_t} := \widetilde{U}_L(t)\ket{\widetilde{\psi}_0}$ denote the time evolved quantum state w.r.t.~the exact discretized evolution operator and let $\ket{\widetilde{\widetilde{\psi}}_t} := W(t)\ket{\widetilde{\psi}_0}$ denote the time evolved quantum state w.r.t.~the $\widetilde{\epsilon}$-precise approximate discretized evolution operator. Then we have that
    \begin{equation}
        \norm{\ket{\widetilde{\widetilde{\psi}}_{t}} - \ket{\widetilde{\psi}_t}} \leq \widetilde{\epsilon}.
    \end{equation}
    This implies that
    \begin{equation}
    \begin{split}
        \norm{\ev{\widetilde{\bxt}}_{t} - \ev{\bxt}_{t}}^2 &= \sum_{j=1}^N \left| \bra{\widetilde{\widetilde{\psi}}_t}\hat{\xt}_j\ket{\widetilde{\widetilde{\psi}}_t} - \bra{\widetilde{\psi}_t}\hat{\xt}_j\ket{\widetilde{\psi}_t} \right|^2 \\
        &\leq \sum_{j=1}^N \left| \bra{\widetilde{\widetilde{\psi}}_t}\hat{\xt}_j\ket{\widetilde{\widetilde{\psi}}_t} - \bra{\widetilde{\widetilde{\psi}}_t}\hat{\xt}_j\ket{\widetilde{\psi}_t} \right|^2 + \sum_{j=1}^N \left| \bra{\widetilde{\widetilde{\psi}}_t}\hat{\xt}_j\ket{\widetilde{\psi}_t} - \bra{\widetilde{\psi}_t}\hat{\xt}_j\ket{\widetilde{\psi}_t} \right|^2 \\
        &\leq 2N x_{\max} \norm{\ket{\widetilde{\widetilde{\psi}}_{t}} - \ket{\widetilde{\psi}_t}}^2 \\
        &\leq 2N x_{\max} \widetilde{\epsilon}^2,
    \end{split}
    \end{equation}
    where we used the Cauchy-Schwarz inequality in going from the second to the third line. 
    Therefore, in order for $\epsilon_{\mathrm{sim}}$ to be at most $\epsilon_x/3$, it suffices to ensure that
    \begin{equation}
        \widetilde{\epsilon} \leq \frac{\epsilon_x/3}{\sqrt{2N x_{\max}}} \leq \frac{1}{6}\sqrt{\frac{\epsilon}{N x_{\max} \lambda_{\max}}}.
    \end{equation}
    Furthermore,
    \begin{equation}
    \begin{split}
        \left| \sum_j \lb \ev{\widetilde{\xt}_j^2}_t - \ev{\xt_j^2}_t \rb \right| &= \left| \sum_j \lb \bra{\widetilde{\widetilde{\psi}}_t}\hat{\xt}^2_j\ket{\widetilde{\widetilde{\psi}}_t} - \bra{\widetilde{\psi}_t}\hat{\xt}^2_j\ket{\widetilde{\psi}_t} \rb \right| \\
        &\leq  \left| \sum_j \lb \bra{\widetilde{\widetilde{\psi}}_t}\hat{\xt}^2_j\ket{\widetilde{\widetilde{\psi}}_t} - \bra{\widetilde{\widetilde{\psi}}_t}\hat{\xt}^2_j\ket{\widetilde{\psi}_t} \rb \right| + \left| \sum_j \lb \bra{\widetilde{\widetilde{\psi}}_t}\hat{\xt}^2_j\ket{\widetilde{\psi}_t} - \bra{\widetilde{\psi}_t}\hat{\xt}^2_j\ket{\widetilde{\psi}_t} \rb \right| \\
        &\leq \left|\bra{\widetilde{\widetilde{\psi}}_t} \sum_j \hat{\xt}^2_j\lb \ket{\widetilde{\widetilde{\psi}}_t} - \ket{\widetilde{\psi}_t} \rb \right| + \left|\lb \bra{\widetilde{\widetilde{\psi}}_t} - \bra{\widetilde{\psi}_t} \rb \sum_j \hat{\xt}^2_j \ket{\widetilde{\psi}_t} \right| \\
        &\leq 2 N x_{\max}^2 \norm{\ket{\widetilde{\widetilde{\psi}}_t} - \ket{\widetilde{\psi}_t}} \\
        &\leq 2 N x_{\max}^2 \widetilde{\epsilon}.
    \end{split}
    \end{equation}
    Similarly,
    \begin{equation}
    \begin{split}
        \left| \sum_j \lb \ev{\widetilde{\xt}_j}^2_t - \ev{\xt_j}^2_t \rb \right| &= \left| \sum_j \lb \ev{\widetilde{\xt}_j}_t - \ev{\xt_j}_t \rb \lb \ev{\widetilde{\xt}_j}_t + \ev{\xt_j}_t \rb  \right| \\
        &\leq 2 x_{\max} \left| \sum_j \lb \ev{\widetilde{\xt}_j}_t - \ev{\xt_j}_t \rb \right| \\
        &\leq 4 N x_{\max}^2 \widetilde{\epsilon}.
    \end{split}
    \end{equation}
    Therefore, 
    \begin{equation}
    \begin{split}
        \left| \sum_j \widetilde{\widetilde{\sigma}}_{j}^2 (t) - \sum_j \widetilde{\sigma}_{j}^2(t) \right| &= \left| \sum_j \lb \ev{\widetilde{\xt}_j^2}_t - \ev{\widetilde{\xt}_j}^2_t \rb - \sum_j \lb \ev{\xt_j^2}_t - \ev{\xt_j}^2_t \rb \right| \\
        &\leq \left| \sum_j \lb \ev{\widetilde{\xt}_j^2}_t - \ev{\xt_j^2}_t \rb \right| + \left| \sum_j \lb \ev{\widetilde{\xt}_j}^2_t - \ev{\xt_j}^2_t \rb \right| \\
        &\leq 6 N x_{\max}^2 \widetilde{\epsilon}.
    \end{split}
    \end{equation}
    This shows that in order for $\epsilon_{\mathrm{sim}'}$ to be at most $\epsilon_x^2/12$
    it suffices to ensure that
    \begin{equation}
        \widetilde{\epsilon} \leq \frac{\epsilon_x^2/12}{6Nx_{\max}^2} \leq \frac{\epsilon}{144 \lambda_{\max} N x_{\max}^2}.
    \end{equation}

    Putting everything together, we thus require a total of
    \begin{equation}
        \widetilde{O} \lb \alpha_{A_L} N \frac{m}{\beta} e^{\beta t^*/m} f'_{\max} \norm{D_p} t^* \log^3 \lb \frac{1}{\widetilde{\epsilon}} \rb \log \lb \frac{1}{\delta} \rb \rb 
    \label{queries_phase_oracles}
    \end{equation}
    queries to the phase oracles $O_{f',k}^{(p)}$ and their controlled versions in order to find an $\bx'$ such that $\left| f\lb \bx' \rb - f\lb \bx^* \rb \right| \leq \epsilon$ with probability at least $1 - \delta$. The above expression can be simplified by using the following identity and upper bound on the $-1$ branch of the Lambert $W$ function:
    \begin{equation}
        e^{-W_{-1}(y)} = \frac{W_{-1}(y)}{y} \in O \lb \frac{\log(-y)}{y} \rb.
    \end{equation}
    Specifically, we have that
    \begin{equation}
        e^{\beta t^*/m} \in \widetilde{O} \lb \frac{\lambda_{\max} \widetilde{\chi}_0}{ \lambda_{\min} \epsilon} \rb.
    \end{equation}
    Additionally, note that the upper bound on the partial derivatives of $f$ obeys
    \begin{equation}
        f'_{\max} \in O \lb \sqrt{N} \lambda_{\max} x_{\max} \rb,
    \end{equation}
    which follows from the following matrix norm inequality: $\norm{A}_1 := \max_k \sum_{l = 1}^N |A_{kl}| \leq \sqrt{N} \norm{A}_2 = \sqrt{N} \lambda_{\max}$.
    Setting $m=1$ and recalling that $\beta \in \Theta \lb \sqrt{\lambda_{\min}} \rb$ and $t^* \in \widetilde{O} \lb 1/\sqrt{\lambda_{\min}} \rb$ then allows us to upper bound the query complexity in Eq.~\eqref{queries_phase_oracles} as follows:
    \begin{equation}
        \widetilde{O} \lb \frac{N^{3/2}}{\epsilon} \frac{\lambda_{\max}^2}{\lambda_{\min}^2} \alpha_{A_L} x_{\max} \norm{D_p}  \widetilde{\chi}_0 \log \lb 1/\delta \rb \rb.
    \end{equation}
    This completes the proof.
\end{proof}

\subsection{Quantum Hamiltonian Approach}
\label{sec:quantum_ham_opt}

Let us now discuss how simulating the dynamics governed by the quantum friction Hamiltonian given in Definition~\ref{def:friction_ham} allows us to efficiently find the optimum of a convex quadratic function.
Before presenting our results, let us briefly fix some notation. Given an initial wave function $\psi_0(x)$ in real space, we write $\ev{B}_t$ to denote the expectation value of some operator $B$ w.r.t.~the time-evolved wave function $\psi_t(x)$ after time $t$. In particular, $\ev{B}_0 = \bra{\psi_0}B\ket{\psi_0}$.

The first result below provides a bound on the $\epsilon$-equilibration time of a single critically or underdamped quantum harmonic oscillator.

\begin{lem}[Equilibration time of a damped quantum harmonic oscillator]
    Consider the following time-dependent quantum Hamiltonian in continuous space:
    \begin{equation}
        H(t) := e^{- \beta t/m} \frac{p^2}{2m} + e^{\beta t/m} \frac{a}{2}(x-x^*)^2 + c,
    \end{equation}
    where $a, m > 0$ and $x^*, c \in \mathbb{R}$ are constants and assume that $\beta^2 \leq 4ma$. Given an initial wave function $\psi_0(x)$ the $\epsilon$-equilibration time $t^*$ can be upper bounded as follows:
    \begin{equation}
        t^* \leq \max \left\{ t^{\ev{x}}_1, t^{\ev{x}}_2,  t^{\sigma}_1, t^{\sigma}_2, t^{\sigma}_3 \right\},
    \end{equation}
    where 
    \begin{align}
        t^{\ev{x}}_1 &:= \frac{1}{\gamma} \log \lb \sqrt{\frac{8a}{\epsilon}} |\ev{x}_0 - x^*| \rb \\
        t^{\ev{x}}_2 &:= -\frac{1}{\gamma} \widetilde{W} \lb - \sqrt{\frac{\epsilon}{8a}} \frac{\gamma}{|\ev{r}_0|} \rb \\
        t^{\sigma}_1 &:= \frac{1}{2 \gamma} \log \lb \frac{18 a \sigma^2_x}{\epsilon} \rb \\
        t^{\sigma}_2 &:= - \frac{1}{2 \gamma} \widetilde{W} \lb - \frac{\gamma \epsilon}{18 a | \mathrm{cov}_0 \lb x, r \rb|} \rb \\
        t^{\sigma}_3 &:= - \frac{1}{\gamma} \widetilde{W} \lb - \gamma \sqrt{\frac{\epsilon}{18 a \sigma^2_{r}}} \rb,
    \end{align}
    with $\gamma := \frac{\beta}{2m}$, $r := \frac{p}{m} + \gamma (x - x^*)$, $\sigma_x^2 := \ev{x^2}_0 - \ev{x}_0^2$, $\widetilde{\mathrm{cov}} \lb x, r \rb_0 := \frac{1}{2} \lb \ev{x  r}_0 + \ev{r x}_0 - 2 \ev{x}_0\ev{r}_0 \rb$ and $\sigma^2_{r} := \ev{r^2}_0 - \ev{r}_0^2$.
\label{lem:damped_QHO}
\end{lem}

\begin{proof}
    The main idea is to show that the equations of motion for $\ev{x}_t$ and $\sigma^2_x(t) = \ev{x^2}_t - \ev{x}^2_t$ are essentially equivalent to the equations of motion of a damped classical harmonic oscillator. This then allows us to use the results on the equilibration time from Lemma~\ref{lem:underdamped_time_bound}.
    In order to simplify the notation, let us consider a coordinate system where all positions are shifted by $x^*$, i.e. $\xt := x - x^*$, such that the minimum is at $\xt_{\min} = 0$. Note that the momentum operator $p$ remains unaffected since $\partial_x = \partial_{\xt}$.
    The Hamiltonian in this shifted coordinate system thus reads
    \begin{equation}
        H(t) = e^{- \beta t/m} \frac{p^2}{2m} + e^{\beta t/m} \frac{a}{2} \xt^2 + c.
    \end{equation}
    Now, recall that for any observable $B$ it holds that
    \begin{equation}
        \frac{d \ev{B}}{dt} = -i \ev{[B,H]} + \ev{\frac{\partial B}{\partial t}}.
    \end{equation}
    Therefore,
    \begin{equation}
        \frac{d \ev{\xt}}{dt} = -i \ev{[\xt,H]} = -i \ev{\left[ \xt, \frac{ e^{-\beta t/m} p^2}{2m} \right]} = - \frac{i}{2m}  e^{-\beta t/m} \ev{\left[\xt, p^2\right]}
        = e^{-\beta t/m} \frac{\ev{p}}{m},
        \label{dxdt}
    \end{equation}
    where we used the fact that $\ev{\left[\xt, p^2 \right]} = 2i \ev{p}$.
    Additionally,
    \begin{equation}
        \frac{d \ev{p}}{dt} = -i \ev{[p,H]} = -i \ev{\left[ p, e^{\beta t/m} \frac{a}{2} \xt^2 \right]} = - a e^{\beta t/m} \ev{\xt}.
    \end{equation}
    This implies that
    \begin{equation}
         \frac{d^2 \ev{\xt}}{dt^2} = \frac{d}{dt} \lb e^{-\beta t/m} \frac{\ev{p}}{m} \rb = 
         -\frac{\beta}{m} \frac{d \ev{\xt}}{dt} - \frac{a}{m} \ev{\xt},
    \label{eom_ev_xt}
    \end{equation}
    which is exactly the same equation of motion as that of a damped harmonic oscillator. The initial conditions can be determined from Eq.~\eqref{dxdt} and Eq.~\eqref{eom_ev_xt} and are identical to the classical case. Hence,
    \begin{equation}
        \ev{\xt}_t = e^{-\gamma t} \lb \ev{\xt}_0  \cos \lb \omega t \rb + \ev{r}_0 \frac{\sin \lb \omega t \rb}{\omega} \rb,
    \end{equation}
    where $\omega := \frac{\sqrt{4ma - \beta^2}}{2m}$. This shows that the position expectation value approaches the minimum exponentially fast.
    As in the classical case, it is not sufficient for our purposes to show that the average position is close to the minimum as we might have to obtain a lot of samples to compute the average position. Ideally, we only need a small number of high-quality samples. This can be guaranteed by picking a large enough $t$ such that the position variance at time $t$, $\sigma^2_x(t) = \ev{x^2}_t - \ev{x}_t^2$, is sufficiently small.
    In order to show that the position variance decreases exponentially with time as well, let us consider the equations of motion for $\ev{\xt^2}$ and $\ev{p^2}$.
    Direct computation reveals that
    \begin{align}
        \frac{d \ev{\xt^2}}{dt} &= -i \ev{[\xt^2,H]} = - \frac{i}{2m} e^{-\beta t/m} \ev{\left[\xt^2, p^2\right]} = \frac{1}{m} e^{-\beta t/m} \ev{\left\{ \xt, p \right\}} \\
        \frac{d \ev{p^2}}{dt} &= -i \ev{[p^2,H]} = - \frac{i a}{2} e^{\beta t/m} \ev{\left[ p^2, \xt^2 \right]} = -a e^{\beta t/m} \ev{\left\{ \xt, p \right\}}
    \end{align}
    where we used the fact that $\left[ \xt^2, p^2 \right] = 2i \left\{ \xt, p \right\}$.
    Then
    \begin{equation}
    \begin{split}
         \frac{d^2 \ev{\xt^2}}{dt^2} &= \frac{d}{dt} \lb \frac{d \ev{\xt^2}}{dt} \rb \\
         &= -\frac{\beta}{m} \frac{d \ev{\xt^2}}{dt} + \frac{1}{m} e^{-\beta t/m} \frac{d \ev{\left\{ \xt, p \right\}}}{dt} \\
         &= -\frac{\beta}{m} \frac{d \ev{\xt^2}}{dt} - \frac{i}{m} e^{-\beta t/m}  \ev{\left[\left\{ \xt, p \right\}, H \right]} \\
        \frac{d^2 \ev{\xt^2}}{dt^2} &= -\frac{\beta}{m} \frac{d \ev{\xt^2}}{dt} - \frac{2a}{m} \ev{\xt^2} + \frac{2}{m^2} e^{-2\beta t/m} \ev{p^2}
    \end{split}
    \label{d2x2dt2}
    \end{equation}
    where we used the fact that
    \begin{equation}
    \begin{split}
        \ls \lc \xt, p \rc, H \rs &= \ls \xt p, H \rs + \ls p \xt, H \rs \\
        &= \xt \ls p, H \rs + \ls \xt, H \rs p + p \ls \xt, H \rs + \ls p, H \rs \xt \\
        &= -2i a e^{\beta t/m} \xt^2 + 2i e^{-\beta t/m} \frac{p^2}{m}.
    \end{split}
    \end{equation}
    Similarly, 
    \begin{equation}
    \begin{split}
         \frac{d^2 \ev{p^2}}{dt^2} &= \frac{d}{dt} \lb \frac{d \ev{p^2}}{dt} \rb \\
         &= \frac{\beta}{m} \frac{d \ev{p^2}}{dt} - a e^{\beta t/m} \frac{d \ev{\left\{ \xt, p \right\}}}{dt} \\
         &= \frac{\beta}{m} \frac{d \ev{p^2}}{dt} + i a e^{\beta t/m} \ev{\left[\left\{ \xt, p \right\}, H \right]} \\
          \frac{d^2 \ev{p^2}}{dt^2}&= \frac{\beta}{m} \frac{d \ev{p^2}}{dt} + 2 a^2 e^{2 \beta t/m} \ev{\xt^2} - \frac{2a}{m}\ev{p^2}.p
    \end{split}
    \label{d2p2dt2}
    \end{equation}
    
    The set of coupled ODEs given by Eq.~\eqref{d2x2dt2} and Eq.~\eqref{d2p2dt2} is exactly the same as in the classical Liouvillian setting (not explicitly shown in this paper). The only difference appears in the initial conditions due to the noncommutativity of $\xt$ and $p$ in the quantum setting.
    While we weren't able to solve the coupled ODEs in the classical Liouvillian setting directly, we found an indirect approach to obtain an exact expression for $\ev{x^2}_t$ based on $\delta$-distributions, see Section~\ref{sec:liouvillian_approach} for details. The same logic can be applied to find an exact expression for $\ev{p^2}_t$. Specifically, we have that
    \begin{align}
        \ev{\xt^2}_t &= e^{-2 \gamma t} \lb \ev{\xt^2}_0 \cos^2 \lb \omega t \rb + \lb \ev{\xt r}_0 + \ev{r \xt}_0 \rb \cos \lb \omega t \rb  \frac{\sin \lb \omega t \rb}{\omega} + \ev{r^2}_0 \frac{\sin^2 \lb \omega t \rb}{\omega^2}  \rb  \\
        \ev{p^2}_t &= e^{2 \gamma t} \lb \ev{p^2}_0 \cos^2 \lb \omega t \rb - \lb \ev{p s}_0 + \ev{s p}_0 \rb \cos \lb \omega t \rb  \frac{\sin \lb \omega t \rb}{\omega} + \ev{s^2}_0 \frac{\sin^2 \lb \omega t \rb}{\omega^2}  \rb,
    \end{align}
    where
    \begin{align}
        r &= \frac{p}{m} + \gamma \xt, \\
        s &= a \xt + \gamma p.
    \end{align}
    By plugging the above solutions back into Eq.~\eqref{d2x2dt2} and Eq.~\eqref{d2p2dt2} it can be verified that they are indeed solutions to the coupled ODEs. If we treat $x$ and $p$ as classical variables such that $[x,p] = 0$, we recover exactly the solutions obtained in the Liouvillian setting. This allows us to use the bounds on the equilibration time from Lemma~\ref{lem:critical_time_bound} and Lemma~\ref{lem:underdamped_time_bound} as long as we slightly adjust the definition of $t_2^\sigma$ to account for the noncommutativity of $x$ and $p$. Specifically, $t_2^\sigma$ depends on $\widetilde{\mathrm{cov}} \lb x, r \rb_0 = \frac{1}{2} \lb \ev{x  r}_0 + \ev{r x}_0 - 2 \ev{x}_0\ev{r}_0 \rb$ in both the classical and the quantum case but only in the classical case do we have that $\widetilde{\mathrm{cov}} \lb x, r \rb_0 = \mathrm{cov} \lb x, r \rb_0 = \ev{x  r}_0 - \ev{x}_0\ev{r}_0$.

\end{proof}

Let us now move to the multivariate setting. Below we provide a bound on the $\epsilon$-equilibration time of damped coupled quantum harmonic oscillators which is essentially the same bound as for classical harmonic oscillators.

\begin{lem}[Equilibration time of damped coupled quantum harmonic oscillators]
    Consider the following time-dependent quantum Hamiltonian in continuous space:
    \begin{equation}
        H(t) := e^{- \beta t/m} \sum_{j=1}^N \frac{p_j^2}{2m} + e^{\beta t/m} \lb  \lb \bx- \bx^* \rb^\top A \lb \bx- \bx^* \rb + c \rb,
    \end{equation}
    where $A \in \mathbb{R}^{N \times N}$ is positive definite with eigenvalues $0 < \lambda_0 =: \lambda_{\min} \leq \lambda_1 \leq \dots \leq \lambda_{N-2} \leq \lambda_{N-1} =: \lambda_{\max}$ and $c \in \mathbb{R}$ is some constant. Assume that $\beta^2 \leq 4m \lambda_{\min}$. Given an initial wave function $\psi_0 \lb \bx \rb$ the $\epsilon$-equilibration time $t^*$ can be upper bounded as follows:
     \begin{equation}
        t^* \leq \max \left\{ t^{\ev{x}}_{1}, t^{\ev{x}}_{2},  t^{\sigma}_{1}, t^{\sigma}_{2}, t^{\sigma}_{3} \right\},
    \end{equation}
    where
    \begin{align}
        t^{\ev{x}}_{1} &:= \frac{1}{\gamma} \log \lb \sqrt{\frac{8\lambda_{\max}}{\epsilon}} \norm{\ev{\bx}_0 - \bx^*} \rb \\
        t^{\ev{x}}_{2} &:= - \frac{1}{\gamma} \widetilde{W} \lb - \sqrt{\frac{\epsilon}{8\lambda_{\max}}} \frac{\gamma}{ \norm{\ev{\br}_0}} \rb \\
        t^{\sigma}_{1} &:= \frac{1}{2 \gamma} \log \lb \frac{18 \lambda_{\max} \sum_{j=1}^N \sigma^2_{x,j}}{\epsilon} \rb \\
        t^{\sigma}_{2} &:= - \frac{1}{2 \gamma} \widetilde{W} \lb - \frac{\gamma \epsilon}{18 \lambda_{\max}  \left| \sum_{j=1}^N \widetilde{\mathrm{cov}}\lb x_j, r_j \rb_0 \right|} \rb \\
        t^{\sigma}_{3} &:= - \frac{1}{\gamma} \widetilde{W} \lb - \gamma \sqrt{\frac{\epsilon}{18 \lambda_{\max} \sum_{j=1}^N \sigma^2_{r,j}}} \rb,
    \end{align}
    with $\gamma := \frac{\beta}{2m}$, $\br := \frac{\bp}{m} + \gamma \lb \bx - \bx^* \rb$, $\sigma^2_{x,j} :=  \ev{x_j^2}_0 - \ev{x_j}_0^2$,  $\widetilde{\mathrm{cov}} \lb x_j, r_j \rb_0 := \frac{1}{2} \lb \ev{x  r}_0 + \ev{r x}_0 - 2 \ev{x}_0\ev{r}_0 \rb$ and $\sigma^2_{r,j} := \ev{r_j^2}_0 - \ev{r_j}_0^2$.
\label{lem:multivariate_quantum}
\end{lem}

\begin{proof}
    Follows directly from Lemma~\ref{lem:damped_QHO} and Lemma~\ref{lem:multivariate_time}.
\end{proof}

The above bound on the $\epsilon$-equilibration time depends on the expectation values of various combinations of the initial positions and momenta. In order to simplify the complexity bounds on the running time of our algorithms, we define the following quantity:

\begin{defn}[Upper bound on initial parameters in the quantum setting]
    Let $\gamma = \beta/2m$, let $\bx^* \in \mathbb{R}^N$ be the minimum of $f$ as in Lemma~\ref{lem:multivariate_quantum} and let $\psi_0(\bx)$ be an initial wave function. Further, for all $j \in [N]$, let $\sigma^2_{x,j} = \ev{x_j^2}_0 - \ev{x_j}^2_0$, $\sigma^2_{p,j} = \ev{p_j^2}_0 - \ev{p_j}^2_0$ and $\widetilde{\mathrm{cov}}_0\lb x_j, p_j \rb = \frac{1}{2} \lb \ev{x_j p_j}_0 + \ev{p_j x_j}_0 - 2\ev{x_j}_0 \ev{p_j}_0 \rb$ as in Lemma~\ref{lem:multivariate_quantum}. Then we define
    \begin{equation}
        \chi_0 := \max \lc \gamma^2 \norm{\ev{\bx}_0 - \bx^*}^2, \frac{\norm{\ev{\bp}_0}^2}{m^2}, \gamma^2 \sum_{j=1}^N \sigma^2_{x,j},  \frac{\gamma}{m} \left| \sum_{j=1}^N \widetilde{\mathrm{cov}}\lb x_j, p_j \rb_0 \right|, \frac{1}{m^2}\sum_{j=1}^N \sigma^2_{p,j} \rc.
    \end{equation}
\label{def:initial_conditions}
\end{defn}

The corollary below provides a simpler asymptotic bound on the $\epsilon$-equilibration time of damped coupled quantum harmonic oscillators.

\begin{cor}[Asymptotic bound on the $\epsilon$-equilibration time in the quantum setting]
    Consider the same setting as in Lemma~\ref{lem:multivariate_quantum} and assume that $\gamma \in \Theta \lb \sqrt{\lambda_{\min}} \rb$.
    Then the $\epsilon$-equilibration time $t^*$ in Lemma~\ref{lem:multivariate_quantum} satisfies the following asymptotic bound:
    \begin{equation}
        t^* \in  O \lb \frac{1}{\sqrt{\lambda_{\min}}} \log \lb \frac{\lambda_{\max}}{\lambda_{\min}} \frac{\chi_0}{\epsilon} \rb \rb,
    \end{equation}
    where $\chi_0$ is as in Definition~\ref{def:initial_conditions}.
\label{cor:equilibration_time}
\end{cor}

\begin{proof}
    Essentially the same as the proof of Corollary~\ref{cor:equilibration_time_liouvillian}.
\end{proof}

Before stating the main theorem of this section, let us briefly fix some new notation and also recall some old notation.

\begin{defn}[Expectation values in the continuum and discrete setting]
    Let $U_H(t) = \mathcal{T} e^{-i\int_0^t H_{\mathrm{fric}}(s) \mathrm{d}s}$ be the time evolution operator as defined in Eq.~\eqref{friction_ham_ev}. Further, let $\Ht(t)$ be the discretized friction Hamiltonian associated with $H_{\mathrm{fric}}(t)$ according to Definition~\ref{def:discretized_ham}.
    Then we define the following discretized time evolution operator $\widetilde{U}_H(s): [0,t] \rightarrow \mathbb{C}^{2^{Nn} \times 2^{Nn}}$:
    \begin{equation}
        \widetilde{U}_H(s) := \mathcal{T} e^{-i \int_0^s \Ht(s') \mathrm{d}s'}.
    \end{equation}
    Next, let $\hat{x}_j$ denote the continuum position operator acting on the $j$-th variable and let $\hat{\xt}_j$ denote the corresponding discretized and truncated position operator which is diagonal in the computational basis with diagonal values $x_{\max} - h_x, x_{\max} - 2h_x \dots, 0 , -h_x, \dots, - x_{\max}$. Each position variable is discretized over $2^n$ grid points implying that $\hat{\xt}_j \in \mathbb{C}^{2^n \times 2^n}$ for all $j \in [N]$. When clear from context, we simply write $\hat{\xt}_j$ to mean $\mathbb{1}_1 \otimes \mathbb{1}_2 \otimes \cdots \otimes \mathbb{1}_{j-1} \otimes \hat{\xt}_j \otimes \mathbb{1}_{j+1} \otimes \cdots \otimes \mathbb{1}_{N}$. 
    Now, let $\ket{\psi_0 (\bx)}$ be a quantum state on the Hilbert space associated with $H_{\mathrm{fric}}(t)$ and let $\ket{\widetilde{\psi}_0(\bx)}$ denote the corresponding quantum state on the Hilbert space of $\Ht(t)$ obtained by evaluating $\ket{\psi_0(\bx)}$ on the $2^{Nn}$ grid points associated with $\Ht(t)$.
    Then we define the following quantities:
    \begin{equation}
        \ev{x_j}_t := \bra{\psi_0} U_H^\dagger(t) \hat{x}_j U_H(t)\ket{\psi_0}
    \end{equation}
    denotes the time evolved expectation value of the continuum position operator w.r.t.~the exact evolution in the continuum, 
    \begin{equation}
        \ev{\xt_j}_t := \bra{\widetilde{\psi}_0} \widetilde{U}_H^\dagger(t) \hat{\xt}_j \widetilde{U}_H(t)\ket{\widetilde{\psi}_0}
    \end{equation}
    denotes the time evolved expectation value of the discretized position operator w.r.t.~the exact discretized evolution operator,
    \begin{equation}
        \sigma_{j}^2(t) := \ev{\hat{x}_j^2}_t - \ev{\hat{x}_j}^2_t
    \end{equation}
    denotes the time evolved variance of the $j$-th continuum position variable and
    \begin{equation}
        \widetilde{\sigma}_{j}^2 (t) := \ev{\hat{\xt}_j^2}_t - \ev{\hat{\xt}_j}^2_t
    \end{equation}
    denotes the time evolved variance of the $j$-th position variable w.r.t.~the exact discretized evolution operator $\widetilde{U}_H(t)$.
\label{def:expectation_values}
\end{defn}

Now we are ready to state the main theorem of this section which provides an upper bound on the running time of the dissipative quantum Hamiltonian simulation algorithm for the task of finding the optimum of a convex quadratic function.

\begin{restatable}[Coherent convex quadratic optimization via quantum dynamics]{thm}{CoherentConvex}
    Let $\epsilon > 0$ be an error tolerance. Let $A \in \mathbb{R}^{N \times N}$ be positive with eigenvalues $0 < \lambda_0 =: \lambda_{\min} \leq \lambda_1 \leq \dots \leq \lambda_{N-2} \leq \lambda_{N-1} =: \lambda_{\max}$ and assume that we know $\lambda_{\min}$ and $\lambda_{\max}$ each within some constant factor.
    Further, let $f(\bx) = \lb \bx- \bx^* \rb^\top A \lb \bx- \bx^* \rb + c$ for some constant $c \in \mathbb{R}$ and consider the corresponding discrete quantum friction Hamiltonian $\Ht$ as given in Definition~\ref{def:discretized_ham}.
    Let $\ket{\psi_0 (\bx)}$ be an initial wave function and assume having access to a quantum state $\ket{\widetilde{\psi}_0 \lb \bx \rb} \in \mathbb{C}^{2^{Nn}}$ encoding the discretized initial wave function such that for all $0 \leq t \leq t^*$ with $t^*$ being the equilibration time from Lemma~\ref{lem:multivariate_quantum} it holds that
    \begin{align}
        \norm{\ev{\bxt}_t - \ev{\bx}_t} &\leq \sqrt{\frac{\epsilon}{18\lambda_{\max}}} \\
        \left| \sum_{j=1}^N \widetilde{\sigma}_{j}^2 (t) - \sum_{j=1}^N \sigma_{j}^2(t) \right| &\leq \frac{\epsilon}{24\lambda_{\max}}.
    \end{align}
    Further, assume that we can access $f$ either via an $\epsilon'$-precise bit oracle $O_f^{(b)}$ as given in Definition~\ref{def:bit_oracle} with $1/\epsilon' \in \widetilde{O} \lb  \frac{\alpha_{A_H}^2 N \lambda_{\max}^{2} x_{\max}^2 \chi_0}{\epsilon^{2} \lambda_{\min}^2} \rb$, or via a phase oracle $O_f^{(p)}$ as given in Definition~\ref{def:phase_oracle}.
 
    Under the above assumptions, there exists a quantum algorithm that can solve the Continuous Quantum Optimization Problem by finding an $\bx'$ such that $\left| f\lb \bx' \rb - f\lb \bx^* \rb \right| \leq \epsilon$ with probability at least $1 - \delta$ using 
    using either
    \begin{equation}
        \widetilde{O} \lb \frac{\alpha_{A_H}}{\sqrt{\lambda_{\min}}} \log^2 \lb\frac{N \lambda_{\max} x_{\max} \chi_0}{\epsilon} \rb \log \lb 1/\delta \rb \rb
    \end{equation}
    queries to $O_f^{(b)}$, or
    \begin{equation}
        \widetilde{O} \lb \frac{N}{\epsilon} \frac{\lambda_{\max}^2}{\lambda_{\min}^2} x_{\max}^2 \alpha_{A_H} \chi_0 \log \lb 1/\delta \rb \rb
    \end{equation}
    queries to controlled-$O_f^{(p)}$ and its inverse, where $\alpha_{A_H}$ is given in Definition~\ref{def:discretized_ham} and $\chi_0$ is given in Definition~\ref{def:initial_conditions}.
\label{thm:main}
\end{restatable}
The proof of the above theorem is quite similar to the proof of Theorem~\ref{thm:optimization_liouvillian}. It can be found in Appendix~\ref{app:main_convex_opt}.
To obtain the complexity results claimed in the informal statement of Theorem~\ref{thm:main} in Section~\ref{sec:main}, we use the fact that
\begin{equation}
    \alpha_{A_H} \in O \lb \frac{N}{h_x^2} \rb,
\end{equation}
as discussed in Section~\ref{sec:quantum_ham_approach}, and drop subdominant logarithmic factors in $N$, $\lambda_{\max}$, $\chi_0$ and $1/\epsilon$. Furthermore, we assume $\chi_0 = \lambda_{\min} \norm{\bx_0 - \bx^*}^2$ for simplicity.

Theorem~\ref{thm:main} shows that we can find the optimum of a quadratic convex functions using $\mathrm{poly} \lb \kappa \rb$ queries to the phase oracle $O_f^{(p)}$, where $\kappa = \lambda_{\max}/\lambda_{\min}$, as long as $\alpha_{A_H}$ scales at most polynomially with $\kappa$. 
Below we present a brief argument supporting this claim while a more thorough worst-case analysis is presented in Section~\ref{sec:ill_conditioned}.

In order to get an idea about the scaling of $\alpha_{A_H}$ w.r.t.~$\lambda_{\max}$, $\lambda_{\min}$ and $\epsilon$, let us consider the expectation value of the continuum momentum operator $p$ in the case of a single damped quantum harmonic oscillator with the following time-dependent Hamiltonian:
\begin{equation}
    H(t) = e^{- \beta t /m} \frac{p^2}{2m} + e^{\beta t /m} \frac{a}{2} x^2 .
\end{equation}
Then we have that
\begin{equation}
    \frac{d\ev{p}}{dt} = -e^{\beta t/m} a \ev{x}.
\end{equation}
Furthermore,
\begin{equation}
    \frac{d^2\ev{p}}{dt^2}  = \frac{\beta}{m}  \frac{d\ev{p}}{dt} - \frac{a}{m} \ev{p}.
\end{equation}
Let us assume that we are in the underdamped regime, i.e. $\beta^2 < 4ma$. Then the solution to the above equation of motion is given by
\begin{equation}
    \ev{p}_t = e^{\beta t/2m} \lb \ev{p}_0 \cos \lb \omega t \rb - \frac{a \ev{x}_0 + \frac{\beta}{2m} \ev{p}_0}{\omega} \sin \lb \omega t \rb \rb,
\end{equation}
where $\omega = \frac{\sqrt{4ma - \beta^2}}{2m}$. Now, consider the case where $\ev{p}_0 = 0$ and set $m = 1$. Then
\begin{equation}
\begin{split}
    \ev{p}_t &= -e^{\beta t/2} \frac{a \ev{x}_0}{\omega} \sin \lb \omega t \rb \\
    &= -e^{\beta t/2} \frac{2a \ev{x}_0}{\sqrt{4a - \beta^2}} \sin \lb \omega t \rb \\
    &= -e^{\beta t/2} \sin \lb \omega t \rb \ev{x}_0 \sqrt{a} \lb 1 + \frac{\beta^2}{8a} + O \lb  \lb \frac{\beta^2}{4a} \rb^2 \rb \rb\label{eq:pMean}.
\end{split}
\end{equation}

Assuming that the expectation value of the discretized derivative operator $\ev{D_x}$ obeys approximately the same equations of motion as in the continuum, we then expect that
\begin{equation}
    \norm{D_x} \in O \lb e^{\beta t^*/2} \sqrt{a} \rb,
\end{equation}
since we are only interested in simulating dynamics up to $t=t^*$ and we expect that the asymptotic scaling of $\norm{D_x}$ is similar to that of $\ev{D_x}$. By the same argument, we expect that the norm of the discretized second derivative operator obeys
\begin{equation}
    \norm{D_x^2} \in O \lb e^{\beta t^*} a \rb.
\end{equation}

Next, let us discuss the multivariate setting. Specifically, let us consider the coordinate system in the eigenbasis of the Hessian (here equivalent to the inverse covariance matrix) $A \in \mathbb{R}^{N \times N}$ such that
\begin{equation}
    H(t) = \sum_{j} \lb e^{-\beta t/m}  \frac{p_j^2}{2m} + e^{\beta t/m} \lambda_j x_j^2 \rb,
\end{equation}
where $0 < \lambda_0 =: \lambda_{\min} \leq \lambda_1 \leq \dots \leq \lambda_{N-2} \leq \lambda_{N-1} =: \lambda_{\max}$ are the eigenvalues of $A$ as before.
Assume that $\beta \in \Theta \lb \sqrt{\lambda_{\min}} \rb$ such that $\beta^2 - 4m\lambda_j < 0$ for all $j \in [N]$.
Then the momentum expectation values satisfy
\begin{equation}
    \ev{p_j}_t = e^{\beta t/2m} \lb \ev{p_j}_0 \cos \lb \omega t \rb - \frac{\lambda_j \ev{x_j}_0 + \frac{\beta}{2m} \ev{p_j}_0}{\omega} \sin \lb \omega t \rb \rb.
\end{equation}
Assuming that the expectation values of the discretized derivative operators $\ev{D_{x,j}}$ obey approximately the same equations of motion as in the continuum and going through the same argument as in the 1-dim.~case, we expect that
\begin{equation}
    \norm{D_{x,j}} \in O \lb e^{\beta t^*/2} \sqrt{\lambda_{\max}} \rb
\end{equation}
and 
\begin{equation}
    \norm{D_{x,j}^2} \in O \lb e^{\beta t^*} \lambda_{\max} \rb \subseteq \widetilde{O} \lb \frac{\lambda_{\max}^2 \chi_0}{ \lambda_{\min} \epsilon} \rb,
\end{equation}
where $\chi_0$ is given in Definition~\ref{def:initial_conditions}.
Therefore, we expect that
\begin{equation}
    \alpha_{A_H} \in \widetilde{O} \lb \frac{N \lambda_{\max}^2 \chi_0}{ \lambda_{\min} \epsilon} \rb.
\label{bound_norm_kinetic_term}
\end{equation}
Recall that $\alpha_{A_H} \in O \lb \frac{N}{h_x^2} \rb$. The above bound then suggests the following scaling for the grid spacing $h_x$:
\begin{equation}
    h_x \in O \lb \frac{1}{\lambda_{\max}} \sqrt{\frac{\lambda_{\min} \epsilon}{\chi_0}} \rb.
\end{equation}

Note that the bound on $\alpha_{A_H}$ in Eq.~\eqref{bound_norm_kinetic_term} is rather pessimistic given that $A_H(t)$ contains a factor of $e^{-\beta t^*}$ which should ensure that $\norm{A_H(t)}$ does not grow exponentially with time despite the fact that $\norm{D_{x,j}^2}$ seems to scale exponentially with time. However, since $\alpha_{A_H}$ is technically the block-encoding normalization constant of $A_H$, it is less obvious how to improve the bound since it would require a modification of the block-encoding circuit.

\subsection{Extension to Smooth and Strongly Convex Functions}
\label{sec:smooth_strongly_convex}

In the discussion above we assumed that the function that we aim to optimize is quadratic. This is a good approximation to the behavior of any twice continuously differentiable objective function if we are sufficiently close to a local optimum; however, it is a stronger condition than the one needed for the quantum gradient descent analysis which only requires the objective function to be strongly convex. This is further relevant as in cases where the optimization algorithm is told that the objective function is quadratic then methods such as Newton's method can be used to find the optimum in a single update.  In such cases, the preceding analysis can be seen as applying to either cases where we are promised to have an initial point that is close to the local optimum or is strictly a quadratic optimization problem, but this information is either unknown to the user or intentionally withheld from the optimization algorithm.

In this section, we discuss an extension of the previous results on finding the optimum of a convex quadratic function to smooth and strongly convex functions.  These changes will necessitate migration to a slightly different frictional Hamiltonian for ease of analysis, but the results will be qualitatively the same, although we won't provide rigorous proofs.  Before introducing the new Hamiltonian, let us review the definition of smoothness and strong-convexity for an objective function.

\begin{defn}[Smooth and strongly convex function]
    Let $f:\mathbb{R}^N \rightarrow \mathbb{R}$ be continuously differentiable and let $\ell \geq \mu >0$ be constants. We say that $f$ is $\ell$-smooth and $\mu$-strongly convex if for all $\bx, \mathbf{y}$ in the domain of $f$ it holds that
    \begin{equation}
        f(\bx) + \lb \mathbf{y} -\bx \rb \cdot \nabla f(\bx) + \frac{\mu}{2} \norm{\mathbf{y} - \bx}^2 \leq f(\mathbf{y}) \leq f(\bx) + \lb \mathbf{y} -\bx \rb \cdot \nabla f(\bx) + \frac{\ell}{2} \norm{\mathbf{y} - \bx}^2.
    \end{equation}
    Further, we define $\kappa := \frac{\ell}{\mu}$ to be the condition number of $f$.
\label{def:strongly_convex}
\end{defn}

In the previous section, we showed that in the case of a convex quadratic function, we can solve the resulting equations of motions analytically which then allowed us to prove relatively tight convergence bounds, see e.g.~Lemma~\ref{lem:multivariate_quantum}.
However, solving the equations of motion associated with a general smooth and strongly convex functions in order to prove convergence bounds is difficult. Instead, we follow Refs.~\cite{Wibisono2016continuous_gradient_descent, leng2023qhd} and utilize Lyapunov functions for our convergence analysis. The Hamiltonian studied in these papers is of the following form:
\begin{equation}
    H(t) = e^{\alpha_t - \gamma_t} \sum_{j=1}^N \frac{p_j^2}{2} + e^{\alpha_t + \beta_t + \gamma_t} f(\bx),
\label{qhd_ham}
\end{equation}
where $\alpha_t$, $\beta_t$ and $\gamma_t$ are continuously differentiable functions of time satisfying the following ideal scaling conditions:
\begin{align}
    \dot{\beta}_t &\leq e^{\alpha_t}, \label{ideal_scaling_beta} \\
    \dot{\gamma}_t &= e^{\alpha_t}.
\label{ideal_scaling_gamma}
\end{align}
Without loss of generality, we can assume that $f(\bx) \geq 0$ for all $\bx \in \mathbb{R}^N$ and $f(\mathbf{0}) = 0$, i.e. the optimum of $f$ is at $\bx = \mathbf{0}$. This can always be accomplished by a simple shift of the coordinate system.
Now, consider the following quantum Lyapunov function/operator proposed in~\cite{leng2023qhd}:
\begin{equation}
    \mathcal{W}(t) := \frac{1}{2} \sum_{j=1}^N J_j^2 + e^{\beta_t} f\lb \bx \rb,
\end{equation}
where $J_j := e^{-\gamma_t} p_j + x_j$.
As shown in~\cite{leng2023qhd}, if $f(\bx)$ is continuously differentiable and convex then it holds that 
\begin{equation}
    \frac{d}{dt} \ev{\mathcal{W}} \leq 0, \quad \forall \, t \geq 0.
\end{equation}
Additionally, 
\begin{equation}
    \ev{\frac{1}{2} \sum_{j=1}^N J_j^2}_t = \frac{1}{2} \sum_{j=1}^N \ev{J_j^2}_t \geq 0, \quad \forall t \geq 0.
\end{equation}
This implies that
\begin{equation}
    \ev{f}_t \leq \ev{\mathcal{W}}_t e^{-\beta_t} \leq \ev{\mathcal{W}}_0 e^{-\beta_t}.
\label{lyapunov_bound}
\end{equation}
As pointed out in Refs.~\cite{Wibisono2016continuous_gradient_descent, leng2023qhd}, the above bound seems to indicate that arbitrarily fast convergence rates are possible by choosing $\beta_t$ appropriately. However, the bound in Eq.~\eqref{lyapunov_bound} is derived in continuous space and time. Once we discretize the system in order to simulate its dynamics on a digital computer we cannot achieve arbitrary convergence rates because of numerical instabilities caused by discretization errors. Proving tight bounds on these discretization errors and obtaining provable convergence rates via Lyapunov functions is difficult although recent work~\cite{chakrabarti2025optimization, leng2025subexponential} seems to have mostly resolved this issue.
Regardless, the Lyapunov function approach is certainly useful for getting an idea about the convergence rate for more general convex functions even if we won't prove any tight bounds here.
Let us now discuss how to choose $\alpha_t$, $\beta_t$ and $\gamma_t$ such that the resulting convergence rate should also be attainable in a discretized setting. Specifically, the following choice of parameters leads to a Hamiltonian that closely resembles the friction Hamiltonian from Definition~\ref{def:friction_ham} used previously in the convergence analysis for convex quadratic functions:
\begin{align}
    \alpha_t &= \frac{1}{2}\ln(\mu), \\
    \beta_t &= \sqrt{\mu} t + \ln (\mu) , \\
    \gamma_t &= \sqrt{\mu} t - \frac{3}{2} \ln (\mu),
\end{align}
where $\mu$ is given in Definition~\ref{def:strongly_convex}. The Hamiltonian in Eq.~\eqref{qhd_ham} then becomes
\begin{equation}
    H(t) = \underbrace{\mu^2 e^{-\sqrt{\mu} t} \sum_j \frac{p_j^2}{2}}_{=: A_H'(t)} + \underbrace{e^{2 \sqrt{\mu} t} f(\bx)}_{=: B_H'(t)}.
\label{new_friction_ham}
\end{equation}
Now, let $\mathbf{\bx}'_t \sim |\psi_t(\bx)|^2$ be a position vector sampled from the time-evolved quantum probability distribution. By Markov's inequality we then have that 
\begin{equation}
    P \lb f(\bx'_t) \geq \epsilon \rb \leq \frac{\ev{f}_t}{\epsilon} \leq \frac{\ev{\mathcal{W}}_0  e^{- \beta_t}}{\epsilon} \leq \frac{\ev{\mathcal{W}}_0  e^{- \sqrt{\mu} t}}{ \mu \epsilon}.
\end{equation}
As before, it suffices to ensure that this failure probability is at most $1/3$ since this allows us to boost the success probability to at least $1 - \delta$ using only $O \lb \log \lb 1/\delta \rb \rb$ samples.
We have that $P \lb f(\bx'_t) \geq \epsilon \rb \leq 1/3$ if
\begin{equation}
    t = t^*_{\mu} := \frac{1}{\sqrt{\mu}} \log \lb \frac{3 \ev{\mathcal{W}}_0}{\mu \epsilon} \rb.
\end{equation}

Note that
\begin{equation}
\begin{split}
    \ev{\mathcal{W}}_0 &= \sum_j \lb e^{-2\gamma_0} \ev{p_j^2}_0 + e^{-\gamma_0} \lb \ev{x_j p_j}_0 + \ev{p_j x_j}_0 \rb + \ev{x_j^2}_0  \rb + e^{\beta_0} \ev{f}_0 \\
    &= \sum_j \lb \mu^3 \ev{p_j^2}_0 + \mu^{3/2} \lb \ev{x_j p_j}_0 + \ev{p_j x_j}_0 \rb + \ev{x_j^2}_0 \rb + \mu \ev{f}_0,
\end{split}
\end{equation}
which implies that
\begin{equation}
    e^{2\sqrt{\mu} t^*_{\mu}} = \lb \frac{3 \ev{\mathcal{W}}_0}{\mu \epsilon} \rb^2 = \lb \frac{3 \sum_j \lb \mu^3 \ev{p_j^2}_0 + \mu^{3/2} \lb \ev{x_j p_j}_0 + \ev{p_j x_j}_0 \rb + \ev{x_j^2}_0 + \mu \ev{f}_0\rb}{\mu \epsilon} \rb^2.
\end{equation}

Let $\mu^2 \alpha_{A_H'}$ be the block-encoding constant of the discretized version of $A_H'$ as defined in Eq.~\eqref{new_friction_ham}. Given that the Hamiltonian in Eq.~\eqref{new_friction_ham} has essentially the same form as the friction Hamiltonian in Def.~\ref{def:friction_ham}, we can port over the complexity results from Theorem~\ref{thm:main}, assuming that the discretization errors are sufficiently small. Following the same analysis as in the proof of Theorem~\ref{thm:main}, we thus conclude that we should be able to approximate the optimal value of $f$ within error $\epsilon$ with probability at least $1-\delta$ using
\begin{equation}
    \widetilde{O} \lb \frac{1}{\sqrt{\mu}} e^{2\sqrt{\mu} t^*_\mu} f_{\max} \mu^2 \alpha_{A_H'} t^*_{\mu} \log^3 \lb \frac{1}{\epsilon} \rb \log \lb 1/\delta \rb \rb \subseteq \widetilde{O} \lb \frac{N}{\epsilon^2} \frac{\ell}{\mu} x_{\max}^2 \alpha_{A_H'} \ev{\mathcal{W}}_0^2 \log \lb 1/\delta \rb \rb
\label{query_complexity_strongly_convex}
\end{equation}
queries to the phase oracle $O_f^{(p)}$ where we used the fact that
\begin{equation}
    f_{\max} \in O \lb \ell \norm{\bx}^2 \rb \subseteq O \lb \ell N x_{\max}^2\rb.
\end{equation}

The query complexity bound in Eq.~\eqref{query_complexity_strongly_convex} indicates that we should be able to achieve $\mathrm{poly} \lb \kappa, 1/\epsilon, N \rb$ scaling even for smooth strongly convex functions which are not necessarily quadratic. A more thorough analysis is left for future work.

\section{Quantum Algorithm for Coherent Global Optimization}
\label{sec:global}

The above methods focused on the task of finding a local optimum for a strongly convex function.  While almost all functions can be closely approximated by a convex quadratic function within a neighborhood of the local optima, the task of finding an initial point that is close enough to allow this approximation to hold is non-trivial. Further, vanishing gradient problems in the trajectory can lead to long evolution times (as seen by $\lambda_{\min} \rightarrow 0$ in the above asymptotics).  In this section, we take a different approach that follows the techniques of~\cite{Simon2024Liouvillian}.  Specifically, we consider the time evolution of a classical phase space probability density subject to a carefully chosen Liouvillian operator that equilibrates the probability density to a classical Boltzmann distribution in the limit of sufficiently long evolution.

The basic idea behind our global approach stems from thermal physics.  Imagine for the moment that we have an objective function $f(\mathbf{x})$ where $\mathbf{x}$ can be interpreted as the position of the system and we will interpret $f(\mathbf{x})$ as an energy function.  We can then, analogous to thermalization, find a globally optimal configuration by preparing a thermal distribution for inverse temperature $\beta>0$:
\begin{equation}
    \rho(\mathbf{x}) = \frac{e^{-\beta f(\mathbf{x})}}{\int e^{-\beta f(\mathbf{s})} {d}\mathbf{s}} .
\end{equation}
As $\beta\rightarrow \infty$, $\rho(\mathbf{x})$ will be supported only over the global minima.  In this sense, we can use thermalization to find a global optimum for the function.  This intuition is enshrined in Hamiltonian-based approaches to quantum optimization as well as the Boltzmann model of neural networks.

Our aim here is to provide dynamical methods for preparing these thermal distributions by reducing the thermalization process to a microcanonical simulation in a higher dimensional space.  In this sense, our work is similar to the local method considered above; however, here we will not need to assume convexity but will need to make assumptions about the spectrum to guarantee convergence to the global optimum.

\begin{figure}
    \centering
   \begin{tikzpicture}[>=stealth]

\draw[thick] (0,0) ellipse (4 and 2);
\node at (0,2.6) {\large Extended Phase Space $(q,p,s,p_s)$};

\fill[red!70] (0,0) ellipse (2.6 and 1.1);
\node[text=white, align=center] at (0,0)
{Microcanonical Distribution\\for Nos\'e Hamiltonian};

\draw[thick] (10,0) ellipse (2.2 and 2.2);
\node at (10,2.8) {\large Original Phase Space $(q,p)$};

\fill[blue!65] (10,0) ellipse (1.2 and 0.8);
\node[text=white, align=center] at (10,0)
{Thermal\\Distribution};

\draw[->, thick] (4.2,0) -- (7.8,0);
\node at (6,0.4) {Projection};

\end{tikzpicture}
    \caption{Illustration of the Nosé Hamiltonian's thermal state in the extended coordinates $(x,p,s,p_s)$ and the corresponding thermal distribution  in the original phase space coordinates $(x,p)$.  The idea behind our global optimization method is to equilibrate over the extended microcanonical distribution using Koopman-von Neumann and then trace over $s,p_s$ to construct a low temperature thermal distribution over the parameters $(x,p)$.}
    \label{fig:Nose}
\end{figure}

Recall that the Liouvillian operator as shown in Eq.~\eqref{liouvillian} can be derived from the classical Hamiltonian of the classical system under consideration.
The Hamiltonian that we use in this setting is the Nos\'e Hamiltonian with the external potential chosen to be the objective function: $V(\bx)=f(\bx)$.  Specifically, the system is described by the following classical system Hamiltonian:
\begin{equation}
    H_{\mathrm{sys}} = \frac{\bp^2}{2m} + f(\bx) = \frac{1}{2m} \sum_{j=1}^N p_j^2 + f(\bx).
\end{equation}
If a system is ergodic then we expect that, for any initial phase space density with fixed energy, the dynamics given by the Liouville equation under $H_{\rm sys}$ will lead to a distribution that is uniformly distributed over the accessible phase space.  In this sense, this approach to optimization is analogous to Hamiltonian Monte-Carlo methods~\cite{betancourt2017conceptual} which aim to draw samples from a classical thermal distribution and only became practical after automatic differentiation became a possibility.

We thermalize the system not by emulating friction via a time-dependent Hamiltonian/Liouvillian, as done in the previous sections, but rather by introducing a new bath variable $s$ with mass $Q$ and momentum $p_s$.
For inverse temperature $\beta$, the Nos\'e Hamiltonian takes the following form:
\begin{equation}
    H_{N} = \frac{\mathbf{p}^2}{2m s^2} + V(\mathbf{x}) + \frac{p_s^2}{2Q} + g\beta^{-1} \ln(s).
\end{equation}
Here, $\mathbf{x} \in \mathbb{R}^N$ and $\mathbf{p} \in \mathbb{R}^N$ are the position and momentum of the system and $g$ is a free parameter which we will later choose to scale linearly with $N$.
Now, let $p'_j := sp_j$ be the momentum conjugate to $x_j$ in the extended system, meaning system + bath.
In the continuum, it can be shown that the microcanonical partition function of the extended system reduces to a canonical partition function when restricted to the system variables $\bx$ and $\bp$~\cite{Nose1984partition_function, Huenenberger2005thermostat}:
\begin{equation}
\begin{split}
    \mathcal{Z} &\propto \int d\{x_n\} \int d\{p'_n\} \int ds \int dp_s \, \delta \lb H_N \lb \{x_n\}, \{p'_n\}, s, p_s \rb - E_{\text{ext}}\rb \\
    &\propto \int d\{x_n\} \int d\{p_n\} \int ds \int dp_s \, s^N\delta \lb H_N \lb \{x_n\}, \{p_n\}, s, p_s \rb - E_{\text{ext}}\rb \\
    &\propto  \int d\{x_n\} \int d\{p_n\} e^{-\beta H_{\text{sys}} \lb \{x_n\}, \{ p_n \} \rb},
\end{split}
\label{eq:partition_func}
\end{equation}
where $E_{\text{ext}}$ is the conserved energy of the extended system, $\beta := 1/k_BT$ and in the second line we have done the variable change $p_j=p'_j/s$.  Thus, if a microcanonical state is prepared for the Nos\'e Hamiltonian then the resulting expectation values match those from the canonical ensemble.  In the limit as $T\rightarrow 0$, this corresponds to a Gibbs distribution with support only on the minima of $V(\bx) = f(\bx)$ as we approach the continuum. Such an optimization process is not generally expected to be efficient. However, we would like to characterize the resources needed for convergence, in the worst case setting, in order to understand when this process will converge.  The conclusion that we come to is that under standard assumptions from quantum thermodynamics, we will approach the microcanonical distribution for the Nos\'e Hamiltonian on a timescale that is inverse in the spectral gap of the Hamiltonian constrained to an energy window of the Nos\'e Hamiltonian.

\subsection{Nos\'e Hamiltonian in a Discrete Setting}

It is trivial to show that the microcanonical state for the Nos\'e Hamiltonian in the extended system produces the canonical ensemble for the system after the bath is integrated out. However, the infinite integrals shown in \eqref{eq:partition_func} can not be evaluated directly on a digital computer, classical or quantum. We will show that after a series of approximations that put the formula in a form suitable for machine evaluation, the expectation value of $\bx$, which is the quantity we are trying to optimize, we obtain is close to the expectation value under the original distribution. 
The microcanonical density function of fixed total energy, $E_{\rm ext}$, is given by 
\begin{equation}
     \rho(\bx,\bp,s,p_s) = C\, \delta \lb H_N \lb \{x_n\}, \{p_n\}, s, p_s \rb - E_{\text{ext}}\rb.
\label{eq:global_rho}
\end{equation}

By integrating \eqref{eq:global_rho} over all coordinates, we can recover the normalization constant $C$ 
as:
\begin{equation}
     C^{-1} = g^{-1}\beta \lb\frac{2\pi m}{g^{-1}(N+1)\beta}\rb^{\frac{N}{2}}\lb\frac{2\pi Q}{g^{-1}(N+1)\beta}\rb e^{g^{-1}(N+1)\beta E_e}\int\,\{dx\}\,e^{-g^{-1}(N+1)\beta f(x)}.
\end{equation}

The expectation value of $\bx$ under this distribution is
\begin{equation}
\begin{split}
    \langle \bx \rangle_0 &:= C\, \int d\{x_n\} \int d\{p'_n\} \int ds \int dp_s s^N \bx \; \delta \lb H_N \lb \{x_n\}, \{p'_n\}, s, p_s \rb - E_{\text{ext}}\rb\\
    &= \frac{\int d\{x_n\}\,\bx\, e^{-g^{-1}\beta(N+1)f(\bx)}}{\int d\{x_n\}\, e^{-g^{-1}\beta(N+1)f(\bx)}}\\
    &=\frac{\int d\{x_n\}\,\bx\, e^{-\beta f(\bx)}}{\int d\{x_n\}\, e^{-\beta f(\bx)}},
\end{split}
\label{eq:ch7_nose_xavg}
\end{equation}
where in the last line we set $g$, which is a free parameter, to be $N+1$. The results of this section can be encapsulated in the following theorem:

\begin{thm}[Discretization of the Nos\'e Hamiltonian average]\label{thm:ch7_disc}
    The position average given by \eqref{eq:ch7_nose_xavg} can be approximated to $O(\epsilon)$ error in the Euclidean norm with the following expression:
    \begin{equation}
    \begin{split}
        \langle \bx \rangle_3 &:= C\,\sum_{n_{x,1}=0}^{x_{max}/h_x}h_x\sum_{n_{x,2}}^{x_{max}/h_x}h_x\cdots\sum_{n_{x,N}}^{x_{max}/h_x}h_x (h_x\bm{n_x})\,
        \sum_{n_{p,1}=-P_1/h_p}^{P_1/h_p}h_p\sum_{n_{p,2}=-P_1/h_p}^{P_1/h_p}h_p\cdots\sum_{n_{p,N}=-P_1/h_p}^{P_1/h_p}h_p \\ 
        &\qquad \sum_{n_{p_s}=-P_2/h_{p_s}}^{P_2/h_{p_s}} h_{p_s} \sum_{n=0}^{N_s} h_s (n h_s)^N \text{rect}_\Delta \lb H \lb h_x\bm{n_x}, h_p\bm{n_p}\rb + \frac{(n_{p_s}h_{p_s})^2}{2Q} + (N+1)\beta^{-1}\ln{n h_s} - E_{\text{ext}}\rb,
    \end{split}
    \end{equation}
    where the rectangle function is defined as 
    \begin{align}
    \text{rect}_\Delta(x) := 
    \begin{cases} 
      \frac{1}{\Delta} & |x|\leq \frac{\Delta}{2}, \\
      0 & \text{otherwise} .
        \end{cases}
    \end{align}
    Further, $\bm{n_x}:= n_{x,1} \bm{\hat{e}_1}+ n_{x,2} \bm{\hat{e}_2}+\cdots+ n_{x,N} \bm{\hat{e}_N}$ and $\bm{n_p}:= n_{p,1} \bm{\hat{e}_1}+ n_{p,2} \bm{\hat{e}_2}+\cdots+ n_{p,N} \bm{\hat{e}_N}$ are vectors on discrete lattices, and the tunable parameters $\Delta, P_1, P_2, h_x, h_p, h_{p_s}, h_s$ have the following asymptotic magnitudes:
    \begin{align}
        \Delta &\in O\lb\frac{1}{\beta N^{\frac{1}{4}}} \sqrt{\frac{\epsilon}{x_{\max}}}\rb,\\
        P_1 &\in \tilde{O}\lb\sqrt{\frac{m}{\beta}}\rb,\\
        P_2 &\in\tilde{O}\lb\sqrt{\frac{Q}{\beta}}\rb,\\
        h_s &\in O\lb\frac{\epsilon}{N^\frac{11}{4}x^{\frac{3}{2}}_{\max}}\rb,\\
        h_p &\in \tilde{O}\lb\frac{1}{N^{\frac{5}{4}}}\sqrt{\frac{m}{\beta}}\frac{\sqrt{\epsilon}}{\sqrt{x_{\max}}}\rb,\\
        h_{p_s} &\in \tilde{O}\lb\frac{1}{N^{\frac{1}{4}}}\sqrt{\frac{Q}{\beta}}\frac{\sqrt{\epsilon}}{\sqrt{x_{\max}}}\rb,\\
        h_x &\in O\lb \frac{1}{\beta N |\partial f(\bx)|_{\max}}\rb,
    \end{align}
    where $|\partial f(\bx)|_{\max}:=\max_i\max_{\bx}\left|\frac{\partial f(\bx)}{\partial x_i}\right|$.
\end{thm}
The proof of the above theorem can be found in Appendix~\ref{appendix:discproof}.

\subsection{Approach to Microcanonical Equilibrium}

In order for us to use the Nosé Hamiltonian to reach a low-temperature thermal state for our reduced system, we consider the approach of the extended, closed system to microcanonical equilibrium.  This makes the problem of approaching the microcanonical distribution equivalent to the well studied problem of thermalization for closed quantum systems.

Using the result that in the continuum limit the Nosé Hamiltonian leads to a thermal distribution over the original system when we trace over the bath, we wish to prepare a microcanonical distribution over the positions of the model.  Recall that the momenta here correspond to the dynamical velocities of the parameters that we wish to optimize over and therefore do not play a meaningful role in the optimization.  We wish to reach this state by preparing the following state, which we find by evolving the system for a random amount of time between $0$ and some maximum duration $t$ for the discrete phase space Koopman-von Neumann density operator $\tilde{\rho}$:
\begin{equation}
     \langle \tilde{\rho} \rangle_{t} := \frac{1}{t} \int_0^t e^{-iLt'} \tilde{\rho}_0 e^{iLt'} dt'.
\end{equation}
Here the Liouvillian is given from the Hamiltonian via the operator
\begin{equation}
\label{eq:Ldefinition}
    L = -i \sum_{j=1}^N \lb \frac{\partial H}{\partial p_j} \frac{\partial}{\partial x_j} - \frac{\partial H}{\partial x_j} \frac{\partial}{\partial p_j} \rb -i \lb \frac{\partial H}{\partial p_s} \frac{\partial}{\partial s} - \frac{\partial H}{\partial s} \frac{\partial}{\partial p_s} \rb,
\end{equation}
where in practical applications of these formulas the (partial) derivatives are replaced with finite difference operators.

We anticipate that the correct microcanonical distribution will emerge as $t \rightarrow \infty$ which we denote as $\langle \tilde{\rho}\rangle_\infty$.
Specifically, we assume ergodicity, meaning that the infinite time average of the phase space density is equal to the phase space average over a narrow energy shell $\Delta$.  We further will assume that all eigenstates of the Liouvillian supported within the energy window have position / momentum distributions corresponding to the microcanonical average. 
This is known as the eigenstate thermalization hypothesis~\cite{Linden2009thermal}, applied in this case to the Liouvillian which plays the role of a Hamiltonian for our quantum dynamical system.

In order to discuss the microcanonical state, we need to formally specify an energy band that the initial distribution is prepared inside.  This energy band $\Delta$ is defined to be  the subspace spanned by the set of eigenvectors of $H$ $\{\ket{\lambda_{H,j}}\}$, where $H\ket{\lambda_{H,j}} = \lambda_{H,j} \ket{\lambda_{H,j}}$ such that the eigenvalues of the Hamiltonian fall inside a fixed range $[\Delta_0,\Delta_1]$.  More compactly,
\begin{equation}
    \Delta = {\rm span}(\{\ket{\lambda_{H,j}}: \lambda_{H,j}\in [\Delta_0,\Delta_1]\}).
\end{equation}
As we approach the continuum, the size of this band can be allowed to shrink and microcanonical equilibrium can be thought of as occurring in this limit when the probability of any microstate of position and momentum that is supported in the span of these vectors is equal to all others in the support.

In order to analyze the time needed to equilibrate, we need in effect to argue about the form of the eigenvectors of the Liouvillian, $L$, and the value of $t$ needed to ensure that transient coherences in the Koopman-von Neumann wave function are small.
The question remaining involves how to select $t$ in this evolution.  This generically depends on the structure of the eigenvalues and eigenvectors of the Nosé Hamiltonian $H_N$, which are difficult to analytically compute.  However, we can provide analytic estimates of the time under the assumptions that the eigenvectors of the Hamiltonian are typical of Haar random vectors.  Specifically though, the assumption of Haar randomness is actually too strong and all that we require is that the eigenvectors are typical of those drawn from a unitary $2$ design as such sets of matrices capture the mean and variance of expectation values over the random matrices without requiring exponential computational overhead to implement.  These assumptions are strong, but are needed to make concrete predictions about the equilibration time without invoking assumptions about the observables considered or the structure of the eigenstates; however, despite this we include the following argument to qualitatively provide insight about how the equilibrium time scales in the discrete approximation to the continuous Nos\`e Hamiltonian.

\begin{lem}
\label{lem:equilibrium}
Let $\tilde{\rho}_0 = \sum_{j,k} \alpha_{j,k} \ketbra{\lambda_j}{\lambda_k} \in \mathbb{C}^{D_x^2D_s^2 \times D_x^2D_s^2}$, where $D_x$ is the dimension of the position/momentum spaces and $D_s$ is the dimension of the position/momentum for the bath, be the initial density matrix of the (extended) classical system where $\{\lambda_j \}$ are the eigenvalues of the Liouvillian $L\in \mathbb{C}^{D_x^2D_s^2 \times D_x^2D_s^2}$ constrained to be in the energy window $\Delta$.   
Let $D_s \in O(1/\delta)$ for $\delta>0$. Then there exists $t \in O(1 /\gamma \delta^2)$ such that if $\tilde{\rho}_{mc}$ is the  microcanonical distribution for the discrete Liouvillian $L$ over the system and bath for the energy window $\Delta$ and if $B(\Delta)$ is a set of projectors onto microstates, which are in $\mathbb{C}^{D_x^2D_s^2 \times D_x^2D_s^2}$, such that ${\rm Tr}(H_N \ketbra{\phi}{\phi}) $, then the maximum error in the probability of sampling a given position coordinate of the extended system is 
$$
\sup_{\phi \in B(\Delta)}\left|{\rm Tr}\left( {\rm Tr}_{s,p,p_s}(\ketbra{\phi}{\phi})\left(\mathbb{E}_{{\tilde{\rho}_0}}({\rm Tr}_{s,p,p_s}(\langle\tilde{\rho}\rangle_{t})) - {\rm Tr}_{s,p,p_s}(\tilde{\rho}_{mc})\right)\right)\right|\le \delta,
$$
where we assume the following:
\begin{enumerate}
    \item Let $\tilde{\rho}_0$ be a random pure state drawn uniformly from states supported only on an energy window which is a convex set $\Delta\subset \mathbb{R}$ such that if $\ket{E_j}$ is an eigenvector of $H$ with eigenvalue $E_j$ then $\bra{E_j}\tilde{\rho}_0\ket{E_j}=0$ if $\lambda_j \not\in \Delta$.  We denote this for brevity as $\tilde{\rho}_0 \in \Delta$. 
    
    \item There exists a minimum spectral gap $\gamma>0$ such that for all $\ketbra{\lambda_j}{\lambda_j},\ketbra{\lambda_k}{\lambda_k}$ such that $\lambda_j,\lambda_k \in \Delta$, 
    $${\substack{\min\\j,k,\lambda_j\ne \lambda_k}}|\lambda_j - \lambda_k|\ge \gamma .$$
    \item The eigenvectors of $L$ are typical of the columns of random unitaries drawn from a unitary $2$-design.
\end{enumerate}
\end{lem}
\begin{proof}
 The time averaged density matrix is given by
\begin{equation}
\begin{split}
    \langle \tilde{\rho} \rangle_{t} &= \frac{1}{t} \int_0^t e^{-iLt'} \tilde{\rho}_0 e^{iLt'} dt' \\
    &= \frac{1}{t} \sum_{j,k} c_{j,k} \int_0^t e^{-i\lambda_jt'} \ketbra{\lambda_j}{\lambda_k} e^{i\lambda_kt'} dt' \\
    &= \sum_{\substack{j,k \\ \lambda_j = \lambda_k}} c_{j,k} \ketbra{\lambda_j}{\lambda_k} + \underbrace{\frac{i}{t} \sum_{\substack{j,k \\ \lambda_j \neq \lambda_k}} c_{j,k} \frac{ e^{-i(\lambda_j - \lambda_k) t} - 1}{\lambda_j - \lambda_k} \ketbra{\lambda_j}{\lambda_k}}_{=: R_t}.
\end{split}
\label{avg_rho_T}
\end{equation}
In the limit $t \rightarrow \infty$, we have that
\begin{equation}
    \tilde{\rho}_{\infty} := \lim_{t \rightarrow \infty} \langle \tilde{\rho} \rangle_{t} = \sum_{\substack{j,k \\ \lambda_j = \lambda_k}} c_{j,k} \ketbra{\lambda_j}{\lambda_k}.
\end{equation}

Next, let us assume that our initial state is a pure state.  This corresponds to assuming that
\begin{equation}
    c_{jk} = \sqrt{c_j c_k^*},
\end{equation}
where $\sum_j |c_j| =1$.  Now, let us assume that we are interested in the measurement of a probability of the system being in a position state $\ket{\bx}$.  We do not actually care about the momentum that we have at a given position, since the positions correspond to the parameters for the optimization whereas the momentum does not have a clear interpretation.  With this in mind, let us then define a notion of variance for the expectation values of the state in a total dimension $D=D_x^2 D_s^2$ considering the position, momentum and bath.  Further, let us consider a microstate $\ketbra{\phi}{\phi}$ inside the energy window $\Delta$ to be a projector onto a combination of positions and momenta we are interested in. Then the expectation value of any such projector onto the position and momentum of the system is
\begin{equation}
    V(\Delta):=\max_{\phi \in B(\Delta)}\left(\frac{1}{D} \left(\sum_{j} \bra{\lambda_j} \ketbra{\phi}{\phi}\ket{\lambda_j} - (\bra{\lambda_j} \ketbra{\phi}{\phi} \ket{\lambda_j})^2 \right)\right).
\end{equation}

We can now bound the error in any probability calculation using this information.  In particular, let us assume that we are interested in the probability of measuring the system in a state  $\ket{\phi}$ in the physical subspace.  We then have that (after truncating all states from consideration that have zero amplitude) the error operator is
\begin{align}
    |{\rm Tr}(R_t \ketbra{\phi}{\phi})|&=\frac{1}{t}  \left|\sum_{\substack{j,k \\ \lambda_j \neq \lambda_k}} c_{j,k} \frac{\left(e^{-i(\lambda_j-\lambda_k)t}-1\right) \bra{\lambda_k} \ketbra{\phi}{\phi}\ket{\lambda_j} }{\lambda_j-\lambda_k}\right|\nonumber\\
    &\le \frac{2}{\gamma t}\sqrt{\sum_{j,k} |c_{j,k}|^2} \sqrt{\sum_{\substack{j,k \\ \lambda_j \neq \lambda_k}} \bra{\lambda_j} \ketbra{\phi}{\phi}\ket{\lambda_k}\bra{\lambda_k} \ketbra{\phi}{\phi}\ket{\lambda_j} }\nonumber\\
    &\le \frac{2}{\gamma t} \sqrt{\sum_{j} \bra{\lambda_j} \ketbra{\phi}{\phi}\ket{\lambda_j} - (\bra{\lambda_j} \ketbra{\phi}{\phi}\ket{\lambda_j})^2}\nonumber\\
    &\le \frac{2\sqrt{V(\Delta) D}}{\gamma t}.
\end{align}

We then can ensure that $|{\rm Tr}(R_t \ketbra{\phi}{\phi})|\le \delta$ if
\begin{equation}
    t \ge \frac{2\sqrt{V(\Delta) D}}{\gamma\delta}.
\label{eq:Tbd}
\end{equation}

Further, if the $\lambda_j$ contain degeneracies we can note that any linear combination of the eigenvectors is also an eigenvector.  In this case, we can write for constants $c'_j$ and new eigenvectors $\ket{\lambda_j'}$
\begin{equation}
    \langle \tilde{\rho}\rangle_\infty = \sum_{j} c'_j \ketbra{\lambda_j'}{\lambda_j'}.
\end{equation}
Then, using our assumption that each eigenvector of the Liouvillian, $\ket{\lambda_j'}$, is typical of a column vector from a random unitary chosen from a unitary 2-design,
it follows from Theorem 3 from~\cite{Linden2009thermal} and the fact that the momentum is being traced over that
\begin{equation}
 \frac{1}{2}\|\mathbb{E}_{{\tilde{\rho}_0}}({\rm Tr}_{s,p,p_s}(\langle \tilde{\rho}\rangle_\infty)) - {\rm Tr}_{s,p,p_s}(\tilde{\rho}_{mc})\|_1 \le \frac{1}{2}\sqrt{\frac{D_x}{D_s^2 D_p}}.
\end{equation}
The above result and the assumption that the dimension of the momentum register is the same as the dimension of the position register shows that the trace distance between the reduced density matrices obeys 
\begin{equation}
    \frac{1}{2}\|\mathbb{E}_{{\tilde{\rho}_0}}({\rm Tr}_{s,p,p_s}(\langle \tilde{\rho}\rangle_\infty)) - {\rm Tr}_{s,p,p_s}(\tilde{\rho}_{mc})\|_1 \le \frac{1}{2}\frac{1}{\sqrt{D_s^2}}.
\end{equation}
This implies that it suffices to pick 
\begin{equation}
    D_s = O(1/\delta)
\end{equation}
to ensure that the error is at most $\delta$.
This justifies the restriction made in the lemma statement that $D_s\in O(1/\delta)$.
We then have using the sub-additive property of the trace distance~\cite{watrous2018theory}
\begin{align}
    \frac{1}{2}\|\mathbb{E}_{{\tilde{\rho}_0}}({\rm Tr}_{s,p,p_s}(\langle\tilde{\rho}\rangle_{t})) - {\rm Tr}_{s,p,p_s}(\tilde{\rho}_{mc})\|_1 \le &\frac{1}{2}\|\mathbb{E}_{{\tilde{\rho}_0}}({\rm Tr}_{s,p,p_s}(\langle\tilde{\rho}\rangle_{t})) - \mathbb{E}_{{\tilde{\rho}_0}}{\rm Tr}_{s,p,p_s}(\langle\tilde{\rho}\rangle_{\infty})\|_1 \nonumber\\
    &\quad+ \frac{1}{2}\|\mathbb{E}_{{\tilde{\rho}_0}}({\rm Tr}_{s,p,p_s}(\langle\tilde{\rho}\rangle_\infty)) - {\rm Tr}_{s,p,p_s}(\tilde{\rho}_{mc})\|_1\in O(\delta).
\end{align}
Hence we have that for any state $\phi\in B(\Delta)$ that the expectation value of the marginal distribution over the positions of the state obeys 
\begin{equation}
    \sup_{\phi \in B(\Delta)}\left|{\rm Tr}\left( {\rm Tr}_{s,p,p_s}(\ketbra{\phi}{\phi})\left(\mathbb{E}_{{\tilde{\rho}_0}}({\rm Tr}_{s,p,p_s}(\langle\tilde{\rho}\rangle_{t})) - {\rm Tr}_{s,p,p_s}(\tilde{\rho}_{mc})\right)\right)\right|]\le \|\mathbb{E}_{{\tilde{\rho}_0}}({\rm Tr}_{s,p,p_s}(\langle \tilde{\rho}\rangle_t)) - {\rm Tr}_{s,p,p_s}(\tilde{\rho}_{mc})\|_1\le \frac{1}{D_s}.
\end{equation}

Thus it suffices to choose a value of $t$ such that if $D_x$ is the dimension of the position register and we assume that this is the same as the momentum dimension, then~\eqref{eq:Tbd} and the fact that $D=D_s^2 D_x^2$ lead to
\begin{equation}
    t \in O\left( \frac{D_x \sqrt{V(\Delta)}}{\gamma \delta^2} \right) .
\end{equation}

The Haar measure is unitarily invariant and thus the average inner product squared between $\ket{\lambda_j}$ and $\ket{0}\ket{k}$ must be on the order of $1/D$.  Since a unitary $2$-design matches the first two  the expectation value will match the Haar expectation and in turn can be bounded above by 
\begin{equation}
    \mathbb{E}(V(\Delta)) \le \mathbb{E}(D^{-1} \sum_{j,k} |\bra{\lambda_j}\ket{0}\ket{k}|^2)=O(D_s^2/D)=O(1/D_x^2).
\end{equation}

Next we have that 
\begin{align}
    \mathbb{E}(V(\Delta)^2) &= \sum_{j,j',k,k'} D^{-2}\mathbb{E}(|\bra{\lambda_j}\ket{0}\ket{k}|^2|\bra{\lambda_{j'}}\ket{0}\ket{k'}|^2)\nonumber\\
    &= \sum_{(j,k)\ne (j',k')} D^{-2}\mathbb{E}(|\bra{\lambda_j}\ket{0}\ket{k}|^2|\bra{\lambda_{j'}}\ket{0}\ket{k'}|^2)+\sum_{j,k} D^{-2}\mathbb{E}(|\bra{\lambda_j}\ket{0}\ket{k}|^4)
\end{align}
Next, using the well known formulas for the $(2,2)$ moments of the distribution as given in~\cite{Linden2009thermal}, we have that
\begin{equation}
    \mathbb{E}(V(\Delta)^2) \in O\left(D_s^4D^{-2}\right)=O\left(1/D_x^4\right)
\end{equation}
Chebyshev's inequality then shows us that the probability that $|D_x^2 V(\Delta) - D_x^2 \mathbb{E}(V(\Delta))|\ge k$ is at most $O(1/k^2D_x^4)$. Thus for all but a set of zero measure the claimed result holds.
This gives us our final result.
\end{proof}

The value of $t$ yielded by Lemma~\ref{lem:equilibrium} can provide us with resource estimates for the number of queries needed to sample from the position distribution for the approximate Gibbs state over the position and momentum for the system.  Specifically, we follow a strategy similar to the local analysis in that we use quantum simulation algorithms to prepare the distribution.  Unlike the previous discussion, we will be able to simply use qubitization rather than the more complicated truncated Dyson series method because the approach used here only requires a time-independent Liouvillian.

A major limitation of the prior argument is that in the continuum limit that we expect $\gamma\rightarrow 0$ and thus increasing care needs to be made about the definition of equilibration.  Specifically, stronger assumptions of the form that for any observable of interest, $A$, we have that for all $x$ in an set $\Sigma$ with $\lambda_x$ with corresponding eigenvectors $\ket{\lambda_x}$ we have that for all $x,y\in \Sigma$ $|\bra{\lambda_x} A \ket{\lambda_x} - \bra{\lambda_y} A \ket{\lambda_y}|\le \delta$ if $|\lambda_x - \lambda_y|\le \gamma$.  Under such continuity assumptions, we can treat all states within the energy band as behaving essentially the same with respect to our observables of interest and allow the prior arguments to be extended to sufficiently well behaved continuum systems.

\subsection{Block Encoding of Nos\'e Liouvillian}
Our approach to finding the global optimum using the above approach boils down to performing a Hamiltonian simulation of the Liouvillian for long enough for the system to reach microcanonical equilibrium over the extended space.  In order to use quantum simulation methods such as qubitization to simulate the Liouvillian, we need first to construct a block encoding of it.  We will approach the block encoding in two phases.  First we will discuss the block encoding of the part of the Liouvillian that does not include the objective function (that is, the potential in the Hamiltonian) and then later discuss block encoding the potential term and add the results to achieve the total block encoding.  The first of these results is given below.

\begin{lem}[Block-encoding of the discretized Nos\'e Liouvillian for global optimization]
     Assume that the derivative operators in the Koopman-von Neumann Hamiltonian are given by degree $2d_x, 2d_{p'}, 2d_s,2d_{p_s}$ order finite-difference formulas for positive integers $d_x,d_{p'},d_s,d_{p_s}$.  
     There exists an $(\alpha_{NVT}, a_{NVT}, \epsilon)$-block-encoding of the discretized classical Liouvillian $L^{(NVT)}$ with {normalization constant}
    \begin{equation*}
         \alpha_{NVT} \in O \lb N \frac{p'_{\text{max}}}{m s_{\text{min}}^2} \frac{\ln d_x}{h_x} +   \frac{p_{s, \text{max}}}{Q} \frac{\ln d_s}{h_s} + \lb N \frac{{p'}_{\text{max}}^2}{m s_{\text{min}}^3} + \frac{N k_B T}{s_{\text{min}}} \rb \frac{\ln d_{p_s}}{h_{p_s}} \rb
    \end{equation*}
    and {a number of ancilla qubits}
    \begin{equation*}
        a_{NVT} \in O \lb \log \lb \frac{\alpha_{NVT}}{\epsilon} \rb + \log d \rb
    \end{equation*}
    where $d := \max \{ d_x, d_{p'}, d_s, d_{p_s} \}$.
    This block-encoding can be implemented using
    \begin{equation*}
        \widetilde{O} \lb N \log \lb \frac{g \alpha_{NVT}}{\epsilon} \rb + \log^{\log 3}{\lb \frac{\alpha_{NVT}}{\epsilon} \rb}  + d \log g \rb
    \end{equation*}
    Toffoli gates, where $g := \max \{ x_{\max}/h_x, p_{\max}/h_p, s_{\max}/h_s, {p_s}_{\max}/h_{p_s} \}$ is the maximum number of grid points along any of the cardinal directions for the simulation.
\label{lem:bounds_L_class_NVT}
\end{lem}
\begin{proof}
Proof immediately follows by substituting in zero total charge into the result of Lemma 2 of~\cite{Simon2024Liouvillian} to remove the unnecessary Coulomb Hamiltonian from that work.
\end{proof}

\begin{lem}\label{lem:derivBE}
Assume for constant $c\in \mathbb{Z}$ that we are given block encoding oracles for the potential operator $\mathtt{PREP}_V\in \mathbb{C}^{2^c\times 2^c}$ and $\mathtt{SEL}\in \mathbb{C}^{2^cD_x\times 2^cD_x}$ such that $O_V = \sum_{x} \alpha_j U_j$ such that $U_j$ is a diagonal operator and $(\bra{0}\otimes I) (\mathtt{PREP}_V^\dagger \otimes I) \mathtt{SEL}_V (\mathtt{PREP}_V) (\ket{0}\otimes I) = \sum_\bx f(\bx) \ketbra{\bx}{\bx}/\alpha $.  We then have that we can construct an $(\alpha',O(c+\log d) ,\epsilon)$-block encoding of $i\sum_{i=1}^N (\partial_{p_i} \otimes \partial_{x_i} f(\bx))$ where $N=2^m$ for integer $m>0$ and $\partial_{p_i}$  is a degree $2d_p$ approximation to the discrete derivative operator in the $i^{\rm th}$ momentum direction where 
$$
\alpha' \in O\left( N\alpha \ln(d_p)\ln(d_x)/h_ph_x \right),
$$
using a number of queries to the above oracles that are in $O(1)$ and a number of Toffoli gates that are in $\widetilde{O}(N\log(g) +d\log( \alpha'/\epsilon))$.
\end{lem}

\begin{proof}
First let us define some notation.  Let the finite difference approximation to the momentum derivative of order $2d_p$ be \begin{equation}\partial_p = \sum_{j=0}^{2d_p} \beta_j U_+^{j-d_p}\end{equation} and the corresponding position derivative of order $2d_x$ is 
\begin{equation}
    \partial_x = \sum_{j=0}^{2d_x} \beta'_j U_+^{j-d_x}
\end{equation} and unitary $U_+$ is the modular incrementer.  We use the following algorithm for block encoding the operator.  
\begin{enumerate}
    \item Prepare state $\frac{1}{N} \sum_j \ket{j}_A \otimes \sum_k \sqrt{\beta_k} \ket{k}_B \otimes \sum_{\ell} \sqrt{\beta'_\ell}\ket{\ell}_C\otimes \ket{0}_D\otimes  \ket{\psi}_s$, 
    \item Construct a Toffoli network that performs for any $\ket{j}_A \ket{\psi}_s$ the permutation $ \ket{j}_A {\rm SWAP}_{x,j0} {\rm SWAP}_{p,j 0}\ket{\psi}$ where ${\rm SWAP}_{x,j0}$ swaps position registers $j$ and $0$ and ${\rm SWAP}_{p,j 0}$ swaps momentum registers $j$ and $0$ for the state $\ket{\psi}_s$.
    \item Controlled on register $B$ apply $\sum_k \ketbra{k}{k} \otimes U_+^{k-d_p}$ to register $s$.
    \item Controlled on register $C$ apply $\sum_\ell \ketbra{\ell}{\ell} \otimes U_+^{\ell-d_x}$ to register $s$.
    \item Apply the transformation $(\mathtt{PREP}_V^\dagger \otimes I) \mathtt{SEL}_V (\mathtt{PREP}_V\otimes I)$ to registers $D$ and $s$.
    \item Apply the inverse of the unitaries in steps $4,3,2$ and $1$.
\end{enumerate}

Let us first consider the one-dimensional case which corresponds to $N=1$.  In this case the SWAP steps in step $2$ are unnecessary because the registers are already in the canonical position.  We have from known results that a $2d_p^{\rm th}$ order divided difference formula~\cite{Simon2024Liouvillian} can be constructed in the form of
\begin{equation}
    \partial_p = \sum_{j=0}^{2d_p} \beta_j U_+^{j-d_p},
\end{equation}
where $U_+$ is an adder circuit and from Lemma 9 of~\cite{Simon2024Liouvillian}
\begin{equation}
    \sum_j |\beta_j| = O(\log(d_p)/h_p).
\end{equation}
Similarly we have that
\begin{equation}
    \sum_\ell |\beta_\ell'| = O(\log(d_x)/h_x).
\end{equation}
Similarly we can construct a $2d_x^{\rm th}$-order approximation to the derivative of the potential, as an operator, in the following fashion
\begin{equation}
    \partial_x f(x) = \sum_j \beta'_j  \sum_x V(x) U_+^{j}\ketbra{x}{x} U_+^{-j} 
\end{equation}
where $\sum_j |\beta'_j| = O(\log(d_x)/h_x)$.  We therefore have that in the one-dimensional case that the Liouvillian for such a term can be written in the form
\begin{equation}
    \partial_p \otimes (\partial_x V(x)) = \sum_{jj'} \beta_j \beta'_{j'} U_+^{j'} \otimes  \sum_x f(x) U_+^{j}\ketbra{x}{x} U_+^{-j}
\end{equation}
Under the assumption that $V(x)$ has a block encoding with constant $\alpha$ we see from the LCU lemma~\cite{childs2012hamiltonian} that we can block encode the above expression with constant $\alpha'_1$ where
\begin{equation}
    \alpha'_1 \in O(\alpha \sum_j |\beta_j| \sum_{j'} |\beta_{j'}'|) \subseteq O\left(\alpha \ln(d_x) \ln(d_p)/h_xh_p \right).
\end{equation}
Thus for the case where $N=1$ our algorithm is correct.

Let us now argue correctness for $N>1$.  Assume that out of all $N$ registers, our algorithm correctly differentiates each component up to component $P\ge 0$.  Let us now consider the action of the algorithm on the case of $P+1$.  Let us now swap registers $P+1$ and $0$ for both position and momentum and then swap back.  By assumption the circuit will correctly transform $P+1$ in this case because it correctly transforms $0$.  Thus the circuit works correctly for all $N$ by induction.

The block-encoding constant for the LCU circuit, $\alpha'$, is straight forward to bound.  It is simply given by the LCU lemma~\cite{childs2012hamiltonian} to be
\begin{equation}
    \alpha' = N \alpha_1' \in O\left(\alpha N \ln(d_x) \ln(d_p)/h_xh_p \right).
\end{equation}

The estimation of the Toffoli gates needed to implement the block encoding is more involved.  We will go through the various stages of the above algorithm to see the ultimate scaling.
Step $1$ of the algorithm requires the preparation of three states on the $A,B,C$ registers.  By assumption, $N=2^m$ and so the state in $A$ requires only Hadamard gates to prepare and thus does not require Toffoli gates.  The states in $B,C$ are generic states and we need to synthesize them within error $\epsilon/(d\alpha')$ and require a number of $T$ gates on the order of~\cite{kliuchnikov2013synthesis} 
\begin{equation}
    \tilde{O}((d_x + d_p)\log((d_x +d_p)\alpha'/\epsilon))=\tilde{O}(d\log(\alpha'/\epsilon))
\end{equation}
A $T$ gate can be implemented using $O(1)$ Toffoli gates using a three qubit resource state of the form ${\rm QFT} \ket{001}$.  Thus the number of Toffoli gates needed to implement step $1$ is given by
\begin{equation}
    {N_{\rm toff,1}} = \tilde{O}\left( d\log\left(\frac{\alpha'}{\epsilon} \right) \right) \label{eq:N1}
\end{equation}

Step $2$ of the algorithm requires us to implement a sequence of $N$, $\log(N)$-controlled Fredkin gates.  Defining $g=\max \{ x_{\max}/h_x, p_{\max}/h_p, s_{\max}/h_s, {p_s}_{\max}/h_{p_s}$, each Fredkin gate swaps registers consisting of at most $\log(g)$ qubits and thus requires $O(\log(g))$ Toffoli gates using the standard controlled SWAP implementation~\cite{nielsen2010quantum}.  An $O(\log(N))$-controlled not gate requires $O(\log(N))$ Toffoli gates to implement it~\cite{nielsen2010quantum} and thus the total number of operations needed to perform all the controlled swap operations is 
\begin{equation}
    N_{\rm toff, 2} = \tilde{O}\left( N \log(N) \log(g) \right) \subseteq \tilde{O}\left( N \log(g) \right).\label{eq:N2}
\end{equation}

Step $3$ requires the application of an in-place adder circuit with input $B$ and output $s$.  The cost of implementing such an adder is $O(d_p \log(g))$~\cite{gidney2018halving} as at most $g$ bits of precision are used for the momentum registers.  This implies that
\begin{equation}
    N_{\rm toff, 3} = O(d_p \log(g)).\label{eq:N3}
\end{equation}
The same argument also implies
\begin{equation}
    N_{\rm toff, 4} = O(d_x \log(g)).\label{eq:N4}
\end{equation}

Step $5$ requires no Toffoli gates and Step $6$ requires precisely as many Toffoli gates as steps $1,2,3,4$.  Thus we have from~\eqref{eq:N1},\eqref{eq:N2}, \eqref{eq:N3} and~\eqref{eq:N4} that the total number of Toffoli gates needed is
\begin{equation}
    N_{\rm toff} = O\left(\sum_i N_{\rm toff,i}\right) \subseteq \tilde O\left( N\log(g) + d \log(\alpha'g/\epsilon)\right) .
\end{equation}
\end{proof}

The above result gives us the cost of performing a block-encoding of the Liouvillian for the Nos\'e Liouvillian.  Our next step, following Koopman-von Neumann, is to simulate the exponential of the operator for the equilibration time proven above for a discrete Hamiltonian in Lemma~\ref{lem:equilibrium}.  This result, given below, is a formal statement of one of our main results which provides asymptotic bounds on the cost of finding a global optimum using our approach.

\begin{thm}[Global Optimization Theorem]\label{thm:mainGlobal}
    Let us assume that the assumptions of Lemmas~\ref{lem:equilibrium} and~\ref{lem:derivBE} hold for $T$ negligibly small and let $t>0$ be an evolution time. Further, let $\delta>0$ be an error tolerance and $\Delta$ be an energy window with the set of positions, $\bx$, within this energy window denoted $B(\Delta)$.  Additionally, let us assume that for evolution under a $2d$-order approximation to the derivative operators of the form given in Lemma~\ref{lem:derivBE} is used to approximate the Liouvillian operator and we have that the position average for any microstate $\ketbra{\phi}{\phi}$ for the discrete microcanonical state, $\tilde{\rho}_{mc}$, and the corresponding average over the same domain in the continuum microcanonical distribution, $\rho_{mc}$, satisfies 
    $$\sup_{\phi \in B(\Delta)}|{\rm Tr}(\ketbra{\phi}{\phi}({\rm Tr}_{s,p,p_s} (\tilde{\rho}_{mc}  - \rho_{mc} )|\le \delta/3.$$ 
    Then there exists a quantum algorithm that can prepare a time averaged distribution, $\langle\tilde{\sigma}\rangle_{t}$, over the discrete position variables that is a $\delta$-close approximation to the microcanonical distribution: 
    $$\sup_{\phi \in B(\Delta)}\left|{\rm Tr}\left( \ketbra{\phi}{\phi}\left(\mathbb{E}_{{\tilde{\rho}_0}}({\rm Tr}_{s,p,p_s}(\langle\tilde{\sigma}\rangle_{t})) - {\rm Tr}_{s,p,p_s}({\rho}_{mc})\right)\right)\right|\le \delta$$,
     that uses a number of queries to the $O_{f',k}^{(p)}$ oracles (and  their inverses) as well as a number of Toffoli gates that are in
    \begin{equation}
        N_{\rm query,MC}\in \tilde{O}\left( \frac{\alpha'' N}{\gamma \delta^2}\right),\qquad N_{\rm Toff, MC} \in \tilde{O}\left(\frac{
        \alpha''\left((N+d)\log(g) \right)}{\gamma \delta^2} \right),
    \end{equation}
    respectively, where $\alpha''$ is defined via
    $$
    \alpha'' := \frac{\ln d}{h_{\min}}\left(  N\left(\frac{p'_{\text{max}}}{m s_{\text{min}}^2} +{\frac{\alpha \ln(d)}{h_{\min}}} +\frac{{p'}_{\text{max}}^2}{m s_{\text{min}}^3} \right) +   \frac{p_{s, \text{max}}}{Q}  \right), 
    $$
    where we assume that $V = \sum_j \alpha_j U_j$ for unitary $U_j$ with $\sum_j |\alpha_j| = \alpha$.
\end{thm}
\begin{proof}

Lemma~\ref{lem:equilibrium} implies that we can achieve a total variational distance in the state of $\delta/3$ by randomizing the evolution time $t'$ to be in $[0,t]$ and simulate the dynamics of
\begin{equation}
    \tilde{\rho}_0 \rightarrow e^{-i L_{TOT} t'} \tilde{\rho}_0 e^{i L_{TOT} t'},
\end{equation}
where $L_{TOT}$ is the total Liouvillian operator acting on the Koopman-von Neumann wave function and
\begin{equation}
    t'\le t\in O(D_x \sqrt{V(\Delta)} /\gamma \delta^2).
\label{eq:tbound}
\end{equation}
We perform the simulation using qubitization~\cite{low2019hamiltonian}.  The dominant contribution to the cost is given by the number of queries made to the $\mathtt{PREPARE}$ and $\mathtt{SELECT}$ which block encode the Liouvillian.  Specifically, these operations are combined to construct the following walk operator~\cite{low2019hamiltonian,gilyen2019quantum}
\begin{equation}
    W = \mathtt{SELECT} (I - 2 \mathtt{PREPARE}\ketbra{0}{0} \mathtt{PREPARE}^\dagger\otimes I)) 
\end{equation}
The action that $W$ takes on eigenvectors $\ket{\lambda_j}$ augmented by the state $\ket{\alpha} :=\mathtt{PREPARE}\ket{0}$ lies in a two-dimensional subspace spanned by~\cite{low2019hamiltonian,babbush2018encoding}
\begin{equation}
    W: {\rm span}(\ket{\alpha}\ket{\lambda_j}, W\ket{\alpha}\ket{\lambda_j}) \mapsto {\rm span}(\ket{\alpha}\ket{\lambda_j}, W\ket{\alpha}\ket{\lambda_j})
\end{equation}
The eigenvalues can be found using an analysis that is nearly identical to that of Grover's algorithm to see that they are of the form $\exp(\pm i\phi_j)$ where
\begin{equation}
    \phi_j = \arccos(\lambda_j/\alpha''),
\end{equation}
where $\alpha''$ is the overall block-encoding constant of the operator.
Generalized quantum signal processing can be used to transform the eigenvalues $\exp(\pm i \phi_j)$ to $\exp(-i \lambda_j t' )$ through the Jacobi-Anger expansion~\cite{low2019hamiltonian,motlagh2024generalized}. This process requires $K$ calls to $W$ and $O(K)$ single qubit rotations where
\begin{equation}
    K\in O(\alpha'' t' + \log(1/\epsilon)/\log\log(1/\epsilon)).
\label{eq:Kbd}
\end{equation}

If we let $L_{NVT}$ refer to the Liouvillian for the case where the external potential is zero that is studied in~Lemma~\ref{lem:bounds_L_class_NVT} then the total Liouvillian is of the form
\begin{equation}
    L_{TOT}= L_{NVT} + i\sum_{i=1}^N (\partial_{p_i} \otimes (\partial_{x_i} V(x))
\end{equation}
Recall here that because position and momentum commute the ordering of the derivative is irrelevant in the above expression.
From Lemma~\ref{lem:derivBE} we have that if we take $\partial_{p_i}$ and $\partial_{x_i}$ to refer to a $2p$-order finite difference approximation to the derivative then
\begin{equation}
    \|\alpha' (\bra{0}\otimes I) (\mathtt{PREP}_V^\dagger \otimes I) \mathtt{SEL}_V (\mathtt{PREP}_V) (\ket{0}\otimes I) - (\partial_{p_i} \otimes (\partial_{x_i} V(x))\|\le \epsilon.
\end{equation}
Thus if $U_{NVT}$ is the unitary that block-encodes $L_{NVT}$ given in Lemma~\ref{lem:bounds_L_class_NVT} it then follows that
\begin{equation}
    \|(\bra{0}\bra{0}\otimes I)\alpha_{NVT} U_{NVT} + \alpha' I\otimes (\mathtt{PREP}_V^\dagger \otimes I) \mathtt{SEL}_V (\mathtt{PREP}_V\otimes I) (\ket{0}\otimes I) (\ket{0}\ket{0}\otimes I)-L_{NVT}\|\le 2\epsilon.
\end{equation}
Thus using the LCU Lemma~\cite{childs2012hamiltonian} we can build a unitary that provides a block encoding of $L_{NVT}$ with constant $\alpha' + \alpha_{NVT}$ and uses $O(1)$ queries to $\mathtt{PREP}_V$ and controlled $\mathtt{SEL}_V$ along with a single qubit rotation that can be implemented using $O(\log(1/\epsilon))$ Toffoli gates which is sub-dominant to the other costs involved in the block-encoding.

Before going into the number of Toffoli gates and queries to the phase oracle needed for the algorithm, let us digress and discuss the scaling of the block-encoding constant which we denote $\alpha''$ for the total Liouvillian $L_{TOT}$.  We have from Lemmas \ref{lem:bounds_L_class_NVT} and~\ref{lem:derivBE} that at temperature $0$ for the heat bath taking $d_s=d_p=d_x=d$
\begin{align}
    \alpha'' &= \alpha' + \alpha_{NVT}\nonumber\\
    &= O \lb N \frac{p'_{\text{max}}}{m s_{\text{min}}^2} \frac{\ln d_x}{h_x} +   \frac{p_{s, \text{max}}}{Q} \frac{\ln d_s}{h_s} + \lb N \frac{{p'}_{\text{max}}^2}{m s_{\text{min}}^3}  \rb \frac{\ln d_{p_s}}{h_{p_s}} \rb +O\left( N\alpha \ln(d_p)\ln(d_x)/h_ph_x \right)\nonumber\\
    &\subseteq O\left(\frac{\ln(d)}{h_{\min}}\left(N\left(\frac{p_{\max}'^2}{m_{\min} s_{\min}^3} +\frac{p'_{\text{max}}}{m s_{\text{min}}^2}+\frac{\alpha\ln(d)}{h_{\min}}\right) \right) + \frac{p_{s,\max}}{Q} \right)
\end{align}
which conforms to our claims about the value of $\alpha''$.  From this we immediately see that the number of queries to controlled $\mathtt{SEL}_V$ and $\mathtt{PREP}_V$ and their inverses needed to perform a simulation within $\delta$ error in the operator norm scales as~\cite{low2019hamiltonian}
\begin{align}
    O(\alpha''t + \log(1/\epsilon))&= \tilde{O}\left( \alpha''\left(\frac{D_x \sqrt{V(\Delta)}}{\gamma \delta^2}\right) +\log(1/\delta) \right)\nonumber\\
    &= \tilde{O}\left( \alpha''\left(\frac{D_x \sqrt{V(\Delta)}}{\gamma \delta^2}\right)\right)
\end{align}
Next, using the fact that 
\begin{equation}
    \|e^{-i Lt} - e^{-i\tilde{L}t}\| \le \|L-\tilde{L}\|t,
\end{equation}
we have that if the error in the block-encoded Liouvillian is $\epsilon$ then we need to provide a block-encoding with error $\delta/3$ in aggregate for the simulation.  This can be achieved by choosing
\begin{equation}
    \epsilon = O(\delta/t)= O\left(\frac{\delta^3 \gamma}{D_x \sqrt{V(\Delta)}} \right)= O\left({\delta^3 \gamma}{} \right)\label{eq:simerrbd}
\end{equation} to achieve our desired final error of $\delta/3$.

Using the logarithmic block-encoding discussed in Lemma~\ref{lem:hamt_liouvillian}, we have that the number of queries to $O_{f'}$ needed to simulate a query our block encoding unitary $(\mathtt{PREP}_V^\dagger \otimes I) \mathtt{SEL}_V (\mathtt{PREP}_V\otimes I)$  $\mathtt{SEL}_V$ is in $O(N \log(1/\epsilon))$.  This implies that the number of queries needed to the phase oracles, $O_{f',k}^{(p)}$, is in
\begin{equation}
    \widetilde{O}\left(\frac{\alpha'' N}{\delta^2 \gamma} \right).
\end{equation}

The number of Toffoli gates needed to block-encode the potential term from the Liouvillian is given by Lemma~\ref{lem:derivBE} to be
\begin{equation}
    \tilde{O}\left( \left(\frac{\alpha''D_x \sqrt{V(\Delta)}}{\delta^2 \gamma} \right)\left(N \log(g) +d \log \lb \frac{NT\alpha \log^2d}{\delta h_{\min}^2} \rb \right) \right)=\tilde{O}\left( \left(\frac{\alpha''D_x \sqrt{V(\Delta)}}{\delta^2 \gamma} \right)(N+d) \log\left({g}\right)  \right).
\end{equation}
The use of the block-encoding requires a constant number of additional controls on each Toffoli gate, which leads to a constant factor extra number of Toffoli gates~\cite{barenco1995elementary}.

By applying the exact same reasoning, we can see that the number of Toffoli gates needed to implement the NVT Liouvillian is
given by Lemma~\ref{lem:bounds_L_class_NVT} to be
\begin{align}
     &\tilde{O}\left( \left(\frac{\alpha''D_x \sqrt{V(\Delta)}}{\delta^2 \gamma} \right)\left(N \log \lb \frac{g \alpha_{NVT}}{\epsilon} \rb + \log^{\log 3}{\lb \frac{\alpha_{NVT}}{\epsilon} \rb}  + d \log g \right)\right)\nonumber\\
     &\quad=\tilde{O}\left( \left(\frac{\alpha''D_x \sqrt{V(\Delta)}}{\delta^2 \gamma} \right)(N+d) \log\left({g}\right)  \right)=\tilde{O}\left( \left(\frac{\alpha''}{\delta^2 \gamma} \right)(N+d) \log\left({g}\right)  \right).
\end{align}
Summing both of these block encodings requires a single qubit rotation using the LCU lemma~\cite{childs2012hamiltonian}.  This additional rotation requires a number of $T$ gates that varies poly-logarithmically with $\alpha''/(\gamma\delta)$ and thus does not contribute to the overall $\tilde{O}$ scaling.

The final cost that we need to consider in qubitization is the rotations needed to transform the eigenvalues of the walk operator.  The total number of these rotations are from~\eqref{eq:Kbd} in~\cite{low2019hamiltonian} 
\begin{equation}
    K\in O(\alpha'' t + \log(\alpha'' t/\delta)),
\end{equation} which is subdominant to the cost of the oracle queries.

Next let us consider the error tolerances used here.  Let $\tilde{\sigma}$ be the approximation to the evolved state that we have in the discretized space.  We then have from the triangle inequality that
\begin{align}
    \sup_{\phi \in B(\Delta)}|{\rm Tr}(\ketbra{\phi}{\phi}({\rm Tr}_{s,p,p_s} (\langle\tilde{\sigma}\rangle_{t}  - \rho_{mc} )|&\le \sup_{\phi \in B(\Delta)}|{\rm Tr}(\ketbra{\phi}{\phi}({\rm Tr}_{s,p,p_s} (\mathbb{E}_{\rho_0}(\langle\tilde{\sigma}\rangle_{t}  - \langle\tilde{\rho}\rangle_{t}) )|\nonumber\\
    &\qquad+ \sup_{\phi \in B(\Delta)}|{\rm Tr}(\ketbra{\phi}{\phi}({\rm Tr}_{s,p,p_s} (\mathbb{E}_{\rho_0}(\langle\tilde{\rho}\rangle_{t}  - \tilde{\rho}_{mc}) )|\nonumber\\ &\qquad + \sup_{\phi \in B(\Delta)}|{\rm Tr}(\ketbra{\phi}{\phi}({\rm Tr}_{s,p,p_s} (\tilde{\rho}_{mc}  - \rho_{mc} )|.
\end{align}
Note the last term ignores the average over the initial state because microcanonical distribution is independent of the initial state chosen for the average over $\Delta$.
From our assumptions on the discretization error in the microcanonical distribution we then have that 
\begin{align}
    \sup_{\phi \in B(\Delta)}|{\rm Tr}(\ketbra{\phi}{\phi}({\rm Tr}_{s,p,p_s} \mathbb{E}_{\rho_0}(\langle\tilde{\sigma}\rangle_{t}  - \rho_{mc} )|&\le \sup_{\phi \in B(\Delta)}|{\rm Tr}(\ketbra{\phi}{\phi}({\rm Tr}_{s,p,p_s} \mathbb{E}_{\rho_0}(\langle\tilde{\sigma}\rangle_{t}  - \langle\tilde{\rho}\rangle_{t} )|\nonumber\\
    &\qquad+\sup_{\phi \in B(\Delta)}|{\rm Tr}(\ketbra{\phi}{\phi}({\rm Tr}_{s,p,p_s} \mathbb{E}_{\rho_0}(\langle\tilde{\rho}\rangle_{t}  - \tilde{\rho}_{mc} )|+\delta/3.
\end{align}
Then from~\eqref{eq:tbound} we have that the value of $t$ is chosen such that the error in the discrete microcanonical average is $\delta/3$ as well implying that
\begin{align}
    \sup_{\phi \in B(\Delta)}|{\rm Tr}(\ketbra{\phi}{\phi}({\rm Tr}_{s,p,p_s} \mathbb{E}_{\rho_0}(\langle\tilde{\sigma}\rangle_{t}  - \rho_{mc} )|&\le \sup_{\phi \in B(\Delta)}|{\rm Tr}(\ketbra{\phi}{\phi}({\rm Tr}_{s,p,p_s} \mathbb{E}_{\rho_0}(\langle\tilde{\sigma}\rangle_{t}  - \langle\tilde{\rho}\rangle_{t} )|+2\delta/3.
\end{align}
Then from~\eqref{eq:simerrbd} we have that
\begin{align}
    \sup_{\phi \in B(\Delta)}|{\rm Tr}(\ketbra{\phi}{\phi}({\rm Tr}_{s,p,p_s} \mathbb{E}_{\rho_0}(\langle\tilde{\sigma}\rangle_{t}  - \rho_{mc} ))|&\le \delta,
\end{align}
as required.
\end{proof}

This shows that if we assume that the eigenvectors within the $\Delta$-window are distributed according to a unitary $2$-design then we can achieve the microcanonical distribution for the discretized dynamics.  This does not directly imply though that the algorithm is capable of finding the minima as the continuum dynamics is promised to be close to the $T\approx 0$ canonical distribution.  

This result has a number of advantages as well as disadvantages relative to the local approach which uses dissipation to prepare the target state.  Our global optimization algorithms requires the Hamiltonian of the extended system to be ergodic, which is a consequence of the eigenstate thermalization hypothesis taken above. In the event that the discretization error in the simulation is negligible we see that the quantum algorithm will reach the microcanonical distribution provided that the Liouvillian is gapped.  If the Liouvillian is not gapped, then this approach will not necessarily equilibrate, but the expectation values of observables may still approach their microcanonical expectations depending on the structure of the eigenvectors.  Regardless, this dependence on the gap is a major weakness of the global rather than the local approach and so there is not necessarily a clear victor when we compare this approach to the prior approaches involving frictional Hamiltonians.

\section{Comparison to Gradient-Based Methods}
\label{sec:barren}

Perhaps the most vexing problem facing gradient descent optimization is the barren plateau effect.  The barren plateau effect refers the observation that a randomly chosen variational model with sufficiently complicated dependence on its parameters will, with high probability, have derivatives that are exponentially small~\cite{mcclean2018barren}.  There are a number of related effects that can lead to exponentially small gradients such as entanglement between the visible and latent spaces in quantum machine learning models~\cite{ortiz2021entanglement} or the presence of noise in the system~\cite{wang2021noise}.  This is problematic in part because the optimal complexity for computing a gradient scales as $\Theta(1/\epsilon)$~\cite{gilyen2019optimizing} and so learning even the sign of a gradient requires $O(2^n)$ queries to an oracle that yields the objective function as a phase.

Barren plateau effects reveal that the cost of gradient computation can be exponentially large for problems that exhibit barren plateaus.  As our algorithms do not explicitly compute the gradient of the objective function there is a potential for substantial computational advantages in cases with vanishing gradients. 

We consider two settings in the following.  First, we apply both our general purpose dynamical simulation algorithms and hybrid quantum classical gradient descent to the benchmark problem of a quadratic objective function.  This benchmark is chosen because convergence guarantees can be made here, despite the fact that special purpose algorithms can use this structure to accelerate convergence~\cite{Jordan2005gradient}. Recall that almost all twice differentiable functions can be closely approximated by a quadratic convex function in a neighborhood about a local optimum so this assumption is more generic than it first seems.  In this case, the barren plateau effect is not appropriate as the interesting barren plateau cases do not exhibit strongly convex landscapes and assuming that we start near the local optimum violates the typicality assumptions.  Instead, here we will consider the effects of what happens when the condition number of the Hessian matrix is large.  This corresponds to settings where the function changes much more rapidly in certain directions than others.  We will then provide evidence of substantial advantages for cases where the Hessian matrix is typical of an ensemble of random positive-semi definite Gaussian matrices known as Wishart matrices.  In this sense, the analysis carried out here is analogous to the analysis of the barren plateau effect in~\cite{mcclean2018barren} but now with a random matrix model that is more appropriate for our setting of convex optimization.

Second, we will consider the impact of barren plateaus through the lens of our global approach to coherent quantum optimization.  We focus on the global approach here because the optimization landscape implicit in the barren plateau problem is highly non-convex.  As a result, it is much more appropriate for us to consider such a method to try to find an optimum.  We will see that, while the method of Theorem~\ref{thm:mainGlobal} still suffers from barren plateau effects, the upper bounds that we can prove for the cost are significantly better for the dynamical approach of Theorem~\ref{thm:mainGlobal} than its hybrid quantum-classical analogue developed below.

\subsection{Performance for Ill-Conditioned Convex Optimization}
\label{sec:ill_conditioned}

Let us now consider the case of solving the continuous quantum optimization problem under the assumptions of a strongly convex objective function $f$ with an ill-conditioned Hessian matrix.  In such cases where one of the directions of descent is much more rapid than the other directions, the learning rate needed for convergence of gradient descent can be quite small which leads to a prohibitive number of steps.  
 
Below, we provide a result that upper bounds the complexity of gradient descent optimization using the gradient calculation method of~\cite{gilyen2019optimizing}, which is known to be the optimal quantum gradient estimation method for a wide class of functions~\cite{gilyen2019optimizing,van2020convex}.

\begin{thm}[Cost of Gradient Optimization Using Quantum Algorithms]
\label{thm:quantumGradDesAlg}
    Let $f: \mathbb{R}^N \rightarrow \mathbb{R}$ be a twice differentiable strongly convex function with a global minimum located at $\bx = \bx^*$ and let it be promised that $\sup_{\bx} \|\nabla f(\bx)\|_\infty \le f'_{\max}$. Further, let the eigenvalues of the Hessian matrix of $f$ lie in the interval $[\lambda_{\min},\lambda_{\max}]$ for all $\bx$ and let $\bx_0$ be a vector of initial coordinates. If the gradient calculation method of~\cite{gilyen2019optimizing} is used to compute the gradients of $f$ and gradient descent is used with a learning rate of $\eta = 2/(\lambda_{\min}+\lambda_{\max})$ then the number of queries needed to the phase oracle $O_f^{(p)}$ to solve the Continuous Quantum Optimization Problem with high probability of success scales in the worst case as
    $$
         N_{\rm queries, hyb}\in O\lb\frac{N f'_{\max}}{\lambda_{\max} } \left( \frac{2L\|\bx_0 - \bx^*\|}{\epsilon}\right)^{(\lambda_{\max}/\lambda_{\min}) \log(3)/4}\rb,
    $$
    where $L$ is the Lipschitz constant for $f$.
    Alternatively, the above result can be upper bounded by
    $$
        N_{\rm queries, hyb} \in O\lb{N \|\bx_0 - \bx^*\| } \left( \frac{2\lambda_{\max}\|\bx_0 - \bx^*\|^2}{\epsilon}\right)^{(\lambda_{\max}/\lambda_{\min}) \log(3)/4}\rb.
    $$
\end{thm}
\begin{proof}
Under the assumption that we have a strongly convex objective function, we have from~\cite{gilyen2019optimizing} that the number of queries to $O^{(p)}_{f}$ needed to estimate the value of the gradient vector within error $\epsilon_0$ in Euclidean norm with high probability is in
\begin{equation}
    O\left(\frac{\sqrt{N}}{\epsilon_0} \right).
\label{eq:fstGradBd}
\end{equation}
Assuming that we have a strongly convex optimization problem, the number of queries needed to such a gradient oracle to find an $\epsilon_x$-approximation to the optimal value $f(\bx^*)$ using a learning rate of $2/(\lambda_{\max}+\lambda_{\min})$ can be found using the following argument.  Standard bounds on the error in gradient descent given in terms of the maximum and minimum eigenvalues of the Hessian matrix of $f$ for all $\bx$ yield~\cite{ryu2016primer}
\begin{equation}
    \epsilon_x\le \left(\frac{\lambda_{\max}-\lambda_{\min}}{\lambda_{\max}+\lambda_{\min}} \right)^{2K} |f(\bx_0) - f(\bx^*)|\le \left(\frac{\lambda_{\max}-\lambda_{\min}}{\lambda_{\max}+\lambda_{\min}} \right)^{2K} L\|\bx_0 - \bx^*\|,
\end{equation}
where $L$ is the Lipschitz constant for the objective function $f$ and $K$ is the number of iterations.  Solving for $K$ yields
\begin{align}
    K \le \frac{1}{2}\frac{\log\lb\frac{L \|\bx_0-\bx^*\|}{\epsilon_x} \rb}{\log\lb\frac{\lambda_{\max}+\lambda_{\min}}{\lambda_{\max}-\lambda_{\min}}\rb}= \frac{1}{4}\frac{\log\lb\frac{L \|\bx_0-\bx^*\|}{\epsilon_x} \rb}{\frac{\lambda_{\min}}{\lambda_{\max}}+ \frac{1}{3}\lb\frac{\lambda_{\min}}{\lambda_{\max}} \rb^3 + \frac{1}{5}\lb\frac{\lambda_{\min}}{\lambda_{\max}} \rb^5+\cdots}\le \frac{1}{4} \frac{\lambda_{\max}}{\lambda_{\min}}\log\lb\frac{L \|\bx_0-\bx^*\|}{\epsilon_x} \rb.
\end{align}

The update rule in gradient descent for learning rate $\eta$ and objective function $f$ is 
\begin{equation}
    \bx \mapsto \bx - \eta \nabla f(\bx) .
\end{equation}
As the norm of the Hessian is at most $\lambda_{\max}$ we have that if $\gamma$ is an error, then
\begin{equation}
    \|\nabla f(\bx) - \nabla f(\bx + \gamma)\| \le \lambda_{\max}\|\gamma\|.
\end{equation}
Thus, if instead we compute an approximate gradient $\tilde{\nabla}f(\bx)$ that is $\epsilon_0$-close to the original gradient we have that
\begin{equation}
    \|{\nabla f(\bx)} - \tilde{\nabla} f(\bx + \gamma)\|\le \| {\nabla f(\bx)} - \nabla f(\bx+\gamma)\| + \|{\nabla} f(\bx+\gamma) -\tilde{\nabla} f(\bx + \gamma) \| \le  \lambda_{\max}\|\gamma\| + \epsilon_0 .
\end{equation}

Now, if $\bx_j$ is the ideal position we have at step $j$ of the gradient descent in absentia of errors and $\bx_j + \gamma_j$ is the actual position that we see at step $j$, the error at the next step is
\begin{align}
    \|\bx_{j+1} - (\bx_{j+1} + \gamma_{j+1})\| = \|\gamma_{j+1}\| =\|\gamma_j - \eta (\tilde{\nabla} f(\bx_j + \gamma_j) - \nabla f(\bx_j))\|\le \|\gamma_j\|(1+\eta \lambda_{\max}) + \eta \epsilon_0.
\end{align}
Thus, summing the formula, we arrive at the following expression for the error in the value of the optimized parameters after $K$ steps of the gradient descent algorithm:
\begin{equation}
    \|\gamma_K\| \le  \eta \epsilon_0 \sum_{j=0}^{K-1} (1+ \eta \lambda_{\max})^j \leq \frac{\epsilon_0(3^{K}-1)}{\lambda_{\max}} \le \frac{\epsilon_0}{\lambda_{\max} \lb \epsilon_x/L\|\bx_0 - \bx^*\| \rb^{(\lambda_{\max}/\lambda_{\min}) \log(3)/4}},
\end{equation}
where we used the fact that $\eta = 2/(\lambda_{\max}+\lambda_{\min})$.
We then have that the error in the final value of the function, under the assumption that each of the partial derivatives are upper bounded by $f'_{\max}$ and using the fact that $\norm{\nabla f} \leq \sqrt{N} f'_{\max}$, obeys 
\begin{equation}
    | f(\bx_K + \gamma_K) - f(\bx_K)| \le \sqrt{N} f'_{\max} \|\gamma_K\|\le \frac{2\sqrt{N} f'_{\max}\epsilon_0}{\lambda_{\max}(\epsilon_x/L\|\bx_0 - \bx^*\|)^{(\lambda_{\max}/\lambda_{\min})\log(3)/4}}.
\end{equation}

If we demand that the error is $ | f(\bx_K + \gamma_K) - f(\bx_K) | = \epsilon/2$ and that $\epsilon_x=\epsilon/2$ we then need to take
\begin{equation}
    \epsilon_0 \in \Theta\left(\frac{\epsilon \lambda_{\max} (\epsilon/[2L\|\bx_0 - \bx^*\|])^{(\lambda_{\max}/\lambda_{\min})\log(3)/4}}{\sqrt{N} f'_{\max}} \right) .
\end{equation}
Substituting this value into~\eqref{eq:fstGradBd} gives us that the number of queries to $O_f^{(p)}$ needed to solve the continuous quantum optimization problem within error $\epsilon$ is in 
\begin{equation}
    O\left(\frac{N f'_{\max}}{\lambda_{\max} } \left( \frac{2L\|\bx_0 - \bx^*\|}{\epsilon}\right)^{(\lambda_{\max}/\lambda_{\min})\log(3)/4}\right)\subseteq O\left(\frac{N f'_{\max}}{\lambda_{\max} } \left( \frac{2\lambda_{\max}\|\bx_0 - \bx^*\|^2}{\epsilon}\right)^{(\lambda_{\max}/\lambda_{\min})\log(3)/4}\right)
\end{equation}
This immediately shows the first of our two results.  The second result then follows from Taylor's theorem since
\begin{equation}
    f'_{\max} \le f'(\bx^*) + \lambda_{\max} \max_\bx\|\bx-\bx^*\|\in O\lb \lambda_{\max} \|\bx_0 -\bx^*\| \rb.
\end{equation}
\end{proof}
This result shows that gradient descent optimization, even when advanced quantum algorithms for gradient evaluation are used, leads to horrendous scaling when dealing with ill-conditioned Hessian matrices for the objective function.  Such ill-conditioned Hessian matrices are expected to emerge in situations where the condition numbers are large.  
By equating the upper bounds given by gradient descent and that of Theorem~\ref{thm:main} we expect to see that the results of Theorem~\ref{thm:main} will provide a tighter upper bound on the cost of finding the optimum in the limit of small $\epsilon$ (keeping all other variables constant and assuming $h_x \sim \sqrt{\epsilon}$) if
\begin{align}
    \frac{\lambda_{\max}}{\lambda_{\min}} \ge \frac{8}{\log(3)} \approx 5.
\end{align}

This suggests that generically, as $\epsilon$ tends to zero, the quantum gradient method of Theorem~\ref{thm:main} will provide an advantage so long as the condition number is not bounded below by a constant.  Similarly, if we consider the limit of large $\|\bx_0 -\bx^*\|$, we note that $\chi_0 \in O({\rm poly}(\|\bx_0-\bx^*\|)$.  Hence, in this limit we also expect to see a substantial advantage.  Similar conclusions can be drawn from appropriate limits of the remaining variables. This suggests that even in the relatively easy case of strongly convex functions, the approaches in Section~\ref{sec:local} provide far lower upper bounds in the query complexity, relative to gradient descent, if the condition number of the Hessian matrix is large and the ratio of the distance to the optimum and the error tolerance are appropriately large.
But when do we expect the condition numbers to be large?

\subsubsection{Typicality Argument for Ill-Conditioning for Wishart Ensembles}

The answer to the question of the scaling behavior of the condition number for a typical matrix depends on the random matrix ensemble considered.  The ensemble we choose for this benchmark comparison are the Wishart matrices which are formed by choosing the matrix elements to be standard normal, $B_{i,j} = \mathcal{N}(0,1)$, and then defining the matrix $A=BB^\dagger \in \mathbb{R}^{N \times N}$.  The simplicity of these matrices is such that we can not only compute the joint distribution of the eigenvalues, as is possible for the celebrated Gaussian Unitary Ensemble, but the order statistics such as the largest and smallest eigenvalues can also be studied~\cite{edelman1988eigenvalues}. Furthermore, Ref.~\cite{anschuetz2025wishart} shows that Wishart matrices describe randomly initialized quantum neural networks.
If we therefore assume that the matrix is well approximated by a Wishart matrix~\cite{edelman1988eigenvalues} then the probability distribution of the Demmel condition number $\kappa_D(A):= \|A\|_F \|A^{-1}\| \ge \lambda_{\max}/\lambda_{\min}$ is
\begin{align}
P(\kappa_D\ge x) =\frac{N(N^2-1)}{x^2}  +o(1/x^2). 
\end{align}
This means that with high probability the condition number for such random matrices will be $\lambda_{\max}/\lambda_{\min} \le \kappa_D \le O(N^{3/2})$.  
Under these circumstances, we anticipate that the number of queries made to $O_f^{(p)}$ scales as
\begin{equation}
    N_{\rm queries, hyb} \in N \left(\frac{2 \lambda_{\max} \|\bx_0 -\bx^*\|^2}{\epsilon} \right)^{O(N^{3/2})}.
\label{eq:gradientEstimate}
\end{equation}
Thus, under these assumptions existing gradient descent methods are expected to be inefficient even if optimal quantum methods for gradient evaluation~\cite{gilyen2019optimizing} are used.  While the Demmel condition number is used as an upper bound in place of the actual condition number, other work shows qualitatively similar bounds (at the price of increased complexity) on the condition number without going through the Demmel condition number~\cite{chen2005condition}. Further, the data in Figure~\ref{fig:wishart} also supports this scaling and so we conjecture that the scaling of the condition number as $O(N^{3/2})$ is tight up to sub-polynomial factors.

\begin{figure}
    \centering
    \includegraphics[width=0.66\linewidth]{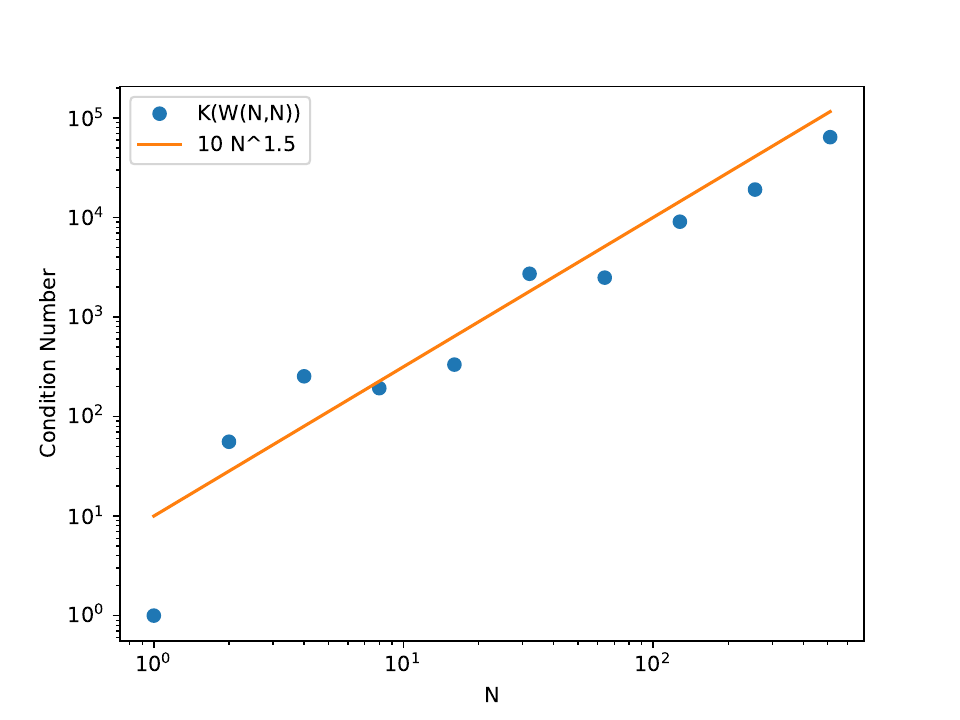}
    \caption{Mean condition numbers for $1000$ Wishart Matrices of size $N\times N$ for $N$ varying exponentially from $2^0$ to $2^{9}$.  Data is consistent with an $N^{3/2}$ powerlaw dependence on the dimension $N$.  Note that the condition number for $N=1$ is always $1$ for any non-zero matrix.}
    \label{fig:wishart}
\end{figure}

Now let us consider the complexity of solving the analogous problem using the result of Theorem~\ref{thm:main}.  Treating $x_{\max}$ as a constant, the cost of performing this optimization using quantum dynamics with high probability of success is
\begin{align}
    N_{\rm queries, qdyn} &\in  \widetilde{O} \lb \frac{N}{\epsilon} \frac{\lambda_{\max}^2}{\lambda_{\min}^2} \alpha_{A_H} \chi_0  \rb\nonumber\\
    &\subseteq \widetilde{O} \lb \frac{N^2}{\epsilon} \frac{\lambda_{\max}^2}{\lambda_{\min}^2} {\norm{D_x^2}}  \chi_0  \rb\nonumber\\
    &\subseteq \widetilde{O} \lb \frac{N^2}{\epsilon} \frac{\lambda_{\max}^2}{h_x^2\lambda_{\min}^2}   \chi_0  \rb
\end{align}
where $h_x$ is the spatial grid spacing in one dimension for the simulation.

Making the assumption that the distance scales for the initial dynamical distribution for position and momentum are dominated by the Euclidean distance between the initial point $\bx_0$ and the optimum $\bx^*$, we have that
\begin{align}
     N_{\rm queries, qdyn} \in \widetilde{O} \lb \frac{N^2}{\epsilon} \frac{\|\bx_0 - \bx^*\|^2 \lambda_{\max}^2}{h_x^2\lambda_{\min}}  \rb.
\end{align}
Let us now assume that the Hessian is a Wishart random matrix. This leads to the probability distribution of the largest eigenvalue satisfying
\begin{equation}
    P(\lambda_{\max} \ge x) \le \frac{\sqrt{\pi} 2^{(1-2N)/2}x^{N-3/2}e^{-x/2}}{\Gamma(2^n/2)^2},
\end{equation}
which in turn implies with high probability~\cite{edelman1988eigenvalues} that $\lambda_{\max} \in O(N).$  Combining this with the same bound on $\lambda_{\max}/\lambda_{\min} \le \kappa_D \in O(N^{3/2})$ leads to the following asymptotic scaling:
\begin{align}
     N_{\rm queries, qdyn}\in \widetilde{O} \lb \frac{N^2}{\epsilon} \frac{\|\bx_0 - \bx^*\|^2 N^{5/2}}{h_x^2}    \rb.
\end{align}

The final question that we need to address involves determining the scaling of the spatial discretization $h_x$.  Intuitively, we expect that the momentum distribution will typically vary over to scale as $1/\max_{\bx}f'(\bx)$ from dimensional analysis.  Thus we expect that for typical distributions,
\begin{equation}
    N_{\rm queries, qdyn} \in \widetilde{O} \lb \frac{N^2}{\epsilon} (\max_{\bx} |f'_{\max}(\bx)|^2\|\bx_0 - \bx^*\|^2 N^{5/2})\rb,
\end{equation}
which is quadratically worse in the derivative of $f$ than the hybrid approach.  This scaling, however, is not rigorous as the evolved probability density could in principle have significantly finer scale variations than that expected from dimensional analysis.  A rigorous bound can be established using the discretization bounds of~\cite{kivlichan2017bounding}. We see that, if we promise that the maximum wave number for the simulation is bounded by $k_{\max}$, and the size of the box for the simulation is $O(\|\bx_0 - \bx^*\|)$, the grid spacing $h_x$ should be chosen such that
\begin{align}
    h_x \le \frac{2\epsilon_d}{3N(k_{\max} + f'_{\max}t^*)} \left(\frac{k_{\max} \|\bx_0 - \bx^*\|}{\pi} \right)^{-N}.
\end{align}
This bound assumes the absolute worst case configurations of the quantum particles given the constraint on $k_{\max}$ for the initial distribution. As the dynamics are unitary, we do not have the exponential propagation of error that we can have in the worst case analysis of the gradient descent algorithm. Thus we can take $\epsilon_d = \Theta(\epsilon)$. Substituting this bound in and using the expression for $t^*$ from Corollary~\ref{cor:equilibration_time} we obtain
\begin{align}
    N_{\rm queries, qdyn}&\in \widetilde{O} \lb  \frac{N^{13/2}\|\bx_0 - \bx^*\|^{2(N+1)} k_{\max}^{2N}(k_{\max}^2 + {f'}_{\max}^2{t^*}^2) }{\epsilon^3}    \rb\nonumber\\
    &\subseteq \widetilde{O} \lb  \frac{N^{13/2}\|\bx_0 - \bx^*\|^{2(N+1)} k_{\max}^{2N}(k_{\max}^2 + {f'}^2_{\max}\lambda_{\min})}{\epsilon^3}\rb\nonumber\\
    &\subseteq \widetilde{O} \lb  \frac{N^{13/2}\|\bx_0 - \bx^*\|^{2(N+1)} k_{\max}^{2N}(k_{\max}^2 + {f'}^2_{\max}N)}{\epsilon^3}\rb.
\label{eq:qdnestimate}
\end{align}

Even under the worst case assumptions made in this discussion, we see that the quantum dynamical optimization can provide substantial advantages.  If we neglect sub-dominant multiplicative factors and assume that $k_{\max} \in O(N/\epsilon)$, as suggested by the scaling in~\eqref{eq:pMean}, we then see that the number of queries in the quantum dynamical approach is in 
\begin{equation}
    \widetilde{O}(\|\bx_0 - \bx^*\|^{2+2N} N^{2N}/\epsilon^{5+2N}),
\end{equation}
whereas the optimization method given in~\eqref{eq:gradientEstimate} scales as
\begin{equation}
    \lb \|\bx_0 - \bx^*\|/\epsilon \rb^{O(N^{3/2})}.
\end{equation}
This constitutes a super-polynomial separation in the upper bounds which means that our fully coherent quantum method provides a substantial improvement in accuracy scaling over gradient descent for the ill-conditioned matrices expected from the Wishart ensemble under the above worst case bounds.
Thus, if this conjecture holds, we could expect a large separation to hold generically for large families of random matrices, but lower bounds are needed to definitively show a separation between quantum and classical methods here.

\subsection{Comparison of Global Quantum Optimization to Hybrid Optimization and Barren Plateau Effects}

We will now use the result of the last section to provide a hybrid quantum/classical algorithm that can approximately solve the global optimization problem using a strategy that is analogous to that of Theorem~\ref{thm:mainGlobal}.  Such algorithms are the natural analogues of existing gradient-based methods that are presently used in variational approaches, except here we use a dynamical system to perform the optimization rather than gradient descent to make closer ties to the global optimizer of Theorem~\ref{thm:mainGlobal}.  The key idea behind this algorithm is to use Newton's equations of motion rather than the Liouville equation to solve for the dynamics of the Nos\'e Hamiltonian.  We choose this path, rather than solving the Liouville equation, because storing the probability density explicitly on a classical computer would require exponential memory in the number of variables.  Despite the intractability of representing the probability distribution as a bit string, a perfect solution to Newton's equations of motion would yield a sample drawn from the same probability distribution as Theorem~\ref{thm:mainGlobal}.  For this reason, these two algorithms are natural analogues and the equilibration time needed to reach the microcanonical distribution will be the same for either.  For this reason, we can directly take the desired simulation time to be $t=1/\delta^2 \gamma$. 

There are many families of integrators that we could consider. Most broadly, the family of linear-multistep methods allow high-order approximations to the solution of the differential equation that allow us to achieve with a step size of $h$ error that scales as $O(h^p)$ using a $p$-step formula if the differential equation is sufficiently smooth.  The analysis of the specific dependence of the error on the properties of the non-linear differential equation is complicated (See for example Theorem 8.10 of~\cite{wanner1996solving} and so, for simplicity, we focus our discussion on the forward Euler integrator despite its dramatically inferior scaling with simulation time and error tolerance.  We then discuss the aspects of the strategy that can, and those that cannot, be improved by using high-order multistep methods.

The forward Euler method takes the following form for a differential equation $\dot{z} = F(z)$ for a non-linear function $F$ with Lipschitz constant $L$ and timestep $h$:
\begin{equation}
    z(t+h) -z(t) = hF(z(t)) +O(h^2)
\end{equation}
Standard bounds on the Euler method give us that if we define $\tilde{z}(t)$ with initial condition $z(0)=\tilde{z}(0)$ to be the solution yielded by iterating the Euler recursion relation~\cite{wanner1996solving,suli2003introduction}
\begin{equation}
|z(t)-\tilde{z}(t)| \le \frac{{\rm dim}(\mathbf{z})h\max_{\mathbf{z},k,j}(| \partial_{z_k} F_j(\mathbf{z})|)(e^{Lt}-1)}{2L}\label{eq:eulerBd}
\end{equation}
This bound is an explicit example of the linear-multistep scaling for $p=1$.  

The bound in~\eqref{eq:eulerBd} does not, however, bound the impact that errors in the derivative operator have on the evolution.
Error propagation in the forward Euler method is similarly straight forward as shown in the following lemma.
\begin{lem}\label{lem:Lipschitz}
    Let $\partial_t \mathbf{z}(t) = F(\mathbf{z}(t))$ for Lipschitz-continuous $F:\mathbb{R}^{2(N+1)}\rightarrow \mathbb{R}^{2(N+1)}$ such that $\|F(\mathbf{z}) - F(\mathbf{z}^*)\| \le L \|\mathbf{z}-\mathbf{z}^*\|$ next let $\tilde{\mathbf{z}}(t)$ satisfy $\partial_t \tilde{\mathbf{z}}(t) = \tilde{F}(\tilde{\mathbf{z}}(t))$ such that $\|\tilde{F}(\mathbf{z}) - F(\mathbf{z})\|\le \epsilon_F$ and let $\tilde{\mathbf{z}}(0) = \mathbf{z}(0)$.   Let $\mathbf{z}'$ be the forward Euler approximation to the solution with the exact function $F(\mathbf{z})$  and let and $\tilde{\mathbf{z}}'$ be approximate solutions to the differential equation yielded by the forward Euler method with stepsize $h$ and differential operator $\tilde{F}$.  We then have that for any integer $q\ge 0$
    $$
\| \mathbf{z}'(qh) - \tilde{\mathbf{z}}'(qh) \| \le \frac{\epsilon_F(e^{Lqh}-1)}{L}
    $$
\end{lem}
\begin{proof}
    Proof follows explicitly. Let us assume that we wish to find the solution at step $q+1$ for integer $q\ge 0$.  We then have from the triangle inequality that
    \begin{align}
        \|\mathbf{z}'((q+1)h) - \tilde{\mathbf{z}}'((q+1)h)\| &= \|\mathbf{z}(qh) - \tilde{\mathbf{z}}'(qh) + hF(\mathbf{z})' - h\tilde{F}(\tilde{\mathbf{z}})\|\nonumber\\
        &\le \|\mathbf{z}'(qh) - \tilde{\mathbf{z}}'(qh)\| + h\|F(\mathbf{z}'(qh)) - \tilde{F}(\mathbf{z}'(qh)) \|+h\|\tilde{F}(\mathbf{z}'(qh)) - \tilde{F}(\tilde{\mathbf{z}}'(qh)) \|\nonumber\\
        &\le \|\mathbf{z}'(qh) - \tilde{\mathbf{z}'}(qh)\|(1+Lh) + \epsilon_Fh\nonumber\\
        &\le \|\mathbf{z}'(qh) - \tilde{\mathbf{z}}'(qh)\|e^{Lh} + \epsilon_Fh.
    \end{align}
    This provides an explicit recursion relation that can be solved.  The solution, subject to the assumption that $\mathbf{z}'(0)=\mathbf{z}(0)$, is
    \begin{align}
        \|\mathbf{z}'(qh) - \tilde{\mathbf{z}}'(qh)\| \le \frac{\epsilon_F h (e^{Lqh}-1)}{e^{Lh}-1} \le \frac{\epsilon_F (e^{Lqh}-1)}{L}.
    \end{align}
\end{proof}
\begin{thm}[Hybrid Quantum/Classical Global Optimization]\label{thm:hybridGlobal}
    There exists a hybrid quantum classical algorithm that for an optimization function $f:\mathbb{R}^N \mapsto \mathbb{R}$ that under the assumptions of Theorem~\ref{thm:mainGlobal} can draw a sample that is $\Delta_z$-close to one drawn from a distribution that has total variational distance $\delta$ from the zero-temperature Gibbs distribution $\lim_{\beta \rightarrow \infty} e^{-\beta f(\bx)} /{\rm Tr}(e^{-\beta f(\bx)})$ and a spectral gap of the Liouvillian of $\gamma$ using a number of queries that scales to the phase oracle $O_f^{(p)}$ that scales (with high probability) as
    $$
     \tilde{O}\left( \frac{N^{3/2} M e^{2L/\gamma\delta^2}}{\gamma^3 \delta^6\Delta^2_z} \right)
    $$
    where $M$ is an upper bound on the second derivatives of the force function such that $$M\ge \max\lb \frac{1}{ms_{\min}^2}, \frac{2p_{\max}}{ms_{\min}^3}, \frac{3p_{\max}^2}{ms_{\min}^4},|f''_{\max}|,\frac{1}{Q}  \rb$$
    and $L$ is a Lipschitz constant such that
    $$
L\ge \sqrt{N}\left(\frac{p_{\max}}{ms_{\min}^2}+|f'_{\max}| \right) +\frac{p_{s,\max}}{Q}+ \frac{p_{\max}^2}{ms_{\min}^3}
    $$
\end{thm}
\begin{proof}
Let us begin by taking the time prescribed from Theorem~\ref{thm:mainGlobal}, $T=1/\gamma\delta^2$ to be an integer multiple of the time spacing used in the Euler method.  That is to say we choose $1/\gamma\delta^2 = q_{\max} h$ for positive integer $q_{\max}$.  Then for any such integer $q\le q_{\max}$ we have that the error in the differential equation obeys the following from~\eqref{eq:eulerBd} 
    \begin{align}
        \|\mathbf{z}(qh) -\tilde{\mathbf{z}}'(qh)\|&\le \|\mathbf{z}(qh)-\mathbf{z}'(qh)\| +\|\mathbf{z}'(qh) -\tilde{\mathbf{z}}'(qh)\|\nonumber\\
        &\le \frac{[(N+1)h \max_{\mathbf{z},k,j}(|\partial_{z_k} F_k(\mathbf{z})|)+\epsilon_F](e^{qLh}-1)}{L}
    \end{align}
    Now assuming that we wish to find a solution within a $\Delta_z$-Euclidean ball about $\mathbf{z}(qh)$ it then suffices to take
    \begin{align}
        h &\le \frac{\Delta L}{2(N+1)\max_{\mathbf{z},k,j}(|\partial_{z_k} F_k(\mathbf{z})|)(e^{qLh}-1) }\\
        \epsilon_F &\le \frac{\Delta L}{2(e^{qLh}-1) }
    \end{align}
    Using the gradient method of~\cite{gilyen2019optimizing} we have that the number of queries needed to $O_f^{(p)}$ to compute the gradients, with high probability, is in
    \begin{equation}
         \tilde{O}\left(\frac{\sqrt{N}}{\epsilon_F} \right).
    \end{equation}
    We then have that if we choose $qh = 1/\gamma \delta^2$ that the total number of Euler steps that we need is $(h\gamma\delta^2)^{-1}$.  This implies that the total number of queries to $O_f^{(p)}$ is
    \begin{align}
        \tilde{O}\left(\frac{\sqrt{N}}{\epsilon_F h\gamma \delta^2} \right)\subseteq \tilde{O}\left( \frac{N^{3/2} \max_{\mathbf{z},k,j}(|\partial_{z_k} F_k(\mathbf{z})|) e^{2L/\gamma\delta^2}}{\gamma^3 \delta^6\Delta_z^2} \right).\label{eq:globalDerivIntEq}
    \end{align}
Next we need to relate the quantities $F$ and $L$ to the parameters of the Nos\'e Hamiltonian.  We have that
\begin{equation}
    F(\mathbf{z}) = \left[\frac{d\bx}{d t},\frac{d\bp}{d t},\frac{d s}{d t}, \frac{dp_s}{d t}\right]^T
\end{equation}
and so the the continuity of $F$ can be studied by looking at each term in the force vector independently.  From Hamilton's equations of motion we have that
\begin{equation}
    \left\|\frac{d\bx}{dt}\right\| = \left\|\frac{\partial H_{N}}{\partial \bp}\right\|\le \frac{\sqrt{N} p_{\max}}{ms_{\min}^2}
\end{equation}
Similarly, as the absolute value of the partial derivative of the objective function $|\partial_{x_k}V(\bx)|=|\partial_{x_k}f(\bx)| \le |f'_{\max}|$ for all $\bx$ we then have from standard norm inequalities that
\begin{equation}
    \left\|\frac{d \bp}{d t} \right\| = \left\|\frac{\partial H_N}{\partial \bx} \right\|\le \sqrt{N}f'_{\max}
\end{equation}
Finally we have
\begin{align}
    \left\|\frac{d s}{d t} \right\| &= \left|\frac{\partial H_n}{\partial p_s} \right| \le \frac{p_{s,\max}}{Q}\\
    \left\|\frac{d p_s}{d t} \right\| &= \left|\frac{\partial H_n}{\partial s} \right| \le \frac{p_{\max}^2}{ms_{\min}^3}.
\end{align}
Thus we have from Taylor's theorem that for any two $\mathbf{z},\mathbf{z}^*$
\begin{equation}
    \|F(\mathbf{z}) - F(\mathbf{z}^*)\| \le \left[\sqrt{N}\left(\frac{p_{\max}}{ms_{\min}^2}+|f'_{\max}| \right) +\frac{p_{s,\max}}{Q}+ \frac{p_{\max}^2}{ms_{\min}^3}\right] \| \mathbf{z} - \mathbf{z}^*\|. 
\end{equation}
Thus we may take our Lipschitz constant for the differential equation to be
\begin{equation}
    L=\sqrt{N}\left(\frac{p_{\max}}{ms_{\min}^2}+|f'_{\max}| \right) +\frac{p_{s,\max}}{Q}+ \frac{p_{\max}^2}{ms_{\min}^3}.
\end{equation}
The derivative $\max_{\mathbf{z},k,j}(|\partial_{z_k} F_k(\mathbf{z})|)$ can be similarly bounded.  Following the exact same reasoning used above we can see that
\begin{align}
    \max_{\mathbf{z},k,j}(|\partial_{z_k} F_k(\mathbf{z})|) \le \max\lb \frac{1}{ms_{\min}^2}, \frac{2p_{\max}}{ms_{\min}^3}, \frac{3p_{\max}^2}{ms_{\min}^4},|f''_{\max}|,\frac{1}{Q}  \rb
\end{align}
where $|f''_{\max}|\ge \max_{\bx,j,k} |\partial_{x_j} \partial_{x_k} f(\bx)|$.  Our theorem then follows by substitution of these results into~\eqref{eq:globalDerivIntEq}.  
\end{proof}
This shows that the upper bounds on the forward Euler method's performance are less than inspiring.  We note immediately that the complexity of this problem scales exponentially with $L/\gamma \delta^2$.  This scaling is an inevitable consequence of the non-unitarity of the underlying dynamics of Newton's equations.  Despite this issue, it is worth noting that this argument simply bounds the distance between the approximated trajectory and the correct one.  It does not necessarily show that the resulting positions of $\tilde{\mathbf{z}}'(1/\gamma\delta^2)$ is nonetheless typical of those that we would expect from samples from the phase space distribution discussed in Theorem~\ref{thm:mainGlobal}.  Regardless, this upper bound is likely the best that we can provide without making further assumptions on the underlying differential equation.

For the sake of comparison with Theorem~\ref{thm:mainGlobal}, let us assume that $|f'_{\max}|,|f''_{\max}|,|p_{\max}|, |p_{s,\max}|$, as well as $s_{\min},d,h_{\min},m,Q,\Delta_z$ are all constants and focus on the scaling with $N,\delta,\gamma,\Delta$.  In this case we have that
\begin{equation}
    L = O(\sqrt{N}/\gamma\delta^2),\quad M = O(1).
\end{equation}
In this case, the best upper bound that we can show for the method of Theorem~\ref{thm:hybridGlobal} is
\begin{equation}
    N_{\rm queries,hyb}=\tilde{O}\lb \frac{N^{3/2}e^{2\sqrt{N}/\gamma\delta^2}}{\gamma^3 \delta^6}\rb.\label{eq:hybUB}
\end{equation}
In contrast we have under these assumptions that $\alpha = O(N)$ and so the scaling that we need to find an analogous solution from Theorem~\ref{thm:mainGlobal} is
\begin{equation}
    N_{\rm queries,quant} = O\lb\frac{N^3}{\gamma \delta^2} \rb
\end{equation}
The scaling in all variables is exponentially better here than in the upper bound~\eqref{eq:hybUB}.    Further, any improvement of this exponential separation is unlikely to arise from the use of high-order linear multistep methods because the exponential scaling in the hybrid algorithm arises from Lemma~\ref{lem:Lipschitz} and such optimizations would likely only therefore improve the polynomial factors of the form $1/\gamma^3 \delta^6$ to $\tilde{O}(1/\gamma^2 \delta^4)$ at best.  This shows that we are capable of achieving an exponentially tighter upper bound for the fully quantum methods than we would be with a hybrid algorithm for global optimization.

\subsubsection{Barren Plateau Effects}
With this discussion in place, let us now change gears to discuss how these algorithms perform in the presence of barren plateaus.  The barren plateau effect is a consequence of the complicated ansatzae that emerge in variational quantum optimization problems~\cite{mcclean2018barren}.  Specifically, let us assume in this case that we have an optimization function of the form
\begin{equation}
    f(\bx) = \bra{\psi_0} \prod_{j=1}^N e^{i x_j H_j} \mathcal{F} \prod_{j=N}^1 e^{-i x_j H_j} \ket{\psi_0}
\end{equation}
for Hermitian $\mathcal{F}$ and $H_j$ and $\ket{\psi}\in \mathbb{C}^{2^n}$.  This is typical in chemistry applications and other areas.  Such models can be explicitly differentiated using the product rule.   
Under the assumption that the product of exponentials (and the unitaries in the decomposition of $\mathcal{F}$) in the derivative of $f(\bx)$ are typical of those drawn from a unitary $2$-design then we have that for any $\bx$ that the expectation value of the derivative operator then it can be shown that for any directional derivative, $\partial_{x_k}$, that 
\begin{equation}
    \mathbb{E}_{\bx}\left(|\partial_{x_k} f(\bx)|^2 \right) \in O(2^{-2n}).
\end{equation}
which suggests that for all but a negligible fraction of the parameter space that the derivative of the objective function is exponentially small.  This means that randomly guessing parameters is unlikely to provide a usable gradient for cases where the number of qubits $n$ is large.

We see in these cases for the dynamical simulation that the size of the gradients do not directly enter the discussion of the complexity of the global method.  The gradients instead indirectly enter through the spectral gap of the Liouvillian, $\gamma$, which dictates the time required for the distribution to equilibrate to the microcanonical distribution.  The influence of these small gradients on the spectral gap of the discretized Liouvillian is difficult to assess especially given that the gap itself is challenging to bound alone.  Instead, we will use a dynamical argument to bound the time needed to escape from a barren plateau.

Let us assume that we begin with an initial distribution that has with probability $1$ that $\bp=0,p_s =0$ and assume that $s\in \Theta(1)$.  Further let us assume that $\bx_0$ is drawn uniformly according to the Haar measure.  Under these circumstances we have from Markov's inequality that for $C\in \Theta(2^{2n})$
\begin{equation}
    {\rm Pr}(|f(\bx_0)|^2 \ge C\mathbb{E}_\bx(|\partial_{x_k} f(\bx)|^2))\in O(2^{-2n}).
\end{equation}
Thus with overwhelming probability the initial point drawn will have exponentially small gradients for the objective function $f$ and assume that the optimal value obeys $f(\bx^*)=0$.  Let us further assume that the function $f(\bx)$ is Lipschitz with constant $L'$.  This assumption is needed so that we can talk about the behavior of the objective function along the path taken in the optimization process.  If that is true then for all $\bx'$ such that $\|\bx' - \bx_0\|\in O(\mathbb{E}_{\bx}(f(\bx))/L')$ we must have that
\begin{equation}
    |f(\bx') - f(\bx_0)| \le L'\|\bx' - \bx_0\| \in O(\mathbb{E}_{\bx}(f(\bx))).
\end{equation}
This implies that in order for us to meaningfully reduce the expectation value of the objective function from its average we need to travel a distance of at least $\|\bx' - \bx_0\|\in O(\mathbb{E}_{\bx}(f(\bx))/L')$.

Given that we have assumed that the initial momenta are zero and that $|\partial \bp / \partial t| \in O(\sqrt{N}\max_{k} |\partial_{x_k} f(\bx))|)$ the time we find from Newton's equations that
\begin{equation}
\partial_t^2 x_j = -\partial_{x_j} f(\bx) ~\text{ subject to $\bx(0) = \bx_0$ and $\partial_t \bx_0 =0$}.
\end{equation}
Under the assumption that the objective function does not meaningfully change over this distance.
This implies that the time required to change by a fraction of $\mathbb{E}_\bx(f(\bx))/L'$ is with high probability over $\bx_0$
\begin{equation}
    t\in \Omega\left(\frac{\sqrt{\mathbb{E}_\bx(f(\bx))/L'}}{{N}^{1/4}\sqrt{\max_{j} |\partial_{x_j} f(\bx)|}} \right)=\Omega\left(\frac{\sqrt{2^n\mathbb{E}_\bx(f(\bx))/L'}}{{N}^{1/4}} \right) 
\end{equation}
Now let us assume seeking a contradiction that $1/\gamma \in o(t)$.  If that were the case then the equilibration time for the microcanonical ensemble would be shorter than the time it takes to meaningfully change the objective function.  This is a contradiction hence we must have that 
\begin{equation}
    \frac{1}{\gamma} \in \Omega\left(\frac{\sqrt{2^n\mathbb{E}_\bx(f(\bx))/L'}}{{N}^{1/4}} \right).
\end{equation}
This shows that, unless the Lipschitz constant for the derivatives of the objective function is exponentially large, the mixing time $1/\gamma\delta^2$ will be exponentially large for any constant target total variational distance $\delta$.  This shows, unsurprisingly, that barren plateau effects will typically be present in the global optimization approach of Theorem~\ref{thm:mainGlobal} just as they appear in traditional hybrid algorithms~\cite{mcclean2018barren}.

In particular we see that if we neglect discretization error for the Nos\'e Liouvillian and assume that the minimum time shown above is sufficient for equilibration for constant $N$ and $\delta$ the bounds satisfy
\begin{align}
    N_{\rm queries, quant} &\in O(2^{n/2}/\sqrt{L'})\nonumber\\
    N_{\rm queries, hybrid} &\in 2^{3n/2}e^{O(2^{n/2}/\sqrt{L'})}/{L'}^{3/2}
\end{align}
While barren plateau effects are not alleviated by the approach of Theorem~\ref{thm:mainGlobal}, the upper bounds provided in Theorem~\ref{thm:mainGlobal} are much more tolerant of small gradients than their analogous hybrid algorithms in Theorem~\ref{thm:hybridGlobal}. 
Thus we see a super-exponential separation between the upper bounds in this case which demonstrates a potentially immense advantage in performance of gradient free approaches, such as ours, to optimization relative to hybrid quantum/classical algorithms.  However, lower bounds are necessary to conclusively decide whether these differences are fundamental or artifacts of the upper bounds used in Lemma~\ref{lem:Lipschitz} and elsewhere.

\section{Numerical Experiments for Local Optimization}

Our convergence guarantees cover the quadratic and strongly convex cases for optimization tasks initialized near a local minimum. Rigorous bounds for globally non-convex objective functions are hard to come by, so we complement the theory with numerical studies in small dimensions. The experiments that follow are not intended as a comparison against classical methods; their purpose is to confirm that the predicted convergence behavior manifests in the quadratic case with the expected sensitivity to $\beta$, and to demonstrate that the underlying mechanism works for non-convex landscapes where rigorous bounds are not available.

We illustrate the behavior of the algorithm through simulations in $d=2$. All experiments evolve the wavefunction under the Hamiltonian described in~\ref{def:discretized_ham} using a Trotter-Suzuki decomposition, with hard-wall boundary conditions on $[-2, 2]^2$. The initial state is a Gaussian wave packet $\psi(\mathbf{x},0) \propto \exp(-\|\mathbf{x}-\mathbf{x}_0\|^2/2\sigma_0^2)$ with $\sigma_0 = 0.3$, and we track the expected position $\langle \mathbf{x}\rangle(t)$ as the algorithm iterates. Convergence is declared when $\|\langle\mathbf{x}\rangle(t) - \mathbf{x}^*\| < \varepsilon$ persists for three consecutive checks, where $\mathbf{x}^*$ denotes the known global minimum and with $\varepsilon=0.08$.

Since Theorem~\ref{lem:multivariate_time} provides convergence guarantees only for quadratic objectives, we begin with a detailed study of this case. We take $f(\mathbf{x}) = (\mathbf{x}-\mathbf{x}^*)^\top A(\mathbf{x}-\mathbf{x}^*)$ with $\mathbf{x}^* = (1/2, 1/2)$ and $A = R^\top\operatorname{diag}(100,1)\,R$, where $R$ is a rotation by $\pi/4$, giving condition number $\kappa = 100$. The starting point is $\mathbf{x}_0 = (-1.5, 0.5)$, which lies along the stiff eigendirection and thus probes the regime where ill-conditioning most affects convergence. Figure~\ref{fig:quad} displays the trajectory $\langle\mathbf{x}\rangle(t)$ for four values of the annealing rate $\beta$. At the optimal value $\beta^* \approx 1.11$, the algorithm converges in parameter time $t_{\mathrm{conv}} = 1.68$. Slower annealing ($\beta = 0.29$) still converges but takes roughly four times longer, while a very slow rate ($\beta = 0.05$) fails to converge within the allotted time, with $\langle\mathbf{x}\rangle$ oscillating broadly. A fast rate ($\beta = 2.00$) also fails: the kinetic energy is suppressed too rapidly, freezing the wavefunction before it can traverse the narrow valley to the minimum.

\begin{figure}[h]
    \centering
    \includegraphics[width=\linewidth]{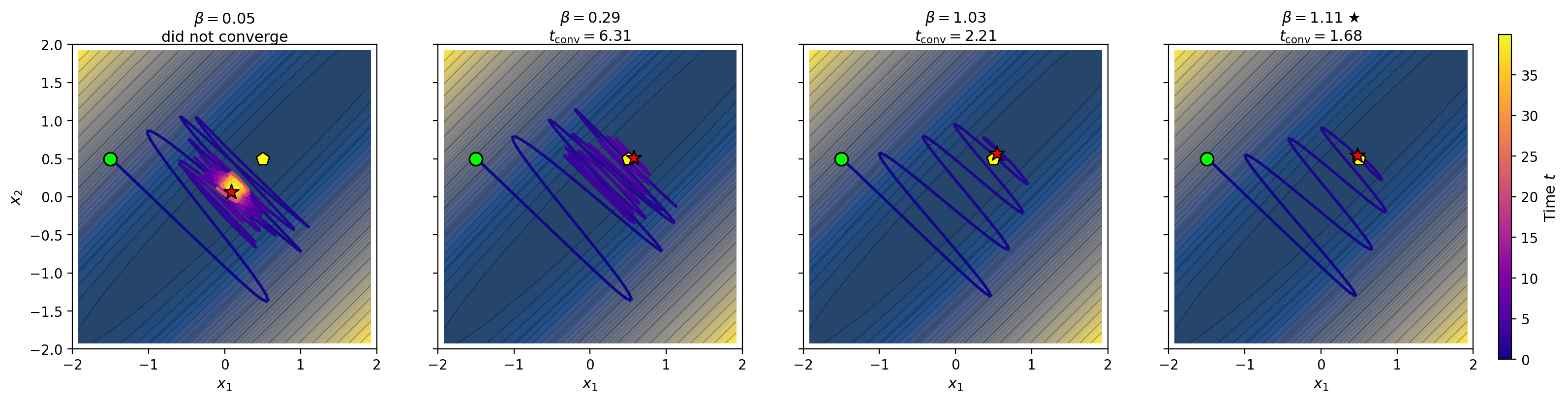}
    \caption{Trajectories for local optimization with a localized wave packet with hard wall boundaries. The target function is quadratic with the minimum at $(0.5,0.5)$ and $\kappa=100$. The green circle indicates the starting point, while the minimum is marked with a yellow pentagon and the final point of the trajectory is marked with a red star.}
    \label{fig:quad}
\end{figure}

For non-globally-convex functions, rigorous convergence bounds cannot be expected in general, so we instead provide numerical evidence that the algorithm localizes correctly on two standard test functions with qualitatively different landscapes. The first is a sum of three Gaussian wells with a quadratic confining term,
\begin{equation}
    \label{eq:numerics_gauss}
    f(\mathbf{x}) = -\sum_{i=1}^3 A_i\exp\!\left(-\frac{\|\mathbf{x}-\mathbf{c}_i\|^2}{2\sigma_i^2}\right) + 0.1\,\|\mathbf{x}\|^2,
\end{equation}
which has a global minimum at $(0.5, 0.5)$ flanked by two local minima at $(-0.6,-0.4)$ and $(-0.3,0.7)$ for $A_i=(10,6,4)$ and $\sigma_i=(0.3,0.25,0.2)$. The second is the scaled Rastrigin function
\begin{equation}
    \label{eq:numerics_rastrigin}
    f(\mathbf{x}) = 2 + \sum_{i=1}^{d}\left[x_i^2 - \cos(2\pi x_i)\right], 
\end{equation}
whose global minimum at the origin is surrounded by a regular lattice of local minima at integer coordinates.

Figure~\ref{fig:3d_1} shows the probability density $|\psi(\mathbf{x},t)|^2$ at $t=0$ and $t=T$ for the three-well potential with $\mathbf{x}_0 = (-1.5,-1.5)$, $\beta=0.3$, and $T=15$. The initial Gaussian is well-separated from all wells. By $t=15$, the wavefunction has split and localized into the three wells, with the dominant peak at the global minimum. Figure~\ref{fig:3d_2} shows the analogous result for the Rastrigin function under the same parameters. It can be seen that the final state distributes across the lattice of local minima. The concentration can be seen clearly in~\ref{fig:prob_concentrate} which shows the probability for a measurement outcome to give a result $<\epsilon=0.1$ close to some local minimum. The concentration increases until $t=15$, at which $P(\textrm{measurement outcome $\epsilon$ close to local minimum}) \approx 0.85$. It can be seen that the dynamics are largely frozen after that point.

\begin{figure}[h]
    \centering
    \includegraphics[width=0.8\linewidth]{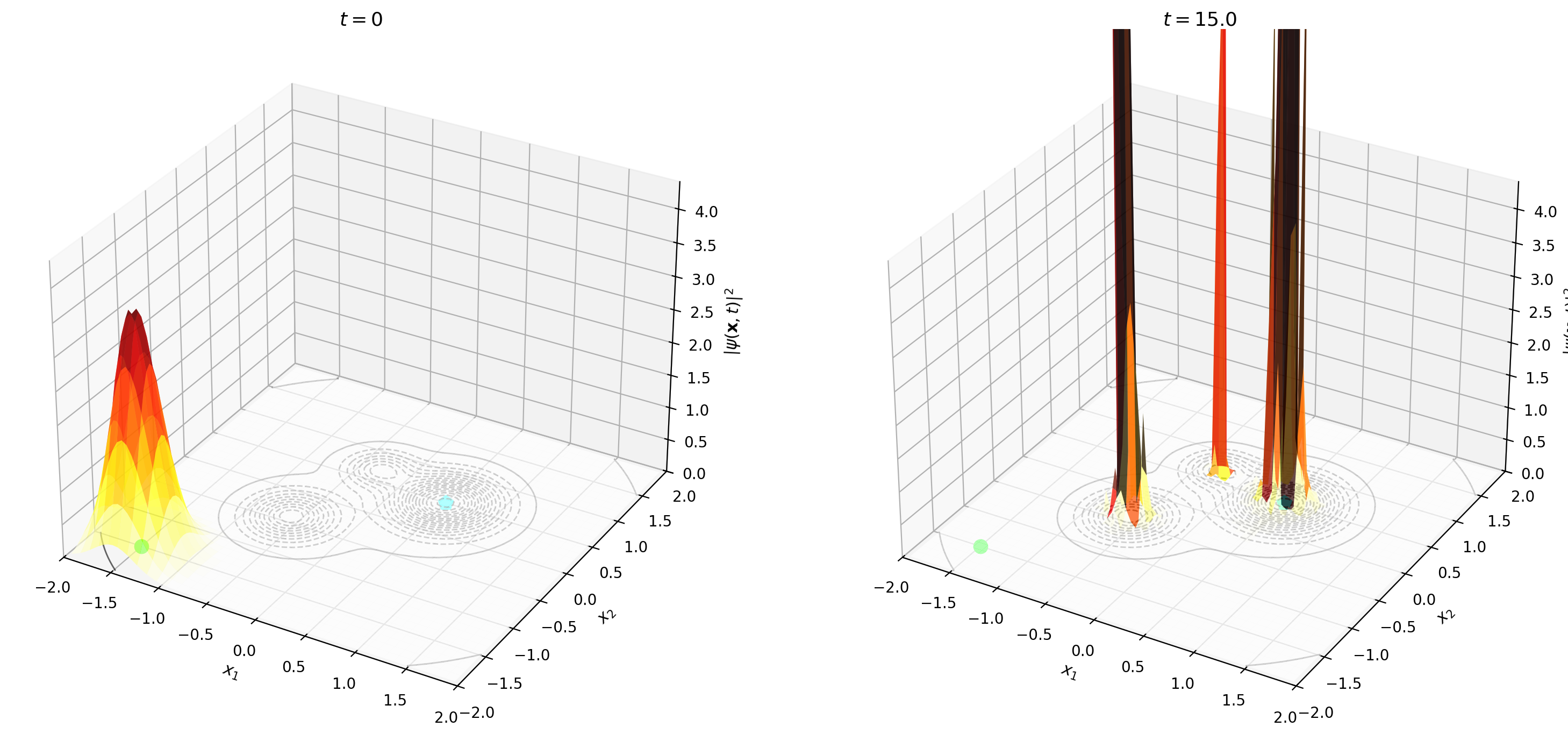}
    \caption{The evolution of probability density under our local optimization process for a target function of Gaussian wells (Equation~\ref{eq:numerics_gauss}) and hard wall boundaries. The blue point indicates the global minimum.}
    \label{fig:3d_1}
\end{figure}

\begin{figure}[h]
    \centering
    \includegraphics[width=0.8\linewidth]{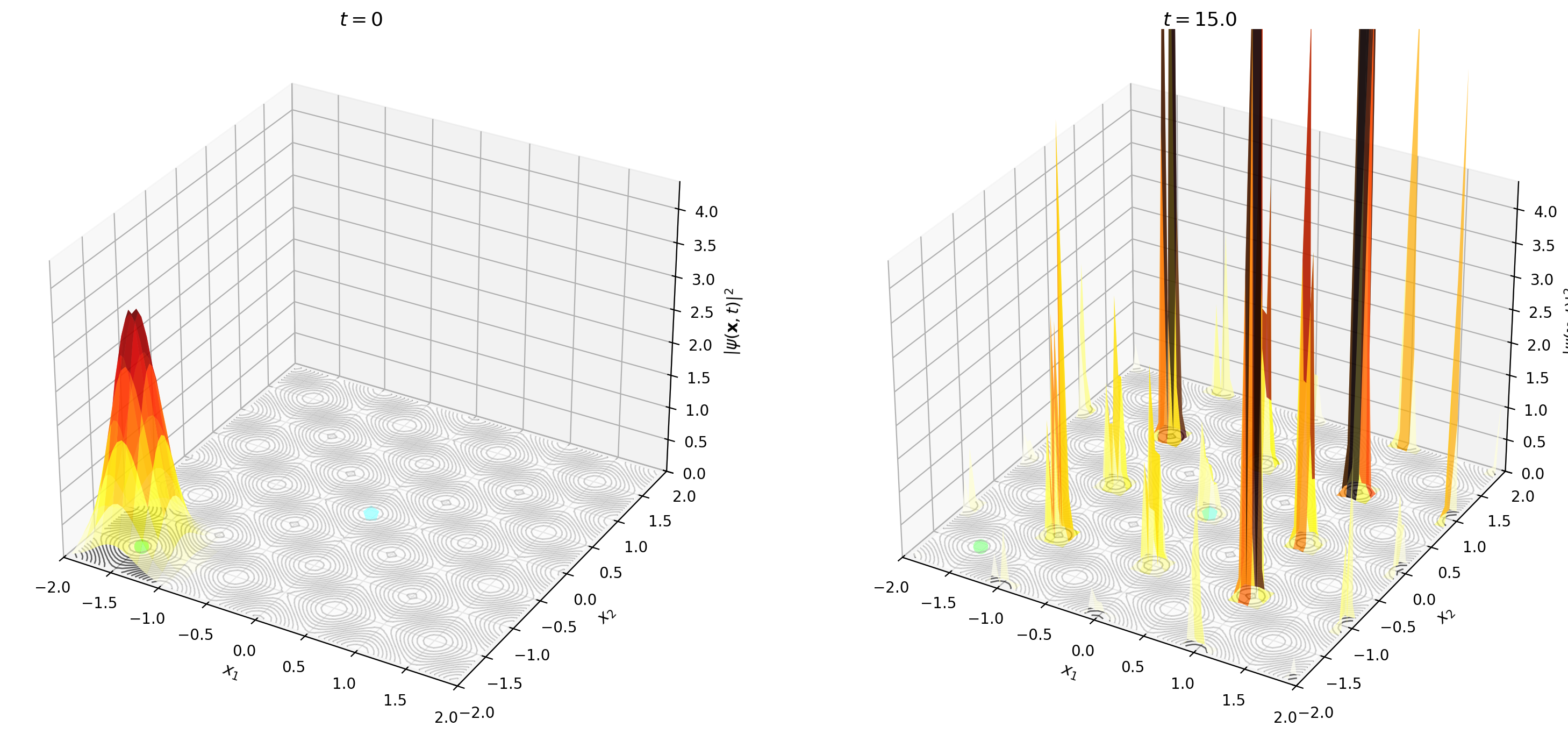}
    \caption{The evolution of probability density under our local optimization process for a Rastrigin target function (Equation~\ref{eq:numerics_rastrigin}) and hard wall boundaries. The blue point indicates the global minimum.}
    \label{fig:3d_2}
\end{figure}

\begin{figure}[h]
    \centering
    \includegraphics[width=0.5\linewidth]{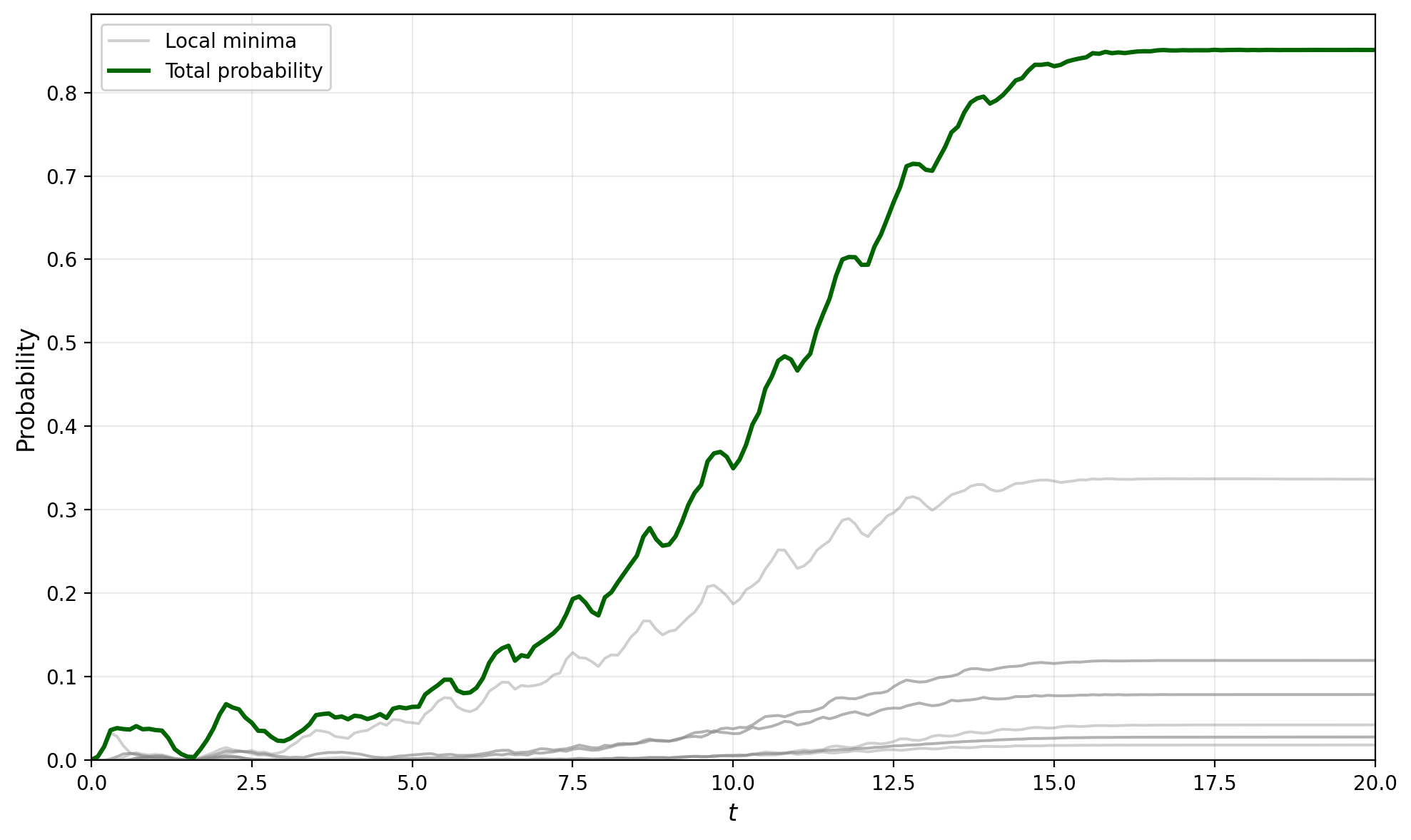}
    \caption{Probability concentration for the Rastrigin function (Equation~\ref{eq:numerics_rastrigin}) under our optimization process. The dark green plot shows the probability of a measurement outcome being $\epsilon=0.1$ close to a local minimum, while the gray plots are the probabilities for outcomes to be found $\epsilon$ close to individual local minima. It can be seen that the dynamics freeze out and plateau around t=15 to the maximum value of $P\approx0.85$.}
    \label{fig:prob_concentrate}
\end{figure}

\section{Conclusion}
\label{sec:conclusion}

We provide three quantum algorithms for continuous optimization that reduce the optimization problem to the problem of simulating a dynamical system. This allows us to leverage quantum simulation ideas to provide a quantum advantage in query complexity for solving optimization problems relative to analogous hybrid quantum-classical optimization schemes that use standard gradient descent or classical dynamical simulation to optimize the dynamics. We prove query upper bounds for several oracle settings, including bit and phase oracle access to the objective function. 
In the phase oracle scenario, the assumption is that the objective function can \emph{only} be accessed via a phase oracle.

Our first algorithm can find a local optimum of a given differentiable function $f:\mathbb{R}^N \rightarrow \mathbb{R}$ by simulating classical dynamics with friction via time-dependent Hamiltonian simulation. Our second algorithm that leverages dissipative quantum dynamics is closely related to the first algorithm but simulates quantum rather than classical dynamics. In order to compare our local optimization algorithms to existing hybrid quantum-classical strategies we consider optimization of a convex quadratic function with phase oracle access as a benchmark problem. Specifically, we show that our local optimization algorithms can find the optimum of a convex quadratic function within error $\epsilon$ using $\widetilde{O}\left(\frac{N^2
\lambda_{\max}^2 \norm{\bx_0 - \bx^*}^2}{h_x^2 \epsilon \lambda_{\min}}\right)$ queries to a phase or probability oracle, where $\lambda_{\max}$ and $\lambda_{\min}$ are the largest and smallest eigenvalues of the Hessian matrix of $f$, $h_x$ is the grid spacing for the discretized dynamics and $\|\bx_0 - \bx^*\|$ is the distance between the initial point and the optimal point. In contrast, the analogous result for a hybrid gradient descent algorithm scales as $\widetilde{O} \lb N^{3/2}\left(\frac{2\lambda_{\max}\|\bx_0 - \bx^*\|^2}{\epsilon}\right)^{(\lambda_{\max}/\lambda_{\min}) \log(3)/4} \rb$, which has exponentially worse dependence on $\lambda_{\max}$ and $\lambda_{\min}$. We also argue that similar scalings hold in a more general setting where the objective function is only promised to be strongly convex.

The fact that the $\epsilon$ scaling of the gradient-based algorithm depends on the condition number of the Hessian leads us to the conclusion that the dynamical approach will give better bounds even for moderately well conditioned Hessians. In order to quantify this, we propose the use of Wishart matrices as a random matrix ensemble to model complicated Hessian matrices that could emerge from realistic optimization problems. We find that the condition number of Wishart matrices scales like $O(N^{3/2})$ which suggests a super-exponential separation between upper bounds assuming that the discretization error is negligible. Even under worst case assumptions on the discretization error of our quantum dynamics, we provide evidence that a super-polynomial separation between the upper bounds persists.

The recently proposed quantum optimization algorithm called ``Quantum Hamiltonian Descent'' (QHD) \cite{leng2023qhd, leng2023nonconvex} is similar to our local optimization approach in the sense that it also simulates the time evolution under a carefully chosen time-dependent Hamiltonian in order to find a minimum of a given objective function encoded in the potential term of the Hamiltonian. However, QHD aims at finding the global minimum whereas we use the time-dependent Hamiltonian to find a local minimum. This results in different constructions of the time-dependent Hamiltonian with QHD being closer to an adiabatic optimization algorithm. More specifically, QHD is based on the continuous-time gradient descent approach developed in~\cite{Wibisono2016continuous_gradient_descent} and utilizes Lyapunov functions for the convergence analysis. This allows them to study the convergence rate of their algorithm for a broader class of convex objective functions. However, as their Hamiltonian does not directly recover Newton's equations of motion with friction in the classical limit, it becomes much more difficult to provide precise query upper bounds in terms of fundamental problem parameters such as the condition number $\kappa$ of the Hessian matrix of a given local optimum or the initial distance $\norm{\bx_0 - \bx^*}$ from the minimum. Indeed, the convergence rates proved in~\cite{leng2023qhd, leng2023nonconvex} depend on an unspecified parameter. 

Our approach on the other hand allows us to prove relatively tight convergence rates and query upper bounds in terms of fundamental problem parameters.
Also note that QHD is only analyzed in the setting where the objective function $f$ can be accessed via a bit oracle, whereas we also consider the case where the objective function is accessed via a phase oracle. In particular, we provide evidence that in cases where the objective function can only be accessed via a phase oracle, our fully quantum algorithms for local optimization can have significantly better asymptotic scaling than hybrid quantum-classical approaches based on gradient descent. In contrast, when having bit oracle access to the objective function, the advantage of a fully coherent quantum algorithm over classical algorithms is more subtle~\cite{chakrabarti2025optimization}.

Our third algorithm, which can find the global optimum of $f$, works differently in that it directly prepares a low temperature Boltzmann distribution to sample from low energy states using the Nos\'e Hamiltonian. We first show that the Nos\'e Hamiltonian still recovers the Boltzmann distribution from an approximation of a global microcanonical state  under discretization, and then show that the algorithm converges in time that scales with the spectral gap of the Liouvillian operator that describes the discrete diffusion in the system. We then show that using the Koopman-von Neumann formalism that let us express the classical dynamics unitarily; we can prepare the microcanonical state, and by extension the classical Boltzmann distribution using quantum simulation. This allows us to borrow ideas from quantum thermodynamics to argue about the time needed to find the global, rather than local, optima for generic systems but requires more stringent assumptions on the dynamics to prove that the runtime is finite. We then use a phase oracle to construct the Koopman-von Neumann equivalent of the Liouvillian operator, which acts like a Hamiltonian as the generator of unitary evolution, and simulate the resulting Hamiltonian using qubitization. Much of our work focuses on estimating the evolution time needed for the evolution. We make assumptions such that the eigenvectors of the Liouvillian are typical of random pure states found by applying a unitary $2$-design to an initial state, and the quantum Hamiltonian generated by the discretized classical Liouvillian is gapped in the energy shell in order to be able to utilize existing quantum thermodynamics result. This let's us find that the queries needed to get within $\delta$ total variational distance with the thermal distribution is in $\widetilde{O}\left( \frac{N^3}{\gamma \delta^2}\right)$ where $\gamma$ is the spectral gap of the Liouvillian.  To compare with a hybrid approach, we propose a hybrid quantum-classical algorithm that solves the problem using Newton's equations of motion and show that such an algorithm would run using a number of queries in $\widetilde{O}\left(\frac{N^{3/2}e^{2\sqrt{N}/\gamma\delta^2}}{\gamma^3 \delta^6}\right)$.  This shows an exponentially better upper bound for our quantum dynamical approach relative to its hybrid analogue.

The above results reveal that variational algorithms may be more promising beyond the NISQ era of computing than previously thought.  This is because the gradient evaluation, which is the major bottleneck of existing variational and quantum machine learning approaches, does not explicitly occur in our techniques.  This has the potential in some cases to exponentially improve the scaling of the algorithms because of the fact that extremely accurate gradient evaluations can be needed for some quantum simulation algorithms, which techniques such as~\cite{gilyen2019optimizing} can only estimate at cost that scales inversely with the error tolerance.  

A major question that this work leaves open involves the issue of when these methods provide a practical advantage over existing variational strategies in chemistry and machine learning. Our work shows that substantial advantages may exist for cases of local optimization in the event that the condition number is large and may provide an advantage for the global optimization too under similar circumstances.  Despite these potential asymptotic advantages, constant factors are not provided in this work.  Further, the number of qubits needed to solve a large scale optimization problem could be prohibitive.  For example, using these ideas to train a large language model would require trillions of qubits to represent the weights that our quantum algorithms would aim to optimize over.  A full and thorough analysis of the constant factors and optimized circuit decompositions would be needed  for us to understand the conditions where our approach may become practical for chemistry and machine learning.  In turn, this will alow us to understand the speed and size that a quantum computer would need to reach in order for us to have an impact on these problems.

More broadly, our work has revealed that gradient descent optimization may not be optimal for quantum optimization. This suggests that approaches to optimization may need to be rethought in quantum settings.  These insights may reveal better tailored strategies for performing optimization in quantum settings and may provide us with novel approaches to address optimization in quantum regimes that hitherto have not been considered.

\section*{Acknowledgments}

The authors would like to thank Matthew Hagan for useful comments and feedback on this work.  This material is primarily based upon work supported by the U.S. Department of Energy, Office of Science, National Quantum Information Science Research Centers, Co-design Center for Quantum Advantage (C2QA) under contract number DE- SC0012704 (PNNL FWP 76274).  This research is also supported by PNNL’s Quantum Algorithms and Architecture for Domain Science (QuAADS) Laboratory Directed Research and Development (LDRD) Initiative. The Pacific Northwest National Laboratory is operated by Battelle for the U.S. Department of Energy under Contract DE-AC05-76RL01830.  NW also acknowledges support from Google Inc. and NW and SS acknowledge support from Boehringer Ingelheim Inc. SS further acknowledges support from an Ontario Graduate Scholarship.

%

\appendix

\section{Proof of Lemma~\ref{lem:intpic-t-dep-sim}}
\label{app:intpic-t-dep}

For convenience, let us restate Lemma~\ref{lem:intpic-t-dep-sim} here.

\TdepIntpic*

\begin{proof}
    We can closely follow the proof of Lemma 6 in~\cite{Low2019InteractionPic}. First, note that in order to switch back from the interaction picture to the Schrödinger picture we require a single application of $U_B(t)$, see Eq.~\eqref{intpic_to_schrodinger}.
    Next, according to Lemma~\ref{lem:time-dep-sim}, the maximum evolution time in each segment of the evolution under $A_I(s)$ is limited to $\tau \in O \lb 1/\alpha_A \rb$. Thus, there are a total of $S \in O \lb \alpha_A t \rb$  segments. By the triangle inequality, it suffices to simulate each segment within error $\epsilon/S$ for an overall simulation error of at most $\epsilon$.
    The number of queries to all $\hamt_j$ is therefore in
    \begin{equation}
        O \lb \alpha_A t \frac{\log \lb \alpha_A t/\epsilon \rb}{\log \log \lb \alpha_A t/\epsilon \rb} \rb.
    \end{equation}
    Note that
    \begin{equation}
    \begin{split}
        \frac{\mathrm{d} A_I(s)}{\mathrm{d}s} &= \frac{\mathrm{d} U_B^\dagger(s)}{\mathrm{d} s} A(s) U_B(s) + U_B^\dagger(s) \frac{\mathrm{d} A(s)}{\mathrm{d}s} U_B(s) + U_B^\dagger(s) A(s) \frac{\mathrm{d} U_B(s)}{\mathrm{d} s} \\
        &= i U_B^\dagger(s) B(s) A(s) U_B(s) + U_B^\dagger(s) \frac{\mathrm{d} A(s)}{\mathrm{d}s} U_B(s) - i U_B^\dagger(s) A(s) B(s) U_B(s).
    \end{split}
    \end{equation}
    Hence,
    \begin{equation}
    \begin{split}
        \langle \| \dot{A_I} \| \rangle &:= \frac{1}{t} \int_0^t \norm{\frac{\mathrm{d} A_I(s)}{\mathrm{d}s}} \mathrm{d} s \leq \max_s \norm{\frac{\mathrm{d} A_I(s)}{\mathrm{d}s}} \\
        &\leq  \max_s \left\{ 2 \norm{A(s)} \norm{B(s)} + \norm{\frac{\mathrm{d} A(s)}{\mathrm{d}s}} \right\} \\
        &\leq 2 \alpha_A \alpha_B + \alpha_A',
    \end{split} 
    \end{equation}
    where we used the fact that $\norm{U_B(s)} = \norm{U_B^\dagger(s)} = 1 \quad \forall s \in [0,t]$.
    Additionally, we have that
    \begin{equation}
    \begin{split}
        \max_s \norm{A_I(s)}^2  &= \max_s \norm{ U_B^\dagger(s) A(s) U_B(s)}^2 \\
        &\leq \max_s \norm{A(s)}^2 \leq \alpha_A^2.
    \end{split}
    \end{equation}
    It then follows from Lemma~\ref{lem:time-dep-sim} that the number of discretization points $M$ for each segment satisfies
    \begin{equation}
    \begin{split}
         M &\in O \lb  \frac{t}{\alpha_A \epsilon} \lb \langle \| \dot{A_I} \| \rangle + \max_s \norm{A_I(s)}^2 \rb \rb \\
         &\in O \lb  \frac{t}{\alpha_A \epsilon} \lb \alpha_A \alpha_B + \alpha_A' + \alpha_A^2 \rb \rb \\
         &\in O \lb  \frac{t}{\epsilon} \lb \alpha_A + \alpha_B + \frac{\alpha_A'}{\alpha_A} \rb \rb.
    \end{split}
    \end{equation}
    Thus, the number of qubits is
    \begin{equation}
        n + O \lb n_a + \log \lb \frac{t}{\epsilon} \lb \alpha_A + \alpha_B + \frac{\alpha_A'}{\alpha_A} \rb \rb \rb
    \end{equation}
    and the number of primitive gates is in
    \begin{equation}
        O \lb \alpha_A t \lb n_a + \log \lb \frac{t}{\epsilon} \lb \alpha_A + \alpha_B + \frac{\alpha_A'}{\alpha_A} \rb \rb \rb \frac{\log \lb \alpha_A t/\epsilon \rb}{\log \log \lb \alpha_A t/\epsilon \rb} \rb.
    \end{equation}
\end{proof}

\section{Proof of Lemma~\ref{lem:robust_intpic_sim}}
\label{app:robust_intpic}

For convenience, let us restate Lemma~\ref{lem:robust_intpic_sim} here.

\Robustintpic*

In order to prove Lemma~\ref{lem:robust_intpic_sim}, we make use of the following result proved in Ref.~\cite{Low2019InteractionPic}.
\begin{lem}[Error from truncating and discretizing the Dyson series~\cite{Low2019InteractionPic}]
    Let $H(s) : [0,t] \mapsto \mathbb{C}^{2^{n_s} \times 2^{n_s}}$  be differentiable and $\langle \| \dot{H} \| \rangle := \frac{1}{t} \int_0^t \norm{\frac{\mathrm{d} H(s)}{\mathrm{d}s}} \mathrm{d} s$. For any $\epsilon \in (0, 2^{1-e}] $, an approximation to the time ordered operator exponential of $-iH(s)$ can be constructed such that
    \begin{equation}
        \norm{\mathcal{T} \left[ e^{-i \int_0^t H(s) \mathrm{d}s} \right] - \sum_{k=0}^K \lb \frac{-it}{M} \rb^k C_k } \leq \epsilon, \quad C_0 := \mathbb{1}, \quad C_k := \sum_{0 \leq m_1 < \cdots < m_k < M} H(m_k t/M) \cdots H(m_1 t/M),
    \end{equation}
    if we take all the following are true.
    \begin{enumerate}
        \item $\max_s \norm{H(s)}t \leq \ln 2$.
        \item $K = \left\lceil -1 + \frac{2 \ln (2/\epsilon)}{\ln \ln (2/\epsilon) + 1} \right\rceil$.
        \item $M \geq \max \left\{ \frac{16 t^2}{\epsilon}  \lb \langle \| \dot{H} \| \rangle + \max_s \norm{H(s)}^2 \rb, K^2 \right\}$.
    \end{enumerate}
\label{lem:truncated_dyson_error}
\end{lem}

\begin{proof}[Proof of Lemma~\ref{lem:robust_intpic_sim}]
    The proof is similar to the proof of Lemma~\ref{lem:intpic-t-dep-sim}. In particular, from Eq.~\eqref{intpic_to_schrodinger} we have that
    \begin{equation}
    \begin{split}
          \mathcal{T} e^{-i \int_0^t \lb A(s) + B(s) \rb \mathrm{d}s} &= \mathcal{T} \left[ e^{- i \int_0^t B(s) \mathrm{d}s} \right] \mathcal{T} \left[ e^{- i \int_0^t A_I(s) \mathrm{d}s} \right] \\
          &= U_B(t) \mathcal{T} \left[ e^{- i \int_0^t A_I(s) \mathrm{d}s} \right].
    \end{split}
    \end{equation}
    Hence, we require a single application of $U_B(t)$ in order to switch back from the interaction picture to the Schrödinger picture. Note that here we assume that we can implement $U_B(t)$ exactly. We implement an approximation to $\mathcal{T} \left[ e^{- i \int_0^t A_I(s) \mathrm{d}s} \right]$ using a truncated Dyson series as in the proof of Lemma~\ref{lem:intpic-t-dep-sim}. The main difference here is that each segment in the time evolution under $A_I(s)$ now has an additional error arising from the imperfect $\hamt_j$ oracles $\widetilde{\hamt}_j$. 
    Let $S = \frac{t}{\tau} \in \Theta \lb \alpha_A t \rb$ be the number of segments and let $K \in \Theta \lb \log \lb S/\epsilon \rb \rb$ be the truncation order of the truncated Dyson series.
    For all $j \in \{1, 2, \dots, S \}$ and $k \in \{1, \dots, K\}$ let 
    \begin{align}
        C_k^{(j)} :&= \sum_{0 \leq m_1 < \cdots < m_k < M} A_I(t_j + \tau m_k/M) \cdots A_I(t_j + \tau m_1/M) \\
        \widetilde{C}_k^{(j)} &:= \sum_{0 \leq m_1 < \cdots < m_k < M} \widetilde{A}_I(t_j + \tau m_k/M) \cdots \widetilde{A}_I(t_j + \tau m_1/M).
    \end{align}
    Further, let $C_0^{(j)} = \widetilde{C}_0^{(j)} = \mathbb{1}$ and recall that $\alpha = \max \left\{ \alpha_A, \widetilde{\alpha}_A \right\}$.
    By the triangle inequality and standard bounds on the binomial coefficient we then have that
    \begin{equation}
        \norm{\widetilde{C}_k^{(j)} - C_k^{(j)}} \leq \binom{M}{k} \alpha^k \frac{\epsilon}{4} \leq \lb \frac{eM}{k} \rb^k \alpha^k \frac{\epsilon}{4 S}.
    \end{equation}
    Thus,
    \begin{equation}
    \begin{split}
        \norm{\sum_{k=0}^K \lb \frac{-i\tau}{M} \rb^k \lb \widetilde{C}_k^{(j)} - C_k^{(j)} \rb} &\leq \sum_{k=0}^K \lb \frac{\tau}{M} \rb^k \norm{\widetilde{C}_k^{(j)} - C_k^{(j)}} \\
        &\leq \sum_{k=1}^K \lb \frac{\tau}{M} \rb^k \lb \frac{eM}{k} \rb^k \alpha^k \frac{\epsilon}{4 S} = \frac{\epsilon}{4 S} \sum_{k=1}^K \lb \frac{e \alpha \tau}{k} \rb^k  \\
        &\leq  \frac{\epsilon}{4 S} \sum_{k=1}^\infty \lb e \alpha \tau \rb^k \\
        &\leq  \frac{\epsilon}{4 S} \sum_{k=1}^\infty \lb \frac{1}{2} \rb^k \\
        &\leq \frac{\epsilon}{4 S} \frac{1}{1- 1/2} \\
        &\leq \frac{\epsilon}{2 S}.
    \end{split}
    \end{equation}
    From Lemma~\ref{lem:intpic-t-dep-sim} and Lemma~\ref{lem:truncated_dyson_error} we have that the simulation error when given access to error-free $\hamt_j$ oracles obeys
    \begin{equation}
        \norm{\mathcal{T} \left[ e^{-i \int_0^t A_I(s) \mathrm{d}s} \right] - \prod_{j=1}^S \lb \sum_{k=0}^K \lb \frac{-i\tau}{M} \rb^k C_k^{(j)} \rb} \leq \frac{\epsilon}{2},
    \end{equation}
    as long as $ S \in \Theta \lb \alpha_A t \rb$, $K \in \Theta \lb \log \lb S/\epsilon \rb \rb$, $\tau \in \Theta \lb \alpha_A^{-1} \rb$ and $M \in \Theta \lb \frac{t}{\epsilon} \lb \alpha_A + \alpha_B + \frac{\alpha_A'}{\alpha_A} \rb \rb$.
    Therefore, the overall error obeys
    \begin{equation}
    \begin{split}
        \norm{\mathcal{T} \left[ e^{-i \int_0^t A_I(s) \mathrm{d}s} \right] - \prod_{j=1}^S \lb \sum_{k=0}^K \lb \frac{-i\tau}{M} \rb^k \widetilde{C}_k^{(j)} \rb} &\leq \norm{\mathcal{T} \left[ e^{-i \int_0^t A_I(s) \mathrm{d}s} \right] - \prod_{j=1}^S \lb \sum_{k=0}^K \lb \frac{-i\tau}{M} \rb^k C_k^{(j)} \rb} \\
        &\quad + \norm{\prod_{j=1}^S \lb \sum_{k=0}^K \lb \frac{-i\tau}{M} \rb^k C_k^{(j)} \rb - \prod_{j=1}^S \lb \sum_{k=0}^K \lb \frac{-i\tau}{M} \rb^k \widetilde{C}_k^{(j)}\rb} \\
        &\leq \frac{\epsilon}{2} + S \frac{\epsilon}{2S} = \epsilon,
    \end{split}
    \end{equation}
    where we used the fact that $\norm{\lb \sum_{k=0}^K \lb \frac{-i\tau}{M} \rb^k C_k^{(j)}\rb} \leq 1$ and $\norm{\lb \sum_{k=0}^K \lb \frac{-i\tau}{M} \rb^k \widetilde{C}_k^{(j)}\rb} \leq 1$ for all $j \in [S]$.
    The required number of queries to all $\widetilde{\hamt}_j$, the number of qubits and the number of primitive gates follows then directly from Lemma~\ref{lem:intpic-t-dep-sim}.
\end{proof}

\section{Proof of Lemma~\ref{lem:hamt}}
\label{app:hamt}

For convenience, let us restate Lemma~\ref{lem:hamt} here.

\HamT*

\begin{proof}
    First, note that $B(s)$ commutes with $B(s')$ for all $s, s' \in [0,t]$. This means that
    \begin{equation}
        U_B(t) = \mathcal{T} e^{-i \int_0^t B(s) \mathrm{d}s} = \sum_{\bx} e^{-i \frac{m}{\beta} \lb e^{\beta t/m} - 1 \rb f(\bx)} \ketbra{\bx}{\bx}.
    \end{equation}
    Similarly,
    \begin{equation}
        U_B^\dagger(t) = \sum_{\bx} e^{i \frac{m}{\beta} \lb e^{\beta t/m} - 1 \rb f(\bx)} \ketbra{\bx}{\bx}.
    \end{equation}
    Then we can decompose the error-free $\hamt_j$ oracles as follows:
    \begin{equation}
    \begin{split}
        \hamt_j &= \lb \sum_{m'=0}^{M-1} \ketbra{m'}{m'} \otimes \mathbb{1}_a \otimes U_B^\dagger(t_j + \tau m'/M) \rb \times \lb \sum_{m'=0}^{M-1} \ketbra{m'}{m'} \otimes O_A(t_j + \tau m'/M) \rb \\
        &\qquad \times \lb \sum_{m'=0}^{M-1} \ketbra{m'}{m'} \otimes \mathbb{1}_a \otimes U_B(t_j + \tau m'/M) \rb \\
        &= \underbrace{\lb \sum_{m'=0}^{M-1} \ketbra{m'}{m'} \otimes \mathbb{1}_a \otimes \sum_{\bx} e^{i \frac{m}{\beta} \lb e^{\beta t_{j,m'}/m} - 1 \rb f(\bx)} \ketbra{\bx}{\bx} \rb}_{=: V_B^\dagger(j)} \times \underbrace{\lb \sum_{m'=0}^{M-1} \ketbra{m'}{m'} \otimes O_A(t_j + \tau m'/M) \rb}_{=: V_A(j)} \\
        &\qquad \times \underbrace{\lb \sum_{m'=0}^{M-1} \ketbra{m'}{m'} \otimes \mathbb{1}_a \otimes \sum_{\bx} e^{-i \frac{m}{\beta} \lb e^{\beta t_{j,m'}/m} - 1 \rb f(\bx)} \ketbra{\bx}{\bx} \rb}_{=: V_B(j)},
    \end{split}
    \end{equation}
    where $t_{j,m'} := (j-1) \tau + m' \tau/M$ and $\lb \bra{0}_a \otimes \mathbb{1}_{n_s} \rb O_A(t_j + \tau m'/M) \lb \ket{0}_a \otimes \mathbb{1}_{n_s} \rb = A_H(t_j + \tau m'/M)/\alpha_{A_H}$.
    Let us now show how to efficiently implement a unitary approximation $\widetilde{V}_B(j)$ to $V_B(j)$ using 2 queries to the bit oracle $O_f^{(b)}$. For simplicity, we ignore the ancilla register labeled $a$ used in the implementation of $V_A(j)$ as the $V_B(j)$'s act trivially on that register.
    \begin{enumerate}
        \item Query $O_f^{(b)}$: $\ket{m'}\ket{\bx}\ket{0} \xrightarrow{O_f^{(b)}} \ket{m'}\ket{\bx}\ket{\widetilde{f}(\bx)}$
        \item Compute an $\epsilon''$-precise approximation $\widetilde{g}(j,m',
        \bx)$ to $g(j,m',
        \bx) :=\lb e^{\beta t_{j,m'}/m} - 1 \rb f(\bx)$ into an ancilla register of size $Q \in O \lb \log \lb e^{\beta t/m} f_{\max}/\epsilon'' \rb \rb$ where $\epsilon'' \leq \frac{\epsilon/4S}{2 \alpha_{A_H} \frac{m}{\beta}}$:
        \begin{equation}
            \ket{m'}\ket{\bx}\ket{\widetilde{f}(\bx)} \ket{0} \rightarrow \ket{m'}\ket{\bx}\ket{\widetilde{f}(\bx)}\ket{\widetilde{g}(j,m',\bx)}.
        \end{equation}
        By the triangle inequality it then suffices for $\widetilde{f}$ to approximate $f$ within error
        \begin{equation}
            \epsilon' \leq \frac{\epsilon''}{2\lb e^{\beta t/m} - 1  \rb}
            \leq \frac{\epsilon/4S}{4\alpha_{A_H} \frac{m}{\beta} \lb e^{\beta t/m} - 1  \rb}.
        \end{equation}
        \item Controlled by the $q$-th qubit of the last ancilla register comprised of $Q$ qubits, apply $R_Z(2^q \theta)$ to an additional ancilla qubit set to $\ket{0}$ where $\theta := \frac{m}{\beta} \lb e^{\beta t/m} - 1 \rb f_{\max}/2^Q$ such that 
        \begin{equation}
            \ket{m'}\ket{\bx}\ket{\widetilde{f}(\bx)}\ket{\widetilde{g}(j,m',\bx)}\ket{0} \rightarrow e^{-i \frac{m}{\beta} \widetilde{g}(j,m',\bx)}\ket{m'}\ket{\bx}\ket{\widetilde{f}(\bx)}\ket{\widetilde{g}(j,m',\bx)}\ket{0}.
        \end{equation}        
        \item Uncompute the ancilla registers by applying all operations, except for the controlled-$R_Z$ rotations, in reverse order such that
        \begin{equation}
            e^{-i \frac{m}{\beta} \widetilde{g}(j,m',\bx)}\ket{m'}\ket{\bx}\ket{\widetilde{f}(\bx)}\ket{\widetilde{g}(j,m',\bx)}\ket{0} \rightarrow e^{-i \frac{m}{\beta} \widetilde{g}(j,m',\bx)}\ket{m'}\ket{\bx}\ket{0}\ket{0}\ket{0}.
        \end{equation}
        This requires another query to $O_f^{(b)}$.
    \end{enumerate}
    By Duhamel's principle we then have that
    \begin{equation}
        \norm{\widetilde{V}_B(j) - V_B(j)} \leq \frac{m}{\beta} \epsilon'' \leq \frac{\epsilon/4S}{2 \alpha_{A_H}}.
    \end{equation}
    The same analysis holds for $\widetilde{V}^\dagger_B(j)$ meaning
    \begin{equation}
        \norm{\widetilde{V}^\dagger_B(j) - V^\dagger_B(j)} \leq \frac{m}{\beta} \epsilon'' \leq \frac{\epsilon/4S}{2 \alpha_{A_H}}.
    \end{equation}

    Next, let us show how to implement a unitary approximation $\widetilde{V}_B(j)$ to $V_B(j)$ via a phase oracle $O_f^{(p)}$ instead of a bit oracle. For simplicity, we again ignore the ancilla register labeled $a$ used in the implementation of $V_A(j)$ as the $V_B(j)$'s act trivially on that register. Now, let
     \begin{equation}
        F := \sum_\bx \frac{f(\bx)}{2 f_{\max}} \ketbra{\bx}{\bx}
    \end{equation}
    such that $O_f^{(p)} = e^{iF}$. Note that $\norm{F} \leq \frac{1}{2}$. This allows us to apply Corollary 71 of Ref.~\cite{gilyen2019quantum} which states that we can implement a $\lb \frac{2}{\pi}, 2, \widetilde{\epsilon} \rb$-block-encoding $U_F$ of $F$ using $O \lb \log \lb 1/\widetilde{\epsilon} \rb \rb$ controlled applications of $O_f^{(p)}$ and its inverse where we demand that
    \begin{equation}
        \widetilde{\epsilon} \in \Theta \lb \frac{\epsilon}{\alpha_{A_H} S \frac{m}{\beta} e^{\beta t/m} f_{\max}} \rb.
    \end{equation}
    Once we have $U_F$ we can use Corollary 62 of Ref.~\cite{gilyen2019quantum} to construct a $ \lb 1, 4, \frac{\epsilon}{16 \alpha_{A_H} S} \rb$-block-encoding of 
    \begin{equation}
        e^{-i \frac{m}{\beta} \lb e^{\beta t_{j,m'}/m} - 1 \rb 2f_{\max} F} = e^{-i \frac{m}{\beta} \lb e^{\beta t_{j,m'}/m} - 1 \rb \sum_{\bx} f(\bx) \ketbra{\bx}{\bx}} = U_B(t_{j,m'}),
    \end{equation}
    using $O \lb \frac{m}{\beta} e^{\beta t/m} f_{\max} + \log \lb \alpha_{A_H} S/\epsilon \rb \rb$ queries to $U_F$ and controlled-$U_F$. Note though that we require controlled applications of $U_B(t_{j,m'})$ for the implementation of $V_B(j)$ since $t_{j,m'}$ is controlled by the $\ket{m'}$ register. The idea for implementing $U_B(t_{j,m'})$ in a controlled fashion is as follows:
    \begin{enumerate}
        \item Compute an $\widetilde{\widetilde{\epsilon}}$-precise approximation $\widetilde{h}(j,m')$ of $h(j, m') := \frac{m}{\beta} \lb e^{\beta t_{j,m'}/m} - 1 \rb f_{\max}$ into an ancilla register of size $\widetilde{Q} = \left\lceil \log \lb \frac{m}{\beta} e^{\beta t/m}  f_{\max}/\widetilde{\widetilde{\epsilon}} \rb \right\rceil$, i.e.
        \begin{equation}
            \ket{m'}\ket{\bx}\ket{0} \rightarrow \ket{m'}\ket{\bx}\ket{\widetilde{h}(j,m')},
        \end{equation}
        where $ \widetilde{\widetilde{\epsilon}} \leq \frac{\epsilon}{16 \alpha_{A_H} S}$.

        \item For all $q \in \ls \widetilde{Q} \rs$ do:

        Controlled by the $q$-th ancilla qubit implement
        \begin{equation}
            \sum_{\bx }e^{-i\widetilde{\widetilde{\epsilon}} \, 2^q \frac{f(\bx)}{f_{\max}}} \ketbra{\bx}{\bx}
        \end{equation}
        within error $\frac{\epsilon}{16 \alpha_{A_H} S}$ using 
        \begin{equation}
            O \lb \lb  \widetilde{\widetilde{\epsilon}} \, 2^q + \log \lb \alpha_{A_H} S /\epsilon \rb \rb \log \lb \frac{\alpha_{A_H} S \frac{m}{\beta} e^{\beta t/m} f_{\max}}{\epsilon} \rb \rb \subseteq \widetilde{O} \lb \frac{m}{\beta} e^{\beta t/m} f_{\max}  \log^2 \lb \frac{\alpha_{A_H} S}{\epsilon} \rb \rb
        \end{equation}
        queries to controlled-$O_f^{(p)}$ and its inverse.
    \end{enumerate}
    
    By the triangle inequality, we then have that
    \begin{equation}
        \norm{\widetilde{V}_B(j) - V_B(j)} \leq \frac{\epsilon}{16 \alpha_{A_H} S} + \frac{\epsilon}{16 \alpha_{A_H} S} = \frac{\epsilon/4S}{2 \alpha_{A_H}}.
    \end{equation}
    The same analysis holds for $\widetilde{V}^\dagger_B(j)$ meaning
    \begin{equation}
        \norm{\widetilde{V}^\dagger_B(j) - V^\dagger_B(j)} \leq \frac{\epsilon}{16 \alpha_{A_H} S} + \frac{\epsilon}{16 \alpha_{A_H} S} = \frac{\epsilon/4S}{2 \alpha_{A_H}}.
    \end{equation}

    So far, we have shown how to implement $V_B(j)$ within error $\frac{\epsilon/4S}{2 \alpha_{A_H}}$ using either a bit or a phase oracle of the objective function $f$. Let us now bound the overall error associated with implementing $\hamt_j$.
    By the triangle inequality we have that
    \begin{equation}
    \begin{split}
        \norm{\widetilde{\hamt}_j - \hamt_j} &= \norm{\widetilde{V}_B^\dagger(j) V_A(j)\widetilde{V}_B(j) - V_B^\dagger(j) V_A(j) V_B(j)} \\
        &\leq \norm{\widetilde{V}_B^\dagger(j) V_A(j)\widetilde{V}_B(j) - \widetilde{V}_B^\dagger(j) V_A(j)V_B(j)} + \norm{\widetilde{V}_B^\dagger(j) V_A(j)V_B(j) - V_B^\dagger(j) V_A(j) V_B(j)} \\
        &\leq \norm{\widetilde{V}_B(j) - V_B(j)} + \norm{\widetilde{V}^\dagger_B(j) - V^\dagger_B(j)} \\
        &\leq \frac{\epsilon/4S}{\alpha_{A_H}}.
    \end{split}
    \end{equation}
    This implies that for all $s \in [0,t]$,
    \begin{equation}
        \norm{\widetilde{A}_I(s) - A_I(s)} \leq \alpha_{A_H} \norm{ \lb \bra{0}_a \otimes \mathbb{1}_n \rb \widetilde{\hamt}_j \lb \ket{0}_a \otimes \mathbb{1}_n \rb -  \lb \bra{0}_a \otimes \mathbb{1}_n \rb \hamt_j \lb \ket{0}_a \otimes \mathbb{1}_n \rb}
        \leq \frac{\epsilon}{4S},
    \end{equation}
    as desired.
\end{proof}

\section{Equilibration Time of an Underdamped Harmonic Oscillator}
\label{app:underdamped}

Here we provide a bound on the $\epsilon$-equilibration time of an underdamped harmonic oscillator.

\begin{lem}[Equilibration time of an underdamped harmonic oscillator]
    Let $f(x) = \frac{a}{2} (x-x^*)^2 + c$ with $a~>~0$ and $c \in \mathbb{R}$ constant and let $\beta^2 < 4ma$. Given an initial phase space density $\rho_0(x, p)$, the $\epsilon$-equilibration time $t^*$ can be upper bounded as follows:
    \begin{equation}
        t^* \leq \max \left\{ t^{\ev{x}}_1, t^{\ev{x}}_2,  t^{\sigma}_1, t^{\sigma}_2, t^{\sigma}_3 \right\},
    \end{equation}
    where 
    \begin{align}
        t^{\ev{x}}_1 &:= \frac{1}{\gamma} \log \lb \sqrt{\frac{8a}{\epsilon}} |\ev{x}_0 - x^*| \rb \\
        t^{\ev{x}}_2 &:= -\frac{1}{\gamma} \widetilde{W} \lb - \sqrt{\frac{\epsilon}{8a}} \frac{\gamma}{|\ev{r}_0|} \rb \\
        t^{\sigma}_1 &:= \frac{1}{2 \gamma} \log \lb \frac{18 a \sigma^2_x}{\epsilon} \rb \\
        t^{\sigma}_2 &:= - \frac{1}{2 \gamma} \widetilde{W} \lb - \frac{\gamma \epsilon}{18 a | \mathrm{cov}_0 \lb x, r \rb|} \rb \\
        t^{\sigma}_3 &:= - \frac{1}{\gamma} \widetilde{W} \lb - \gamma \sqrt{\frac{\epsilon}{18 a \sigma^2_{r}}} \rb,
    \end{align}
    with $\gamma := \frac{\beta}{2m}$, $r := \frac{p}{m} + \gamma (x - x^*)$, $\sigma_x^2 := \ev{x^2}_0 - \ev{x}_0^2$, $\mathrm{cov}_0 \lb x, r \rb := \ev{x  r}_0 - \ev{x}_0\ev{r}_0$ and $\sigma^2_{r} := \ev{r^2}_0 - \ev{r}_0^2$.
\label{lem:underdamped_time_bound}
\end{lem}

\begin{proof}
    Let us again consider the coordinate system where all positions are shifted by $x^*$ such that the minimum is at $\xt_{\min} = 0$. Then the general solution for an underdamped harmonic oscillator with initial position $\xt_0 := x_0 - x^*$ and initial momentum $p_0$ is of the form 
    \begin{equation}
    \begin{split}
        \Xt_{\xt_0, p_0}(t) &= e^{-\frac{\beta t}{2m}} \lb \xt_0  \cos \lb \omega t \rb + \lb \frac{\beta}{2m} \xt_0 + \frac{p_0}{m} \rb \frac{\sin \lb \omega t \rb}{\omega} \rb \\
        &= e^{-\gamma t} \lb \xt_0  \cos \lb \omega t \rb + \lb \gamma \xt_0 + \frac{p_0}{m} \rb \frac{\sin \lb \omega t \rb}{\omega} \rb \\
        &= e^{-\gamma t} \lb \xt_0  \cos \lb \omega t \rb + r_0 \frac{\sin \lb \omega t \rb}{\omega} \rb,
    \end{split}
    \end{equation}
    where $\omega := \frac{\sqrt{4ma - \beta^2}}{2m}$.
    According to Lemma~\ref{lem:expectation} we thus have that
    \begin{equation}
    \begin{split}
        \ev{\xt(t)} &= \int_{\mathbb{R}^{2}} e^{-\gamma t} \lb \xt_0 \cos \lb \omega t \rb + r_0 \frac{\sin \lb \omega t \rb}{\omega} \rb \rho(\xt_0, p_0) d\xt_0 dp_0 \\
        &= e^{-\gamma t} \lb \ev{\xt}_0 \cos \lb \omega t \rb + \ev{r}_0 \frac{\sin \lb \omega t \rb}{\omega} \rb \\
        &= e^{-\gamma t} \lb \ev{\xt}_0 \cos \lb \omega t \rb + t \ev{r}_0  \frac{\sin \lb \omega t \rb}{\omega t} \rb.
    \end{split}
    \end{equation}
    Therefore,
    \begin{equation}
        |\langle \xt \rangle_t - \xt_{\min}| = |\ev{\xt}_t| \leq e^{-\gamma t} |\ev{\xt}_0| + t e^{-\gamma t} |\ev{r}_0|,
    \label{underdamped_upper_bound}
    \end{equation}
    where we used the fact that $\left| \frac{\sin \lb \omega t \rb}{\omega t} \right| \leq 1$ for all $t$.
    Note that the above inequality is exactly the same as in Eq.~\eqref{critical_upper_bound} in the proof of the equilibration time of the critically damped harmonic oscillator.
    Thus, to get $\epsilon'$-close to the minimum at $\xt_{\min} = 0$ such that
    \begin{equation}
        |\ev{\xt}_t - \xt_{\min}| = |\ev{x}_t - x^*| \leq \epsilon' \quad \forall t \geq t^*,
    \end{equation}
    it suffices to ensure that
    \begin{align}
        e^{-\gamma t} |\ev{\xt_0}| &\leq \frac{\epsilon'}{2} \; \forall t \geq t^* \implies t^* \geq \frac{1}{\gamma} \log \lb \frac{2 |\ev{\xt}_0 |}{\epsilon'} \rb = \frac{1}{\gamma} \log \lb \frac{2 |\ev{x}_0 - x^*|}{\epsilon'} \rb \label{underdamped_t1_x} \\
        e^{-\gamma t} |\ev{r}_0| t &\leq \frac{\epsilon'}{2}  \; \forall t \geq t^* \implies t^* \geq - \frac{1}{\gamma} \widetilde{W} \lb -\frac{\gamma \epsilon'}{2 |\ev{r}_0|} \rb. \label{underdamped_t2_x}
    \end{align}
    
    Showing that the average position is close to the minimum is not sufficient for our purposes as we might have to obtain a lot of samples to compute the average position. Ideally, we only need a small number of high-quality samples. This can be guaranteed by picking a large enough $t$ such that the position variance at time $t$, $\sigma^2(t) := \ev{x^2}_t - \ev{x}_t^2$, is sufficiently small. Specifically, Chebyshev's inequality states that a position sample $x'_t \sim R(x, p, t)$ from the time evolved phase space density satisfies
    \begin{equation}
         P\lb |x'_t - \ev{x}_t| \geq \epsilon' \rb \leq \frac{\sigma^2(t)}{\epsilon'^2}.
    \end{equation}
    We want this failure probability to be at most $1/3$ such that for all $t \geq t^*$ we can boost the success probability close to at least $1 - \delta$ using only $\log \lb 1/\delta \rb$ samples and taking the sample that yields the smallest value of $f$. Therefore, it suffices to ensure that
    \begin{equation}
        \sigma^2(t) = \ev{x^2}_t - \ev{x}_t^2 \leq \frac{\epsilon'^2}{3}. 
    \label{var_inequality_underdamped}
    \end{equation}
    Note that $\ev{\xt^2}_t - \ev{\xt}_t^2 = \ev{(x - b)^2}_t - \ev{x - b}_t^2 = \ev{x^2}_t - \ev{x}_t^2 = \sigma^2(t)$.
    Direct computation reveals that 
    \begin{equation}
    \begin{split}
        \ev{\xt^2(t)} - \ev{\xt(t)}^2 &= \int_{\mathbb{R}^{2}} e^{- 2 \gamma t} \lb \xt_0 \cos \lb \omega t \rb + r_0 \frac{\sin \lb \omega t \rb}{\omega} \rb^2 d\xt_0 dp_0 \\
        &\quad - \lb \int_{\mathbb{R}^{2}} e^{-\gamma t} \lb \xt_0 \cos \lb \omega t \rb + r_0 \frac{\sin \lb \omega t \rb}{\omega} \rb \rho(\xt_0, p_0) d\xt_0 dp_0 \rb^2 \\
        &= e^{-2 \gamma t} \lb \ev{\xt^2}_0 \cos^2 \lb \omega t \rb + 2 \cos \lb \omega t \rb \frac{\sin \lb \omega t \rb}{\omega} \ev{\xt r}_0 + \frac{\sin^2 \lb \omega t \rb}{\omega^2} \ev{r^2}_0 \rb \\
        &\quad - e^{-2 \gamma t} \lb \ev{\xt}^2_0 \cos^2 \lb \omega t \rb + 2 \cos \lb \omega t \rb \frac{\sin \lb \omega t \rb}{\omega} \ev{\xt}_0\ev{r}_0 + \frac{\sin^2 \lb \omega t \rb}{\omega^2} \ev{r}_0^2 \rb \\
        &= e^{-2\gamma t} \lb \sigma^2_x \cos^2 \lb \omega t \rb + 2 t \cos \lb \omega t \rb \frac{\sin \lb \omega t \rb}{\omega t} \, \mathrm{cov}_0 \lb \xt, r \rb + t^2 \frac{\sin^2 \lb \omega t \rb}{\omega^2 t^2} \sigma^2_r \rb
    \end{split}
    \end{equation}
    where $\sigma_x^2 = \sigma^2(0)$, $\mathrm{cov}_0 \lb \xt, r \rb = \ev{\xt r}_0 - \ev{\xt}_0\ev{r}_0$ and $\sigma^2_r = \ev{r^2}_0 - \ev{r}_0^2$. Note that $\mathrm{cov}_0 \lb \xt, r \rb = \mathrm{cov}_0 \lb x, r \rb$. Further, note that $\left|  \frac{\sin \lb \omega t \rb}{\omega t} \right| \leq 1$.
    Thus,
    \begin{equation}
         \sigma^2(t) \leq e^{-2 \gamma t} \sigma^2_x + 2 t e^{-2 \gamma t}  | \mathrm{cov}_0 \lb x, r \rb| + t^2 e^{-2 \gamma t} \sigma^2_r,
    \label{var_bound_underdamped}
    \end{equation}
    which is exactly the same upper bound found for the critically damped harmonic oscillator in Eq.~\eqref{var_bound}.
    Hence, the inequality in Eq.~\eqref{var_inequality_underdamped} is satisfied for all $t \geq t^*$ if
    \begin{align}
        e^{-2 \gamma t} \sigma^2_x &\leq \frac{\epsilon'^2}{9} \implies t^* \geq \frac{1}{2 \gamma} \log \lb \frac{9 \sigma^2_x}{\epsilon'^2} \rb \label{underdamped_t1} \\
        2 t e^{-2 \gamma t}  | \mathrm{cov}_0 \lb x, r \rb| &\leq \frac{\epsilon'^2}{9} \implies t^* \geq - \frac{1}{2\gamma} \widetilde{W} \lb - \frac{\gamma \epsilon'^2}{9 | \mathrm{cov}_0 \lb x, r \rb|} \rb \label{underdamped_t2}\\
        t^2 e^{-2 \gamma t} \sigma^2_r &\leq \frac{\epsilon'^2}{9} \implies t^* \geq - \frac{1}{\gamma} \widetilde{W} \lb - \gamma \sqrt{\frac{\epsilon'^2}{9 \sigma^2_r}} \rb. \label{underdamped_t3}
    \end{align}
    Lastly, let us bound $\epsilon'$ in terms of $\epsilon$. Again, note that with probability at least $2/3$
    \begin{equation}
        |x'_t - x^*| \leq |x'_t - \ev{x}_t| + |\ev{x}_t - x^*| \leq 2 \epsilon'.
    \end{equation}
    In case of success we have that
    \begin{equation}
    \begin{split}
        |f(x_t') - f\lb x^* \rb| &= \left| \frac{a}{2} \lb x'_t - x^* \rb^2 + c - c \right| \\
        &= \frac{a}{2} \left| \lb x_t' - x^* \rb^2  \right| \\
        &= \frac{a}{2} \left|x'_t - x^* \right|^2 \\
        &\leq 2 a \epsilon'^2.
    \end{split}
    \end{equation}
    To ensure that $|f(x'_t) - f\lb x^* \rb| \leq \epsilon$ with probability at least $2/3$, it then suffices to choose $\epsilon' = \sqrt{\frac{\epsilon}{2 a}}$. Plugging this back into the expressions in Eqs.~\eqref{underdamped_t1_x}, \eqref{underdamped_t2_x}, \eqref{underdamped_t1}, \eqref{underdamped_t2} and \eqref{underdamped_t3} and taking the maximum yields the final bound. 
\end{proof}

\section{Proof of Lemma~\ref{lem:multivariate_time}}
\label{app:multivariate_time}

For convenience, let us restate Lemma~\ref{lem:multivariate_time} here.

\Multivariate*

\begin{proof}
    The main idea is to reduce the $N$-variable optimization problem to $N$ independent single-variable optimization problems.
    First, note that $A$ is symmetric and thus can be diagonalized by an orthogonal matrix $Q \in \mathbb{R}^{N \times N}$ such that $Q A Q^\top = \Lambda$ for some positive diagonal matrix $\Lambda \in \mathbb{R}^{N \times N}$.
    Then
    \begin{equation}
    \begin{split}
        \lb \bx - \bx^* \rb^\top A \lb \bx - \bx^* \rb &= \lb \bx - \bx^* \rb^\top Q^\top Q A Q^\top Q \lb \bx - \bx^* \rb \\
        &= \lb Q \bx - Q \bx^* \rb^\top \Lambda \lb Q \bx - Q \bx^* \rb \\
        &= \bxt^\top \Lambda \bxt,
    \end{split}
    \end{equation}
    where we defined $\bxt := Q \bx - Q \bx^*$ such that the minimum of $f$ in the shifted eigenbasis of $A$ is at $\bxt_{\min} = \mathbf{0}$. 
    When mapping the optimization problem to the Liouvillian setting, we need to construct appropriate momentum operators for all $N$ variables. Let us now show that $\bxt$ and $\bpt := Q \bp$ obey the correct Poisson bracket relations. By linearity of the transformation we have that
    \begin{align}
        \left\{ \bxt_i, \bxt_j \right\} &= \sum_{k,l} \left\{ Q_{ik} \bx_k, Q_{jl} \bx_l \right\} = \sum_{k,l} Q_{ik} Q_{jl}  \underbrace{\left\{ \bx_k, \bx_l \right\}}_{=0} \\
        \left\{ \bpt_i, \bpt_j \right\} &= \sum_{k,l} \left\{ Q_{ik} \bp_k, Q_{jl} \bp_l \right\} = \sum_{k,l} Q_{ik} Q_{jl}  \underbrace{\left\{ \bp_k, \bp_l \right\}}_{=0} \\
        \left\{ \bxt_i, \bpt_j \right\} &= \sum_{k,l} \left\{ Q_{ik} \bx_k, Q_{jl} \bp_l \right\} = \sum_{k,l} Q_{ik} Q_{jl}  \underbrace{\left\{ \bx_k, \bp_l \right\}}_{= \delta_{k,l}} = \sum_k Q_{ik} Q^\top_{kj} = \delta_{ij},
    \end{align}
    which shows that $\bxt$ and $\bpt$ obey the correct Poisson bracket relations. 
    The time-dependent classical Hamiltonian with friction associated with $f(\bx)$ is therefore given by
    \begin{equation}
    \begin{split}
        H(t) &= e^{-\beta t/m} \frac{\bp^\top \bp}{2m} + e^{\beta t/m} \lb \frac{1}{2} \lb \bx - \bx^* \rb^\top A \lb \bx - \bx^* \rb + c \rb \\
        &= e^{-\beta t/m} \frac{\bp^\top Q^\top Q \bp}{2m} +  e^{\beta t/m} \lb \frac{1}{2} \lb \bx - \bx^* \rb^\top Q^\top Q A Q^\top Q \lb \bx - \bx^* \rb + c \rb \\
        &=  e^{-\beta t/m}\frac{\bpt^\top \bpt}{2m} + e^{\beta t/m} \lb \frac{1}{2}\bxt^\top \Lambda \bxt + c \rb \\
        &= \sum_{j=1}^N \lb  e^{-\beta t/m} \frac{\widetilde{p}_j^2}{2m} + e^{\beta t/m} \frac{1}{2} \lambda_j \widetilde{x}_j^2  \rb + e^{\beta t/m} c.
    \end{split}
    \end{equation}
    The corresponding classical Liouvillian takes the following form:
    \begin{equation}
        L(t) = -i \sum_{j=1}^N \lb e^{-\beta t/m}\frac{\widetilde{p}_j}{m} \partial_{\widetilde{x}_j} - e^{\beta t/m} \lambda_j \widetilde{x}_j \partial_{\widetilde{p}_j} \rb.
    \end{equation}
    Since the above Hamiltonian and Liouvillian describe the dynamics of $N$ independent variables, we can use our previous results for the 1-dim.~case.
    Specifically, for all $j \in [N]$ it holds that
    \begin{equation}
        |\langle \xt_j \rangle_t - \xt_{\min,j}| = |\ev{\xt_j}_t| \leq e^{-\gamma t} |\ev{\xt_j}_0| + t e^{-\gamma t} |\ev{\rt_j}_0|,
    \end{equation}
    where $\rt_j := \frac{\widetilde{p}_j}{m} + \gamma \xt_j$.
    Let $\epsilon' > 0$ be an error parameter to be bounded later. We demand that
    \begin{equation}
        |\ev{\widetilde{x}_j}_t - \widetilde{x}_{\min, j}| = |\ev{\widetilde{x}_j}_t| \leq \frac{\epsilon'}{4} \lb \frac{|\ev{\xt_{j}}_0|}{\norm{\ev{\bxt}_0}} + \frac{|\ev{\widetilde{r}_{j}}_0|}{\norm{\ev{\widetilde{\mathbf{r}}}_0}} \rb \quad \forall t \geq t^* \quad \mathrm{and} \quad \forall j \in [N].
    \label{error_single_component}
    \end{equation}
    The inequality in Eq.~\eqref{error_single_component} can be satisfied by ensuring that
    \begin{align}
        e^{-\gamma t} |\ev{\widetilde{x}_{j}}_0| &\leq \frac{\epsilon'}{4} \frac{|\ev{\xt_{j}}_0|}{\norm{\ev{\bxt}_0}} \;\; \forall t \geq t^* \implies t^* \geq \frac{1}{\gamma} \log \lb \frac{4 \norm{\ev{\bxt}_0}}{\epsilon'} \rb = \frac{1}{\gamma} \log \lb \frac{4 \norm{\ev{\bx}_0}}{\epsilon'} \rb \label{multi_t1_x} \\
        e^{-\gamma t} |\ev{\widetilde{r}_{j}}_0| t &\leq \frac{\epsilon'}{4} \frac{|\ev{\widetilde{r}_{j}}_0|}{\norm{\ev{\widetilde{\mathbf{r}}}_0}}  \;\; \forall t \geq t^* \implies t^* \geq - \frac{1}{\gamma} \widetilde{W} \lb -\frac{\gamma \epsilon'}{4 \norm{\ev{\widetilde{\mathbf{r}}}_0}} \rb = - \frac{1}{\gamma} \widetilde{W} \lb -\frac{\gamma \epsilon'}{4 \norm{\ev{\mathbf{r}}_0}} \rb. \label{multi_t2_x}
    \end{align}
    If the inequality in Eq.~\eqref{error_single_component} is satisfied then
    \begin{equation}
    \begin{split}
        \norm{\ev{\bx}_t - \bx^*} = \norm{\ev{\bxt}_t - \bxt_{\min}} = \norm{\ev{\bxt}_t} 
        &\leq \sqrt{\frac{\epsilon'^2}{16} \sum_j \lb \frac{|\ev{\xt_{j}}_0|}{\norm{\ev{\bxt}_0}} + \frac{|\ev{\widetilde{r}_{j}}_0|}{\norm{\ev{\widetilde{\mathbf{r}}_0}}} \rb^2} \\
        &\leq \frac{\epsilon'}{4} \sqrt{2 + 2 \sum_j \frac{|\ev{\xt_{j}}_0|}{\norm{\ev{\bxt}_0}}  \frac{|\ev{\widetilde{r}_{j}}_0|}{\norm{\ev{\widetilde{\mathbf{r}}}_0}}} \\
        &\leq \frac{\epsilon'}{4} \sqrt{2 + 2 \sqrt{\lb \sum_j \frac{|\ev{\xt_{j}}_0|^2}{\norm{\ev{\bxt}_0}^2} \rb \lb \sum_k \frac{|\ev{\widetilde{r}_{k}}_0|^2}{\norm{\ev{\widetilde{\mathbf{r}}}_0}^2} \rb}} \\
        &\leq \frac{\epsilon'}{4} \sqrt{2 + 2} = \frac{\epsilon'}{2},
    \end{split}
    \end{equation}
    where we used the Cauchy-Schwarz inequality in going from the second to the third line.
     
    As in the single-variable case, we need to make sure that the final phase space probability density is concentrated around the minimum such that we only need a small number of samples in order to determine the minimum within error $\epsilon$.
    Let $\bxt'_t \sim R(\bxt, \bpt, t)$ be a sample position vector from the time evolved phase space density in the shifted eigenbasis of $A$. According to the multivariate Chebyshev inequality the probability that $\bx'_t$ is far from the mean $\ev{\bx}_t$ is upper bounded as follows:
    \begin{equation}
         P\lb \norm{\bx'_t - \ev{\bx}_t} \geq \frac{\epsilon'}{2} \rb = P\lb \norm{\bxt'_t - \ev{\bxt}_t} \geq \frac{\epsilon'}{2} \rb \leq \frac{4 \sum_j\widetilde{\sigma}^2_j(t)}{\epsilon'^2},
    \end{equation}
    where $\widetilde{\sigma}^2_j(t) = \ev{\widetilde{x}_j^2}_t - \ev{\widetilde{x}_j}^2_t$ is the variance associated with the $j$-th position coordinate in the shifted eigenbasis of $A$ at time $t$. 
    We want this failure probability to be at most $1/3$ such that we can boost the success probability to at least $1 - \delta$ using only $O \lb \log \lb 1/\delta \rb \rb$ samples and then taking the sample that yields the smallest value of $f$. 
    Therefore, it suffices to ensure that
    \begin{equation}
        \sum_j \widetilde{\sigma}^2_j(t) \leq \frac{\epsilon'^2}{12}. 
    \end{equation}
    Using Lemma~\ref{lem:expectation} it can be verified that
    \begin{equation}
        \widetilde{\sigma}^2_j(t) = e^{-2\gamma t} \lb \widetilde{\sigma}^2_{x,j} \cos^2 \lb \omega t \rb + 2 t \cos \lb \omega t \rb \frac{\sin \lb \omega t \rb}{\omega t} \, \mathrm{cov}_0 \lb \xt_{j}, \widetilde{r}_{j} \rb + t^2 \frac{\sin^2 \lb \omega t \rb}{\omega^2 t^2} \widetilde{\sigma}^2_{r,j} \rb
    \end{equation}
    in the underdamped case where $\widetilde{\sigma}^2_{x,j} = \ev{\xt_j^2}_0 - \ev{\xt_j}_0^2$ and $\widetilde{\sigma}^2_{r,j} = \ev{\rt_j^2}_0 - \ev{\rt_j}_0^2$. Similarly, in the critically damped case it holds that
    \begin{equation}
        \widetilde{\sigma}^2_j(t) = e^{-2\gamma t} \lb \sigma^2_{x,j} + 2 t \, \mathrm{cov}_0 \lb \xt_{j}, \widetilde{r}_{j} \rb + t^2 \sigma^2_{r,j} \rb.
    \end{equation}
    For a more detailed derivation see the proof of Lemma~\ref{lem:critical_time_bound} or Lemma~\ref{lem:underdamped_time_bound}.
    In either case, we obtain the following upper bound:
    \begin{equation}
        \sum_j \widetilde{\sigma}^2_j(t) \leq  e^{-2\gamma t} \sum_j \widetilde{\sigma}^2_{x,j} + 2t e^{-2\gamma t} \left| \sum_j \mathrm{cov}_0 \lb \xt_{j}, \widetilde{r}_{j} \rb \right| + t^2 e^{-2\gamma t} \sum_j \widetilde{\sigma}^2_{r,j}.
    \end{equation}
    The above inequality is satisfied for all $t\geq t^*$ if for all $j$
    \begin{align}
        e^{-2\gamma t} \sum_j \widetilde{\sigma}^2_{x,j} &\leq \frac{\epsilon'^2}{36} \implies t^* \geq \frac{1}{2 \gamma} \log \lb \frac{36 \sum_j \widetilde{\sigma}^2_{x,j}}{\epsilon'^2} \rb \label{multi_t1} \\
        2 t e^{-2 \gamma t}  \left| \sum_j \mathrm{cov} \lb \xt_{0,j}, \widetilde{r}_{0,j} \rb \right| &\leq \frac{\epsilon'^2}{36} \implies t^* \geq - \frac{1}{2 \gamma} \widetilde{W} \lb - \frac{\gamma \epsilon'^2}{36 \left| \sum_j \mathrm{cov}_0 \lb \xt_{j}, \widetilde{r}_{j} \rb \right|} \rb  \label{multi_t2} \\
        t^2 e^{-2 \gamma t} \sum_j \widetilde{\sigma}^2_{r,j} &\leq \frac{\epsilon'^2}{36} \implies t^* \geq - \frac{1}{\gamma} \widetilde{W} \lb - \gamma \sqrt{\frac{\epsilon'^2}{36 \sum_j \widetilde{\sigma}^2_{r,j}}} \rb. \label{multi_t3}
    \end{align}
    Note that
    \begin{equation}
    \begin{split}
        \mathrm{cov}_0 \lb \xt_{j}, \widetilde{r}_{j} \rb &= \ev{\xt_j \rt_j }_0 - \ev{\xt_j}_0 \ev{\rt_j}_0 \\
        &= \ev{\lb \sum_kQ_{jk} \lb x_k - x^*_k \rb \rb \lb \sum_l Q_{jl} r_l \rb}_0 - \ev{\sum_k Q_{jk} \lb x_k - x^*_k \rb}_0 \ev{\sum_l Q_{jl} r_l }_0 \\
        &= \sum_k \sum_l Q_{jk} Q_{jl}   \ev{x_k r_l}_0 - \sum_k \sum_l Q_{jk} Q_{jl}   \ev{x_k}_0 \ev{r_l}_0.
    \end{split}
    \end{equation}
    Thus,
    \begin{equation}
    \begin{split}
        \sum_j \mathrm{cov}_0 \lb \xt_{j}, \widetilde{r}_{j} \rb &= \sum_j \lb \sum_k \sum_l Q_{jk} Q_{jl}   \ev{x_k r_l}_0 - \sum_k \sum_l Q_{jk} Q_{jl}   \ev{x_k}_0 \ev{r_l}_0 \rb \\
        &= \sum_j \sum_k \sum_l Q_{jk} Q_{jl}  \lb \ev{x_k r_l}_0 - \ev{x_k}_0 \ev{r_l}_0 \rb \\
        &= \sum_k \sum_l \underbrace{\sum_j Q_{kj}^\top Q_{jl}}_{= \delta_{k,l}}  \lb \ev{x_k r_l}_0 - \ev{x_k}_0 \ev{r_l}_0 \rb \\
        &= \sum_k \lb \ev{x_k r_k}_0 - \ev{x_k}_0 \ev{r_k}_0 \rb \\
        &= \sum_j \mathrm{cov}_0 \lb x_{j}, r_{j} \rb.
    \end{split}
    \end{equation}
    By the same argument, we have that
    \begin{align}
        \sum_j \widetilde{\sigma}^2_{x,j} &= \sum_j \sigma^2_{x,j} \\
        \sum_j \widetilde{\sigma}^2_{r,j} &= \sum_j \sigma^2_{r,j}.
    \end{align}
    This allows us to express the bounds in Eqs.~\eqref{multi_t1}, \eqref{multi_t2} and \eqref{multi_t3} in terms of the original coordinate system.
    
    So far, we shown that a sample position vector $\bx'_t \sim R(\bx, \bp, t)$ from the time evolved phase space density satisfies
    \begin{equation}
        \norm{\bx'_t - \bx^*} \leq \norm{\bxt'_t - \ev{\bx}_t} + \norm{\ev{\bx}_t  - \bx^*} \leq \epsilon'
    \end{equation}
    with probability at least $2/3$ as long as $t \geq t^*$.
    Let us now discuss how $\epsilon'$ needs to be chosen.
    By the Cauchy-Schwarz inequality we have that
    \begin{equation}
    \begin{split}
        \left| f(\ev{\bx}_t) - f(\bx^*) \right| &= \left| \frac{1}{2}\lb \ev{\bx}_t - \bx^* \rb^\top A \lb \ev{\bx}_t - \bx^* \rb \right| \\
        &\leq \frac{1}{2} \norm{\ev{\bx}_t - \bx^*} \norm{A \lb \ev{\bx}_t - \bx^* \rb} \\
        &\leq \frac{1}{2} \norm{A} \norm{\ev{\bx}_t - \bx^*}^2 \\ 
        &\leq \frac{1}{2} \norm{A} {\epsilon'}^2 = \frac{1}{2} \lambda_{\max} \epsilon'^2.
    \end{split}
    \label{function_error_bound}
    \end{equation}
    To ensure that $|f(\bx'_t) - f\lb \bx^* \rb| \leq \epsilon$ with probability at least $2/3$, it then suffices to choose $\epsilon' = \sqrt{\frac{2\epsilon}{\lambda_{\max}}}$. Plugging this back into the expressions in Eqs.~\eqref{multi_t1_x}, \eqref{multi_t2_x}, \eqref{multi_t1}, \eqref{multi_t2} and \eqref{multi_t3} and taking the maximum yields the final bound on the $\epsilon$-equilibration time $t^*$.
\end{proof}

\section{Proof of Corollary~\ref{cor:equilibration_time_liouvillian}}
\label{app:equilibration_time}

For convenience, let us restate Corollary~\ref{cor:equilibration_time_liouvillian} here.

\EquilibrationTime*

\begin{proof}
    First, note that
    \begin{align}
         t^{\ev{x}}_{1} &= \frac{1}{\gamma} \log \lb \sqrt{\frac{8\lambda_{\max}}{\epsilon}} \norm{\ev{\bx}_0 - \bx^*} \rb \in O \lb \frac{1}{\sqrt{\lambda_{\min}}} \log \lb \frac{\lambda_{\max}}{\lambda_{\min}} \frac{ \widetilde{\chi}_0}{\epsilon} \rb \rb, \\
          t^{\sigma}_{1} &= \frac{1}{2 \gamma} \log \lb \frac{18 \lambda_{\max} \sum_{j=1}^N \sigma^2_{x,j}}{\epsilon} \rb \in O \lb \frac{1}{\sqrt{\lambda_{\min}}} \log \lb \frac{\lambda_{\max}}{\lambda_{\min}} \frac{ \widetilde{\chi}_0}{\epsilon} \rb \rb.
    \end{align}
    For the other three bounds, we will use the following known upper bound on the $-1$ branch of the Lambert $W$ function:
    \begin{equation}
        -W_{-1} \lb e^{- u - 1} \rb \leq 1 + \sqrt{2u} + u, \quad \text{for } \; u>0.
    \end{equation}
    By the triangle inequality we have that
    \begin{equation}
        \norm{\ev{\br}_0} \leq \frac{\norm{\ev{\bp}_0}}{m} + \gamma \norm{\ev{\bx}_0 - \bx^*} \leq 2 \max \left\{ \frac{\norm{\ev{\bp}_0}}{m}, \gamma \norm{\ev{\bx}_0 - \bx^*} \right\}.
    \end{equation}
    Further, note that
    \begin{equation}
    \begin{split}
        \mathrm{cov}\lb x_j, r_j \rb_0 &= \ev{\gamma x_j \lb x_j - x_j^* \rb + \frac{x_j p_j}{m}}_0 - \ev{x_j}_0\ev{\gamma \lb x_j - x_j^* \rb + \frac{p_j}{m}}_0 \\
        &= \gamma \sigma_{x,j}^2 + \frac{1}{m} \mathrm{cov}\lb x_j, p_j \rb_0.
    \end{split}
    \end{equation}
    Similarly,
    \begin{equation}
    \begin{split}
        \sigma_{r,j}^2 &= \ev{\lb \gamma \lb x_j - x_j^* \rb + \frac{p_j}{m} \rb^2}_0 - \ev{\gamma \lb x_j - x_j^* \rb + \frac{p_j}{m}}_0^2 \\
        &= \gamma^2 \sigma_{x,j}^2 + \frac{1}{m^2} \sigma_{p,j}^2 + \frac{2\gamma}{m} \, \mathrm{cov}\lb x_j, p_j \rb_0.
    \end{split}
    \end{equation}
    Hence,
    \begin{equation}
    \begin{split}
         t^{\ev{x}}_{2} &= - \frac{1}{\gamma} \widetilde{W} \lb - \sqrt{\frac{\epsilon}{8\lambda_{\max}}} \frac{\gamma}{ \norm{\ev{\br}_0}} \rb \\
         &\leq \frac{1}{\gamma} \lb \sqrt{2 \lb \log \lb \sqrt{\frac{8\lambda_{\max}}{\epsilon}} \frac{\norm{\ev{\br}_0}}{\gamma} \rb -1 \rb} +  \log \lb \sqrt{\frac{8\lambda_{\max}}{\epsilon}} \frac{\norm{\ev{\br}_0}}{\gamma} \rb \rb \\
         &\in O \lb \frac{1}{\gamma} \log \lb \sqrt{\frac{\lambda_{\max}}{\epsilon}} \frac{\norm{\ev{\br}_0}}{\gamma} \rb \rb \subseteq  O \lb \frac{1}{\sqrt{\lambda_{\min}}} \log \lb \frac{\lambda_{\max}}{\lambda_{\min}} \frac{ \widetilde{\chi}_0}{\epsilon} \rb \rb.
    \end{split}
    \end{equation}
    Similarly,
    \begin{equation}
    \begin{split}
         t^{\sigma}_{2} &= - \frac{1}{2 \gamma} \widetilde{W} \lb - \frac{\gamma \epsilon}{18 \lambda_{\max}  \left| \sum_{j=1}^N \mathrm{cov}\lb x_j, r_j \rb_0 \right|} \rb \\
         &\leq \frac{1}{2\gamma} \lb \sqrt{2 \lb \log \lb \frac{18 \lambda_{\max}}{\epsilon} \frac{\left| \sum_{j=1}^N \mathrm{cov}\lb x_j, r_j \rb_0 \right|}{\gamma} \rb -1 \rb} + \log \lb \frac{18 \lambda_{\max}}{\epsilon} \frac{\left| \sum_{j=1}^N \mathrm{cov}\lb x_j, r_j \rb_0 \right|}{\gamma} \rb \rb \\
         &\in O \lb \frac{1}{\gamma} \log \lb \frac{\lambda_{\max}}{\epsilon} \frac{\max \lc \gamma \sum_j \sigma_{x,j}^2, \frac{1}{m} \left| \sum_j \mathrm{cov}\lb x_j, p_j \rb_0 \right| \rc}{\gamma} \rb \rb \subseteq  O \lb \frac{1}{\sqrt{\lambda_{\min}}} \log \lb \frac{\lambda_{\max}}{\lambda_{\min}} \frac{ \widetilde{\chi}_0}{\epsilon} \rb \rb.
    \end{split}
    \end{equation}
    Lastly,
    \begin{equation}
    \begin{split}
         t^{\sigma}_{3} &= - \frac{1}{\gamma} \widetilde{W} \lb - \gamma \sqrt{\frac{\epsilon}{18 \lambda_{\max} \sum_{j=1}^N \sigma^2_{r,j}}} \rb \\
         &\leq \frac{1}{\gamma} \lb \sqrt{2 \lb \log \lb \sqrt{\frac{18 \lambda_{\max}}{\epsilon} \frac{\sum_{j=1}^N \sigma^2_{r,j}}{\gamma^2}} \rb -1 \rb} + \log \lb \sqrt{\frac{18 \lambda_{\max}}{\epsilon} \frac{\sum_{j=1}^N \sigma^2_{r,j}}{\gamma^2}} \rb \rb \\
         &\in O \lb \frac{1}{\gamma} \log \lb \frac{\lambda_{\max}}{\lambda_{\min}} \frac{\max \lc \gamma^2 \sum_j \sigma_{x,j}^2, \frac{1}{m^2} \sum_j \sigma_{p,j}^2, \frac{\gamma}{m} \left| \sum_j \mathrm{cov}\lb x_j, p_j \rb_0 \right| \rc}{\epsilon} \rb \rb \\
         &\subseteq  O \lb \frac{1}{\sqrt{\lambda_{\min}}} \log \lb \frac{\lambda_{\max}}{\lambda_{\min}} \frac{ \widetilde{\chi}_0}{\epsilon} \rb \rb.
    \end{split}
    \end{equation}
    Combining all the asymptotic upper bounds, we thus find that
    \begin{equation}
        t^* \in  O \lb \frac{1}{\sqrt{\lambda_{\min}}} \log \lb \frac{\lambda_{\max}}{\lambda_{\min}} \frac{ \widetilde{\chi}_0}{\epsilon} \rb \rb.
    \end{equation}
\end{proof}

\section{Proof of Theorem~\ref{thm:main}}
\label{app:main_convex_opt}

For convenience, let us restate Theorem~\ref{thm:main} here.

\CoherentConvex*

\begin{proof}
    Lemma~\ref{lem:multivariate_quantum} provides an upper bound on the equilibration time $t^*$ of damped coupled quantum harmonic oscillators in continuous space whose potential is given by $f(\bx)$ and whose equilibrium configuration corresponds to the minimum $\bx^*$ of $f(\bx)$. The main idea is to simulate the dynamics of such damped coupled quantum harmonic oscillators, which are governed by $U(t) := \mathcal{T} e^{-i \int_0^t H(s) \mathrm{d}s}$, on a quantum computer and then use the fact that the probability distribution after time $t^*$ is strongly localized around the minimum at $\bx^*$ such that we only need to draw a small number of samples in order to determine $f(\bx^*)$ within error $\epsilon$.

    In order to simulate the dynamics of such a continuous system on a quantum computer, we need to discretize it. Let us consider the discrete quantum Hamiltonian $\widetilde{H}(t)$ as given in Definition~\ref{def:discretized_ham} and let $\widetilde{U}_H(t) = \mathcal{T} e^{-i \int_0^t \Ht(s) \mathrm{d}s}$ denote the corresponding time evolution operator. Lemma~\ref{lem:friction}  shows how to implement an $\widetilde{\epsilon}$-precise approximation $W(t)$ to $\widetilde{U}_H(t)$ using
    \begin{equation}
        O \lb \alpha_{A_H} t \frac{\log \lb \alpha_{A_H} t/\widetilde{\epsilon} \rb}{\log \log \lb \alpha_{A_H} t/\widetilde{\epsilon} \rb} \rb
    \end{equation}
    queries to an $\epsilon'$-precise bit oracle $O_f^{(b)}$ of $f$ as given in Definition~\ref{def:bit_oracle} with $1/\epsilon' \in O \lb \frac{\alpha_{A_H}^2 t \frac{m}{\beta} e^{\beta t/m}}{\widetilde{\epsilon}} \rb$. 
    Lemma~\ref{lem:friction_phase}, on the other hand, shows how to implement an $\widetilde{\epsilon}$-precise approximation $W(t)$ to $\widetilde{U}_H(t)$ using
    \begin{equation}
        \widetilde{O} \lb \frac{m}{\beta} e^{\beta t/m} f_{\max} \alpha_{A_H} t \log^3 \lb \frac{1}{\widetilde{\epsilon}} \rb \rb
    \end{equation}
    queries to a phase oracle $O_f^{(p)}$ of $f$ as given in Definition~\ref{def:phase_oracle} and its controlled version. Note that
    \begin{equation}
        \alpha_{A_H} \in \lb \frac{N}{m h_x^2} \rb,
    \end{equation}
    according the discussion after Lemma~\ref{lem:props_disc_H}. We use the above bound in Theorem~\ref{thm:main_informal} to simplify the expressions.
    The remainder of the proof is essentially the same as the proof of Theorem~\ref{thm:optimization_liouvillian}. Nonetheless, in the following, we spell out the details for completeness.

    Now, let us consider the expectation value of the position operator in the continuum and in the discrete setting according to Definition~\ref{def:expectation_values} and let
    \begin{equation}
        \ev{\widetilde{\xt}_j}_t := \bra{\widetilde{\psi}_0} W^\dagger(t) \hat{\xt}_j W(t)\ket{\widetilde{\psi}_0}
    \end{equation}
    denote the time evolved expectation value of the discretized position operator for the $j$-th variable w.r.t.~the approximate discretized evolution operator $W(t)$.
    By the triangle inequality, we then have that 
    \begin{equation}
        \norm{\ev{\widetilde{\bxt}}_{t^*} - \bx^*} \leq \underbrace{\norm{ \ev{\widetilde{\bxt}}_{t^*} - \ev{\bxt}_{t^*}}}_{=: \epsilon_{\mathrm{sim}}} + \underbrace{\norm{\ev{\bxt}_{t^*} - \ev{\bx}_{t^*}}}_{=: \epsilon_{\mathrm{dis}}} + \underbrace{\norm{\ev{\bx}_{t^*} - \bx^*}}_{=: \epsilon_{\mathrm{eq}}}.
    \end{equation}
    Let $\epsilon_x$ be an error tolerance to be bounded later.
    In order for the above error to be at most $\epsilon_x$ such that $\norm{\ev{\widetilde{\bxt}}_{t^*} - \bx^*} \leq \epsilon_x$, it suffices to ensure that
    \begin{align}
        \epsilon_{\mathrm{sim}} &\leq \epsilon_x/3 \\
        \epsilon_{\mathrm{dis}} &\leq \epsilon_x/3 \\
        \epsilon_{\mathrm{eq}} &\leq \epsilon_x/3.
    \end{align}
    We also need to ensure that the concentration bounds used in Lemma~\ref{lem:multivariate_quantum} are still valid in the discrete setting since otherwise we might have to draw a lot of samples to obtain a good estimate of $\bx^*$.
    Let $\widetilde{\bxt}'_t \sim |\widetilde{\widetilde{\psi}}_t(\bx)|^2$ be a sample position vector from the time-evolved discrete probability distribution associated with $\ket{\widetilde{\widetilde{\psi}}_t} := W(t) \ket{\widetilde{\psi}_0}$. According to the multivariate Chebyshev inequality the probability that $\widetilde{\bxt}'_t$ is far from the mean vector $\ev{\widetilde{\bxt}}_{t}$ is upper bounded as follows:
    \begin{equation}
         P\lb \norm{\widetilde{\bxt}'_t - \ev{\widetilde{\bxt}}_{t}} \geq \epsilon_x \rb \leq \frac{\sum_j \widetilde{\widetilde{\sigma}}_{j}^2 (t)}{\epsilon_x^2},
    \end{equation}
    where
    \begin{equation}
        \widetilde{\widetilde{\sigma}}_{j}^2 (t) := \ev{\widetilde{\widetilde{x}}_j^2}_t - \ev{\widetilde{\widetilde{x}}_j}^2_t
    \end{equation}
    denotes the time-evolved variance of the $j$-th position variable w.r.t.~the approximate discretized evolution operator $W(t)$.
  
    By the triangle inequality we then have that
    \begin{equation}
        \left| \sum_j \widetilde{\widetilde{\sigma}}_{j}^2 (t) - \sum_j \sigma_{j}^2(t) \right| \leq \underbrace{\left| \sum_j \widetilde{\widetilde{\sigma}}_{j}^2 (t) - \sum_j \widetilde{\sigma}_{j}^2(t) \right|}_{=: \epsilon'_{\mathrm{sim}}} + \underbrace{\left| \sum_j \widetilde{\sigma}_{j}^2 (t) - \sum_j \sigma_{j}^2(t) \right|}_{=: \epsilon'_{\mathrm{dis}}}.
    \end{equation}
    Now, as long as
    \begin{align}
        \sum_j \sigma_j^2(t) &\leq \frac{\epsilon_x^2}{6}, \\
        \epsilon'_{\mathrm{sim}} &\leq \frac{\epsilon_x^2}{12}, \\
        \epsilon'_{\mathrm{dis}} &\leq \frac{\epsilon_x^2}{12},
    \end{align}
    it holds that
    \begin{equation}
    \begin{split}
        P\lb \norm{\widetilde{\bxt}'_t - \ev{\widetilde{\bxt}}_{t}} \geq \epsilon_x \rb &\leq \frac{\sum_j \widetilde{\widetilde{\sigma}}_{j}^2 (t)}{\epsilon_x^2} \\
        &\leq \frac{1}{\epsilon_x^2} \lb \sum_j \sigma_j^2(t) + \epsilon'_{\mathrm{sim}} + \epsilon'_{\mathrm{dis}} \rb \\
        &\leq \frac{1}{3}.
    \end{split}
    \end{equation}
    This means that with probability at least $2/3$ we have that
    \begin{equation}
        \norm{\widetilde{\bxt}'_t - \bx^*} \leq \norm{\widetilde{\bxt}'_t - \ev{\widetilde{\bxt}}_{t}} + \norm{\ev{\widetilde{\bxt}}_{t} - \bx^*} \leq 2\epsilon_x.
    \end{equation}
    Drawing $O \lb \log \lb 1/\delta \rb \rb$ samples and taking the sample that leads to the smallest value of $f$ allows us to boost the success probability to at least $1 - \delta$. Specifically, the probability that not a single sample out of $v$ many samples is $\epsilon_x$-close to $\ev{\widetilde{\bxt}}_{t}$ is at most $\lb 1/3 \rb^v$. Thus, $1/{3^v} \leq \delta$ if $v \geq \log \lb 1/\delta \rb/\log \lb 3\rb$.

    By the Cauchy-Schwarz inequality we then have that
    \begin{equation}
    \begin{split}
        \left| f \lb \widetilde{\bxt}'_{t} \rb - f(\bx^*) \right| &= \left| \frac{1}{2}\lb \widetilde{\bxt}'_{t} - \bx^* \rb^\top A \lb \widetilde{\bxt}'_{t} - \bx^* \rb \right| \\
        &\leq \frac{1}{2} \norm{\widetilde{\bxt}'_{t} - \bx^*} \norm{A \lb \widetilde{\bxt}'_{t} - \bx^* \rb} \\
        &\leq \frac{1}{2} \norm{A} \norm{\widetilde{\bxt}'_{t} - \bx^*}^2 \\ 
        &\leq 2 \norm{A} {\epsilon_x}^2 = 2 \lambda_{\max} {\epsilon_x}^2.
    \end{split}
    \end{equation}
    Hence, in order to ensure that $\left| f \lb \widetilde{\bxt}'_{t} \rb - f(\bx^*) \right| \leq \epsilon$, it suffices to choose $\epsilon_x = \sqrt{\frac{\epsilon}{2\lambda_{\max}}}$.

    Let us now discuss how to achieve the various error bounds.
    First, according to Corollary~\ref{cor:equilibration_time}, we require
    \begin{equation}
        t^* \in  O \lb \frac{1}{\sqrt{\lambda_{\min}}} \log \lb \frac{\lambda_{\max}}{\lambda_{\min}} \frac{ \chi_0}{\epsilon_x} \rb \rb \subseteq O \lb \frac{1}{\sqrt{\lambda_{\min}}} \log \lb \frac{\lambda_{\max}}{\lambda_{\min}} \frac{\chi_0}{\epsilon} \rb \rb,
    \end{equation}
    in order for $\epsilon_{\mathrm{eq}} \leq \epsilon_x/3$ and $\sum_j \sigma_j^2(t) \leq \epsilon_x^2/6$. Note that in the above bound we picked the friction coefficient $\beta$ such that $\beta \in \Theta \lb \sqrt{\lambda_{\min}} \rb$. 

    The conditions on the discretization errors, $\epsilon_{\mathrm{dis}} \leq \epsilon_x/3 \leq \sqrt{\frac{\epsilon}{18\lambda_{\max}}}$ and $\epsilon_{\mathrm{dis}}' \leq \epsilon_x^2/12 \leq \frac{\epsilon}{24\lambda_{\max}}$ are true by assumption. Implicitly, this requires us to choose a sufficiently small grid spacing $h_x$ for the finite difference approximations of the discretized partial derivatives $\partial_{x,j}$.
   
    Next, let us discuss how to ensure that $\epsilon_{\mathrm{sim}} \leq \epsilon_x/3$ and $\epsilon'_{\mathrm{sim}} \leq \epsilon_x^2/12$. Let $\ket{\widetilde{\psi}_t} := \widetilde{U}_L(t)\ket{\widetilde{\psi}_0}$ denote the time evolved quantum state w.r.t.~the exact discretized evolution operator and let $\ket{\widetilde{\widetilde{\psi}}_t} = W(t)\ket{\widetilde{\psi}_0}$ again denote the time evolved quantum state w.r.t.~the $\widetilde{\epsilon}$-precise approximate discretized evolution operator. Then we have that
    \begin{equation}
        \norm{\ket{\widetilde{\widetilde{\psi}}_{t}} - \ket{\widetilde{\psi}_t}} \leq \widetilde{\epsilon}.
    \end{equation}
    This implies that
    \begin{equation}
    \begin{split}
        \norm{\ev{\widetilde{\bxt}}_{t} - \ev{\bxt}_{t}}^2 &= \sum_{j=1}^N \left| \bra{\widetilde{\widetilde{\psi}}_t}\hat{\xt}_j\ket{\widetilde{\widetilde{\psi}}_t} - \bra{\widetilde{\psi}_t}\hat{\xt}_j\ket{\widetilde{\psi}_t} \right|^2 \\
        &\leq \sum_{j=1}^N \left| \bra{\widetilde{\widetilde{\psi}}_t}\hat{\xt}_j\ket{\widetilde{\widetilde{\psi}}_t} - \bra{\widetilde{\widetilde{\psi}}_t}\hat{\xt}_j\ket{\widetilde{\psi}_t} \right|^2 + \sum_{j=1}^N \left| \bra{\widetilde{\widetilde{\psi}}_t}\hat{\xt}_j\ket{\widetilde{\psi}_t} - \bra{\widetilde{\psi}_t}\hat{\xt}_j\ket{\widetilde{\psi}_t} \right|^2 \\
        &\leq 2N x_{\max} \norm{\ket{\widetilde{\widetilde{\psi}}_{t}} - \ket{\widetilde{\psi}_t}}^2 \\
        &\leq 2N x_{\max} \widetilde{\epsilon}^2,
    \end{split}
    \end{equation}
    where we used the Cauchy-Schwarz inequality in going from the second to the third line. 
    Therefore, in order for $\epsilon_{\mathrm{sim}}$ to be at most $\epsilon_x/3$, it suffices to ensure that
    \begin{equation}
        \widetilde{\epsilon} \leq \frac{\epsilon_x/3}{\sqrt{2N x_{\max}}} \leq \frac{1}{6}\sqrt{\frac{\epsilon}{N x_{\max} \lambda_{\max}}}.
    \end{equation}
    Furthermore,
    \begin{equation}
    \begin{split}
        \left| \sum_j \lb \ev{\widetilde{\xt}_j^2}_t - \ev{\xt_j^2}_t \rb \right| &= \left| \sum_j \lb \bra{\widetilde{\widetilde{\psi}}_t}\hat{\xt}^2_j\ket{\widetilde{\widetilde{\psi}}_t} - \bra{\widetilde{\psi}_t}\hat{\xt}^2_j\ket{\widetilde{\psi}_t} \rb \right| \\
        &\leq  \left| \sum_j \lb \bra{\widetilde{\widetilde{\psi}}_t}\hat{\xt}^2_j\ket{\widetilde{\widetilde{\psi}}_t} - \bra{\widetilde{\widetilde{\psi}}_t}\hat{\xt}^2_j\ket{\widetilde{\psi}_t} \rb \right| + \left| \sum_j \lb \bra{\widetilde{\widetilde{\psi}}_t}\hat{\xt}^2_j\ket{\widetilde{\psi}_t} - \bra{\widetilde{\psi}_t}\hat{\xt}^2_j\ket{\widetilde{\psi}_t} \rb \right| \\
        &\leq \left|\bra{\widetilde{\widetilde{\psi}}_t} \sum_j \hat{\xt}^2_j\lb \ket{\widetilde{\widetilde{\psi}}_t} - \ket{\widetilde{\psi}_t} \rb \right| + \left|\lb \bra{\widetilde{\widetilde{\psi}}_t} - \bra{\widetilde{\psi}_t} \rb \sum_j \hat{\xt}^2_j \ket{\widetilde{\psi}_t} \right| \\
        &\leq 2 N x_{\max}^2 \norm{\ket{\widetilde{\widetilde{\psi}}_t} - \ket{\widetilde{\psi}_t}} \\
        &\leq 2 N x_{\max}^2 \widetilde{\epsilon}.
    \end{split}
    \end{equation}
    Similarly,
    \begin{equation}
    \begin{split}
        \left| \sum_j \lb \ev{\widetilde{\xt}_j}^2_t - \ev{\xt_j}^2_t \rb \right| &= \left| \sum_j \lb \ev{\widetilde{\xt}_j}_t - \ev{\xt_j}_t \rb \lb \ev{\widetilde{\xt}_j}_t + \ev{\xt_j}_t \rb  \right| \\
        &\leq 2 x_{\max} \left| \sum_j \lb \ev{\widetilde{\xt}_j}_t - \ev{\xt_j}_t \rb \right| \\
        &\leq 4 N x_{\max}^2 \widetilde{\epsilon}.
    \end{split}
    \end{equation}
    Therefore, 
    \begin{equation}
    \begin{split}
        \left| \sum_j \widetilde{\widetilde{\sigma}}_{j}^2 (t) - \sum_j \widetilde{\sigma}_{j}^2(t) \right| &= \left| \sum_j \lb \ev{\widetilde{\xt}_j^2}_t - \ev{\widetilde{\xt}_j}^2_t \rb - \sum_j \lb \ev{\xt_j^2}_t - \ev{\xt_j}^2_t \rb \right| \\
        &\leq \left| \sum_j \lb \ev{\widetilde{\xt}_j^2}_t - \ev{\xt_j^2}_t \rb \right| + \left| \sum_j \lb \ev{\widetilde{\xt}_j}^2_t - \ev{\xt_j}^2_t \rb \right| \\
        &\leq 6 N x_{\max}^2 \widetilde{\epsilon}.
    \end{split}
    \end{equation}
    This shows that in order for $\epsilon_{\mathrm{sim}'}$ to be at most $\epsilon_x^2/12$
    it suffices to ensure that
    \begin{equation}
        \widetilde{\epsilon} \leq \frac{\epsilon_x^2/12}{6Nx_{\max}^2} \leq \frac{\epsilon}{144 \lambda_{\max} N x_{\max}^2}.
    \end{equation}

    Lastly, we will use the following identity and upper bound on the $-1$ branch of the Lambert $W$ function in order to simplify the final complexity bounds:
    \begin{equation}
        e^{-W_{-1}(y)} = \frac{W_{-1}(y)}{y} \in O \lb \frac{\log(-y)}{y} \rb.
    \end{equation}
    Specifically, this implies that
    \begin{equation}
        e^{\beta t^*/m} \in \widetilde{O} \lb \frac{\lambda_{\max} \chi_0}{ \lambda_{\min} \epsilon} \rb.
    \end{equation}
    Additionally, note that
    \begin{equation}
        f_{\max} \in O \lb N \lambda_{\max} x^2_{\max} \rb.
    \end{equation}

    Putting everything together, we thus require either a total of
    \begin{equation}
    \begin{split}
        O \lb \alpha_{A_H} t^{*} \frac{\log \lb \alpha_{A_H} t/\widetilde{\epsilon} \rb}{\log \log \lb \alpha_{A_H} t/\widetilde{\epsilon} \rb} \log \lb 1/\delta \rb \rb &\subseteq \widetilde{O} \lb  \alpha_{A_H} t^* \log \lb 1/\widetilde{\epsilon} \rb \log \lb 1/\delta \rb \rb \\
        &\subseteq \widetilde{O} \lb \frac{\alpha_{A_H}}{\sqrt{\lambda_{\min}}} \log^2 \lb\frac{N \lambda_{\max} x_{\max} \chi_0}{\epsilon} \rb \log \lb 1/\delta \rb \rb
    \end{split}
    \end{equation}
    queries to an $\epsilon'$-precise bit oracle $O_f^{(b)}$ with
    \begin{equation}
        1/\epsilon' \in O \lb \frac{\alpha_{A_H}^2 t \frac{m}{\beta} e^{\beta t/m}}{\widetilde{\epsilon}} \rb  \subseteq \widetilde{O} \lb  \frac{\alpha_{A_H}^2 N \lambda_{\max}^{2} x_{\max}^2 \chi_0}{\epsilon^{2} \lambda_{\min}^2} \rb,
    \end{equation}
    or
    \begin{equation}
        \widetilde{O} \lb \frac{m}{\beta} e^{\beta t^*/m} f_{\max} \alpha_{A_H} t^* \log^3 \lb \frac{1}{\widetilde{\epsilon}} \rb \log \lb 1/\delta \rb \rb \subseteq \widetilde{O} \lb \frac{N}{\epsilon} \frac{\lambda_{\max}^2}{\lambda_{\min}^2} x_{\max}^2 \alpha_{A_H} \chi_0 \log \lb 1/\delta \rb \rb
    \end{equation}
    queries to the phase oracle $O_f^{(p)}$ and its controlled version where we picked $m=1$. 
    This completes the proof.
\end{proof}

\section{Proof of Theorem~\ref{thm:ch7_disc}}
\label{appendix:discproof}

Below we provide a proof of Theorem~\ref{thm:ch7_disc} which bounds the spatial discretization error of our quantum algorithm for global continuous optimization based on classical Liouvillian dynamics.

\begin{proof}[{Proof of Theorem~\ref{thm:ch7_disc}}]
We are going to consider the rectangle function as an approximation to the $\delta$-function:

\begin{align}
    \text{rect}_\Delta(x)=\begin{cases} 
      \frac{1}{\Delta} & |x|\leq \frac{\Delta}{2}, \\
      0 & \text{otherwise} .
   \end{cases}
\end{align}
Observe that this function recovers the Dirac delta distribution as $\Delta \rightarrow 0$. If we substitute the delta function in~\eqref{eq:global_rho} with the rectangle function, the $\bx$ average under the modified density function can be calculated easily as
\begin{align}
    \langle \bx \rangle_1 &:= C\, \int d\{x_n\} \int d\{p'_n\} \int ds \int dp_s s^N \text{rect}_\Delta \lb H_N \lb \{x_n\}, \{p'_n\}, s, p_s \rb - E_{\text{ext}}\rb \\
    &=\langle \bx \rangle_0 \frac{\sinh{(\frac{\Delta(N+1)}{2}g^{-1}\beta)}}{\frac{\Delta(N+1)}{2}g^{-1}\beta}\\
    &= \langle \bx \rangle_0 \, \text{sinhc}\lb\frac{g^{-1}\beta\Delta(N+1)}{2}\rb.
\end{align}
One of our key aims here is to understand how small $\Delta$ needs to be to accurately approximate the expectation value of $\bx$. Since $\text{sinhc}(0) = 1$, it can be seen easily that $\langle \bx \rangle_1 = \langle \bx \rangle_0$ as $\Delta \rightarrow 0$. We can Taylor expand to get upper bounds on $\Delta$:
\begin{equation}
    \langle \bx \rangle_1 = \langle \bx \rangle_0 + \langle \bx \rangle_0 \frac{g^{-2}(N+1)^2\beta^2\Delta^2}{24} + \langle \bx \rangle_0 O(\Delta^4\beta^4).\label{eq:global_error1}
\end{equation}
Thus we can guarantee $\|\langle \bx\rangle_1 - \langle \bx\rangle_0\|\le \epsilon$ if
\begin{equation}
    \Delta \in O \lb \sqrt{\frac{\epsilon}{\beta \norm{\ev{\bx}_0} }} \rb.
\end{equation} 
This finite energy shell substitution lets us use discretization bounds for the integrals. The integral over s is given by integrating between the turning points for the rect function which we define to be $s_1,s_2$:
\begin{align}
    \int_0^\infty ds\,s^N \text{rect}_\Delta \lb H_{\mathrm{sys}} \lb \bx, \bp \rb + \frac{p_s^2}{2Q} + g\beta^{-1}\ln{s} - E_{\text{ext}}\rb &= \frac{1}{\Delta}\int_{s_1}^{s_2}ds\,s^N\nonumber\\
    &=\frac{s_2^{N+1}-s_1^{N+1}}{\Delta(N+1)}\nonumber\\
    &= s_0^{N+1}\frac{2\,\text{sinh}(g^{-1}\beta(N+1)\frac{\Delta}{2})}{\Delta(N+1)},
\end{align}
where we define for convenience
\begin{align}
    s_0&:= e^{-g^{-1}\beta\lb H_{\mathrm{sys}}(\bx,\bp)-E_{ext}+\frac{p_s^2}{2Q}\rb}\\
    s_1&:=s_0\,e^{-g^{-1}\beta\frac{\Delta}{2}}\\
s_2&:=s_0\,e^{g^{-1}\beta\frac{\Delta}{2}}.    
\end{align} 
This integral can not be discretized directly with a method like the midpoint rule due to the discontinuous nature of the rect function. Instead, we assume an underlying grid with spacing $h_s$, which will have at least $n=\floor{\frac{s_2-s_1}{h_s}}$ points in the interval. We assume that the first $n-1$ points constitute a left Riemann sum and treat the intervals $[s_1,h_s\ceil{\frac{s_1}{h_s}})$ and $(h_s\floor{\frac{s_2}{h_s}},s_2]$ as pure error. We also assume that the upper limit of our discrete sum is greater than $s_2$. Then, assuming that $h_{s}/(s_2-s_1)\le 1/4$ such that $(1- 2h_s/(s_2-s))^{-1} \le 2$, the error, defined as

\begin{equation}
    \varepsilon_s := \int \,ds\,s^N\,\text{rect}_\Delta \lb H_N-E_e\rb -\sum_{n=0}^{n=N_s}h_s (n h_s)^N \text{rect}_\Delta \lb H_{\mathrm{sys}} \lb \{x_n\}, \{p_n\}\rb + \frac{p_s^2}{2Q} + g\beta^{-1}\ln{n h_s} - E_{\text{ext}}\rb,
\end{equation}

is upper bounded by
\begin{align}
    \varepsilon_s &\leq \frac{M_1(s_2-s_1)^2}{2(n-1)} + \frac{h_s}{\Delta} (s_2^N + (s_1+1)^N)\nonumber\\
    &\leq \frac{N\,s_2^{N-1}(s_2-s_1)^2}{2\Delta(\frac{s_2-s_1}{h_s}-2)} + \frac{h_s}{\Delta} (s_2^N + (s_1+1)^N)\nonumber\\
    &= \frac{h_s}{\Delta}\frac{N\,s_2^{N-1}(s_2-s_1)}{2} \lb1-\frac{2h_s}{s_2-s_1} \rb^{-1} + \frac{h_s}{\Delta} (s_2^N + (s_1+1)^N).
\end{align}
Here $M_1$ is an upper bound on the absolute value of the derivative of the integrand which is upper bounded by $Ns_2^{N-1}/\Delta$. Then, substituting in the above expressions for $s_1$ and $s_2$, we obtain
\begin{equation}
    \varepsilon_s \leq \frac{h_s}{\Delta}s_0^N (N+2)e^{g^{-1}\beta N\frac{\Delta}{2}},
\end{equation}
such that the discretization error is:
\begin{align}
    &\int \,ds\,s^N\,\text{rect}_\Delta \lb H_N-E_e\rb -\sum_{n=0}^{n=N_s}h_s (n h_s)^N \text{rect}_\Delta \lb H_{\mathrm{sys}} \lb \{x_n\}, \{p_n\}\rb + \frac{p_s^2}{2Q} + g\beta^{-1}\ln{n h_s} - E_{\text{ext}}\rb \nonumber \\
    &\qquad= s_0^{N+1}\frac{2\,\text{sinh}(g^{-1}\beta(N+1)\frac{\Delta}{2})}{\Delta(N+1)}-\sum_{n=0}^{n=N_s}h_s (n h_s)^N \text{rect}_\Delta \lb H_{\mathrm{sys}} \lb \{x_n\}, \{p_n\}\rb + \frac{p_s^2}{2Q} + g\beta^{-1}\ln{n h_s} - E_{\text{ext}}\rb \nonumber\\ 
    &\qquad\le \frac{h_s}{\Delta}s_0^N (N+2)e^{g^{-1}\beta N\frac{\Delta}{2}}.
\end{align}

This method consistently underestimates the integral due to throwing away portions of the support and using the left Riemann sum on a monotonically increasing function. We can then split apart the dependencies in $s_0$ since they all appear in the exponent. Let us define $\langle \bx \rangle_2$ as follows:
\begin{align}
    \langle \bx \rangle_2&:=C\, \sum_{\{n_x\}} \, \sum_{\{\bm{n_x}\}} \{h_x\} (h_x \bm{n_x})\sum_{\{\bm{n_p}\}} \{h_p\} \sum_{n_{p_s}} h_{p_s}\int \,ds\,s^N\,\text{rect}_\Delta \lb H_N\lb \{x_n\}, \{p_n\}, s, p_s \rb-E_e\rb \nonumber\\
    &= C\frac{2\,\text{sinh}(g^{-1}\beta(N+1)\frac{\Delta}{2})}{\Delta(N+1)}\, \sum_{\{\bm{n_x}\}} \{h_x\} (h_x \bm{n_x}) \sum_{\{\bm{n_p}\}} \{h_p\} \sum_{n_{p_s}} h_{p_s} s_0^{N+1}\nonumber\\
    &=C\frac{2\,\text{sinh}(g^{-1}\beta(N+1)\frac{\Delta}{2})}{\Delta(N+1)}\, \sum_{\{\bm{n_x}\}} \{h_x\} (h_x \bm{n_x}) \sum_{\{\bm{n_p}\}} \{h_p\} \sum_{n_{p_s}} h_{p_s} e^{ -g^{-1}\beta(N+1)\lb H(\bx,\bp)-E_{ext}+\frac{p_s^2}{2Q}\rb}.
\end{align}
Then we can break down the error owing to the fact that $\bx, \bp$ and $p_s$ dependencies only appear in the exponent in the integrand and thus can be distributed:
\begin{align}
    \langle \bx \rangle_2 &=C\frac{2\,\text{sinh}(g^{-1}\beta(N+1)\frac{\Delta}{2})}{\Delta(N+1)}\,e^{g^{-1}\beta (N+1)E_e}\lb \sum_{\{\bm{n_x}\}} \{h_x\} (h_x \bm{n_x})e^{ -g^{-1}\beta(N+1)f(x)}\rb\nonumber\\
    &\qquad\lb \sum_{n_p} h_pe^{ -g^{-1}\beta(N+1)\frac{p^2}{2m}}\rb^N\lb\sum_{n_{p_s}} h_{p_s} e^{ -g^{-1}\beta(N+1)\frac{p_s^2}{2Q}}\rb\\
    &=C\,\frac{2\,\text{sinh}(g^{-1}\beta(N+1)\frac{\Delta}{2})}{\Delta(N+1)}\,e^{g^{-1}\beta (N+1)E_e} \lb \int\,\{dx\}\,\bx\,e^{-g^{-1}(N+1)\beta f(x)}+\bep\rb \nonumber\\
    &\qquad\lb \int\,dp e^{-g^{-1}(N+1)\beta\frac{p^2}{2m}}+\varepsilon_p\rb^N \lb \int\,dp_s e^{-g^{-1}(N+1)\beta\frac{p_s^2}{2Q}}+\varepsilon_{p_s}\rb\\
    &= \langle \bx \rangle_1 +\bep \,C\frac{2\,\text{sinh}(g^{-1}\beta(N+1)\frac{\Delta}{2})}{\Delta(N+1)}\,e^{g^{-1}\beta (N+1)E_e}\lb \int\,dp e^{-g^{-1}(N+1)\beta\frac{p^2}{2m}}\rb^N \lb \int\,dp_s e^{-g^{-1}(N+1)\beta\frac{p_s^2}{2Q}}\rb\nonumber\\
    &+ \varepsilon_p \,NC\frac{2\,\text{sinh}(g^{-1}\beta(N+1)\frac{\Delta}{2})}{\Delta(N+1)}\,e^{g^{-1}\beta (N+1)E_e}\lb \int\,dp e^{-g^{-1}(N+1)\beta\frac{p^2}{2m}}\rb^{N-1}\nonumber\\
    &\qquad\qquad\qquad\qquad\qquad\qquad\qquad\qquad\lb \int\,\{dx\}\,\bx\,e^{-g^{-1}(N+1)\beta f(x)}\rb \lb \int\,dp_s e^{-g^{-1}(N+1)\beta\frac{p_s^2}{2Q}}\rb\nonumber\\
    &+ \varepsilon_{p_s} \,C\frac{2\,\text{sinh}(g^{-1}\beta(N+1)\frac{\Delta}{2})}{\Delta(N+1)}\,e^{g^{-1}\beta (N+1)E_e}\lb \int\,\{dx\}\,\bx\,e^{-g^{-1}(N+1)\beta f(x)}\rb \lb \int\,dp e^{-g^{-1}(N+1)\beta\frac{p^2}{2m}}\rb^N\nonumber\\
&+O(\varepsilon_p^2+\norm{\bep}\varepsilon_p+\norm{\bep}\varepsilon_{p_s}+\varepsilon_p\varepsilon_{p_s}).
\end{align}
Here we defined
\begin{align}
    \bep &:= \sum_{\{\bm{n_x}\}} \{h_x\} (h_x \bm{n_x}) e^{-\beta g^{-1}(N+1)f(h_x \bm{n_x})} - \int\,\{dx\}\,\bx\,e^{-g^{-1}(N+1)\beta f(x)}, \\
    \varepsilon_p &:= \sum_{\{\bm{n_p}\}} \{h_p\} e^{-\beta g^{-1}(N+1)\frac{n_p^2h_p^2}{2m}} - \int\,dp e^{-\beta g^{-1}(N+1)\frac{p^2}{2m}},\\
    \varepsilon_{p_s} &:=\sum_{n_{p_s}} h_{p_s} e^{-\beta g^{-1}(N+1)\frac{n_{p_s}^2h_{p_s}^2}{2Q}} - \int\,dp_s e^{-\beta g^{-1}(N+1)\frac{p_s^2}{2Q}},
\end{align}
with $\sum_{\{\bm{n_x}\}} \{h_x\} (h_x \bm{n_x})$ referring to the $N$-dimensional sum $\sum_{n_1}h_x\sum_{n_2}h_x\cdots\sum_{n_N}h_x (h_x n_1 \bm{\hat{e}_1}+h_x n_2 \bm{\hat{e}_2}+\cdots+h_x n_N \bm{\hat{e}_N})\,$ in shorthand. We can solve the integrals and substitute in $C$:
\begin{align}
    \langle \bx \rangle_2 &=  \langle \bx \rangle_1 + \frac{\bep \text{sinhc}\lb g^{-1}\beta\Delta(N+1)/2\rb}{\int\,\{dx\}\,e^{-g^{-1}(N+1)\beta f(x)}}+\frac{N\varepsilon_p\langle \bx \rangle_1}{\int\,dp e^{-g^{-1}(N+1)\beta\frac{p^2}{2m}}}+\frac{\varepsilon_{p_s}\langle \bx \rangle_1}{\int\,dp_s e^{-g^{-1}(N+1)\beta\frac{p_s^2}{2Q}}}\\
    &= \langle \bx \rangle_1 + \frac{\bep\text{sinhc}\lb g^{-1}\beta\Delta(N+1)/2\rb}{\int\,\{dx\}\,e^{-g^{-1}(N+1)\beta f(x)}}+N\varepsilon_p\langle \bx \rangle_1\sqrt{\frac{g^{-1}(N+1)\beta}{2\pi m}}+\varepsilon_{p_s}\langle \bx \rangle_1\sqrt{\frac{g^{-1}(N+1)\beta}{2\pi Q}}\\
    &\qquad+O(\varepsilon_p^2 +\norm{\bep}\varepsilon_p +\norm{\bep}\varepsilon_{p_s} +\varepsilon_p\varepsilon_{p_s}),
\end{align}
which lets us upper bound the following quantity:
\begin{align}
    \norm{\langle \bx \rangle_2 -\langle \bx \rangle_1} &\leq \frac{\norm{\bep}\text{sinhc}\lb g^{-1}\beta\Delta(N+1)/2\rb}{\int\,\{dx\}\,e^{-g^{-1}(N+1)\beta f(x)}}+\norm{\langle \bx \rangle_1}N\varepsilon_p\sqrt{\frac{g^{-1}(N+1)\beta}{2\pi m}}+\norm{\langle \bx \rangle_1}\varepsilon_{p_s}\sqrt{\frac{g^{-1}(N+1)\beta}{2\pi Q}}\nonumber\\
    &\qquad +O(\varepsilon_p^2 +\norm{\bep}\varepsilon_p +\norm{\bep}\varepsilon_{p_s} +\varepsilon_p\varepsilon_{p_s}).
\end{align}
Lastly, if we define $\langle \bx \rangle_3$ to be

\begin{equation}
\begin{split}
     \langle \bx \rangle_3 &:= C\,\sum_{\{\bm{n_x}\}} \{h_x\} (h_x \bm{n_x}) \sum_{\{\bm{n_p}\}} \{h_p\} \sum_{n_{p_s}} h_{p_s} \sum_n h_s (n h_s)^N \times \\
     &\qquad  \times \text{rect}_\Delta \lb H_{\mathrm{sys}} \lb \{x_n\}, \{p_n\}\rb + \frac{p_s^2}{2Q} + g\beta^{-1}\ln{n h_s} - E_{\text{ext}}\rb,
\end{split}
\end{equation}
we can complete the error approximation, since this expression is solely comprised of sums and is therefore amenable for numerical evaluation. Substituting in the expression for the $s$ sum, we get:
\begin{align}
    \langle \bx \rangle_3 = C\, \sum_{\{\bm{n_x}\}} \{h_x\} (h_x \bm{n_x})\sum_{\{n_p\}} \{h_p\} \sum_{n_{p_s}} h_{p_s}\lb s_0^{N+1}\frac{2\,\text{sinh}(g^{-1}\beta(N+1)\frac{\Delta}{2})}{\Delta(N+1)} + \varepsilon_s\rb,
\end{align}
which then implies
\begin{equation}
\begin{split}
    &\|  \langle \bx \rangle_3 - C\frac{2\,\text{sinh}(g^{-1}\beta(N+1)\frac{\Delta}{2})}{\Delta(N+1)}\, \sum_{\{\bm{n_x}\}} \{h_x\} (h_x \bm{n_x}) \sum_{\{n_p\}} \{h_p\} \sum_{n_{p_s}} h_{p_s} s_0^{N+1}\| \\
    &= \|C\, \sum_{\{\bm{n_x}\}} \{h_x\} (h_x \bm{n_x}) \sum_{\{n_p\}} \{h_p\} \sum_{n_{p_s}} h_{p_s} \varepsilon_s\|.
\end{split}
\end{equation}
Observe that the second term in the norm in the first line is just $\langle \bx \rangle_2$. Using this and substituting in the upper bound for $\varepsilon_s$ leads us to:
\begin{align}
    \|  \langle \bx \rangle_3 - \langle \bx \rangle_2\|&\leq \left\|C\frac{h_s}{\Delta}(N+2)e^{g^{-1}\beta N\frac{\Delta}{2}}\, \sum_{\{\bm{n_x}\}} \{h_x\} (h_x \bm{n_x}) \sum_{\{n_p\}} \{h_p\} \sum_{n_{p_s}} h_{p_s} s_0^N \right\| \\
    &= C\frac{h_s}{\Delta}(N+2)e^{g^{-1}\beta N\frac{\Delta}{2}}e^{g^{-1}N\beta E_e}\,\left\| \int\,\{dx\}\,\bx\,e^{-g^{-1}N\beta f(x)}+\varepsilon_x\right\| \nonumber\\
    &\qquad\lb \int\,dp e^{-g^{-1}(N+1)\beta\frac{p^2}{2m}}+\varepsilon_p\rb^N \lb \int\,dp_s e^{-g^{-1}N\beta\frac{p_s^2}{2Q}}+\varepsilon_{p_s}\rb  \\
    &\leq C\frac{h_s}{\Delta}(N+2)e^{g^{-1}\beta N\frac{\Delta}{2}}e^{g^{-1}N\beta E_e}\,\left\| \int\,\{dx\}\,\bx\,e^{-g^{-1}N\beta f(x)}\right\| \times \\
    &\qquad \times \lb \int\,dp e^{-g^{-1}N\beta\frac{p^2}{2m}}\rb^N \lb \int\,dp_s e^{-g^{-1}N\beta\frac{p_s^2}{2Q}}\rb\nonumber\\
    &\qquad +O(h_s(\varepsilon_x+\varepsilon_p+\varepsilon_{p_s}))\\
    &\leq C\frac{h_s}{\Delta}(N+2)e^{g^{-1}\beta N\frac{\Delta}{2}}e^{g^{-1}N\beta E_e}\,\left\| \int\,\{dx\}\,\bx\,e^{-g^{-1}N\beta f(x)}\right\|\lb \frac{2\pi m}{g^{-1}N\beta}\rb^\frac{N}{2} \sqrt{\frac{2\pi Q}{g^{-1}\beta N}}\nonumber\\
    &\qquad +O(h_s(\varepsilon_x+\varepsilon_p+\varepsilon_{p_s}))\\
    &\leq h_s\frac{e^{g^{-1}\beta N\frac{\Delta}{2}}}{\Delta g^{-1}\beta}(N+2)e^{-g^{-1}\beta E_e}\frac{\left\| \int\,\{dx\}\,\bx\,e^{-g^{-1}N\beta f(x)}\right\|}{\int\,\{dx\}\,e^{-g^{-1}(N+1)\beta f(x)}}\lb \frac{N+1}{N}\rb^\frac{N+1}{2} \nonumber\\
    &\qquad +O(h_s(\varepsilon_x+\varepsilon_p+\varepsilon_{p_s})).
\end{align}

The error in discretizing the quadratic integrals for the variables $p$ and $p_s$ may be quantified easily. For an integral such as $\int_{-\infty}^\infty dp\, e^{-\alpha p^2}$, if we truncate the integral at $P_{\max}$, we incur an error upper bounded by $E_{\text{trunc}}:=\frac{1}{\alpha P_{\max} e^{\alpha P_{\max}^2}}$ (obtained by a simple upper bound on the erfc function~\cite{erfc_uoft}), and discretizing the finite integral $\int_{-P_{\max}}^{P_{\max}} dp\, e^{-\alpha p^2}$ using the midpoint rule, the error in the quadrature is a function of the second derivative of the integrand.  Evaluating the maximum value of the second derivative of the Gaussian and substituting the result in to standard error bounds for the midpoint rule yields an error of  
\begin{equation}
    E_{\text{mp}}:=\frac{h^2 \alpha P_{\max}}{6},
\end{equation}
with $h$ being the grid spacing of the discretization. We will set $P_{\max} = \sqrt{\frac{1}{\alpha}\log\lb h^2\alpha/6\rb}$ to upper bound $E_{\text{trunc}}$ with $E_{\text{mp}}$. This results in $|\varepsilon|=|E_{\text{trunc}}+E_{\text{mp}}|\leq\frac{h^2 \alpha^{\frac{1}{2}}}{3}  \sqrt{\frac{1}{\alpha}\log\lb h^2\alpha/6\rb}$. Plugging in values of $\alpha$ and $h$ for the $p$ and $p_s$ integrals, we get an upper bound on the magnitudes of $\varepsilon_p$ and $\varepsilon_{p_s}$:
\begin{align}
    |\varepsilon_p| &\leq \frac{1}{3\sqrt{2}}h_p^2 \sqrt{\frac{g^{-1}(N+1)\beta}{m}}\sqrt{\log \lb \frac{h_p^2g^{-1}(N+1)\beta}{12m}\rb}\\
    |\varepsilon_{p_s}| &\leq \frac{1}{3\sqrt{2}}h_{p_s}^2 \sqrt{\frac{g^{-1}(N+1)\beta}{Q}}\sqrt{\log \lb \frac{h_{p_s}^2g^{-1}(N+1)\beta}{12Q}\rb}.
\end{align}
Using these bounds and the triangle inequality, the distance between $\langle \bx \rangle_3$ and $\langle \bx \rangle_1$ is:
\begin{align}
     \|\langle \bx \rangle_3-\langle \bx \rangle_1\| &\in O\biggl( \frac{\norm{\bep}\text{sinhc}\lb g^{-1}\beta\Delta(N+1)/2\rb}{\int\,\{dx\}\,e^{-g^{-1}(N+1)\beta f(x)}} \nonumber\\
    &+ \norm{\langle \bx \rangle_1} h_p^2 \frac{g^{-1}\beta(N+1)}{m}\nonumber \sqrt{\log \lb \frac{h_p^2g^{-1}(N+1)\beta}{12m}\rb}\\
    &+ \norm{\langle \bx \rangle_1} h_{p_s}^2 \frac{g^{-1}\beta(N+1)}{Q}\sqrt{\log \lb \frac{h_{p_s}^2g^{-1}(N+1)\beta}{12Q}\rb}\nonumber\\
    &+ h_s\frac{e^{g^{-1}\beta N\frac{\Delta}{2}}}{\Delta g^{-1}\beta}(N+2)e^{-g^{-1}\beta E_e}\frac{\left\| \int\,\{dx\}\,\bx\,e^{-g^{-1}N\beta f(x)}\right\|}{\int\,\{dx\}\,e^{-g^{-1}(N+1)\beta f(x)}}\lb \frac{N+1}{N}\rb^\frac{N+1}{2}\biggr).
\end{align}
Setting $g=N+1$ and $\langle \bx \rangle_1= \langle \bx \rangle_0 \,\text{sinhc}\lb\frac{g^{-1}\beta\Delta(N+1)}{2}\rb$, we obtain a complete upper bound for the discretization error:
\begin{align}
     \|\langle \bx \rangle_3-\langle \bx \rangle_1\| &\in O\biggl( \frac{\norm{\bep}\text{sinhc}\lb \beta\Delta/2\rb}{\int\,\{dx\}\,e^{-\beta f(x)}} \nonumber\\
    &+ \norm{\langle \bx \rangle_1} h_p^2 N \frac{\beta}{m} \sqrt{\log \lb \frac{h_p^2\beta}{m}\rb}\nonumber \\
    &+ \norm{\langle \bx \rangle_1} h_{p_s}^2 \frac{\beta}{Q}\sqrt{\log \lb \frac{h_{p_s}^2\beta}{Q}\rb}\nonumber\\
    &+ h_s\frac{e^{\beta \frac{N}{N+1}\frac{\Delta}{2}}}{\Delta \beta}(N+2)(N+1)e^{-\frac{\beta E_e}{N+1}}\frac{\left\| \int\,\{dx\}\,\bx\,e^{-\frac{N}{N+1}\beta f(x)}\right\|}{\int\,\{dx\}\,e^{-\beta f(x)}}\lb \frac{N+1}{N}\rb^\frac{N+1}{2}\biggr).
\end{align}
Simplifying and incorporating the upper bound on $\|\langle \bx \rangle_1-\langle \bx \rangle_0\|$ from~\eqref{eq:global_error1} by using the triangle inequality as $\norm{\langle\bx \rangle_3-\langle \bx \rangle_0}=\norm{\langle\bx \rangle_3-\langle\bx \rangle_1+\langle\bx \rangle_1-\langle \bx \rangle_0}\leq\norm{\langle\bx \rangle_3-\langle\bx \rangle_1}+\norm{\langle\bx \rangle_1-\langle \bx \rangle_0}$ one last time, we arrive at: 
\begin{align}
     \|\langle \bx \rangle_3-\langle \bx \rangle_0\| &\in O\biggl(\norm{\langle \bx \rangle_0}\beta^2 \Delta^2+ \frac{\norm{\bep}}{\int\,\{dx\}\,e^{-\beta f(x)}} \nonumber\\
    &+ \norm{\langle \bx \rangle_0} h_p^2 N \frac{\beta}{m} \sqrt{\log \lb \frac{h_p^2\beta}{m}\rb}\nonumber \\
    &+ \norm{\langle \bx \rangle_0} h_{p_s}^2 \frac{\beta}{Q}\sqrt{\log \lb \frac{h_{p_s}^2\beta}{Q}\rb}\nonumber\\
    &+ h_s\frac{e^{\beta \frac{\Delta}{2}}}{\Delta \beta}N^2 \frac{\left\| \int\,\{dx\}\,\bx\,e^{-\beta f(x)}\right\|}{\int\,\{dx\}\,e^{-\beta f(x)}}\biggr).
\end{align}

If we want each term to be $O(\epsilon)$, the asymptotic behaviors for the grid spacings are given by
\begin{align}
    \Delta &\in O\lb\beta^{-1} \sqrt{\frac{\epsilon}{\norm{\ev{\bx}_0}}}\rb,\\
    h_s &\in O\lb\frac{\epsilon}{N^2\norm{\ev{\bx}_0}^{\frac{3}{2}} e^{\epsilon/\norm{\ev{\bx}_0}}}\rb,\\
    h_p &\in \tilde{O}\lb\frac{1}{N}\sqrt{\frac{m}{\beta}}\frac{\sqrt{\epsilon}}{\sqrt{\norm{\ev{\bx}_0}}}\rb,\\
    h_{p_s} &\in \tilde{O}\lb\sqrt{\frac{Q}{\beta}}\frac{\sqrt{\epsilon}}{\sqrt{\norm{\ev{\bx}_0}}}\rb,\\
    \label{eq:ch7_bepbigo}\norm{\bep} 
    &\in O\lb\epsilon\left\| \int\,\{dx\}\,\bx\,e^{-\beta f(x)}\right\|\rb.
\end{align}

Until this point we have avoided specifying what vector norm was being used and only invoked the triangle inequality which is a common property of vector norms. For a specific norm, such as the Euclidean norm, we can refine the constraint on $\norm{\bep}$ into a constraint on the grid spacing for the discretization of $\int \{dx\}$. Without loss of generality, we are going to assume that the interval which we're looking at for each coordinate $x_i$ is $[0,x_{\max}]$, and we're going to define $V:=x_{\max}^N$ as the volume of the space we're considering for the optimization problem. Note that
\begin{align}
    \bep&=\sum_{\{\bm{n_x}\}} \{h_x\} (h_x \bm{n_x}) e^{-\beta f(h_x \bm{n_x})} - \int\,\{dx\}\,\bx\,e^{-\beta f(x)} \\
    &=\sum_{i=1}^N \hat{\bm{e}}_i \lb\sum_{\{\bm{n_x}\}} h_x^2 n_i e^{-\beta f(h_x \bm{n_x})} - \int\,\{dx\}\,x_i\,e^{-\beta f(x)}\rb \\
    &:=\sum_{i=1}^N \hat{\bm{e}}_i \varepsilon_{i}.
\end{align}
For a certain component $\varepsilon_i$, let us define $g_i(x_1,x_2\cdots,x_N):=x_ie^{-\beta f(x_1,x_2\cdots,x_N)}$ such that
\begin{align}
    \varepsilon_i &= \sum_{\{\bm{n_x}\}} h_x^2 n_i e^{-\beta f(h_x \bm{n_x})} - \int\,\{dx\}\,x_i\,e^{-\beta f(x)}\\
    &= \sum_{n_1=0}^{x_{\max}/h_x-1}\cdots  \sum_{n_N=0}^{x_{\max}/h_x-1} \int_{h_xn_1}^{h_x(n_1+1)}dx_1\cdots\int_{h_xn_N}^{h_x(n_N+1)}dx_N\lb g_i(n_1h_x,n_2h_x\cdots, n_Nh_x)-g_i(x_1,x_2\cdots,x_N)\rb.
\end{align}
Taking the absolute value of both sides:
\begin{align}\label{eq:ch7_epsiloni}
    |\varepsilon_i| &\leq \sum_{n_1=0}^{x_{\max}/h_x-1}\cdots  \sum_{n_N=0}^{x_{\max}/h_x-1} \int_{h_xn_1}^{h_x(n_1+1)}dx_1\cdots\int_{h_xn_N}^{h_x(n_N+1)}dx_N \left|g_i(x_1,x_2\cdots,x_N)-g_i(n_1h_x,n_2h_x\cdots, n_Nh_x)\right|.
\end{align}

Next, we invoke the multi-variable mean value theorem for $f:\mathbb{R}^N\rightarrow\mathbb{R}$ as stated in~\cite{mvt_multi}. Let $L$ be a line segment with endpoints $\bm{a},\bm{b} \in \mathbb{R}^N$. Then there exists a point $\bm{c} \in [\bm{a}, \bm{b}]$ such that:
\begin{equation}
    f(\bm{b})-f(\bm{a}) = \nabla f(\bm{c})\cdot (\bm{b}-\bm{a}).
\end{equation}
Invoking this in \eqref{eq:ch7_epsiloni} yields
\begin{align}
    |\varepsilon_i| &\leq \sum_{n_1=0}^{x_{\max}/h_x-1}\cdots  \sum_{n_N=0}^{x_{\max}/h_x-1} \int_{h_xn_1}^{h_x(n_1+1)}dx_1\cdots\int_{h_xn_N}^{h_x(n_N+1)}dx_N \left|\sum_{j=1}^N\frac{\partial g_i}{\partial x_j}\Biggr\rvert_c(x_j-n_jh_x)\right|\\
    &= \sum_{n_1=0}^{x_{\max}/h_x-1}\cdots  \sum_{n_N=0}^{x_{\max}/h_x-1} \sum_{j=1}^N\left|\partial_jg_i\right|_{\max}\int_{h_xn_1}^{h_x(n_1+1)}dx_1\cdots\int_{h_xn_N}^{h_x(n_N+1)}dx_N (x_j-n_jh_x)\\
    &= \sum_{n_1=0}^{x_{\max}/h_x-1}\cdots  \sum_{n_N=0}^{x_{\max}/h_x-1} \sum_{j=1}^N\left|\partial_jg_i\right|_{\max}\frac{h_x^{N+1}}{2}\\
    &= \frac{x_{\max}^N}{h_x^N}\frac{h_x^{N+1}}{2}\sum_{j=1}^N\left|\partial_jg_i\right|_{\max}\\
    &=\frac{Vh_x}{2}\sum_{j=1}^N\left|\partial_jg_i\right|_{\max}.
\end{align}
The $c$ in the first line is understood to be inside the volume unit defined by the limits of the integrals. We then upper bound that value using the maximum value of the partial derivative in the entire optimization volume. The rest follows trivially. The partial derivatives of $g_i$ are ($j\neq i$ for the second case):
\begin{align}
    \partial_ig_i(\bx) &= (1-\beta x_i \partial_if(\bx))e^{-\beta f(\bx)}\\
    \partial_jg_i(\bx) &= -\beta x_i \partial_jf(\bx)e^{-\beta f(\bx)}.
\end{align}
Using these, we can upper bound $\sum_{j=1}^N\left|\partial_jg_i\right|_{\max}$:
\begin{align}
    \sum_{j=1}^N\left|\partial_jg_i\right|_{\max} &\leq x_{\max}\beta e^{-\beta \min f(\bx)} \sum_{j\neq i} |\partial_jf(\bx)|_{\max} + (1+\beta x_{\max}|\partial_if(\bx)|_{\max})e^{-\beta \min f(\bx)}\\
    &\leq e^{-\beta \min f(\bx)} (x_{\max}\beta N|\partial f(\bx)|_{\max}+1).
\end{align}
In the second line we upper bounded each of the maximum partial derivatives with the maximum partial derivative among all coordinates to make the argument clean. Since we are treating a general function $f(\bx)$ without specific knowledge of the structure, we are leaving a lot on the table as far as the headrooms in these inequalities go. With constraints on the function $f(\bx)$, the upper bounds can be refined greatly. The above equation lets us bound $|\varepsilon_i|$:
\begin{align}
    |\varepsilon_i| \leq \frac{Vh_x}{2}e^{-\beta \min f(\bx)} (x_{\max}\beta N|\partial f(\bx)|_{\max}+1).
\end{align}
Since the upper bound does not depend on $i$, we can extend the bound to $\norm{\bep}_2$ trivially:
\begin{align}
    \norm{\bep}_2 &\leq \frac{Vh_x\sqrt{N}}{2}e^{-\beta \min f(\bx)} (x_{\max}\beta N|\partial f(\bx)|_{\max}+1).
\end{align}
Or, asymptotically:
\begin{align}
    \norm{\bep}_2 &\in O\lb h_x\beta NV N^{\frac{3}{2}}x_{\max}|\partial f(\bx)|_{\max}\rb.
\end{align}

Recalling \eqref{eq:ch7_bepbigo}, if we want to recover the scaling $O\lb\epsilon\left\| \int\,\{dx\}\,\bx\,e^{-\beta f(x)}\right\|_2\rb$, which itself is in \\
$O\lb \epsilon V x_{\max}\sqrt{N}e^{-\beta \min f(x)}\rb$, it suffices for $h_x$ to have the following asymptotic scaling:
\begin{equation}
    h_x \in O\lb \frac{1}{\beta N |\partial f(\bx)|_{\max}}\rb.
\end{equation}
\end{proof}

\end{document}